\documentclass[10pt,reqno]{amsart}

\usepackage[margin=1in]{geometry}
\usepackage{amsmath}
\usepackage{amsthm, amsfonts,amssymb,euscript}
\usepackage{mathtools}
\usepackage{enumitem}
\usepackage{bm}
\usepackage{amssymb}
\usepackage{graphicx}
\usepackage{caption}
\usepackage{subcaption}
\usepackage{amsfonts}
\usepackage{float}
\usepackage{color}
\usepackage{slashed}
\usepackage[affil-it]{authblk}
\usepackage{mathrsfs}
\usepackage{upgreek}
\usepackage{pifont}
\usepackage{bbm}
\usepackage[colorlinks=true]{hyperref}

\usepackage[usenames,dvipsnames,svgnames,table]{xcolor}
\hypersetup{linkcolor=BrickRed, urlcolor=green, citecolor=blue, linktoc=page}

\numberwithin{equation}{section}

\newtheorem*{proposition*}{Proposition}
\newtheorem*{theorem*}{Theorem}
\newtheorem*{conjecture*}{Conjecture}
\newtheorem*{claim*}{Claim}
\newtheorem*{lemma*}{Lemma}
\newtheorem*{corollary*}{Corollary}
\newtheorem{theorem}{Theorem}[section]
\newtheorem{proposition}[theorem]{Proposition}

\newtheorem{lemma}[theorem]{Lemma}
\newtheorem{corollary}[theorem]{Corollary}

\newtheorem*{definition*}{Definition}
\newtheorem{definition}{Definition}[section]
\newtheorem*{assumption*}{Assumption}

\newtheorem*{remark*}{Remark}
\newtheorem{remark}{Remark}[section]
\newtheorem*{condition*}{Condition}
\newtheorem{condition}{Condition}

\newcommand{\la}{\langle}
\newcommand{\ra}{\rangle}
\newcommand{\R}{\mathbb{R}}
\newcommand{\s}{\mathbb{S}}
\newcommand{\C}{\mathbb{C}}

\newcommand{\Z}{\mathbb{Z}}
\newcommand{\N}{\mathbb{N}}

\DeclareMathOperator{\supp}{\textnormal{supp}}
\newcommand{\snabla}{\slashed{\nabla}}

\newcommand{\Lbar}{\underline{L}}

\newcommand{\tomega}{\widetilde{\omega}}
\newcommand{\homega}{\widehat{\omega}}

\newcommand{\h}{\mathbbm{h}}
\newcommand{\re}{\textnormal{Re}\,}
\newcommand{\im}{\textnormal{Im}\,}
\newcommand{\tV}{\widetilde{V}}
\newcommand{\sign}{\,\textnormal{sign}\,}
\newcommand{\q}{q}

\usepackage[colorinlistoftodos,prependcaption,textsize=tiny]{todonotes}

\DeclareFontFamily{U}{mathx}{}
\DeclareFontShape{U}{mathx}{m}{n}{<-> mathx10}{}
\DeclareSymbolFont{mathx}{U}{mathx}{m}{n}
\DeclareMathAccent{\widecheck}{0}{mathx}{"71}

\allowdisplaybreaks

\usepackage[foot]{amsaddr}

\begin{document}

\author{Dejan Gajic$^*$$^1$}
\address{$^{1}$\small Institut f\"ur Theoretische Physik, Universit\"at Leipzig, Br\"uderstrasse 16, 04103 Leipzig, Germany }
\email{$^*$dejan.gajic@uni-leipzig.de}
\title[Charged scalar fields on Reissner--Nordstr\"om I: integrated energy estimates]{Charged scalar fields on Reissner--Nordstr\"om spacetimes I: integrated energy estimates}

\maketitle

\begin{abstract} This is the first part of a series of papers deriving the precise, late-time behaviour and (in)stability properties of charged scalar fields on near-extremal Reissner--Nordstr\"om spacetimes via energy estimates. In this paper, we establish global, weighted integrated energy decay and energy boundedness estimates for solutions to the charged scalar field equation on (near-)extremal Reissner--Nordstr\"om(--de Sitter) spacetimes. These estimates extend to Reissner--Nordstr\"om spacetimes away from extremality under the assumption of mode stability on the real axis. 

Together with the companion paper \cite{gaj26b}, this paper forms the first global  quantitative analysis of the charged scalar field equation on asymptotically flat black hole spacetimes, without a smallness assumption on the scalar field charge. Due to a coupling of the degeneration of the red-shift effect with the presence of superradiance at the linearized level, charged scalar fields on Reissner--Nordstr\"om spacetimes also probe some of the main difficulties encountered when studying the (neutral) wave equation on extremal Kerr spacetimes.\end{abstract}

\tableofcontents
\section{Introduction}
Integrated energy estimates are a core ingredient for deriving global decay results for nonlinear, geometric wave equations. Schematically, they take the following form:
\begin{equation}
\label{eq:introschemied}
	\int_0^{\infty}\int_{\Sigma_{\tau}}w_{1} |\partial \phi|^2+w_2 |\phi|^2\,d\mu_{\Sigma_{\tau}} d\tau\leq  \int_{\Sigma_{0}}\tilde{w}_{1}|\partial \phi|^2+\tilde{w}_2|\phi|^2\,d\mu_{\Sigma_0}+\int_0^{\infty}\int_{\Sigma_{\tau}}(\hat{w}_{1} |\partial \phi|+\hat{w}_2|\phi|)\cdot |\mathcal{L} \phi|\,\,d\mu_{\Sigma_{\tau}} d\tau,
\end{equation}
with $\mathcal{L}$ a linearized wave operator, $\{\Sigma_{\tau}\}_{\tau\in[0,\infty)}$ a foliation by spacelike hypersurfaces, $w_i,\tilde{w}_i,\hat{w}_i:\Sigma_{\tau}\to \R$ weight functions with possible degeneracies.

Estimates of the above type with $\mathcal{L}=\square$, the standard wave operator, originate in the work of Morawetz \cite{mor2}, where they were used to study nonlinear wave equations on the Minkowski spacetime. 

In the context of asymptotically flat black hole spacetimes, the nature and validity of \eqref{eq:introschemied} is intimately related to the following obstructions to decay that each have an analogue in the dynamics of null geodesics:
\begin{enumerate}[label=(\Alph*)]
	\item degeneration of energy control near trapped null geodesics,
	\item superradiance,
	\item degeneration of red-shift near the event horizon,
	\item degeneration of red-shift near null infinity.\footnote{This may be interpreted as the degeneration of the red-shift effect along the future cosmological horizon $\mathcal{C}^+$ as the cosmological constant $\Lambda$ goes to $0$, with future null infinity $\mathcal{I}^+$ playing the role of $\mathcal{C}^+$ in the limit. The existence of hierarchies of $r^p$-weighted energy estimates in the far-away region, see \cite{newmethod}, may be thought of as a ``degenerate red-shift effect''. }
\end{enumerate}
The above mechanisms all act as obstructions for proving decay, and they may be (strongly) coupled. In particular, if there exist trapped null geodesics that remain trapped under generic perturbations in phase space (stably trapped null geodesics), see \cite{molog, bur98}, or if there exists superradiance without an event horizon, see \cite{mo18}, then an estimate of the form \eqref{eq:introschemied} fails to hold.

Extremal black holes feature a coupling of (a subset) of (A)--(D), which is delicately balanced so that the validity of \eqref{eq:introschemied} cannot immediately be ruled out, in contrast with the examples in the above paragraph, but it is significantly harder to establish compared to the sub-extremal black hole setting, where only (A), (B) and (D) are present and they are moreover decoupled. This is closely related to the critical nature of extremal black holes as members of families of sub-extremal stationary black holes. 

Not all extremal black holes feature a coupling between all the phenomena (A)--(D). For example, in the case of the wave equation on extremal Kerr black holes (A)--(C) are coupled, but (D) is decoupled. 

In the case of the (neutral) wave equation on extremal Reissner--Nordstr\"om, (B) is not present and (A), (C) and (D) are all decoupled.
 
In the present paper, we will investigate a setting which features all the difficulties (A)--(D) and in which (B)--(D) are moreover coupled.\footnote{While trapping of null geodesics is present and more complicated than in the extremal Reissner--Nordstr\"om case, it can be decoupled from the remaining difficulties.}

We consider the conformally covariant, inhomogeneous charged scalar field equation on Reissner--Nordstr\"om spacetime backgrounds:
\begin{align}
\label{eq:CSFintro}
	(g^{-1}_{M,Q})^{\mu \nu}\:(^AD)_{\mu}\:(^AD)_{\nu}\phi=G_A,
\end{align}
where $g_{M,Q}$ is the Reissner--Nordstr\"om metric and $^AD=\nabla-i\mathfrak{q}A\otimes(\cdot)$ is the electromagnetic gauge derivative, with $A$ a 1-form and $dA=F_Q=-Q^2r^{-2}dt\wedge dr$ the Reissner--Nordstr\"om Faraday tensor. 

While our main interest are the \emph{extremal} Reissner--Nordstr\"om ($|Q|=M$) and \emph{near}-extremal Reissner--Nordstr\"om spacetimes ($0<1-\frac{|Q|}{M}\ll 1$), we also establish conditional results for the remaining members of the Reissner--Nordstr\"om family. Furthermore, nearby members of the sub-extremal Reissner--Nordstr\"om--\emph{de Sitter} family also play a role in the analysis. 

The present paper is the first part of a series of papers concerning charged scalar fields on Reissner--Nordstr\"om black holes. A precise late-time analysis of the decay behaviour of solutions to \eqref{eq:CSFintro} is carried out in the companion paper \cite{gaj26b}, using the results of the present paper as a starting point.

The main motivation for considering \eqref{eq:CSFintro} is two-fold:
\begin{itemize}
\item 	With the restriction of spherical symmetry, equation \eqref{eq:CSFintro} is the linearization of the spherically symmetric Einstein--Maxwell--charged scalar field system of equations. A robust analysis of its linearization is the first step towards an analysis of 1) the nonlinear coupled Maxwell--charged scalar field equations, and then 2) the full Einstein--Maxwell--charged scalar field system. Already in spherical symmetry, we expect that the tools developed in the present paper and the companion paper \cite{gaj26b} will play an important role when studying the Einstein--Maxwell--charged scalar field system, in particular near extremality.
\item The difficulties encountered in \eqref{eq:CSFintro} resemble closely the difficulties encountered in the study of the wave equation on extremal Kerr. In future work, the methods introduced in the present paper will be used to sharpen the instability results in \cite{gaj23} and upgrade them to a full understanding of the late-time asymptotic behaviour of fixed azimuthal modes. This leaves open an understanding of the infinite sum of azimuthal modes, which is also expected to make use of the ideas developed in the present paper.
\end{itemize}

In the present paper, we prove energy boundedness and (global) integrated energy estimates in the following form:
\begin{theorem}[Energy boundedness and integrated energy estimates; informal version of Theorem \ref{thm:main}]
\label{thm:intromain}
Let $\epsilon>0$. Assume that $0\leq 1-\frac{|Q|}{M}\ll 1$ or $|\mathfrak{q}Q|\ll 1$. Write:
	\begin{equation*}
		G_A=(g^{-1}_{M,Q})^{\mu \nu}\:(^AD_{\mu})\:(^AD_{\nu})(r^{-1}\psi).
	\end{equation*}
Then there exists a uniform constant $C>0$ such that for all $\tau_1\leq \tau_2$:
	\begin{multline}
		\label{eq:introied}
		\sup_{\tau\in[\tau_1,\tau_2]}\int_{\Sigma_{\tau}} \mathcal{E}_{1-\epsilon}[\psi]\, r^{-2}d\mu_{\Sigma_{\tau}}+\int_{\tau_1}^{\tau_2}\int_{\Sigma_{\tau}} \left[\upzeta (\rho_++\kappa_+)r^{-1}\mathcal{E}_{1-\epsilon}[\psi]+(1-\upzeta)r^{-2}|\psi|^2\right]\,r^{-2}d\mu_{\Sigma_{\tau}}d\tau\\
		\leq C\int_{\Sigma_{\tau_1}} \mathcal{E}_{1+\epsilon}[\psi]\, r^{-2}d\mu_{\Sigma_{\tau}}\\
		+C\int_{\tau_1}^{\tau_2}\int_{\Sigma_{\tau}}\left[ (\Omega^{-1}r)^{1+\epsilon}\rho_+^{1-2\epsilon}r^{-1+2\epsilon}|r^2G_A|^2+(1-\upzeta)|^AD_{\tau}G_A|^2\right]r^{-2}\,d\mu_{\Sigma_{\tau}} \,d\tau,
		\end{multline}
		where:
		\begin{itemize}
			\item $\Sigma_{\tau}$ are horizon-intersecting, asymptotically hyperboloidal level sets of the time function $\tau$ with induced volume for $d\mu_{\tau}$, equipped with coordinates $(r,\theta,\varphi)$, where $r$ is the standard radius function,
			\item $\Omega^2=1-\frac{2M}{r}+\frac{Q^2}{r^2}$,
			\item $r_+$ is the event horizon radius and $\rho_+(r)=r_+^{-1}-r^{-1}$,
			\item $\kappa_+$ is the surface gravity of the event horizon,
			\item $\mathcal{E}_p[\psi]$ are weighted energy densities, which can be uniformly estimated follows:
			\begin{equation*}
				\mathcal{E}_{p}[\psi]\sim  (\Omega^{-1}r)^{p}\Omega^2|^AD_{r}\psi|^2+r^{-2}(|^AD_{\tau}\psi|^2+|^AD_{\s^2}\psi|^2+|\psi|^2),
			\end{equation*}
			\item $\upzeta$ is a cut-off function that is identically 1 outside of a small neighbourhood of $\{r=r_{\sharp}\}$, the photon sphere, and is identically zero inside an even smaller neighbourhood of $\{r=r_{\sharp}\}$.
		\end{itemize}
		The estimate holds also when $1-\frac{|Q|}{M}\gtrsim 1$ and $|\mathfrak{q}Q|\gtrsim 1$, under the assumption of \emph{mode stability on the real axis}.
\end{theorem}
Theorem \ref{thm:intromain} may be thought of as a \emph{weak} energy decay statement, which forms a crucial ingredient in the proof of \emph{strong} energy decay statements, precise late-time asymptotics and instabilities in the companion paper \cite{gaj26b}; see \cite{gaj26b}[Theorems 1.1 and 1.2]. A precise version of Theorem \ref{thm:intromain} can be found in the statement of Theorem \ref{thm:main}.

We now provide some remarks.
\begin{remark}[Coupling between Morawetz and $(\Omega^{-1}r)^{p}$-weighted energy estimates]
	The integrated estimate \eqref{eq:introied} should be interpreted as a combination of a local (in space) integrated energy estimate, or Morawetz estimate, with additional $(\Omega^{-1}r)^{p}$-weighted energy estimates near the horizons/null infinity, in the spirit of \cite{newmethod} and \cite{paper4}.
	
	Two related difficulties that are particular to the charged scalar field setting are: 1) the need for an energy flux $\mathcal{E}_p[\psi]$ on the right-hand side with $p>0$ (even in the small $|\mathfrak{\q}Q|$ setting), 2) for sufficiently large $|\mathfrak{q}Q|$, we can only establish $(\Omega^{-1}r)^p$-weighted energy estimates near the horizons that are uniform in $\kappa_+$ and $\kappa_c$ as $\kappa_+\downarrow 0$ and $\kappa_c\downarrow 0$ as part of a global analysis in frequency space. In this sense, the difficulties (C) and (D) introduced above become strongly coupled with (A) and (B) when $|\mathfrak{\q}Q|$ is sufficiently large.
\end{remark}

\begin{remark}[Loss of weights in energy boundedness]
	Note the loss of an $\epsilon$-power in the $\mathcal{E}_{p}[\psi]$-integrals along $\Sigma_{\tau_i}$ when comparing the two sides of \eqref{eq:introied}. This loss can be removed by restricting to $|\mathfrak{q}Q|<\frac{1}{4}$; see the precise statement of Theorem \ref{thm:main}, using the fact that (C) and (D) decouple from (A) and (B) in that case. Alternatively, one can fix an arbitrarily $|\mathfrak{q}|$, but restrict to $\psi$ supported on sufficiently large angular frequencies, $\psi_{\geq \ell}$, to remove the $\epsilon$-loss (and also increase the range of $|p|$).
	
	Note also that an analogue of \eqref{eq:introied} can be derived without the above $\epsilon$-loss, using purely physical-space-based methods, but \underline{including} spacetime integrals of zeroth order terms on the right-hand side (that are moreover supported away from the photon sphere); see Corollary \ref{cor:iedmoduleo0th}. This reflects the fact that the $\epsilon$-loss may be thought of as a \emph{low-frequency} phenomenon. 
	
	Integrated energy estimates (modulo lower-order terms) obtained purely from an application of the divergence theorem with appropriate energy currents in physical space are relevant when extending the analysis from linear wave equations to quasilinear wave equations, and are therefore relevant for the Einstein equations. We refer the reader to the philosophy introduced in \cite{dhrt22, dhrt24} for more details.
	\end{remark}

\begin{remark}[Small $|\mathfrak{q}|$]
	In the special case that $\mathfrak{q}=0$, \eqref{eq:introied} follows straightforwardly from by generalizing the methods valid in the $G_A=0$ case (see for example \cite{redshift} for the sub-extremal setting and \cite{aretakis1} for the extremal setting). In the case $|\mathfrak{q}Q|\ll 1$ we moreover prove an analogue of \eqref{eq:introied} without resorting to Fourier transforms in time; see Corollary \ref{cor:iedsmallq}. This can be interpreted as the statement that the coupling between (B), (C) and (D) is weak for very small $|\mathfrak{q}|$. \textbf{The bulk of the analysis in the present paper is required to treat the case where $|\mathfrak{q}Q|$ is \underline{not} small!}
\end{remark}

\begin{remark}[Mode stability]
Note that Theorem \ref{thm:intromain} only applies \underline{unconditionally} to Reissner--Nordstr\"om solutions with $0\leq 1-\frac{|Q|}{M}\ll 1$ (or alternatively, with $\frac{|Q|}{M}\leq 1$, but $|\mathfrak{q}Q|\ll 1$). Outside of this parameter range, the estimate becomes \underline{conditional} on the validity of \emph{mode stability on the real axis} (away from $\omega=0$). Qualitatively, this corresponds to the non-existence of purely oscillating, finite energy solutions, which remans an open problem on the full sub-extremal Reissner--Nordstr\"om family.\footnote{Note that, under the assumption of mode stability on the real axis, Theorem \ref{thm:intromain} implies mode stability in the upper-half complex plane, as it excludes the existence of exponentially growing modes.}

Indeed, currently available methods in the literature for establishing mode stability \cite{whiting,sr15,costa20, cate22} apply only in a restricted frequency range when $\frac{|Q|}{M}$ is strictly bounded away from extremality, see Proposition \ref{thm:modestabsubext}, but they fail to hold for all superradiant frequencies\footnote{For non-superradiant frequencies, mode stability follows straightforwardly; see for example Lemma \ref{lm:nonsuperradwronsk}.}. Fortuitously, the methods of \cite{costa20} do remain applicable in the \underline{extremal} Reissner--Nordstr\"om case, which, as we will show, implies also mode stability for sufficiently small $\kappa_+$ (and also for Reissner--Nordstr\"om--de Sitter with sufficiently small $M^2\Lambda>0$).\footnote{Even if oscillating or exponentially growing modes are present, one could in principle apply the methods in \cite{gajwar19a, gw24} to consider a finite-codimension subset of initial data that \emph{excludes} these modes from being excited in the time evolution.}

Let us also that in \cite{beha21}, it was shown that there exist exponentially growing modes for the charged Klein--Gordon equation on Reissner--Nordstr\"om--de Sitter with $\frac{2}{3}\Lambda$ replaced by $\mathfrak{m}^2(\Lambda)$, and $\mathfrak{m}^2(\Lambda)$ sufficiently small depending on $\Lambda$ (together with a corresponding largeness assumption on $|\mathfrak{q}Q|$).
\end{remark}

\begin{remark}[The role of $\Lambda$]
Solutions to the conformally covariant charged scalar field equation:
\begin{equation*}
		(g^{-1}_{M,Q, \Lambda})^{\mu \nu}\:(^AD)_{\mu}\:(^AD)_{\nu}\phi-\frac{2\Lambda}{3}\phi=G_A,
	\end{equation*}
with respect to Reissner--Nordstr\"om--de Sitter metrics $g_{M,Q,\Lambda}$ with small $M^2\Lambda$ play an important role for establishing \eqref{eq:introied} in the $\Lambda=0$ case. More precisely, the positive surface gravity  of the cosmological horizon $\kappa_c>0$ (as well as $\kappa_+>0$) is used to derive future-integrability in time of solutions to \eqref{eq:CSFintro} with an appropriately decaying $G_A$. Future-integrability is a property that is initially \emph{assumed} to be able to carry out an analysis in Fourier space, but needs to be justified in order to complete the argument and conclude \eqref{eq:introied}. This is first done via a continuity argument in the charge parameter $\mathfrak{q}$ in the $\kappa_+,\kappa_c>0$ setting, making use of the red-shift effect. One may compare this with \cite{part3} where future-integrability is instead justified in the context of neutral scalar fields on sub-extremal Kerr via a continuity argument in the Kerr rotation parameter $a$.

After justifying the future-integrability assumption in the setting where $\kappa_+>0$ and $\kappa_c>0$, we use the uniformity of the estimates in $\kappa_+,\kappa_c$ to be able to take the limit $\kappa_+\downarrow 0$ and $\kappa_c\downarrow 0$ and conclude \eqref{eq:introied}. As part of this argument, we show that an analogue of Theorem \ref{thm:intromain} holds for $0<M^2\Lambda$ suitably small (uniformly in $\Lambda$). The cosmological constant $\Lambda$ therefore has a regularizing effect. It plays a similar regularizing role in the construction of quasinormal modes in \cite{gajwar19a}. See also the uniform-in-$\Lambda$ analysis in \cite{fst26}.
\end{remark}

\begin{remark}[Electromagnetic gauge]
	Note that \eqref{eq:introied} is \emph{gauge-invariant}, in the sense that it is independent of the particular choice of electromagnetic gauge $A$, provided we multiply $G_A$ by an appropriate phase factor; see \S \ref{sec:CSF}. It is convenient to consider the (conformally) irregular 1-form
	\begin{equation*}
		A=\widetilde{A}=-\frac{Q}{r}dt,
	\end{equation*}
	in the Fourier analysis, while the (conformally) smooth 1-form
	\begin{equation*}
		A=\widehat{A}=-\frac{Q}{r}d\tau.
	\end{equation*}
	is more convenient to establish local physical-space-based energy estimates.
\end{remark}

\subsection{Previous results}
In this section, we give an overview of related previous work concerning integrated energy estimates for the wave equation on black holes spacetimes. We provide further references pertaining to precise, late-time (pointwise) decay estimates in our companion paper \cite{gaj26b}.
\subsubsection{Integrated energy estimates on sub-extremal black holes}
In the context of Schwarzschild spacetimes, establishing \emph{local} integrated energy estimates (in a spatially bounded region, supported away from the event horizon) requires an understanding of the obstruction of (unstable) trapped null geodesics at the photon sphere $\{r=3M\}$. Integrated local energy estimates for the wave equation on Schwarzschild were first obtained in \cite{redshift} without a decomposition into time frequencies $\omega$, requiring merely a decomposition into angular frequencies $\ell$. See also the related earlier work \cite{blu0}, as well as \cite{tataru1} for a refined understanding of the effects of trapped null geodesics. Note that the decomposition into angular frequencies in proofs of local integrated energy estimates was removed entirely in \cite{hmv24}.

An important ingredient towards turning the \emph{local} integrated energy estimates into \emph{global} integrated energy estimates that extend to the horizon and future null infinity are \emph{red-shift estimates} (near the horizon), introduced in \cite{redshift} and \emph{$r^p$-weighted energy estimates} (in a far-away region), introduced in \cite{newmethod}. As will be clarified below,  $r^p$-weighted energy estimates may be viewed as a degenerate analogue of red-shift estimates.

In going from Schwarzschild to the more general Kerr setting, one faces several important additional difficulties. First of all, due to the presence of an ergoregion, where the Killing vector fields generating time translation symmetry $\partial_t$ fail to be timelike, one cannot prove an energy boundedness statement independently from an integrated energy estimate, in contrast with the Schwarzschild setting. When decomposing into time frequencies $\omega$, this manifests itself as \emph{superradiance} \cite{zel71,zel72}, see \S\ref{sec:introstephighfreq} for a discussion of superradiance in the charged setting. In addition, the existence of trapped null geodesics arises creates an obstruction at a range of $r$-values when varying the time frequency, rather than a single $r$-value, as in Schwarzschild. 

Integrated energy estimates in the full sub-extremal Kerr range were first obtained in \cite{part3} by carrying out a full decomposition into time frequencies and oblate spheroidal harmonic modes (appropriate generalizations of spherical harmonic modes) and treating different time- and angular-frequency regimes separately. An important step in the proof is a proof of \emph{future-integrability} of solutions, which is necessary to be able to apply the estimates at the level of the Fourier transform in time. See also the estimates in \cite{dhrt24}, which only rely on a decomposition into time frequencies and azimuthal modes, at the expense of introducing lower-order terms on the right-hand side. We also point the reader to the comprehensive analysis on integrated energy estimates in setting of the Teukolsky equation in \cite{shlcosta20, shlcosta23} and earlier works in the very slowly rotating case ($|a|\ll M$) in \cite{dr7,lecturesMD,enadio, tataru2,blukerr}.

In the present paper, we will apply microlocal energy currents as in \cite{dr7,part3} to deal with $(\omega,\ell)$ outside of a compact set; see Step 1 of \S \ref{sec:sketchpf}. 

As mentioned above, energy \emph{boundedness} estimates and \emph{integrated} local energy estimates are coupled in the proof of \cite{part3}, due to the presence of superradiance in Kerr. In \cite{st25} it was, however, shown that energy boundedness estimates can be obtained independently from integrated local energy estimates, when restricting to solutions supported on high, trapped (but non-superradiant) frequencies.

Finally, we note that in the case of (massive) wave equations on sub-extremal Reissner--Nordstr\"om--de Sitter spacetimes, integrated energy estimates were obtained in \cite{gon24}. 

\subsubsection{Integrated energy estimates on extremal black holes}
The first integrated energy estimates for the wave equation on extremal black hole spacetimes were obtained on extremal Reissner--Nordstr\"om in \cite{aretakis1}, employing similar vector fields multipliers to \cite{redshift}, exploiting the presence of a photon sphere at $\{r=2M\}$ and the existence of a degenerate energy boundedness estimate that is independent of local integrated energy estimates. In \cite{aretakis1}, a preliminary analogue of the $r^p$-weighted energy estimates was introduced near the event horizon, which was later fully developed in \cite{paper4}.

Integrated energy estimates for the wave equation on extremal Kerr were established first for axisymmetric solutions ($m=0$) in \cite{aretakis3}, see also \cite{giwa24}, and for solutions supported on a finite number of azimuthal modes ($m\neq 0$) in \cite{gaj23}, the latter assuming future-integrability. As the results in \cite{gaj23} are closely connected to results of the present paper, we provide a more detailed comparison in \S \ref{sec:compwavekerr}.

\subsubsection{Integrated estimates for charged scalar fields}
When viewing the Minkowski spacetime as a solution to the Einstein--Maxwell equations, the corresponding Faraday tensor $F$ vanishes. The linearization of the Maxwell--charged scalar field equations around $(\phi,F)=(0,0)$ results in the standard wave equation on Minkowski.\footnote{One can nevertheless study the charged scalar field equation on the Minkowski background around a non-zero $F$, namely $F=-Qr^{-2}dt\wedge dr$ corresponding to the Reissner--Nordstr\"om solution, see also \cite{gvdm24}. This equation plays a fundamental role in the precise late-time asymptotics derived in \cite{gaj26b}.} The nonlinear Maxwell--charged scalar field system, however, does feature a non-zero $F$, which has as significant long-range effect on the analysis. For this system, there is a large body of literature ranging from global well-posedness results for large initial data, see \cite{em82a,em82b, km94, yan18}, to decay estimates for small initial data \cite{ls06,yan16}. In particular, \cite{yan16, yan18} feature integrated local energy decay estimates and $r^p$-weighted energy estimates.

A key additional difficulty that arises when considering charged scalar fields on black hole backgrounds is the non-decay of $F$ in time, due to the fact that $F$ is non-decaying in the absence of a scalar field. 

This difficulty affects both integrated local energy decay estimates and $r^p$-weighted energy estimates.
It was first addressed in the context of the Maxwell--charged scalar field system in spherical symmetry with small scalar field charge parameters $\mathfrak{q}$ on sub-extremal Reissner--Nordstr\"om backgrounds in \cite{vdm22}. In particular, integrated local energy estimates and $r^p$-weighted energy estimates were established, exploiting both spherical symmetry and smallness of $\mathfrak{q}$.

Let us note that in the context of the Maxwell--charged scalar fields (on both Minkowski and black hole backgrounds), there exists a conserved, non-negative definite (degenerate) energy for the pair $(\phi,F)$. In contrast, for the linearized equation \eqref{eq:CSFintro} there exists no such non-negative definite conserved quantity, allowing in principle for solutions that are growing in time.\footnote{Note that the existence of solutions that grow in time for the linearized system can in principle be compatible with a bounded energy for the nonlinear system, even for small initial data, since mere boundedness of the energy in the nonlinear system does not need to guarantee (sufficiently rapid) convergence to the zero solution, and therefore decay of nonlinear terms.}

Finally, we note that the wave equation with an asymptotically inverse-square potential serves as an important toy model that captures some of the difficulties that appear in the case of large scalar field charge coupling constant $\mathfrak{q}$. In this case, integrated energy estimates (local and $r^p$-weighted) were obtained in \cite{gaj22a}.

\subsection{Sketch of the proof}
\label{sec:sketchpf}
The proof of Theorem \ref{thm:intromain} can be divided into four main steps:
\begin{enumerate}[label=\textbf{Step \arabic*}:]
	\item Derive weighted $L^2$-estimates for the Fourier transform of $\phi$ in $t$ and on $\s^2$, restricted to high frequencies, assuming a priori that $\phi$ is \emph{future integrable} in time to make sense of Fourier-transformed expressions.
	\item Derive weighted $L^2$-estimates for the Fourier transform, restricted to bounded frequencies.
	\item Apply Plancherel's theorem to obtain weighted $L^2$-estimates for $\phi$ (under the assumption of future integrability).
	\item Verify the assumption of future-integrability for $\kappa_+>0$ and $\kappa_c>0$ and then take the limit $\kappa_c, \kappa_+\downarrow 0$.
\end{enumerate}

\subsubsection{Fourier transform}
\label{sec:introsketchfourier}
Let $(t,r,\theta,\varphi)$ denote the standard coordinates on Reissner--Nordstr\"om(--de Sitter) with $\Lambda\geq 0$. Consider the electromagnetic gauge:
\begin{equation*}
	A=-\frac{Q}{r}dt.
\end{equation*}
Let $\psi=r\phi$. In order to Fourier transform in $t$, we first extend $\psi$ as a function on the full manifold and then restrict:
\begin{equation*}
	\uppsi(t,r,\theta,\varphi):=\xi(\tau(t,r))\cdot  \psi(t,r,\theta,\varphi),
\end{equation*}
where $\xi: \R_{\tau}\to [0,1]$ is a smooth cut-off function satisfying $\xi(\tau)=0$ for $\tau\leq 0$ and $\xi(\tau)=1$ for $\tau\geq 1$.

Note that then
	\begin{equation*}
		(g^{-1}_{M,Q,\Lambda})^{\mu \nu}\:(^AD)_{\mu}\:(^AD)_{\nu}(r^{-1}\uppsi)-\frac{2\Lambda}{3}(r^{-1}\uppsi)=\xi \cdot G_A+F_{\xi},
	\end{equation*}
where $F_{\xi}$ is compactly supported in $0\leq \tau\leq 1$.

Assuming that $\uppsi(t,\cdot)\in L^2(\R_t)$, the $t$-Fourier transform $\mathfrak{F}(\uppsi)(\omega,\cdot)$ is well-defined. We perform a Fourier transform on $\s^2_{\theta,\varphi}$ by decomposing into spherical harmonics:
\begin{equation*}
	u_{\ell m}(r_*;\omega):=\int_{\s^2}\mathfrak{F}(\uppsi)(\omega,r(r_*),\theta,\varphi)\overline{Y}_{\ell m}(\theta,\varphi)\,\sin\theta d\theta d\varphi,
\end{equation*} 
where $r_*$ is a tortoise coordinate, satisfying $\frac{dr_*}{dr}=\Omega^{-2}$, and $Y_{\ell m}$ denote spherical harmonics.\footnote{Note that since Reissner--Nordstr\"om spacetimes are spherically symmetric, the Fourier transforms in $t$ and $\s^2$ commute.}

Then $u_{\ell m}$ satisfies a Schr\"odinger equation:
\begin{equation}
\label{eq:introschrod}
	u_{\ell m}''+\left(\omega^2-V_{\ell  \omega}\right)u_{\ell m}=H_{\ell m},
\end{equation}
with $H_{\ell m}(r_*,\omega)$ the Fourier transform of $r\Omega^2(F_{\xi}+\xi G_A)_{\ell m}(t,r_*)$ and $V_{\ell\omega}$ functions of $r_*$.

It is convenient to consider the following dimensionless charge parameter:
\begin{equation*}
	\q:=\mathfrak{q}Q.
\end{equation*}

It will also be more convenient to introduce the following shifted frequencies:
\begin{align*}
	\homega:=\omega-\q r_c^{-1},\\
	\tomega:=\omega-\q r_+^{-1},
\end{align*}
with $r_+$ the event horizon radius and $r_c$ the cosmological horizon radius ($r_c=\infty$ when $\Lambda=0$)) as these show up in the following boundary conditions for $u_{\ell m}$ as $r_*\to \pm \infty$ (which should be interpreted in an $L^2_{\omega}$-sense):
\begin{align*}
	u_{\ell m}(r_*;\omega)\sim &\: e^{-i\tomega r_*-i\q \int_{0}^{r_*}\rho_+(r(r_*'))\,dr_*'}\quad (r_*\to -\infty),\\
		u_{\ell m}(r_*;\omega)\sim &\: e^{+i\homega r_*-i\q \int_{0}^{r_*}\rho_c(r(r_*'))\,dr_*'}\quad (r_*\to \infty),
\end{align*}
where
\begin{align*}
	\rho_+:=&\: r_+^{-1}-r^{-1},\\
	\rho_c:=&\: r^{-1}-r_c^{-1}.
\end{align*}
Then we define the shifted potentials $\widehat{V}_{\ell  \omega}$ and $\widetilde{V}_{\ell  \omega}$ as follows:
\begin{equation*}
	\omega^2-V_{\ell \omega}=:\homega^2-\widehat{V}_{\ell  \homega}=: \tomega^2-\widetilde{V}_{\ell  \tomega}.
\end{equation*}
Since $\omega^2-V_{\ell \omega}$ does not play an important role, we will only consider $\homega$ and $\tomega$ and redefine: $V_{\ell \homega}=\widehat{V}_{\ell  \homega}$.

Then:
\begin{equation*}
V_{   \ell \homega}(r)=\ell(\ell+1)\Omega^2r^{-2}-{\q}^2\rho_c^2(r)+2{\q} \homega \rho_c(r)+\Omega^2\left[r^{-1}\frac{d\Omega^2 }{dr}+\frac{2\Lambda}{3}\right] .
\end{equation*}

The goal at the Fourier-transformed level is to derive appropriate uniform $L^2_{\omega}$-estimates, so that we can then apply Plancherel's theorem to arrive at integrated-in-time energy estimates in physical space.

\subsubsection{Step 1: high-frequency regimes}
\label{sec:introstephighfreq}
We partition the high-frequency regime as follows:
\begin{itemize}
	\item $\homega$-dominated: $|\homega|\gg \ell+1$,
	\item $\ell$-dominated:  $\ell\gg |\homega|+1$,
	\item trapped: $|\homega|\sim \ell\gg 1$.
\end{itemize}
In each of these frequency regimes, we derive weighted $L^2_{r_*}$-estimates via the consideration of \emph{multipliers} of the form (where we omit the subscripts $m\ell$):
\begin{equation*}
	\re(h u\cdot \overline{H}),\quad \re(y u'\cdot \overline{H}),\quad \re(f'u \overline{H}+2fu' \overline{H}),\quad \re(i\homega u\cdot \overline{H}),\quad \re(i\tomega u\cdot \overline{H}),
\end{equation*}
with $h,y,f: \R_{r_*}\to \R$, as well as appropriate linear combinations, and then integrating by parts.

In addition to the above multipliers, we need to consider multipliers corresponding to $(\Omega^{-1}r)^p$-estimates in frequency space:
\begin{equation*}
\re((\Omega^{-1}r)^{p} w'e^{i \homega r_*-i \q \int_{0}^{r_*}\rho_c(r_*')\,dr_*'}\overline{H}),\quad \re((\Omega^{-1}r)^{p} v'e^{-i \tomega r_*-i \q \int_{0}^{r_*}\rho_+(r_*')\,dr_*'}\overline{H}),
\end{equation*}
where
\begin{equation*}
	v=e^{i\tomega r_*+i \q \int_{0}^{r_*}\rho_+(r_*')\,dr_*'}u,\quad w=e^{-i\homega r_*+i \q \int_{0}^{r_*}\rho_c(r_*')\,dr_*'}u.
\end{equation*}

In particular, the multipliers $\re(i\homega u\cdot \overline{H})$ and $\re(i\tomega u\cdot \overline{H})$ may be interpreted as the Fourier-transformed analogues of energy estimates with respect to the multipliers $K_+\psi$ and $K_c\psi$, respectively, where\\ $K_+=T+i\q r_+^{-1}$ and $K_c=T+i\q r_c^{-1}$ (and $K_c=T$ if $\Lambda=0$). Integrating by parts results in:
\begin{align*}
	\homega^2|u|^2(\infty)+\homega \tomega |u|^2(-\infty)=&\:\int_{\R}\homega  \re(iu \overline{H}),\\
	\tomega^2|u|^2(-\infty)+\homega \tomega |u|^2(\infty)=&\:\int_{\R}\tomega  \re(iu \overline{H}).
\end{align*}
Note that the two terms on the left-hand side are each non-negative only if:
\begin{equation*}
	\homega \tomega\geq 0.
\end{equation*}
We then say $\omega$ is \emph{non-superradiant}. Superradiant frequencies correspond to $\homega \tomega <0$, or equivalently:
\begin{equation*}
	\homega (\homega-\q \rho_+(r_c))<0,
\end{equation*}
which corresponds to $\q\homega\in (0,\q^2\rho_+(r_c))$.

Superradiance forms an obstruction, as it prevents us from controlling boundary terms at $r_*=\pm \infty$ using the Fourier transform of $K_+$-energy estimates and $K_c$-energy estimates.

An important property of $V_{\ell \homega}$ is that in the superradiant frequency regime, $\homega \tomega<0$,
\begin{equation*}
	V_{\ell \homega}(r_{\rm max})-\homega^2>0.
\end{equation*}
We say that \emph{superradiant frequencies are not trapped}. This nomenclature is be motivated by the dynamics of free-falling, charged, massless particles. More precisely, we consider massive particles with charge $e$ and mass $m$, restricted to the equatorial plane, satisfying the (spacetime version) of the Lorentz force law and then take the limit $m\downarrow 0$, $e\to 0$, such that $\frac{e}{m}\to\mathfrak{q}\neq 0$. The effective potential $V_{\rm eff}$ appearing in the corresponding equations of motion then satisfies:
\begin{equation*}
	V_{\rm eff}(r)=L^2\Omega^2r^{-2}-{\q}^2\rho_c^2(r)+2{\q} \hat{E}\rho_c(r)+\q^2r_c^{-2},
\end{equation*}
where $L$ is the magnitude of the angular momentum and $\hat{E}:=p_t+qr_c^{-1}$, with $p_t$ the canonical momentum associated to time translations. The pair $(\hat{E},L)$ are conserved along the motion of the particle.

Note that, up to a frequency-independent and charge-independent term, $V_{\rm eff}(r)$ agrees with $V_{\homega \ell}$ under the identification $\homega\mapsto \hat{E}$ and $L^2\mapsto \ell(\ell+1)$. The charged particle is trapped at $r=r_{\rm max}$ if $V_{\rm eff}(r_{\rm \max})-\hat{E}^2=0$.

Superradiant frequencies correspond to charged particles with $\q\hat{E}\in (0,\q^2\rho_+(r_c))$. It can be shown that in this case, $V_{\rm eff}(r_{\rm max})-\hat{E}^2>0$, so these particles are not trapped.

Since superradiant frequencies $\homega$ are bounded, they only affect the $\ell$-dominated frequencies $|\homega|\ll \ell+1$ of our high-frequency regimes.

The form of the potential $V_{\ell \omega}$ when $ \ell\gg |\homega|+1$ is sketched in Figure \ref{fig:introelldom} below.
\begin{figure}[h!]
	\begin{center}
\includegraphics[scale=0.6]{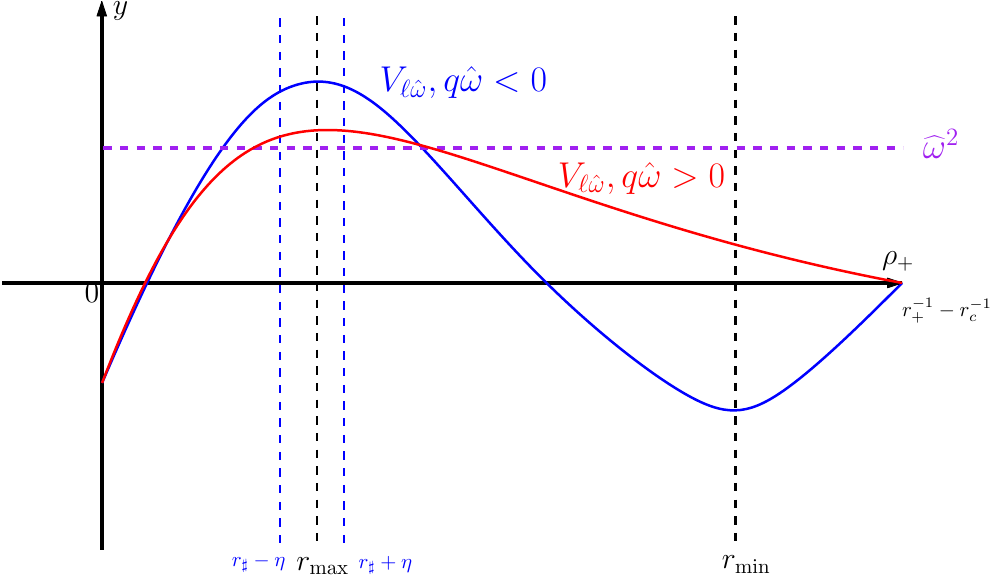}
\end{center}
\caption{A sketch of the graph of the potential $V_{\ell \omega}$ as a function of $\rho_+$ in the high, $\ell$-dominated frequency regime with $\q\tomega<0$.}
	\label{fig:introelldom}
\end{figure}
\\
We make the following observations:
\begin{itemize}
	\item the maximum $r_{\rm max}$ is located very close to $r_{\sharp}$, the \emph{Reissner--Nordstr\"om(-de Sitter) photon sphere radius} and
	\begin{equation*}
		V_{\ell \homega}(r_{\rm max})-\homega^2\gtrsim \ell(\ell+1).
	\end{equation*}
	\item In the non-superradiant case with $\q\tomega<0$ and $\q\homega<0$, $V_{\ell \homega}$ has a minimum to the right of the maximum for suitably small $\kappa_c$. As $\frac{\q\homega}{\ell(\ell+1)}\uparrow 0$, the minimum approaches $\rho_+=r_+^{-1}-r_c^{-1}$. Analogously, it can be shown that when $\kappa_+$ is suitably small and $\q\tomega>0$ and $\q\tomega>0$, $V$ has a minimum to the left of the maximum, which approaches $\rho_+=0$ as $\q\tomega\downarrow 0$. Finally, when $\q\homega\cdot \q\tomega<0$ (the superradiant regime), $V_{\ell \homega}$ has a single extremum, a maximum. 
\end{itemize}
The first observation provides a quantitative version of the statement that superradiant frequencies are not trapped.

The second observation concerns the main difficulty in the high-frequency analysis and may be thought of as the main characteristic feature of the absence of red-shift and the presence of a scalar field charge in the high-frequency regime.

In the $\q\tomega, \q\homega<0$ case, we resolve this difficulty by exploiting the fact that the minimum occurs for large values of $r$ ($\rho_c$ close to zero), where we can consider the Fourier versions of $(\Omega^{-1}r)^p$-weighted energy estimates above. The same strategy works in the $\q\tomega, \q\homega>0$ case. While such weighted energy estimates do not hold in general in physical space, we can exploit the fact that $\ell\gg |\homega|$ to ensure their validity in this frequency regime. It is important to emphasize here that we do \underline{not} make use of the red-shift effect near the event or cosmological horizons, as this would result in estimates that are not uniform as $\kappa_c\downarrow 0$ or $\kappa_+\downarrow 0$.

In the superradiant case ($\q\tomega\cdot \q\homega<0$), the potential has a single extremum, so the above difficulty does not occur. However, in that case boundary terms at $r=r_+$ and $r=r_+$ (or $r=\infty$) do not have a good sign, so we cannot use the Fourier versions of $K_+$-energy estimates and $K_c$-energy estimates and we instead exploit, as in \cite{part3}, the largeness of $V_{\ell \homega}(r_{\rm max})-\homega^2\gtrsim \ell(\ell+1)$ to absorb these boundary terms by employing (Fourier versions of) $K_+$-energy estimates and $K_c$-energy estimates with appropriate cut-off functions. 

The trapped frequencies can be treated analogously. Since these frequencies are non-superradiant, the main remaining difficulty is the presence of a minimum of the potential near $\rho_+=0$ or $\rho_c=0$, as above, which necessitates a coupling with $(\Omega^{-1}r)^p$-weighted estimates as above. The main difference is the appearance of a degenerate factor near $r=r_{\sharp}$ due to the fact that $V_{\ell \homega}(r_{\rm max})$ or $\homega^2-V_{\ell \homega}(r_{\rm max})$ cannot be bounded below.

We finally note that the large $\homega$-frequency regime is largely independent of the precise behaviour of the potential and proceeds in a straightforward manner. In particular, we do not need to appeal to $(\Omega^{-1}r)^p$ weighted energy estimates (which would would  \emph{a priori} also not be valid in this frequency regime!).

The conclusion of Step 1 can be stated very schematically as the following inequality, where we ignore all (important!) $r$-weights on the left-hand side:
\begin{multline}
\label{eq:schematiciedhighfreq}
c\cdot 	\int_{\R_{r_*}}|u'|^2+\homega^2|u|^2+|u|^2\,dr_*\leq \int_{\R_{r_*}} \re(y u'\cdot \overline{H})+\re(h u\cdot \overline{H})+\re(f'u \overline{H}+2fu' \overline{H})+\re(i\homega u\cdot \overline{H})+\re(i\tomega u\cdot \overline{H})\,dr_*\\
-\int_{\R_{r_*}}\re((\Omega^{-1}r)^{p} w'e^{i \homega r_*-i \q \int_{0}^{r_*}\rho_c(r_*')\,dr_*'}\overline{H})+ \re((\Omega^{-1}r)^{p} v'e^{-i \tomega r_*-i \q \int_{0}^{r_*}\rho_+(r_*')\,dr_*'}\overline{H})\,dr_*,
\end{multline}
where $h,y,f$ denote functions that arise from the multipliers described above and may have some $\omega$ dependence in a tractable way, but not to leading-order near $r=r_+$ and $r=r_c$. The constant $c>0$, importantly, does not depend on $\omega$ or $\ell$.

\subsubsection{Step 2: bounded frequency regimes}
\label{intro:sketchstep2}
In the case of bounded frequencies, we can (formally) apply Green's formula to express:
\begin{equation}
\label{eq:greenformula}
u(r_*)=\frac{u_{\infty}(r_*)}{\mathfrak{W}}\int_{-\infty}^{r_*}u_+(r_*')H(r_*')\,dr_*'+\frac{u_{+}(r_*)}{\mathfrak{W}}\int_{r_*}^{\infty}u_{\infty}(r_*')H(r_*')\,dr_*',
\end{equation}
with $u_{\infty}$ and $u_{+}$ solutions to \eqref{eq:introschrod} with $H=0$, satisfying:
\begin{align*}
	u_{+}\sim &\: e^{-i\tomega r_*-i\q \int_{0}^{r_*}\rho_+(r(r_*'))\,dr_*'}\quad (r_*\to -\infty),\\
		u_{\infty}(r_*)\sim &\: e^{+i\homega r_*-i\q \int_{0}^{r_*}\rho_c(r(r_*'))\,dr_*'}\quad (r_*\to \infty),
\end{align*}
and with $\mathfrak{W}=\mathfrak{W}(\omega, \ell)$ the corresponding Wronskian. The goal is to obtain suitable estimates for $u_+$, $u_{\infty}$ and $\mathfrak{W}$ and then convert these into weighted $L^2$-estimates via Green's formula. This problem can be divided into two parts:
\begin{enumerate}
	\item An analysis \textbf{away from} $\homega=0$ and $\tomega=0$, where $u_+$ and $u_{\infty}$ can easily be controlled and boundedness of $|\mathfrak{W}|^{-1}$ may be thought of as a consequence of the fact that $\mathfrak{W}\neq 0$ together with continuity of $\mathfrak{W}$ along the real axis away from $\homega=0$ and $\tomega=0$. The statement $\mathfrak{W}\neq 0$ is known as \emph{mode-stability on the real axis}, because $\mathfrak{W}(\omega,\ell)=0$ would imply the existence of finite energy solutions $\psi_{\ell m}(t,r)=e^{-i\omega t}u(r)$ which are non-decaying.
	\item An analysis \textbf{near} $\homega=0$ and $\tomega=0$. As $\homega\to 0$, $r_*=\infty$ changes from a regular singular point of \eqref{eq:introschrod} to an irregular singular point in the case $\kappa_c=0$. The same change occurs as $\tomega\to 0$ for $r_*=-\infty$ in the case $\kappa_+=0$.
\end{enumerate}
We first establish mode stability on the real axis in the case $\kappa_c=\kappa_+=0$ by applying the methods developed in \cite{costa20} and then extend this to the near-extremal case $\kappa_c\ll \kappa_+\ll 1$. \textbf{This is the only part of the analysis where appeal to (near-)extremality $1-\frac{|Q|}{M}\ll 1$!}

The rest of the analysis concerns small frequencies: $|\homega|\ll 1$ or $|\tomega|\ll 1$. This part contains the main novel aspects of the analysis in the present paper, which includes a mathematically rigorous realization of the small-frequency ``matched asymptotics'' method that is frequently applied in the black hole setting in the physics literature, see for example the foundational works \cite{sta73, staro, tp74} as well as \cite{hp98b,zimmerman4} in the charged scalar field on (extremal) Reissner--Nordstr\"om. In contrast with the above works, we need to provide rigorous estimates for relevant error terms to be able to conclude quantitative estimates.

Consider solutions $u$ \eqref{eq:introschrod} with $H=0$. We can rewrite the ODE in two different ways:
\begin{align}
\label{eq:intromaineqUevent}
\frac{d^2U_+}{ds_+^2}=&\:\rho_+^{4}(r^{-2}\Omega^2)^{-2}\mathcal{V}U_+, \\
\label{eq:intromaineqUcosmo}
\frac{d^2U_c}{ds_c^2}=&\:\rho_c^{4}(r^{-2}\Omega^2)^{-2}\mathcal{V}U_c,
\end{align}
where	$s_+= \frac{1}{\rho_+}-r_+$, $s_c= \frac{1}{\rho_c}-r_+$, $U_+(s_+)=(s_++r_+)(r^{-1}\Omega u)(r_*(s_+))$, $U_c(s_c)=(s_c+r_+)(r^{-1}\Omega u)(r_*(s_c))$ and:
\begin{equation*}
\mathcal{V}_{\homega \ell}(r)=-\homega^2+V_{\homega\ell}(r)+\Omega^3\frac{d^2\Omega}{dr^2}(r).
\end{equation*}
The coordinates $s_+$ and $s_c$ are natural for treating the cases $\kappa_c\geq 0$ and $\kappa_+\geq 0$ in a uniform manner.

The key idea is to approximate the solutions $U_+$ and $U_c$ by appropriate special functions in different regions and establish the existence of an overlap region, where the approximations coincide.

For simplicity, in the discussion below, we will restrict to the extremal case $\kappa_c=\kappa_+=0$ (so $\homega=\omega$) and we will set the length scale to $r_+=1$. Denote:
\begin{equation*}
	\beta_{\ell}=\sqrt{(2\ell+1)^2-4\q^2}.
\end{equation*}

In that case, we can split the $r$-domain $[1,\infty)$ as follows:
\begin{itemize}
\item Region I:
When
\begin{equation*}
\frac{1}{2}|\omega|^{\frac{1}{4}} |\tomega|^{-\frac{1}{4}}\leq s_+ <\infty,
\end{equation*}
we can approximate:
\begin{equation*}
	U_+\approx B_1W_{-i\sigma,\frac{1}{2}\beta_{\ell}}(-2i\tomega s_+)+B_2M_{-i\sigma,\frac{1}{2}\beta_{\ell}}(-2i \tomega s_+),
\end{equation*}
where $W_{\varkappa,\mu},M_{\varkappa,\mu}$ are standard Whittaker functions and $\sigma=-\q+2\tomega$. Note that in the limit $\omega\to 0$, $U_+$ becomes \emph{exactly} a linear combination of Whittaker functions.
\item Region II:
When
\begin{equation*}
\frac{1}{2} |\tomega|^{\frac{1}{4}}|\omega|^{-\frac{1}{4}}\leq s_c <\infty,
\end{equation*}
we can approximate:
\begin{equation*}
	U_c\approx D_1W_{-i\sigma,\frac{1}{2}\beta_{\ell}}(-2i\omega s_c)+D_2M_{-i\sigma,\frac{1}{2}\beta_{\ell}}(-2i \omega s_c),
\end{equation*}
with now $\sigma=-\q-2\tomega$. Note that in the limit $\tomega\to 0$, $U_c$ becomes \emph{exactly} a linear combination of Whittaker functions.
\end{itemize}
Since $s_c=r-1$ and $s_+=\frac{1}{r-1}$, we can equivalently characterize Region II as follows:
\begin{equation*}
	0<s_+\leq 2 |\omega|^{\frac{1}{4}}|\tomega|^{-\frac{1}{4}},
\end{equation*}
so Regions I and II overlap and their intersection corresponds to: $\frac{1}{2}|\omega|^{\frac{1}{4}} |\tomega|^{-\frac{1}{4}}\leq s_+\leq 2 |\omega|^{\frac{1}{4}}|\tomega|^{-\frac{1}{4}}$. Furthermore, $U_c=s_+^{-1}U_+$.

Consider now the case of small $|\tomega|$. Setting $u=u_+$ fixes the constants $B_1$ and $B_2$. In particular, $B_2=0$. Since both approximations are valid in the overlap region \underline{and} the arguments in all the Whittaker functions are small, we can determine the constants $D_1$ and $D_2$ from the small $z$-asymptotics of $W_{\varkappa,\mu}(z),M_{\varkappa,\mu}(z)$. Furthermore, $|\mathfrak{W}|^{-1}\sim |\tomega|^{-1} D_2^{-1}$, which provides a bound for the Wronskian:
\begin{align*}
	|\mathfrak{W}|^{-1} \sim &\: |\tomega|^{-\frac{1}{2}+\frac{1}{2} \re \beta_{\ell}}\quad(\beta_{\ell}\neq 0),\\
	|\mathfrak{W}|^{-1} \sim &\:|\tomega|^{-\frac{1}{2}}\frac{1}{1+\log((1+|\tomega| )^{-1})}\quad(\beta_{\ell}=0),
\end{align*}
where the logarithmic term arises from the logarithm  appearing in the small $z$-asymptotics of the Whittaker function $M_{\varkappa,\mu}(z)$ when $\mu=0$. Furthermore, the estimates on $D_1$ and $D_2$ also result in global, uniform estimates for $u_+$. 

Choosing $u=u_{\infty}$ instead fixes $D_1$ and also fixes $D_2=0$. Then estimates for $B_1$ and $B_2$ follow, which therefore also give global, uniform estimates for $u_{\infty}$. The small-$|\omega|$ case proceeds entirely analogously.

\textbf{The above argument relies on precise properties of Whittaker functions, such as explicit connection formulas between $W_{\varkappa,\mu}(z)$ and $M_{\varkappa,\mu}(z)$!}

While an approximation by Whittaker functions is sufficient for the extremal case, it is insufficient for deriving uniform estimates in $\kappa_+$ and $\kappa_c$ in the near-extremal case, which are required for our analysis. In this case, we need to split the $r$-domain into four regions in total. When $|\tomega|\gg \kappa_+$ and $|\homega|\gg \kappa_c$, we can still use Whittaker functions away from the event horizon and cosmological horizon/null infinity. Near the horizons, we can instead apply a standard WKB or Liouville--Green approximation and approximate the solutions by oscillating exponentials.

When $|\tomega|\lesssim \kappa_+$ or $|\homega|\lesssim \kappa_c$, we instead approximate $U_+$ and $U_c$ by \emph{Gauss hypergeometric functions}. These can then be matched to other hypergeometric or Whittaker functions, as above. Note that the small-$|\tomega|$ analysis requires smallness of $\kappa_+$ and $\kappa_c$, whereas the small-$|\homega|$ analysis requires only smallness of $\kappa_c$. 

We conclude the following uniform estimates for $\mathfrak{W}$, valid in the full bounded frequency regime (under the assumption of mode stability when $\kappa_+$ is not small):
\begin{equation*}
|\mathfrak{W}|^{-1}\sim (|\tomega|+\kappa_+)^{-\frac{1}{2}+\frac{1}{2} \re \beta_{\ell}}(|\homega|+\kappa_c)^{-\frac{1}{2}+\frac{1}{2} \re \beta_{\ell}}\frac{1}{1+\delta_{\beta_{\ell}0}\log((1+|\tomega| +\kappa_+)^{-1}(1+|\homega| +\kappa_c)^{-1})}.
\end{equation*}
In the remainder of the bounded-frequency analysis, we apply \eqref{eq:greenformula} to conclude the following schematic inequality, where we ignore (important!) $r$-dependent weights on both sides:
\begin{equation}
\label{eq:schematiciedboundfreq}
c\cdot 	\int_{|\omega|\lesssim 1}\int_{\R_{r_*}}(\ldots)|u|^2\,dr_*d\omega\leq \int_0^{\infty}\int_{\Sigma_{\tau}} (\ldots)|F_{\xi}|^2+(\ldots)|G_A|^2\,\sin \theta d\theta d\varphi dr d\tau.
\end{equation}
Note that in contrast with \eqref{eq:schematiciedhighfreq}, \eqref{eq:schematiciedboundfreq} already has physical-space norms on the right-hand side, which is necessary in order to obtain $L^2$-estimates via \eqref{eq:greenformula} with the desired $r$-weights; see also \cite{sr15}[\S3.1] for a related analysis.

Let us finally mention that a precise analysis of analogues of $u_+$ and $u_{\infty}$ also plays an important role in the late-time analysis of the massive Klein--Gordon equation on Schwarzschild in \cite{psv23, svdm24}. In that case, the critical frequencies $\omega=\pm m$, with $m$ the Klein--Gordon mass, play an analogous role to $\omega=0$ and $\omega=\mathfrak{q}Qr_+^{-1}$ in the present paper. For $|\omega|<m$,  homogeneous solutions are approximated by matching an approximation by Bessel functions, to a WKB/turning point approximation (in a large-$r$ region). This stands in contrast with earlier work from the physics literature (in the bounded angular frequency regime), see for example \cite{kt02}, where Whittaker functions are used (for bounded angular frequencies) at suitably large values of $r$ and are matched to hypergeometric functions which cover the rest of the region (and describe exactly the behaviour of solutions at $\omega=\pm m$). 

In the Klein--Gordon setting, the turning points of the relevant potentials occur at large absolute values of the relevant rescaled argument $x$ appearing in the Whittaker function $W_{\varkappa,\mu}(x)$ in \cite{kt02} (with $|x|\sim \frac{r}{\sqrt{m^2-\omega^2}}$). Note that for large $|x|$, the Whittaker functions themselves can be approximated via a WKB/turning point analysis. Furthermore, $\varkappa$ scales like $\frac{1}{\sqrt{m^2-\omega^2}}$, which implies that for $|x|\ll \varkappa$, the Whittaker functions can instead be approximated by  Bessel functions. This is what allows \cite{psv23} to circumvent the use of Whittaker functions altogether in the analysis.

In the charged scalar field setting, there are two key differences:  1) the turning points of the potential $V$ do \underline{not} occur at large (or small) values of the relevant rescaled variables $|\tomega s_+|$ and $|\homega s_c|$ of the Whittaker functions, so the Whittaker functions \underline{cannot} be approximated via a WKB/turning point analysis in the relevant regions, and 2) the Whittaker parameter $\varkappa$ is bounded, so Bessel functions also do not form adequate approximations. For this reason, Whittaker/hypergeometric functions play an essential role in the analysis of the present paper. It would nevertheless be interesting to revisit the setting of \cite{psv23} with a matched asymptotics argument centered around Whittaker and hypergeometric functions, as in the present paper.

We refer to the upcoming work \cite{tdcheun} for a related matched-asymptotics analysis in the setting of more general Heun equations. 

\subsubsection{Step 3: Plancherel}
Combining Steps 1 and 2 amounts to summing up the inequalities \eqref{eq:schematiciedhighfreq} and \eqref{eq:schematiciedboundfreq}. We can freely apply Plancherel to the left-hand side of the summed estimate, but since the right-hand side has non-trivial frequency dependence, applying Plancherel to the right-hand side is less straightforward.

The resolution to this difficulty is the observation that the functions in the multipliers appearing in Step 1 can be chosen in such a way that they agree for all frequencies near the horizons up to higher order terms in $\rho_+$ or $\rho_c$, which can be dealt with by applying Young's inequality and absorbing the arising $u$-dependent terms into the left-hand side. In particular, this means that we require also multiplier estimates for the bounded frequencies of Step 2, but since we can use the estimate \eqref{eq:schematiciedboundfreq}, they simplify greatly.

After writing the estimates at the level of $u$ in a sufficiently uniform way, Plancherel can be obtained to derive the desired physical space integrated energy estimates.

\subsubsection{Step 4: verifying future integrability}
Since the final integrated energy estimates of Step 3 are only valid under the assumption of sufficient integrability of $\uppsi$ and $\xi G$ in $\tau$, which follows from \emph{future integrability} (in $\tau\in [0,\infty)$) of $\psi$ (and $G_A$), it remains to verify this assumption on $\psi$. Note that the assumption is \emph{qualitative} in nature, in contrast with the conclusion of Step 3, which was a \emph{quantitative} statement. In particular, we can restrict to fixed angular frequencies $\ell\in \N_0$ without loss of generality.

The idea is to use both $\kappa_+$ and $\kappa_c$ as \emph{regularizers}. That is to say, we can first consider the near-extremal Reissner--Nordstr\"om--de Sitter case, where $\kappa_+>0$ and $\kappa_c>0$ and derive future integrability, using crucially the strict positivity of the surface gravities. Then the quantitative estimates in Step 3, which are uniform $\kappa_+$ and $\kappa_c$, remain valid as $\kappa_+\downarrow 0$ and $\kappa_c\downarrow 0$, allowing us to conclude the validity of the estimates (and as a consequence, future integrability) in the $\kappa_c=\kappa_+=0$ case.

We derive future integrability in the $0<\kappa_c<\kappa_+$ case via a continuity argument in the charge parameter $\q$ (keeping $Q$ fixed). The starting point is future integrability for small $|\q|$; see \cite{bes20}[Theorem 4.2]. Alternatively, we could appeal to the integrated energy estimates for $\q=0$ \cite{gon24}[Proposition 11] from which future integrability for small $|\q|$ follows straightforwardly.

Assuming future integrability for $\q=\q_0$, we consider nearby $q_1$ and the time-dependent charge parameter $q_{\tau}=\xi_{\tau}\q_1+(1-\xi_{\tau}) \q_0$, where $\xi_{\tau}(\tau')$ is a smooth cut-off function that vanishes for $\tau'\geq \tau$ and equal to 1 for $\tau'\leq \tau-\delta$ with $\delta>0$. We then treat the terms involving $\q_{\tau}-\q_0$ as an inhomogeneity and apply integrated estimates for $\q=\q_0$. For sufficiently small $|\q_1-\q_0|$, this inhomogeneity can be absorbed into the left-hand side of the estimate, if we make use of the fact that $\kappa_c>0$ and $\kappa_+>0$. That is to say, we make use of \underline{red-shift estimates} near the horizon. This forms the key step towards extending the range of $\q$ for which future integrability holds.

\subsection{Comparison with wave equation on extremal Kerr}
\label{sec:compwavekerr}
After Fourier transforming in time and decomposing into spheroidal harmonic modes, integrated energy estimates in extremal Kerr are determined by an analysis at the level of an inhomogeneous radial ODE:
\begin{equation*}
	u''+(\omega^2-V_{m\ell \omega})u=H.
\end{equation*} 
We fix $|a|=M$ and $r_+=1$ and define $\tomega^2-\widetilde{V}=\omega^2-V_{m\ell\omega}$ with $\tomega=\omega-\frac{m}{2}$. Then we obtain the following asymptotic behaviour around $r=1$:
\begin{equation*}
	\widetilde{V}_{m\ell \omega}(r(r_*))= 2a m \tomega r_*^{-1}+\left[\Lambda_{m \ell}\left(\frac{m}{2}\right)-2m^2+O(m\tomega )\right]r_*^{-2}+\log r_* O(r_*^{-3}),
\end{equation*}
with $\Lambda_{m\ell}(\omega)=\lambda_{m\ell}(\omega)+a^2\omega^2$, where $\lambda_{m\ell}(\omega)$ are oblate  spheroidal eigenvalues.

In comparison, the potential $\widetilde{V}_{\ell}$ introduced in \S \ref{sec:introsketchfourier} has the following behaviour around $r=r_+=1$:
\begin{equation*}
	\widetilde{V}_{\ell \omega}(r)=2{\q}\tomega r_*^{-1}+[\ell(\ell+1)-\q^2+O(\tomega)]r_*^{-2}+\log r_* O(r_*^{-3})
	\end{equation*}
	
	Under the transformation $am\mapsto \q$ and $\Lambda_{m\ell}(\frac{m}{2})-2m^2\mapsto \ell(\ell+1)-\q^2$, the asymptotic behaviours of the potentials as $r_*\to -\infty$ and $\tomega\to 0$ agree.
	
	In fact, the analogue of the potential $\mathcal{V}_{\homega \ell}(r)$ in extremal Kerr has the same structure  as $\mathcal{V}_{\homega \ell}(r)$ in the setting of the charged scalar field, with the above identifications for bounded $r$.
	
The analysis concerning the radial ODE in a neighbourhood of $r=r_+$ developed in the present paper applies therefore also to the extremal Kerr case. Since we allow the constants in the relevant estimates to depend on $\q$, the analysis is only applicable in the setting of bounded $m$ and can be used to remove the conditional assumptions in the integrated energy estimates in \cite{gaj23}.

The improvement of the results in \cite{gaj23} as well as the consideration of unbounded $m$ is part of future work.

\subsection{Overview of the remainder of the paper}
Here, we provide an outline of the content of the remaining sections of the paper.
\begin{itemize}
	\item In \S \ref{sec:prelim} we provide all the necessary geometric preliminaries and we set notational conventions. We moreover derive different forms of \eqref{eq:CSFintro}.
	\item We give a precise statement of the main theorem of the paper in \S \ref{sec:precisemainthm}.
	\item The analysis in \S \ref{sec_smallqinten} concerns small $|\mathfrak{q}Q|$ and is purely physical-space based and independent from the Fourier analysis in the rest of the paper.
	\item In \S \ref{sec:iledfreq} we explain how to relate the analysis of \eqref{eq:CSFintro} to the analysis of an appropriate Schr\"odinger ODE, by taking a Fourier transform.
	\item We carry out the high-frequency analysis in \S \ref{sec:iledfreq}. 
	\item In \S \ref{sec:intestboundfreq}, we perform the bounded-frequency analysis. This section may be thought of as the core of the paper and it contains most of the new ideas.
	\item We put together the high- and bounded-frequency analysis in \S \ref{sec:iedcomb} and then apply Plancherel's theorem to obtain physical space estimates in \S \ref{sec:iedphy}.
	\item The appendix contains some results that can be read independently from the rest of the paper. In Appendix \ref{sec:purelyphysied} we provide a purely physical-space proof of integrated energy decay estimates, with zeroth order terms on the right-hand side. This result is also used in the small-$|\mathfrak{q}Q|$-analysis. In Appendix \ref{sec:restrmodestab}, we prove a restricted mode stability result away from extremality. This result is not used in the rest of the paper.
	\item Finally, Appendix \ref{sec:ODEest} contains ODE tools that are necessary for carrying out the matched asymptotic analysis in \S \ref{sec:intestboundfreq}.
\end{itemize}
\subsection{Acknowledgments}
 We thank Marc Casals for helpful discussions and for sharing relevant computations and Rita Teixeira da Costa for invaluable discussions on matched asymptotics in a wider context. The author acknowledges funding through the ERC Starting Grant 101115568.

\section{Preliminaries}
\label{sec:prelim}
In this section, we introduce the relevant geometric concepts, as well as the charged scalar field equation. We also briefly discuss local energy estimates and introduce notation that will be used throughout the remaining sections of the paper and is necessary to be able to give a precise statement of the main theorems in \S \ref{sec:precisemainthm}.
\subsection{Reissner--Nordstr\"om--de Sitter spacetimes}
\label{sec:rndsgeom}
Consider the parameters $(M,Q,\Lambda)\in \R_+\times \R \times[0,\infty)$ and the corresponding polynomial
\begin{equation*}
	r^2\Omega^2(r):=r^2-2Mr+Q^2-\frac{\Lambda}{3} r^4.
\end{equation*}
For $Q\neq 0$, $r^2\Omega^2(r)$ has up to three positive roots $r_-,r_+,r_c$, with
\begin{equation*}
0<r_-\leq r_+\leq r_c\leq \infty,
\end{equation*}
where we take the convention that $r_c=\infty$ if $\Lambda=0$. We will restrict to $M^2\Lambda$ sufficiently small, so that $r_+<r_c$. Later in the paper, we will restrict further to $M^2\Lambda\ll 1$.

Let $\mathcal{M}_{M,Q,\Lambda}=\R_v\times [r_+,r_c)_r\times \s^2_{(\theta,\varphi)}$ denote the \emph{Reissner--Nordstr\"om(--de Sitter) exterior manifold with boundary}.
The \emph{Reissner--Nordstr\"om(--de Sitter) exterior spacetimes} are tuples $(\mathcal{M}_{M,Q,\Lambda},g_{M,Q,\Lambda})$ (together with the time orientation determined by the vector field $\partial_v$), with $g_{M,Q,\Lambda}$ a Lorentzian metric on $\mathcal{M}_{M,Q,\Lambda}$ that can be expressed as follows:
\begin{align}
\label{eq:RNmetric}
g_{M,Q,\Lambda}=&-\Omega^2 dv^2+2dvdr+r^2\slashed{g}_{\s^2},\\ \nonumber
\slashed{g}_{\s^2}:=&\:d\theta^2+\sin^2\theta d\varphi^2,
\end{align}
where $v\in \R$ and $r\in [r_+,\infty)$ are called \emph{ingoing Eddington--Finkelstein coordinates}.\footnote{The spherical coordinate chart $(\theta,\varphi)$ feature a degeneration at a half-circle connecting the north and south pole and can be complemented with an additional coordinate chart $(\tilde{\theta},\tilde{\varphi})$ to cover the full sphere $\s^2$.}

When $\Lambda=0$, the spacetimes are called \emph{Reissner--Nordstr\"om} spacetimes. When $\Lambda>0$, they are called \emph{Reissner--Nordstr\"om--de Sitter} spacetimes.

The \emph{future event horizon} $\mathcal{H}^+$ is defined as the boundary:
\begin{equation*}
\mathcal{H}^+:=\partial \mathcal{M}_{M,Q,\Lambda}=\{r=r_+\}.
\end{equation*}

The corresponding inverse or dual metric is:
\begin{equation*}
g_{M,Q,\Lambda}^{-1}=\Omega^2\partial_r\otimes \partial_r+\partial_v\otimes \partial_r+\partial_r\otimes \partial_v+ r^{-2}\slashed{g}_{\s^2}^{-1},
\end{equation*}
with $\slashed{g}_{\s^2}^{-1}$ the inverse metric associated to the unit round sphere metric $\slashed{g}_{\s^2}$.

The \emph{event horizon surface gravity} $\kappa_+$ and the \emph{cosmological horizon surface gravity} $\kappa_c$ are defined as follows:
\begin{align*}
\kappa_+(Q,M,\Lambda):=\frac{1}{2}\frac{d\Omega^2}{dr}(r_+),\\
\kappa_c(Q,M,\Lambda):=-\frac{1}{2}\lim_{r\uparrow r_c}\frac{d\Omega^2}{dr}(r).
\end{align*}
From the lemma below, it follows that $\kappa_c\downarrow 0$ as $\Lambda\downarrow 0$. When $\kappa_+=0$, we will say the spacetime is \emph{extremal}.\footnote{There are another notions of extremality in the Reissner--Nordstr\"om--de Sitter context, which arise when $r_+=r_c$; see for example the discussion in \cite{gaj18}[\S1]. These spacetimes do not play a role in the present paper, since we assumed the strict inequality $r_+<r_c$.} In particular, when $\kappa_c=\kappa_+=0$, we have that $r_+=r_-$ and $r_c=\infty$ and the spacetimes are called \emph{extremal Reissner--Nordstr\"om} spacetimes. In that case, $|Q|=M$.

When $\kappa_+>0$ and $\kappa_c\geq 0$, we say the spacetimes are \emph{sub-extremal}. If in addition $0<M\kappa_+\ll 1$ and $0\leq M\kappa_c\ll 1$, we say the sub-extremal spacetimes are \emph{near-extremal}.

In the lemma below, we relate the parameters $(M,\Lambda,Q)$ to the surface gravities $(\kappa_+,\kappa_c)$.
\begin{lemma}
\label{lm:surfacegrav}
	\mbox{}\begin{enumerate}[label=\emph{(\roman*)}]
\item We can write:
	\begin{equation*}
		r^2\Omega^2(r)=\frac{\Lambda}{3}(r_c-r)(r-r_+)(r-r_-)(r+r_c+r_++r_-),
	\end{equation*}
	with
	\begin{align*}
		\frac{3}{\Lambda}=&\:r_c^2+r_+^2+r_-^2+r_+r_-+r_cr_++r_cr_-,\\
		2M\cdot \frac{3}{\Lambda}=&\:r_c^2(r_++r_-)+r_+^2(r_c+r_+)+r_-^2(r_c+r_-)+2r_cr_+r_-,\\
		Q^2\cdot \frac{3}{\Lambda}=&\:(r_c+r_++r_-)r_cr_+r_-.
	\end{align*}
	Furthermore, we can express:
	\begin{align*}
		M=&\:r_+\left(1-\kappa_+r_+-\frac{2 \Lambda}{3}r_+^2\right)=r_c\left(1-\kappa_c r_c-\frac{2 \Lambda}{3}r_c^2\right),\\
		Q^2=&\:r_+^2\left(1-2\kappa_+ r_+-\Lambda r_+^2\right)=r_c^2\left(1-2\kappa_c r_c-\Lambda r_c^2\right).
	\end{align*}
	\item We can express:
	\begin{align*}
		\kappa_+=&\:\frac{\Lambda}{6} r_+^{-2}(r_c-r_+)(r_+-r_-)(2r_++r_c+r_-),\\
			\kappa_c=&\:\frac{\Lambda}{6}r_c^{-2}(r_c-r_+)(r_c-r_-)(2r_c+r_++r_-).
	\end{align*}
	\item Fix $Q,M>0$. Then
	\begin{align*}
		\lim_{\Lambda\downarrow 0}\frac{\Lambda}{3}r_c^{2}=&\:1,\\
		\lim_{\Lambda\downarrow 0} r_c \kappa_c(Q,M,\Lambda)=&\:1,\\
		\lim_{\Lambda\downarrow 0} \kappa_+(Q,M,\Lambda)=&\:\frac{r_+-r_-}{2r_+^2}.
	\end{align*}
	\item We can expand:
\begin{align*}
	\frac{\Lambda}{3}r_c^2=&\:1-(r_++r_-)\kappa_c+r_+O(\kappa_c),\\
	\kappa_c r_c=&\: 1-\frac{3}{2}(r_++r_-)\kappa_c+r_+^2O(\kappa_c^2).
\end{align*}
		\end{enumerate}
\end{lemma}
\begin{proof}
Straightforward consequences of expanding the following identity, which uses that $r^2\Omega^2$ is a fourth-order polynomial:
\begin{equation*}
	C_0(r_c-r)(r-r_+)(r-r_-)(r-r_0)=r^2\Omega^2(r)=r^2-2Mr+Q^2-\frac{\Lambda}{3} r^4. \qedhere
\end{equation*}	
Furthermore, the expressions for $M$ and $Q$ in terms of $\kappa_+$, $r_+$ and $\Lambda$ follow from solving the system:
\begin{align*}
	\Omega^2(r_+)=&\:0,\\
	\frac{1}{2}\frac{d\Omega}{dr}(r_+)=&\:\kappa_+.
\end{align*}
Similarly, the expressions for $M$ and $Q$ in terms of $\kappa_c$, $r_c$ and $\Lambda$ follow from $\Omega^2(r_+)=0$ and $\frac{1}{2}\frac{d\Omega}{dr}(r_c)=-\kappa_c$.
\end{proof}

It will be convenient to characterize the Reissner--Nordstr\"om spacetimes via the triple $(r_+,\kappa_+,\kappa_c)$ instead of $(\Lambda,M,Q)$. In much of the analysis, we will also simplify the notation by fixing, without loss of generality, $r_+=1$.

\subsection{Notation: inequalities and Big-O notation}
\label{sec:Bignot}
We first introduce the following notation. Let $h:I\to \R$ be a non-polynomial function. We use $O_s(h)$ to group all functions $f: I\to \R$ satisfying the following estimate: let $s,l\in (0,\infty]$, with $l\leq s$. Then there exists a constant $C>0$ such that for all $x\in I$ and $l\leq s$:
\begin{equation*}
\left|\frac{d^lf}{dx^l}\right|(x)\leq C \left|\frac{d^lh}{dx^l}\right|(x).
\end{equation*}
We write $O(h)=O_0(h)$ and when the above estimate holds for any $s\in \N_0$, we write $O_{\infty}(h)$.

We extend the above notation to polynomial $h$ as follows. We write $O_{s}(x^{m})$ for $m\in \N_0$, if
\begin{equation*}
\left|\frac{d^lf}{dx^l}\right|(x)\leq Cx^{-m-l}
\end{equation*}
for all $l\leq s$. When $m=0$, we write $O_s(x^0)$ to group functions $f$ satisfying:
\begin{equation*}
\left|\frac{d^lf}{dx^l}\right|(x)\leq Cx^{-l}
\end{equation*}
for all $l\leq s$. 

When the constant $C$ is not entirely uniform and depends on an additional parameter $\lambda$, we will place the parameter in the superscript and write $O_s^{\lambda}(h)$.

Furthermore, given a constant $\alpha\in \C$ and $s\in \R$, we also use the notation $O(\alpha^s)$ to group \underline{constants} $\beta\in \C$ satisfying $|\beta|\leq C |\alpha|^s$, with $C$ independent of $\alpha$.

We will use $c,C$ to denote constants appearing in inequalities for functions $f$ that are \emph{uniform} in the following sense: $c,C$ only depend on $r_+$,  the choice of foliation defining function $\widetilde{\h}$, or equivalently $\h$, introduced in \S \ref{sec:foliations} and $\q=\mathfrak{q}Q$, the charge parameter introduced in \S \ref{sec:CSF}. If $c$ or $C$ have additional dependencies, we will denote them explicitly. To ease the notion in the paper, we will also adhere to the following convention (the  ``algebra of constants''):
\begin{equation*}
C\cdot C=C,\quad c\cdot c=c, \quad c^{-1}=C,\quad \quad C+C=C,\quad c+c=c,\quad C+c=C.
\end{equation*}
Furthermore, we will use the notation $f\lesssim h$ when $|f|\leq C |h|$, $f\gtrsim h$ when $|f|\geq c |h|$ and $f\sim h$ when $f\lesssim h$ and $f\gtrsim h$.

\subsection{Properties of $r^{-2}\Omega^2$}
In the lemma below, we will describe the asymptotic behaviour of $r^{-2}\Omega^2$ on near-extremal Reissner--Nordstr\"om--de Sitter spacetimes. For this purpose, it is convenient to introduce the following alternative radial coordinates:
\begin{align*}
	\rho_+:=&\: r_+^{-1}-r^{-1},\\
	\rho_c:=&\: r^{-1}-r_c^{-1}.
\end{align*}
Note that $\rho_+,\rho_c\in [0,r_+^{-1}-r_c^{-1})$.
\begin{lemma}
\label{lm:metricest}
The following estimates hold:
	\begin{align*}
	r^{-2}\Omega^2(r(\rho_+))=&\:2\kappa_+ \rho_++(1-6\kappa_+ r_+-2\Lambda r_+^2)\rho_+^2+(-2r_++6\kappa_+ r_+^2+\frac{8}{3}\Lambda r_+^3)\rho_+^3\\
	&+r_+^2\left(1-2\kappa_+ r_+-\Lambda r_+^2\right)\rho_+^4,\\
	\frac{d(r^{-2}\Omega^2)}{d\rho_+}(r(\rho_+))=&\: 2\kappa_+ +2\left[1-6r_+\kappa_+-2\Lambda r_+^2\right]\rho_++2(-3r_++9\kappa_+ r_+^2+4\Lambda r_+^3)\rho_+^2\\
	&+4r_+^2\left(1-2\kappa_+ r_+-\Lambda r_+^2\right)\rho_+^4,\\
	\frac{d^2(r^{-2}\Omega^2)}{d\rho_+^2}(r(\rho_+))=&\:2\left[1-6r_+\kappa_+-2\Lambda r_+^2\right]-4(3r_+-9\kappa_+ r_+^2-4\Lambda r_+^3)\rho_+\\
	&+12r_+^2\left(1-2\kappa_+ r_+-\Lambda r_+^2\right)\rho_+^4,\\
	r^{-2}\Omega^2(r(\rho_c))=&\:2\kappa_c \rho_c+(1+6\kappa_c r_c-2\Lambda r_c^2)\rho_c^2+(2r_c+6\kappa_c r_c^2-\frac{8}{3}\Lambda r_c^3)\rho_c^3\\
	&+r_+^2\left(1-2\kappa_+ r_+-\Lambda r_+^2\right)\rho_c^4\\
	=&\:2\kappa_c \rho_c+(1+O(\kappa_c))\rho_c^2+(-2r_++2\kappa_+ r_+^2+O(\kappa_c))\rho_c^3\\
	&+r_+^2\left(1-2\kappa_+ r_+-\Lambda r_+^2\right)\rho_c^4,\\
	\frac{d(r^{-2}\Omega^2)}{d\rho_c}(r(\rho_c))=&\: 2\kappa_c +2\left[1+6r_c\kappa_c-2\Lambda r_c^2\right]\rho_c+2(3r_c+9\kappa_c r_c^2\\
	&+4r_+^2\left(1-2\kappa_+ r_+-\Lambda r_+^2\right)\rho_c^4\\
	=&\:2\kappa_c +2(1+O(\kappa_c))\rho_c+3(-2r_++2\kappa_+ r_+^2+O(\kappa_c))\rho_c^2\\
	&+4r_+^2\left(1-2\kappa_+ r_+-\Lambda r_+^2\right)\rho_c^4,\\
	\frac{d^2(r^{-2}\Omega^2)}{d\rho_c^2}(r(\rho_c))=&\:2\left[1+6r_c\kappa_c-2\Lambda r_c^2\right]+4(3r_c+9\kappa_c r_c^2-4\Lambda r_c^3)\rho_c\\
	&+12r_+^2\left(1-2\kappa_+ r_+-\Lambda r_+^2\right)\rho_c^4\\
	=&\:2(1+O(\kappa_c))+4(-3r_++3\kappa_+ r_+^2+O(\kappa_c))\rho_c+12r_+^2\left(1-2\kappa_+ r_+-\Lambda r_+^2\right)\rho_c^4.
	\end{align*}
\end{lemma}
\begin{proof}
Follows from a Taylor expansion of $r^{-2}\Omega^2$ in $\rho_+$ or $\rho_c$, combined with the expressions for $\kappa_+$, $\kappa_c$ $M$ and $Q$ in Lemma \ref {lm:surfacegrav}.
\end{proof}

The function $\frac{d}{dr}\left(r^{-2}\Omega^2\right)$ has a single root $r=r_{\sharp}$ in $(r_+,r_c)$, with
\begin{equation*}
r_{\sharp}=\frac{3M}{2}+\frac{1}{2}\sqrt{9M^2-8Q^2}.
\end{equation*}
The hypersurface $\{r=r_{\sharp}\}$ is called the \emph{photon sphere}. Observe that it does not depend on $\Lambda$.

We define the \emph{tortoise coordinate} $r_*: (r_+,r_c)\to \R$ as the solution to the ODE:
\begin{align*}
	\frac{dr_*}{dr}=&\:\frac{1}{\Omega^2},\\
	r_*(r_{\sharp})=&\:0.
\end{align*}
We can write
\begin{equation*}
	\frac{r^2}{\Omega^2}(\rho_+)=\frac{1}{\rho_+(\rho_++2\kappa_+)}(1+O_{\infty}(r_+\rho_+)+\frac{r_+\kappa_c^2}{\kappa_+}O_{\infty}(r_+\rho_+)).
\end{equation*}
The function $r_*$ therefore satisfies the following asymptotic properties:
\begin{multline}
\label{eq:tortoiseevent1}
	r_*(r(\rho_+))=\int_{r_0}^{r(\rho_+)} \frac{1}{\Omega^2}(r')\,dr'=\int_{\rho_+(r_0)}^{\rho_+} \left(\frac{r^2}{\Omega^2}\right)(r(\rho_+'))\,d\rho_+'\\
	=\int_{\rho_+(r_0)}^{\rho_+}\frac{1}{\rho_+'(\rho_+'+2\kappa_+)}(1+O_{\infty}(r_+\rho_+')+\frac{r_+\kappa_c^2}{\kappa_+}O_{\infty}(r_+\rho_+')) \,d\rho_+'\\
	=\frac{1}{2\kappa_+}\log\left(\frac{\rho_+}{\rho_++2\kappa_+}\right)+\log(\rho_++2\kappa_+)\left[O_{\infty}((r_+\rho_+)^0)+\frac{r_+\kappa_c^2}{\kappa_+}O_{\infty}((r_+\rho_+)^0)\right].
\end{multline}
When $\kappa_c\leq \kappa_+<\rho_+$, we can expand further to obtain
\begin{equation}
\label{eq:tortoiseevent2}
	r_*(r(\rho_+))=-\rho_+^{-1}\left[1+O_{\infty}(\kappa_+ \rho_+^{-1})\right]+r_+\log(r_+\rho_+)O((r_+\rho_+)^0).
\end{equation}
When $\kappa_c\leq \kappa_+$ and $\rho_+<\kappa_+$, we can expand further to obtain
\begin{equation}
\label{eq:tortoiseevent3}
	r_*(r(\rho_+))=-\frac{1}{2\kappa_+}\left[\log(2\kappa_+\rho_+^{-1})+O_{\infty}(\kappa_+^{-1} \rho_+)^0\right].
\end{equation}
We repeat the above computations with the role of $\rho_+$ taken on by $\rho_c$ to obtain:
\begin{equation}
	\label{eq:tortoisecosmo1}
	r_*(r(\rho_c))=-\frac{1}{2\kappa_c}\log\left(\frac{\rho_c}{\rho_c+2\kappa_c}\right)+\log(\rho_c+2\kappa_c)O_{\infty}((r_+\rho_c)^0).
\end{equation}
When $\kappa_c<\rho_c$, we can expand further to obtain
\begin{equation}
\label{eq:tortoisecosmo2}
	r_*(r(\rho_c))=\rho_c^{-1}\left[1+O_{\infty}(\kappa_c \rho_c^{-1})\right]+r_+\log(r_+\rho_c )O((r_+\rho_+)^0).
\end{equation}
When $\rho_c<\kappa_c$, we can expand further to obtain
\begin{equation}
\label{eq:tortoisecosmo3}
	r_*(r(\rho_c))=\frac{1}{2\kappa_c}\left[\log(2\kappa_c\rho_c^{-1})+O_{\infty}(\kappa_c^{-1} \rho_+)^0\right]
\end{equation}
We conclude in particular that in all cases $r_*((r_+,r_c))=\R$.

Let $t=v-r_*$. Then the metric $g_{M,Q,\Lambda}$ takes the following form in the interior of $\mathcal{M}_{M,Q,\Lambda}$, $\mathring{\mathcal{M}}_{M,Q,\Lambda}=\R_v\times (r_+,r_c)_r\times \s^2$, with respect to $(t,r,\theta,\varphi)$ coordinates:
\begin{equation*}
g_{M,Q,\Lambda}=-\Omega^2 dt^2+\Omega^{-2}dr^2+r^2(d\theta^2+\sin^2\theta d\varphi^2).
\end{equation*}

\subsection{Foliations}
\label{sec:foliations}
We will foliate $\mathcal{M}_{M,Q,\Lambda}$ by spacelike hypersurfaces $\Sigma_{\tau}$, where $\tau\in \R$. To construct $\Sigma_{\tau}$, we first introduce the foliation-defining function $\widetilde{\mathbbm{h}}:[r_+,r_c)\to  [0,\infty)$, which is a smooth function satisfying:
\begin{equation*}
2-\widetilde{\mathbbm{h}}\Omega^2\geq 0.
\end{equation*}
We moreover define $\h: [r_+,r_c)\to  [0,\infty)$ as follows:
\begin{equation*}
\h(r):=2-\widetilde{\mathbbm{h}}\Omega^2(r).
\end{equation*}
The time function $\tau: \mathcal{M}_{M,Q,\Lambda}\to \R$ is then defined as follows: 
\begin{equation*}
\tau(v,r,\theta,\varphi)=v-\int_{r_+}^r \widetilde{\mathbbm{h}}(r')\,dr'.
\end{equation*}
We denote the level sets of constant $\tau$ by $\Sigma_{\tau}$. Since
\begin{equation*}
g^{-1}_{M,Q,\Lambda}(d\tau,d\tau)=-\widetilde{\h}\h\leq 0,
\end{equation*}
we have that $\Sigma_{\tau}$ must consist of spacelike or null segments. Furthermore, if $\h>0$ and $\widetilde{\h}>0$, then $\Sigma_{\tau}$ are smooth spacelike hypersurfaces for all $\tau\in \R$.

We record the following identity that we will appeal to later:
\begin{equation}
\label{eq:identityhtildeh}
\Omega^2\h \widetilde{\h}=1-(1-\widetilde{\h}\Omega^2)^2.
\end{equation}

We can express:
\begin{equation}
\label{eq:invmetrictaucoord}
g_{M,Q,\Lambda}^{-1}=- \mathbbm{h}\widetilde{\mathbbm{h}} \partial_{\tau}\otimes \partial_{\tau}+(1-\widetilde{\h}\Omega^2)(\partial_{\tau}\otimes \partial_{r}+\partial_r\otimes \partial_{\tau})+\Omega^2\partial_r\otimes \partial_r+r^{-2}\slashed{g}_{\s^2}^{-1}.
\end{equation}
By applying \eqref{eq:identityhtildeh}, it then follows that $\sqrt{-\det g_{M,Q,\Lambda}}=r^2\sin\theta$.

Note that $v|_{\Sigma_{\tau}}(r_+)=\tau$, so $\Sigma_{\tau}\cap \mathcal{H}^+=\{r=r_+\}\cap\{v=\tau\}$. Define $u(\tau,r):=v(\tau,r)-2r_*$. Then $u(\tau,r_{\sharp})=v(\tau,r_{\sharp})=\tau+\int_{r_+}^{r_{\sharp}} \widetilde{\mathbbm{h}}(r')\,dr'$. To conclude that the restriction $\left|u|_{\Sigma_{\tau}}\right|(r)$ is bounded as $r\to r_c$, we express:
\begin{equation*}
u(\tau,r)=\tau+\int_{r_{\sharp}}^r (\widetilde{\mathbbm{h}}-2\Omega^{-2})(r')\,dr'+\int_{r_+}^{r_{\sharp}} \widetilde{\mathbbm{h}}(r')\,dr'=\tau-\int_{r_{\sharp}}^r(\Omega^{-2}\h)(r')\,dr'+\int_{r_+}^{r_{\sharp}} \widetilde{\mathbbm{h}}(r')\,dr'.
\end{equation*}
Hence, $\lim_{r\to r_c}u(\tau,r)$ is well-defined if $\Omega^{-2}\h$ is integrable on $[r_{\sharp},r_c)$. For the sake of convenience, we will occasionally make one of the two following assumptions:
\begin{enumerate}
\item Either there exists $h_0>0$ such that for $r\geq r_{\sharp}$:
\begin{equation*}
\Omega^{-2}\h(r)=h_0 r^{-2}+O_{\infty}(r^{-3}),
\end{equation*}
\item Or there exist $r_+<r_H<r_I<r_c$, such that $\widetilde{\h}(r)=0$ when $r\leq r_H$ and  $\h(r)=0$ when $r\geq r_I$.
\end{enumerate}

Finally, it will be convenient to express the metric in $(u,x,\theta,\varphi)$ coordinates, with $x=\frac{1}{r}$:
\begin{equation*}
g=r^2(\Omega^2x^2du^2+2dudx+\slashed{g}_{\s^2}).
\end{equation*}
We can express with respect to $(u,x)$ coordinates: $\mathcal{M}_{M,Q,\Lambda}\cong \mathcal{H}^+\cup \R_u\times (0,r_+^{-1})_x\times \s^2$. Then we extend $\mathcal{M}_{M,Q,\Lambda}$ by defining:
\begin{equation*}
\widehat{\mathcal{M}}_{M,Q,\Lambda}:=\mathcal{H}^+\cup (\R_u\times [0,r_+^{-1})_x\times \s^2).
\end{equation*}
We denote the level set $\{x=0\}$ of $\widehat{\mathcal{M}}_{M,Q,\Lambda}$ as follows:
\begin{align*}
\mathcal{I}^+:=&\:\{x=0\}\subset \widehat{\mathcal{M}}_{M,Q,\Lambda}\quad (\Lambda=0),\\
\mathcal{C}^+:=&\:\{x=0\}\subset \widehat{\mathcal{M}}_{M,Q,\Lambda}\quad (\Lambda>0).
\end{align*}
We refer to $\mathcal{I}^+$ as \emph{future null infinity} and to $\mathcal{C}^+$ as the \emph{cosmological horizon}.

When integrating over the unit  round sphere $\s^2$, we will make use of the notation:
\begin{equation*}
d\sigma:=\sin\theta d\theta d\varphi.
\end{equation*}

\subsection{Couch--Torrence radial coordinates}
\label{sec:couchtorrencerad}
The following alternative radial coordinates will be convenient to make manifest symmetries between the event and cosmological horizons:
\begin{align*}
	s_+:=&\: \frac{1}{\rho_+}-r_+,\\
	s_c:=&\: \frac{1}{\rho_c}-r_+.
\end{align*}
Since $\rho_+,\rho_c\in [0,r_+^{-1}-r_c^{-1})$, we have that $s_+,s_c\in (\frac{r_+^2}{r_c-r_+},\infty)$. When $\kappa_c=0$, we have that $s_c=r-r_+$ and $s_+=\frac{r_+^2}{r-r_+}$.

Furthermore, we can express:
\begin{align}
\label{eq:scintermss+}
	r_+^{-1}s_c=&\:\frac{1+r_+r_c^{-1}(r_+^{-1}s_++r_+r_c^{-1}) }{(1-r_+r_c^{-1})r_+^{-1}s_++r_+^2 r_c^{-2}},\\
	\label{eq:s+intermssc}
	r_+^{-1}s_+=&\:\frac{1+(r_+ r_c^{-1})^{2}(1-r_+^{-1}s_c)}{(1-r_+r_c^{-1})r_+^{-1}s_c-r_+r_c^{-1}}.
\end{align}

We will refer to $s_+$ and $s_c$ as \emph{Couch--Torrence radial coordinates} and the map $s_c\mapsto s_+$ as a \emph{Couch--Torrence transformation}, see \cite{couch}. This map is equivalent to $\rho_c\mapsto \rho_+$.

As the lemma below shows, the map $s_c\mapsto s_+$ plays a special role in the extremal case $\kappa_c=\kappa_+=0$.

\begin{lemma}
\label{lm:couchtorr}
Let $\kappa_c=\kappa_+=0$ and $r_+=1$. Consider the functions $r, \Omega, r_*: [0,\infty)\to [1,\infty)$, which are defined as follows: $r(s_c):=s_c+1$, $\Omega(s_c):=(s_c^{-1}+1)^{-1}$ and $r_*(s_c):=\int_{1}^{s_c}\frac{1}{\Omega^2(s_c)}\,ds_c$. Then
\begin{align}
\label{eq:couchtorr1}
\Omega(s_+)=&\:\frac{1}{r(s_c)},\\
\label{eq:couchtorr2}
\Omega(s_c)=&\:\frac{1}{r(s_+)},\\
\label{eq:couchtorr3}
\left(r^{-1}\Omega^2\frac{d\Omega^2}{dr}\right)(s_c)=&\:\left(r^{-1}\Omega^2\frac{d\Omega^2}{dr}\right)(s_+),\\
\label{eq:couchtorr1b}
r_*(s_c)=&\:-r_*(s_+(s_c)).
\end{align}
\end{lemma}
\begin{proof}
The identities \eqref{eq:couchtorr1} and \eqref{eq:couchtorr2} are immediate consequences of the identity $s_+=s_c^{-1}$. The identity \eqref{eq:couchtorr3} similarly follows straightforwardly.

We then use that
\begin{equation*}
	r_*(s_c)=-\int_{1}^{s_+(s_c)}\frac{1}{\Omega^2(s_c(s_+))}s_+^{-2}\,ds_+=-\int_{1}^{s_+(s_c)}\frac{1}{\Omega^2(s_+)}\,ds_+
\end{equation*}
to conclude \eqref{eq:couchtorr1b}.
\end{proof}
In fact, from the above lemma, it can easily be shown that the map $(t,s_c,\theta,\varphi)\mapsto (t,s_+,\theta,\varphi)$ is a conformal isometry of $(\mathring{\mathcal{M}}_{M,M,0},g_{M,M,0})$.

\subsection{Key vector fields}
The following vector fields will play key roles in the analysis in the remainder of the article. With respect to $(v,r,\theta,\varphi)$ coordinates:
\begin{align*}
T:=&\:\partial_v,\\
X:=&\: \partial_r+\widetilde{\h}\partial_v,\\
\Lbar:=&\: -\frac{\Omega^2}{2}\partial_r,\\
L:=&\: T-\Lbar,\\
Y_*:=&L-\underline{L}=\Omega^2X+(1-\widetilde{\h}\Omega^2)T.
\end{align*}
With respect to $(\tau,r,\theta,\varphi)$ coordinates, we have that $T=\partial_{\tau}$ and $X=\partial_r$.

We will also denote for $\q\in \R$:
\begin{align*}
K_+:=&\: T+i\q r_+^{-1}\mathbf{1},\\
K_c:=&\: T+i\q r_c^{-1}\mathbf{1}.
\end{align*}
where we will fix $\q$ via \eqref{def:q} in \S \ref{sec:CSF}.

\subsection{Frequency decompositions on $\s^2$}
We can decompose any function $f\in L^2(\s^2;\C)$ as follows:
\begin{equation*}
f=\sum_{\ell\in \N_0}\sum_{m\in \Z, |m|\leq \ell} f_{\ell m}Y_{\ell m},
\end{equation*}
with $f_{\ell m}\in \C$ and $Y_{\ell m}\in L^2(\s^2)$ the $(\ell,m)$-th spherical harmonics, which satisfy $\slashed{\Delta}_{\s^2}Y_{\ell m}=-\ell(\ell+1)Y_{\ell m}$ and $\la Y_{\ell m},Y_{\ell' m'}\ra_{L^2(\s^2)}=\delta_{\ell \ell'}\delta_{mm'}$.

We moreover denote:
\begin{align*}
f_{\ell}=&\:\sum_{m\in \Z, |m|\leq \ell} f_{\ell m}Y_{\ell m},\\
f_{\geq \ell}=&\:\sum_{\ell'=\ell}^{\infty}f_{\ell'}.
\end{align*}

\subsection{Charged scalar field equation}
\label{sec:CSF}
The (electrically charged) Reissner--Nordstr\"om--de Sitter solution to the Einstein--Maxwell equations are pairs $(\mathcal{M}_{M,Q,\Lambda},g_{M,Q,\Lambda},F_Q)$, with $(\mathcal{M}_{M,Q,\Lambda},g_{M,Q,\Lambda})$ defined in \S \ref{sec:rndsgeom} and
\begin{equation*}
F=-\frac{Q}{r^2}dv\wedge dr.
\end{equation*}
The constant $Q\in \R$ can be interpreted as the total electrical charge of the Reissner--Nordstr\"om spacetime.\footnote{More generally, one may consider electromagnetically charged Reissner--Nordstr\"om--de Sitter solutions $(\mathcal{M}_{M,e,\Lambda},g_{M,e,\Lambda},F_{Q,P})$, with $P\in \R$, $e=\sqrt{Q^2+P^2}$ and $F_{Q,P}=-\frac{Q}{r^2}dt\wedge dr+P\sin\theta d\theta\wedge d\varphi$. The analysis of the present paper can be generalized to that setting, with the key difference being the replacement of spherical harmonics $Y_{\ell,m}$ with so-called ``monopole harmonics'' $\widetilde{Y}_{j,m,P\mathfrak{q}}$, with eigenvalues $\lambda_j=j(j+1)-(\mathfrak{q}P)^2$, $j\in \N_0+|\mathfrak{q}P|$, $m\in \Z$, $|m|\leq j$, provided that the scalar field charge $\mathfrak{q}$ satisfies $\mathfrak{q}P\in \frac{1}{2}\Z$ (Dirac quantization). }

We define the following 1-form on $\mathring{\mathcal{M}}_{M,e}$: $\widetilde{A}=-\frac{Q}{r}dt$ and observe that $F=d\widetilde{A}$. Note that the 1-form $\widetilde{A}$ is not well-defined on $\mathcal{H}^+$. Define:
\begin{equation*}
\widehat{A}:=-\frac{Q}{r}d\tau.
\end{equation*}
The 1-form $\widehat{A}$ is well-defined everywhere on $\mathcal{M}_{M,Q,\Lambda}$ and also on the extended manifold $\widehat{M}_{M,Q,\Lambda}$. We can relate  $\widetilde{A}$ and $\widehat{A}$ as follows:
\begin{equation*}
\widehat{A}=\widetilde{A}-\frac{Q}{r}(\Omega^{-2}-\widetilde{\mathbbm{h}})dr=\widetilde{A}-d\left(\int_{r_{\sharp}}^r\frac{Q}{r'}(\Omega^{-2}-\widetilde{\mathbbm{h}})(r')\,dr'\right),
\end{equation*}
so we also have that $d\widehat{A}=F$. In the remainder of the article, we will assume that $P=0$.

We will refer to a choice of 1-form $A$ satisfying $F=dA$ as an \emph{electromagnetic gauge}. Given an electromagnetic gauge $A$, we define the following linear operator in terms of the Levi--Civita covariant derivative $\nabla$:
\begin{equation*}
^AD:=\nabla-i\mathfrak{q} A\otimes (\cdot) ,
\end{equation*}
with $\mathfrak{q}\in \R$.

Given a vector field $Y$, we moreover write $^AD_{Y}:=\nabla_Y-i\mathfrak{q} A(Y)\mathbf{1}$.

Given a tensor field $T$, we moreover apply the notational convention: $(^AD_{\mu}) (^AD_{\nu})T:=((^AD)^2T)_{\mu \nu}$.

The \emph{inhomogeneous charged scalar field equation} with inhomogeneity $G_A$ with respect to $(g_{M,Q,\Lambda},A)$ is then defined as follows:
\begin{equation}
\label{eq:CSF}
(g^{-1}_{Q,M,\Lambda})^{\mu \nu}(^AD_{\mu}) (^AD_{\nu})\phi-\frac{2\Lambda}{3}\phi=G_A.
\end{equation}
In the context of \eqref{eq:CSF}, we refer to $\mathfrak{q}$ as the \emph{scalar field charge parameter}. Note that by definition of $D$ and $A$, $\mathfrak{q}$ has units of inverse length and $Q$ has units of length, so the product 
\begin{equation}
\label{def:q}	
\q:=  \mathfrak{q}\cdot Q
\end{equation}
is dimensionless.

We introduce the parameter $\beta_{\ell}\in [0,\infty)\cup i(0,\infty)$, which will play a vital role in the bounded frequency analysis in Fourier space and appears also in the late-time asymptotics of \cite{gaj26b}. Let $\ell\in \N_0$, then
\begin{equation}
	\beta_{\ell}:=\sqrt{(2\ell+1)^2-4q^2}=\sqrt{1+4\ell(\ell+1)-4q^2}.
\end{equation}

Suppose $A'=A+df$ for some $f\in C^{\infty}(\mathring{\mathcal{M}}_{M,Q,\Lambda})$ and $\phi'=e^{i\mathfrak{q}f}\phi$. Then for any vector field $Y$, we have that:
\begin{equation*}
^{A'}D\phi'=e^{i \mathfrak{q}f}(^AD\phi).
\end{equation*}
Hence $\phi$ is a solution to \eqref{eq:CSF} if and only if $\phi'$ satisfies:
\begin{equation}
\label{eq:csfgaugetransf}
(g^{-1}_{M,Q,\Lambda})^{\mu \nu}(^{A'}D)_{\mu}(^{A'}D)_{\nu}\phi'-\frac{2\Lambda}{3}\phi'=e^{i \mathfrak{q}f}G_A=:G_{A'}.
\end{equation}

We refer to $\phi'$ as a \emph{gauge transformation} of $\phi$.

Let $\gamma\in \N_0^3$. We introduce the shorthand notation:
\begin{equation}
\label{eq:defDZgamma}
D_{\mathbf{Z}}^{\gamma}:=(^{A}\slashed{D}_{\s^2})^{\gamma_1}(r\Omega (^{A} D_X))^{\gamma_2}(^{A}D_T)^{\gamma_3},
\end{equation}
with $^{A}\slashed{D}_{\s^2}=\snabla_{\s^2}-i\mathfrak{q} \slashed{A}_{\s^2}\otimes (\cdot)$ and $ \slashed{A}_{\s^2}=A^{\theta}d\theta+A^{\varphi}d\varphi$. 

In the remainder of the paper, we will assume without loss of generality that $ \slashed{A}_{\s^2}\equiv 0$, so that $\slashed{D}_{\s^2}=\snabla_{\s^2}$.

We will also denote:
\begin{equation}
\label{eq:defZgamma}
\mathbf{Z}^{\gamma}:=\snabla_{\s^2}^{\gamma_1}(r\Omega X)^{\gamma_2}T^{\gamma_3}.
\end{equation}

In the remainder of the article, we will mainly work with the quantity
\begin{equation*}
\psi:= r\cdot \phi
\end{equation*}
instead of $\phi$, as it appears more naturally in the relevant energies and equations.

The following proposition provides more explicit expressions for \eqref{eq:CSF} with respect to $\widehat{A}$ in terms $\psi$:
\begin{proposition}
Let $\phi$ be a solution to \eqref{eq:CSF} with $A=\widehat{A}$. Then $\psi$ satisfies:
\begin{multline}
\label{eq:maineqradfield}
rG_{\widehat{A}}=X(\Omega^2X\psi)+r^{-2}\slashed{\Delta}_{\s^2}\psi-\h \widetilde{\h} T^2\psi-2(1-\h)(T+i\q r^{-1})X\psi\\
+\left(\frac{d\h}{dr}-2i \q r^{-1} \h\widetilde{\h}\right)T\psi-\left[\frac{d\Omega^2}{dr} r^{-1}+\frac{2\Lambda}{3}-i\q r^{-1}\frac{d\h}{dr}-\q^2  \mathbbm{h}\widetilde{\mathbbm{h}}r^{-2}-i \q r^{-2}(1-\h)\right]\psi\\
=X(\Omega^2X\psi)+r^{-2}\slashed{\Delta}_{\s^2}\psi-\h \widetilde{\h} K_c^2\psi-2(1-\h)XK_c\psi+2i \q(1-\h)\rho_c X\psi\\
+\left(\frac{d\h}{dr}-2i \q \rho_c \h\widetilde{\h}\right)K_c\psi-\left[\frac{d\Omega^2}{dr} r^{-1}+\frac{2\Lambda}{3}-i \q \rho_c \frac{d\h}{dr}-\q^2  \mathbbm{h}\widetilde{\mathbbm{h}}\rho_c^2-i \q r^{-2}(1-\h)\right]\psi.
\end{multline}
In particular, for $\widetilde{\h}(r)=\Omega^{-2}$, we have that $\h\equiv 1$, $\widehat{A}=\widetilde{A}$ and:
\begin{equation}
\label{eq:maineqradfieldtilde}
r\Omega^2 G_{\widetilde{A}}=Y_*^2\psi- K_c^2\psi+r^{-2}\Omega^2\slashed{\Delta}_{\s^2}\psi-2i \q \rho_c K_c\psi+\left[\q^2\rho_c^{2}- r^{-1}\Omega^2\frac{d\Omega^2}{dr} -\frac{2\Lambda}{3}\right]\psi.
\end{equation}

With respect to the coordinate chart $(\tau,\rho_c,\theta,\varphi)$ we moreover obtain:
\begin{multline}
\label{eq:maineqradfieldconf}
r^2(rG_{\widehat{A}})=\partial_{\rho_c}(\Omega^2r^{-2}\partial_{\rho_c}\psi)+\slashed{\Delta}_{\s^2}\psi-r^2\h \widetilde{\h} K_c^2\psi+2(1-\h)\partial_{\rho_c}K_c\psi+2i \q(1-\h)\rho_c \partial_{\rho_c}\psi\\
+\left(r^2\frac{d\h}{dr}-2i \q \rho_c r^2\h\widetilde{\h}\right)K_c\psi-\left[r\frac{d\Omega^2}{dr}+\frac{2\Lambda}{3}r^2 -i \q \rho_cr^2 \frac{d\h}{dr}-\q^2  \rho_c^2r^2\mathbbm{h}\widetilde{\mathbbm{h}}-i \q (1-\h)\right]\psi
\end{multline}
and with respect to the coordinate chart $(\tau,\rho_+,\theta,\varphi)$:
\begin{multline}
\label{eq:maineqradfieldconfhor}
r^2(rG_{\widehat{A}})=\partial_{\rho_+}(\Omega^2r^{-2}\partial_{\rho_+}\psi)+\slashed{\Delta}_{\s^2}\psi-r^2\h \widetilde{\h} K_+^2\psi-2(1-\h)\partial_{\rho_+}K_+\psi+ 2i\q(1-\h)\rho_+ \partial_{\rho_+}\psi\\
+\left(r^2\frac{d\h}{dr}+2i \q \rho_+ r^2\h\widetilde{\h}\right)K_*\psi-\left[r\frac{d\Omega^2}{dr}+\frac{2\Lambda}{3}r^2 +i\q  \rho_+r^2 \frac{d\h}{dr}-\q^2 \rho_+^2r^2\mathbbm{h}\widetilde{\mathbbm{h}}-i \q(1-\h)\right]\psi.
\end{multline}
\end{proposition}
\begin{proof}
First, we expand:
\begin{multline*}
(g^{-1})_{M,Q,\Lambda}^{\mu\nu}(^A D)_{\mu}(^AD)_{\nu}\phi=\square_{g_{M,Q,\Lambda}}\phi-2i \mathfrak{q} A^{\mu}\partial_{\mu}\phi+\left[i\mathfrak{q}\nabla_{\mu}A^{\mu}-\mathfrak{q}^2 A^{\mu}A_{\mu}\right]\phi\\
=\square_{g_{M,Q,\Lambda}}\phi-2i \mathfrak{q} A^{\mu}\partial_{\mu}\phi+\left[-i\mathfrak{q}\frac{1}{\sqrt{-\det g_{M,Q,\Lambda}}}\partial_{\mu}(\sqrt{-\det g_{M,Q,\Lambda}}A^{\mu})-\mathfrak{q}^2 A^{\mu}A_{\mu}\right]\phi.
\end{multline*}
Then we take $A=\widehat{A}$ and we use that $\sqrt{-\det g_{M,Q,\Lambda}}=r^2\sin\theta$ and that the dual vector field to $\widehat{A}$ takes the following form with respect to $(\tau,r,\theta,\varphi)$ coordinates,
\begin{equation*}
\widehat{A}^{\sharp}:=\frac{Q}{r} \mathbbm{h}\widetilde{\mathbbm{h}}\partial_{\tau}+\frac{Q}{r}(1-\h)\partial_r
\end{equation*}
to obtain:
\begin{multline*}
(g^{-1}_{M,Q,\Lambda})^{\mu\nu}(^{A}D)_{\mu}(^{A})D_{\nu}\phi=\square_{g_{M,Q,\Lambda}}\phi-2i  \mathfrak{q}  Q r^{-1}\h\widetilde{\h}\partial_{\tau}\phi-2i  \mathfrak{q}  Q r^{-1}(1-\h)\partial_{r}\phi\\
+\left[-i \mathfrak{q} Qr^{-2}\partial_{r}( r(1-\h))+ \mathfrak{q} ^2Q^2 r^{-2} \mathbbm{h}\widetilde{\mathbbm{h}}\right]\phi\\
=\square_{g_{M,Q,\Lambda}}\phi-2i  \mathfrak{q}  Q r^{-2}\h\widetilde{\h}\partial_{\tau}\psi-2i  \mathfrak{q}  Q r^{-2}(1-\h)\partial_{r}\psi+\left[i \mathfrak{q} Qr^{-2}\frac{d\h}{dr}+\mathfrak{q}^2Q^2  \mathbbm{h}\widetilde{\mathbbm{h}}r^{-3}+i  \mathfrak{q}  Qr^{-3}(1-\h)\right]\psi.
\end{multline*}
We conclude that \eqref{eq:maineqradfield} holds by writing $\q=\mathfrak{q}Q$ and combining
\begin{equation*}
r\square_{g_{M,Q,\Lambda}}\phi=\frac{r}{\sqrt{-\det g_{M,Q,\Lambda}}}\partial_{\alpha}\left(\sqrt{-\det g_{M,Q,\Lambda}} (g^{-1}_{M,e})^{\alpha \beta}\partial_{\beta}(r^{-1}\psi)\right)
\end{equation*}
with \eqref{eq:invmetrictaucoord} to obtain:
\begin{equation*}
r\square_{g_{M,Q,\Lambda}}\phi=X(\Omega^2X\psi)+r^{-2}\slashed{\Delta}_{\s^2}\psi-\h \widetilde{\h} T^2\psi-2(1-\h)TX\psi+\frac{d\h}{dr}T\psi-r^{-1}\frac{d\Omega^2}{dr}\psi
\end{equation*}
and writing $T=K_c-i \q r_c^{-1}$.

To obtain \eqref{eq:maineqradfieldconf}, we use that $\partial_{\rho_c}=-r^2X$. Equation \eqref{eq:maineqradfieldconfhor} follows from \eqref{eq:maineqradfield} by using that $\partial_{\rho_+}=r^2X$ and writing $T=K_+-i\q r_+^{-1}$.
\end{proof}

\subsection{Local estimates for the charged scalar field equation}
In this section, we state basic local energy estimates for solutions to \eqref{eq:CSF}.
\begin{theorem}[Existence and uniqueness for the charged scalar field equation]
\label{thm:gwp}
Let $A\in \Omega^1(\mathcal{M}_{M,Q,\Lambda})$. Consider the initial data pair:
\begin{equation*}
(\upphi,T\upphi)\in C^{\infty}(\Sigma_0)\times C^{\infty}(\Sigma_0)
\end{equation*}
and assume that $G_A\in C^{\infty}(\mathcal{M}_{M,Q,\Lambda})$.

Then there exists a unique solution $\phi\in C^{\infty}(\mathcal{M}_{M,Q,\Lambda})$ to \eqref{eq:CSF} with respect to $A$ and $G_A$, such that $(\phi|_{\Sigma_0}, T\phi|_{\Sigma_0})=(\upphi,T\upphi)$.

If we take $A=\widehat{A}$ and assume that $G_{\widehat{A}}\in C^{\infty}(\widehat{\mathcal{M}}_{M,Q,\Lambda})$, then $\psi=r\phi\in C^{\infty}(\widehat{\mathcal{M}}_{M,Q,\Lambda})$.
\end{theorem}
\begin{proof}
This follows from standard local existence and uniqueness for linear wave equations and relies on the existence of the energy estimates in Theorem \ref{thm:localenest} with $p=2$. In the case that $A=\widehat{A}$, we can additionally commute with $(r^2X)^k$ to derive the desired regularity on the extended manifold $\widehat{\mathcal{M}}_{M,Q,\Lambda}$.
\end{proof}

Let $p\in \R$. We define the \emph{energy densities} $\mathcal{E}_p[\psi]$ as follows:
\begin{equation*}
\mathcal{E}_p[\psi]:= (\Omega^{-1}r)^{p}\Omega^{2}|^{A}D_X\psi|^2+(\Omega^{-1}r)^{\min\{p,0\}}(\h\widetilde{\h}|^{A}D_T\psi|^2+r^{-2}|^{A}\slashed{D}_{\s^2}\psi|^2+r^{-2}|\psi|^2).
\end{equation*}
Given $\gamma=(\gamma_1,\gamma_2,\gamma_3)\in \N_0^3$ consider the operators $\mathbf{Z}^{\gamma}$ and $D_{\mathbf{Z}}^{\gamma}$ as defined in \eqref{eq:defDZgamma}. Then we define the following higher-order energy densities:
\begin{align*}
\mathcal{E}_p[\mathbf{Z}^{\gamma}\psi]:=&\:(\Omega^{-1}r)^{p}\Omega^{2}|^{A}D_X \mathbf{Z}^{\gamma}\psi|^2+(\Omega^{-1}r)^{\min\{p,0\}}\left(\h\widetilde{\h}|^{A}D_T\mathbf{Z}^{\gamma}\psi|^2+r^{-2}|^{A}\slashed{D}_{\s^2}\mathbf{Z}^{\gamma}\psi|^2+r^{-2}|\mathbf{Z}^{\gamma}\psi|^2\right),\\
\mathcal{E}_p[D_{\mathbf{Z}}^{\gamma}\psi]:=&\:(\Omega^{-1}r)^{p}\Omega^{2}|^{A}D_X D_\mathbf{Z}^{\gamma}\psi|^2+(\Omega^{-1}r)^{\min\{p,0\}}\left(\h\widetilde{\h}|^{A}D_TD_\mathbf{Z}^{\gamma}\psi|^2+r^{-2}|^{A}\slashed{D}_{\s^2}D_\mathbf{Z}^{\gamma}\psi|^2+r^{-2}|D_\mathbf{Z}^{\gamma}\psi|^2\right).
\end{align*}
We refer to integrals with respect to the volume form $d\sigma dr$
\begin{equation*}
\int_{\Sigma_{\tau}}\mathcal{E}_p[\mathbf{Z}^{\gamma}\psi]\,d\sigma dr
\end{equation*}
as (higher-order, weighted) \emph{energies}. Note that
\begin{equation*}
\int_{\Sigma_{\tau}}\mathcal{E}_p[\mathbf{Z}^{\gamma}\psi]\,d\sigma dr\sim \int_{\Sigma_{\tau}}\mathcal{E}_p[\mathbf{Z}^{\gamma}\psi]r^{-2}\,d\mu_{\Sigma_{\tau}},
\end{equation*}
with $d\mu_{\Sigma_{\tau}}$ a natural choice of induced volume form on $\Sigma_{\tau}$.

\begin{theorem}[Local energy estimates for the charged scalar field equation]
\label{thm:localenest}
Let $\phi$ be a solution to \eqref{eq:CSF} with $A=\widehat{A}$. Let $0\leq p\leq 2$. For all $0\leq \tau_A\leq \tau_B<\infty$ and $N\in \N_0$, there exists a constant $C_{\tau_b-\tau_a}=C_{\tau_b-\tau_a}(r_+,q,\widetilde{\h},p,N,\tau_b-\tau_a)>0$ such that:
\begin{equation}
\label{eq:localen1}
\sum_{|\gamma|\leq N}\int_{\Sigma_{\tau_b}} \mathcal{E}_p[D_\mathbf{Z}^{\gamma}\psi]\,d\sigma dr\leq C_{\tau_b-\tau_A}\left[\sum_{|\gamma|\leq N}\int_{\Sigma_{\tau_a}} \mathcal{E}_p[D_{\mathbf{Z}}^{\gamma}\psi]\,d\sigma dr+\int_{\tau_a}^{\tau_b}\int_{\Sigma_{\tau}}r^{-2}|D_{\mathbf{Z}}^{\gamma}(r^3G_{A})|^2\,d\sigma dr d\tau\right].
\end{equation}
\end{theorem}
\begin{proof}
The proof is a standard local-in-time energy estimate. Without loss of generality, we derive \eqref{eq:localen1} with $A=\widehat{A}$ by multiplying both sides of \eqref{eq:maineqradfield} with $T\overline{\psi}$ and $(r^{-1}\Omega)^{-p}\Omega^2X\overline{\psi}$, integrating by parts, applying a Gr\"onwall inequality to estimate the spacetime integrals in terms on the initial energy flux. After multiplying $rG_{A}$ with $r^2$, we can commute straightforwardly with $\mathbf{Z}^{\gamma}$.
\end{proof}

\begin{remark}
By combining Theorem \ref{thm:gwp} and Theorem \ref{thm:localenest} with a standard density argument, we have global uniqueness and existence of solutions arising from more general initial data that satisfy only $\int_{\Sigma_{0}}\mathcal{E}_p[\mathbf{Z}^{\gamma}\psi]\,d\sigma dr<\infty$ for some $p\in [0,2]$.
\end{remark}

\subsection{Cut-off and bump functions}
\label{sec:cutoffs}
Let $R_i\in (r_+,r_c)$ and $r_i\in (1,\infty)$, where $i\in \N_0$ and $r_3<r_2<r_1$, $R_1>R_2>R_3$. Assume moreover that $2\rho_c(R_i)<\rho_c(r_+)$ and $2\rho_+(r_i)<\rho_+(r_c)$. We define the corresponding cut-off functions $\chi_{R_i}, \chi_{r_i}: [1,r_c)\to [0,1]$ as as smooth functions satisfying:
\begin{equation*}
\chi_{R_i}(r)= \begin{cases}
1\quad \textnormal{for $\rho_c(r)\leq \rho_c(R_i)$},\\
0 \quad \textnormal{for $\rho_c(r)\geq 2\rho_c(R_i)$},
\end{cases}
\end{equation*}
with $0\leq \frac{d\chi_{R_i}}{dr}=-r^{-2}\frac{d\chi_{R_i}}{d\rho_c}\leq C R_i^{-2}(\rho_c(R_i))^{-1}$ and  $\left|\frac{d^2\chi_{R_i}}{dr^2}\right|\leq C R_i^{-4}(\rho_c(R_i))^{-2}$ for for some constant $C>0$, and
\begin{equation*}
\chi_{r_i}(r)= \begin{cases}
1\quad \textnormal{for $\rho_+(r)\leq \rho_+(r_i)$},\\
0 \quad \textnormal{for $\rho_+(r)\geq 2\rho_+(r_i)$},
\end{cases}
\end{equation*}
with $0\leq -\frac{d\chi_{r_i}}{dr}=-r^{-2}\frac{d\chi_{r_i}}{d\rho_+}\leq C r_i^{-2}(\rho_+(r_i))^{-1}$ and  $\left|\frac{d^2\chi_{r_i}}{dr^2}\right|\leq C r_i^{-4}(\rho_+(r_i))^{-2}$ for for some constant $C>0$.

We will also use that $\supp \left(\frac{d\chi_{R_i}}{dr}\right)\subset \left(\frac{1}{2-r_c^{-1}R_i}R_i,R_i\right)$ and $\supp \left(\frac{d\chi_{r_i}}{dr}\right)\subset \left(r_i,\frac{1}{2-r_+^{-1}r_i}r_i\right)$.

Let $\eta>0$ be suitably small. We define $\chi_K,\chi_T: [r_+,r_c)\to \R$ as smooth cut-off functions satisfying the following properties: $\chi_K(r)=1$ for $r\leq r_{\sharp}-\eta$, $\chi_K(r)=0$ for $r\geq r_{\sharp}+\eta$ and $\frac{d\chi_K}{dr}\leq 0$, $|\frac{d\chi_K}{dr}|\leq C\eta^{-1}$ for some numerical constant $C>0$. Similarly, $\chi_T(r)=1$ for $r\geq r_{\sharp}+\eta$, $\chi_T(r)=0$ for $r\leq r_{\sharp}-\eta$, $\frac{d\chi_T}{dr}\geq 0$ and $|\frac{d\chi_T}{dr}|\leq C\eta^{-1}$.

Finally, for $\eta>0$ suitably small, we define $\upzeta: [r_+,r_c)\to [0,1]$ to be a smooth function satisfying $\upzeta(r)=1$ for $r\notin (r_{\sharp}-2\eta,r_{\sharp}+2\eta)$ and $\upzeta(r)=0$ for $r\in (r_{\sharp}-\eta,r_{\sharp}+\eta)$.

\section{Precise statement of the main theorem}
\label{sec:precisemainthm}
In this section, we give a precise version of Theorem \ref{thm:intromain}.
\begin{theorem}
\label{thm:main}
Let $\phi$ be a solution to \eqref{eq:CSF} and let $\psi=r\phi$. Assume that $\kappa_+\geq \kappa_c\geq 0$, $\kappa_c\leq \kappa_1$ and $\q_0\leq |\q|\leq \q_1$ with $\q_0=0$ when $\kappa_c=0$. Let $\epsilon,\delta,\eta,q_0,\kappa_1>0$, $N,\ell\in \N_0$ and $r_+<r_H<r_I<r_c$. Consider $1< p<\min\{1+\re \beta_{\ell},2\}+\epsilon$. Then, for $\kappa_+\leq \kappa_1$ and $\kappa_1$ suitably small, or $\q_1$ suitably small, there exist a constant $C=C(\epsilon,\delta,\eta, \q_1,\kappa_1, p, r_H,r_I,\h,N)>0$, such that:
\begin{multline}
\label{eq:mainied}
\sup_{\tau\geq 0}\sum_{k\leq N}\int_{\Sigma_{\tau}} \mathcal{E}_{p-2\epsilon}[D_T^k\psi_{\geq \ell}]\,d\sigma dr+\sum_{k_1+k_2+k_3=k}\int_{0}^{\infty}\int_{\Sigma_{\tau}\cap\{r_H\leq r\leq r_I\}} |D_{\s^2}^{k_1}D_{Y_*}^{k_2+1}D_T^{k_3}\psi_{\geq \ell}|^2+|D_{\s^2}^{k_1}D_{Y_*}^{k_2}D_T^{k_3}\psi_{\geq \ell}|^2\\
+\upzeta(|D_{\s^2}^{k_1}D_{Y_*}^{k_2}D_T^{k_3+1}\psi_{\geq \ell}|^2+|D_{\s^2}^{k_1+1}D_{Y_*}^{k_2}D_T^{k_3}\psi_{\geq \ell}|^2)\,d\sigma dr d\tau\\
+\int_{0}^{\infty}\int_{\Sigma_{\tau}\setminus\{r_H\leq r\leq r_I\}} (\rho_++\kappa_+) (\rho_c+\kappa_+) (\rho_+ \rho_c)^{2\epsilon}(\Omega^{-1} r)^{p}\Omega^2|D_{X}D_T^k\psi_{\geq \ell}|^2\\
+\rho_+\rho_c \Omega^{-2}\left[(\Omega^{-1} r)^{-\delta}|D_T^{k+1}\psi_{\geq \ell}|^2+ (\rho_+\rho_c)^{2\epsilon}(\Omega^{-1} r)^{p-2}(|{D}_{\s^2}D_T^k\psi_{\geq \ell}|^2+|D_T^k\psi_{\geq \ell}|^2)\right]\,d\sigma drd\tau\\
\leq C\sum_{k\leq N}\int_{\Sigma_0} \mathcal{E}_{p}[D_T^k\psi_{\geq \ell}]\,d\sigma dr+\int_{0}^{\infty}\int_{\Sigma_{\tau}}\max\{ (\Omega^{-1} r)^{p}\rho_+^{1-2\epsilon}\rho_c^{1-2\epsilon},1\}r^{-2}|r^3D_T^kG_{A}|^2+\upzeta|D_T^{k+1}G_A|^2\,d\sigma dr \,d\tau.
\end{multline}
In particular, \eqref{eq:mainied} holds for all $\q$ on extremal Reissner--Nordstr\"om. Furthermore, under the assumption of quantitative mode stability in the form of Condition \ref{cond:quantmodestab} in \S \ref{sec:wronskianestimates}, \eqref{eq:mainied} holds also when $\kappa_+> \kappa_1$ and $\q_1$ is not assumed to be small.

If $|\q|<\left(\frac{1}{4}+\frac{1}{2}\ell\right)$, or equivalently, $\re \beta_{\ell}>\sqrt{3}\left(\ell+\frac{1}{2}\right)$, then \eqref{eq:mainied} holds also with $\epsilon=0$ and $1<p<1+\sqrt{1-\frac{16\q^2}{(2\ell+1)^2}}$.
\end{theorem}

\begin{remark}
The assumption that $|\q|\geq \q_0>0$ when $\kappa_c>0$ is not necessary and can be easily removed. However, since we restrict our estimates in the case $|\q|\leq \q_0$ (with $\q_0>0$ small) to the $\kappa_c=0$ setting (see \S \ref{sec_smallqinten}), for the sake of simplicity, we exclude the case $|\q|<q_0$ and $\kappa_c>0$ from the statement of the theorem.
\end{remark}

We will prove Theorem \ref{thm:main} in \S \ref{sec:finishpfmainthm}.
\section{Integrated estimates for small $|\q|$}
\label{sec_smallqinten}
In this section, we derive integrated energy estimates that are only valid for very small $|\q|$, using only a decomposition into spherical harmonic modes. We easily obtain an integrated energy estimate for solutions to \eqref{eq:CSF} that are supported on bounded angular frequencies as a corollary from ($r$-weighted) integrated energy estimates valid in the $\q=0$ case. This is the content of Theorem \ref{thm:smallqboundfreq}. In Corollary \ref{cor:iedsmallq}, we control the higher angular frequencies, using the purely physical-space based integrated energy estimates \emph{modulo zeroth order terms} from Appendix \ref{sec:purelyphysied}.

In this section, we will restrict to $\kappa_c=0$.
\begin{theorem}
\label{thm:smallqboundfreq}
Let $\kappa_c=0$, $\kappa_+\geq 0$. Let $r^{-1}\psi$ be a solution to \eqref{eq:CSF} arising from the initial value problem in Theorem \ref{thm:gwp} with $\q=0$. Let $0\leq p\leq 2$, $\delta>0$ and $N\in \N_0$.

Then there exists a constant $C=C(\tilde{\h},r_H,r_I,\delta,N)>0$ such that:
\begin{multline}
\label{eq:iedphysicalspaceq0}
\sup_{\tau_1\leq \tau \leq \tau_2}\sum_{k\leq N}\int_{\Sigma_{\tau}} \mathcal{E}_{p}[T^k\psi]\,d\sigma dr\\
+\sum_{k_1+k_2+k_3\leq N}\int_{\tau_1}^{\tau_2}\Bigg[\int_{\Sigma_{\tau}\cap\{r_H\leq r\leq r_I\}}|\snabla_{\s^2}^{k_1}T^{k_2}Y_*^{k_3+1}\psi|^2+ |\snabla_{\s^2}^{k_1}T^{k_2}Y_*^{k_3}\psi|^2\\
+(1-r^{-1}r_{\sharp})^2(|\snabla_{\s^2}^{k_1}T^{k_2+1}Y_*^{k_3}\psi|^2+|\snabla_{\s^2}^{k_1+1}T^{k_2}Y_*^{k_3}\psi|^2)\,d\sigma dr\Bigg] d\tau\\
+\sum_{k\leq N}\int_{\tau_1}^{\tau_2}\int_{\Sigma_{\tau}\cap\{r\leq r_H\}} p(\rho_++\kappa_+)\mathcal{E}_{p}[T^k\psi]+(2-p)(\rho_++\kappa_+)(\Omega^{-1}r)^p|\snabla_{\s^2}T^k\psi|^2+(r-r_+)^{-1+\delta}|T^{k+1}\psi|^2\,d\sigma drd\tau\\
+\sum_{k\leq N}\int_{\tau_1}^{\tau_2}\int_{\Sigma_{\tau}\cap\{r\geq r_I\}} p\mathcal{E}_{p-1}[T^k\psi]+(2-p)r^{-1}(\Omega^{-1}r)^p|\snabla_{\s^2}T^k\psi|^2+r^{-1+\delta}|T^{k+1}\psi|^2\,d\sigma drd\tau\\
\leq C \sum_{|\gamma|\leq N}\int_{\Sigma_{\tau_1}} \mathcal{E}_{p}[T^k\psi]\,d\sigma dr+C\int_{\tau_1}^{\tau_2}\int_{\Sigma_{\tau}} \rho_+ r^{-1} (\Omega^{-1}r)^pr^{-2}|T^k(r^3G_{\widehat{A}})|^2\,d\sigma dr d\tau\\
+C\left|\int_{\tau_1}^{\tau_2}\int_{\Sigma_{\tau} }\re\left(\overline{T^{k+1}\psi} \cdot rG_{\widehat{A}}\right)\,d\sigma dr d\tau\right|.
\end{multline}
Furthermore, for $L_0\in \N_0$ and $\psi=\psi_{\leq L_0}$, we can remove the factor $(1-r^{-1}r_{\sharp})^2$ and the term $\overline{T^{k+1}\psi} \cdot rG_{\widehat{A}}$ above at the expense of making the constant $C$ depend on $L_0$.
\end{theorem}
\begin{proof}
	Consider first the case $N=0$. The control of the integral in $\{r_H\leq r\leq r_I\}$ follows from a \emph{Morawetz estimate} or integrated local energy decay estimate. In the $\kappa_+>0$ case, this follows as a special case from \cite{damon}[Theorem 3.2.1]. The $\kappa_+=0$ case follows from \cite{aretakis1}[Theorem 1]. Strictly speaking, the proof of \cite{aretakis1}[Theorem 1] does not consider simultaneously the extremal and sub-extremal cases.  A purely physical-space-based proof, which is moreover uniform in $\kappa_+$, follows from \cite{gaj22a}[\S 6] with $D(r)=1-\frac{2M}{r}+Q^2r^{-2}$. Indeed, the desired Morawetz estimates for $\ell=0$ follow from \cite{gaj22a}[Proposition 6.1] and for $\psi_{\ell}$ with $\ell\geq 1$, we simply combine \cite{gaj22a}[Proposition 6.1] with \cite{gaj22a}[Proposition 6.2], the latter which remains unchanged when $D(r)=1-\frac{2M}{r}+Q^2r^{-2}$. In order to obtain uniform-in-$\kappa_+$ estimates for higher spherical harmonic modes $\psi_{\geq L_0}$, with $L_0$ arbitrarily large, we can directly apply Corollary \ref{cor:iedmoduleo0th} (or alternatively, modify \cite{gaj22a}[Proposition 6.5]).
	
	Control over the integrals in $\{r\leq r_H\}$ and $\{r\geq r_I\}$ follows from $r^p$-weighted energy estimates and $(r-r_+)^{-p}$-weighted energy estimates; see for example \cite{paper4}[Proposition 6.5] for the $|Q|=M$ case or Proposition \ref{prop:rpestphysspace} below for estimates that are uniform in $\kappa_+$.
	
To conclude the $N\geq 1$ cases, we commute with the Killing vector field $T$ and apply standard elliptic estimates in $\{r_H\leq r\leq r_I\}$.
\end{proof}

\begin{remark}
An analogue of Theorem \ref{thm:smallqboundfreq} can also be shown to hold when $\kappa_c>0$ by applying the results in \cite{gon24}.
\end{remark}

\begin{corollary}
\label{cor:iedsmallq}
	Let $\kappa_c=0$, $\kappa_+\geq 0$ and $N\in \N_0$. Let $r^{-1}\psi$ be a solution to \eqref{eq:CSF} arising from the initial value problem in Theorem \ref{thm:gwp} with 1-form $A=\widehat{A}$ and denote with $r^{-1}\widetilde{\psi}$ the solution corresponding to $A=\widetilde{A}$.
	 
	 Then for $1\leq p<2$ and $|\q|$ suitably small (depending on $p$) there exists a constant $C=C(\tilde{\h},r_H,r_I,p,\q,N)>0$, such that:
\begin{multline}
\label{eq:iedphysicalspaceqsmallboundl}
\sup_{\tau\in [\tau_1,\tau_2]}\sum_{k\leq N}\int_{\Sigma_{\tau}}\mathcal{E}_p[T^k\psi]\,d\sigma dr d\sigma+\int_{\mathcal{I}^+\cap \{\tau_1\leq \tau\leq \tau_2\}}|T^{k+1}\psi|^2\,d\sigma d\tau+\int_{\mathcal{H}^+\cap \{\tau_1\leq \tau\leq \tau_2\}}|K_+T^k\psi|^2\,d\sigma d\tau\\
+\int_{\tau_1}^{\tau_2}\int_{\Sigma_{\tau}\cap \{r\notin (r_H,r_I)\}}\mathcal{E}_{p-1}[T^k\psi]\,d\sigma dr d\tau\\
+\sum_{k_1+k_2+k_3\leq N}\int_{\tau_1}^{\tau_2}\int_{\Sigma_{\tau}\cap \{r\in (r_H,r_I)\}}(1-r^{-1}r_{\sharp})^2(|\snabla_{\s^2}^{k_1+1}T^{k_2}Y_*^{k_3}\psi|^2+|\snabla_{\s^2}^{k_1}T^{k_2+1}Y_*^{k_3}\widetilde{\psi}|^2)\\
+|\snabla_{\s^2}^{k_1}T^{k_2}Y_*^{k_3+1}\widetilde{\psi}|^2+|\snabla_{\s^2}^{k_1}T^{k_2}Y_*^{k_3}\psi|^2\,d\sigma dr d\tau\\
\leq C\sum_{k\leq N}\int_{\Sigma_{\tau_1}}\mathcal{E}_p[T^k\psi]\,d\sigma dr d\sigma+C\int_{\tau_1}^{\tau_2}\int_{\Sigma_{\tau}}\rho_+ r^{-1} (\Omega^{-1}r)^pr^{-2}|T^k(r^3G_{\widehat{A}})|^2\,d\sigma dr d\tau\\
+C\left|\int_{\tau_1}^{\tau_2}\int_{\Sigma_{\tau} }\re\left(\overline{T^{k+1}\psi} rT^kG_{\widehat{A}}\right)\,d\sigma dr d\tau\right|.
\end{multline}
\end{corollary}
\begin{proof}
	We first consider the case of bounded angular frequencies: $\psi=\psi_{\leq L_0}$. We rearrange the terms in \eqref{eq:maineqradfield}, using \eqref{eq:csfgaugetransf} to obtain:
	\begin{equation*}
		\square_{g_{M,Q,\Lambda}} (e^{i\q\chi_{r_1}\tau}\phi)=e^{i\q\chi_{r_1}\tau}G_{\hat{A}}+F_{\q},
	\end{equation*}
	with
	\begin{multline*}
		F_{\q}=-2i\q r^{-1}\left[(1-\h)(\chi_{r_1}r_+^{-1}-r^{-1}) \partial_r(e^{i\q \chi_{r_1}\tau}\psi)+(\chi_{r_1}r_+^{-1}-r^{-1})  \h\widetilde{\h}T(e^{i\q\chi_{r_1}\tau}\psi)\right]\\
		+r^{-1}\left[-i\q r^{-2}\left(1-\h-(\chi_{r_1}r_+^{-1}-r^{-1})r^2 \frac{d\h}{dr}\right)-\q^2(\chi_{r_1}r_+^{-1}-r^{-1})^2  \h\widetilde{\h}\right]e^{i\q\chi_{r_1}\tau}\psi\\
		+\q\frac{d\chi_{r_1}}{dr}O(1)(|X\psi|+|T\psi|)+\q\frac{d^2\chi_{r_1}}{dr^2}O(1)|\psi|.
	\end{multline*}
	Using that $\supp (\frac{d\chi_{r_1}}{dr})\subset (r_1,\frac{1}{2-r_+^{-1}r_1}r_1)$ for $r_1>r_+$, we can estimate
	\begin{multline*}
		\int_{\tau_1}^{\tau_2}\int_{\Sigma_{\tau}} (1-r^{-1})^{1-p}r^{p+1}|rF_{\q}|^2\,d\sigma dr d\tau\leq C\q^2 \int_{\tau_1}^{\tau_2}\int_{\Sigma_{\tau}\cap\{r\leq r_H\}} \rho_+ \mathcal{E}_{p}[e^{i\q\chi_{r_1}\tau}\psi]+\rho_+^{1-p}|\psi|^2\,d\sigma drd\tau\\
+C\q^2\int_{\tau_1}^{\tau_2}\int_{\Sigma_{\tau}\cap\{r\geq r_I\}} \mathcal{E}_{p-1}[e^{i\q\chi_{r_1}\tau}\psi]+r^{p-3}|\psi|^2\,d\sigma drd\tau\\
+C\q^2 \int_{\tau_1}^{\tau_2}\left[\int_{\Sigma_{\tau}\cap\{r_H\leq r\leq r_I\}}|Y_*\psi|^2+ |\psi|^2+|T\psi|^2+|\snabla_{\s^2}\psi|^2\,d\sigma dr\right] d\tau.
	\end{multline*}
	
	When $p<2$, we can apply a standard Hardy inequality to estimate:
	\begin{multline*}
		\int_{\Sigma_{\tau}}\rho_+^{1-p}\chi_{r_1}|\psi|^2\,d\sigma dr+\int_{\Sigma_{\tau}}r^{p-3}\chi_{r_1}|\psi|^2\,d\sigma dr\leq C(2-p)^{-2}\left[\int_{\Sigma_{\tau}}\rho_+^{3-p}\chi_{r_1}|X\psi|^2\,d\sigma dr+\int_{\Sigma_{\tau}}r^{p-1}\chi_{r_1}|X\psi|^2\,d\sigma dr\right]\\
		+C \int_{\Sigma_{\tau}\cap \{r_1\leq r\leq  R_1\} }\psi^2\,d\sigma dr.
	\end{multline*}

	For $\phi$ supported on spherical harmonics with $\ell\leq L_0$, we can therefore apply \eqref{eq:iedphysicalspaceq0} with $0<p<2$ (without the degenerate factor $(1-r^{-1}r_{\sharp})^2$!) and $G_{\widehat{A}}$ replaced by $G_{\hat{A}}+F_{\q}$ and absorb the terms with a factor $\q^2$ into the LHS of \eqref{eq:iedphysicalspaceq0} for $|\q|$ suitably small, depending on $p$.
	
	Now consider $\psi=\psi_{\geq L_0}$ then we can instead apply Corollary \ref{cor:iedmoduleo0th}.
\end{proof}

\section{Integrated energy estimates in frequency space: set-up}
\label{sec:iledfreq}
In order to obtain integrated energy estimates \underline{without} a smallness assumption on $|q|$, we will need to apply a Fourier transform in the time coordinate $t$ and perform a frequency-space-based analysis.

We provide below an outline of the arguments in frequency space.
\begin{itemize}
	\item In the present section, we will introduce the a priori assumption of \emph{future integrability} to make sense of the Fourier transform and derive a Schr\"odinger-type radial ODE \eqref{eq:radialODE}. The analysis in \S \S \ref{sec:iledfreq}--\ref{sec:iedcomb} concerns only \eqref{eq:radialODE} and can therefore be read independently from the rest of the paper. Here, we will moreover derive the main relevant properties of the potentials appearing in the radial ODE.
	\item In \S \ref{sec:iedlargefreq}, we consider large time \emph{or} angular frequencies and derive $L^2$-estimates for solutions to \eqref{eq:radialODE} by employing ``microlocal" energy currents or vector field multiplier, i.e.\ by a systematic integration by parts adapted to the large frequency behaviour of the potential in \eqref{eq:radialODE}.
	\item In \S \ref{sec:intestboundfreq}, we restrict to bounded time \emph{and} angular frequencies. We derive $L^2$-estimates for solutions to \eqref{eq:radialODE} by utilizing Green's formula \eqref{eq:greenformula} and deriving appropriate $L^{\infty}$-estimates for homogeneous solutions to \eqref{eq:radialODE}, as well as uniform estimates for the corresponding Wronskian. 
	\item In \S \ref{sec:iedcomb}, we combine the estimates in the different frequency ranges.
	\item We then derive energy estimates in physical space in \S\ref{sec:iedphy} from the estimates in \S \ref{sec:iedcomb} by applying Plancherel's theorem. 
	\end{itemize}

\subsection{Sufficient integrability, future integrability and the Fourier transform}
\label{sec:deffutureint}
In this section, we will introduce the notion of future integrability of solutions $\psi_{\ell m}: [0,\infty)_{\tau}\times [r_+,\infty)_r\to \C$ to \eqref{eq:CSF} in the gauge $A=\widehat{A}=-Qr^{-1}d\tau$ with fixed angular frequency $\ell$ and azimuthal number $m$. 

Let $\xi: \R \to \R$ be a smooth cut-off function such that $\xi(\tau)=1$ when $\tau\geq 1$ and $\xi(\tau)=0$ when $\tau\leq 0$. Let
\begin{equation*}
\uppsi_{\ell m}: \R_{t}\times [r_+,\infty)_r\to \C
\end{equation*}
be the trivial extension of $\xi\cdot \psi_{\ell m}(\tau(t,r),r)$ from $\{\tau\geq 0\}\subset \R_t\times [r_+,\infty)_r$ to $\R_{t}\times [r_+,\infty)_r$.
Then we can apply \eqref{eq:maineqradfield} to obtain:
\begin{equation}
\label{eq:CSFtimecutoff}
(g^{-1}_{M,Q,\Lambda})^{\mu \nu}(^{\widehat{A}}D)_{\mu}(^{\widehat{A}}D)_{\nu}(r^{-1}\uppsi_{\ell m}Y_{\ell m})=\xi (G_{\widehat{A}})_{\ell m}Y_{\ell m}+(F_{\widehat{A}, \xi})_{\ell m}Y_{\ell m},
\end{equation}
with
\begin{equation}
\label{eq:Ftimecutoff}
r\Omega^{2}(F_{\widehat{A},\xi})_{\ell m}=-2\dot{\xi} \Omega^2\h \widetilde{\h} T\psi_{\ell m}-\ddot{\xi} \Omega^2\h \widetilde{\h}\psi_{\ell m}+ 2\dot{\xi}(1-\widetilde{\h}\Omega^2)\Omega^2X\psi_{\ell m}+\Omega^2\left(\frac{d\h}{dr}-2i \q r^{-1} \h\widetilde{\h}\right)\dot{\xi}\psi_{\ell m},
\end{equation}
where $\dot  \xi:=\frac{d\xi}{d\tau}, \ddot  \xi:=\frac{d^2\xi}{d\tau^2}$ and where we also identify $\xi (G_{\widehat{A}})_{\ell m}$ with its trivial extension to $\R_{t}\times [r_+,\infty)_r$.

In what follows, we will frequently omit the subscript $\ell m$ in the notation.

With respect to the gauge $\widetilde{A}=-Qr^{-1}dt=\widehat{A}+Qr^{-1}dr_*-Qr^{-1}\widetilde{\h}dr$, we then obtain: 
\begin{equation}
(g^{-1}_{M,Q,\Lambda})^{\mu \nu}(^{\widetilde{A}}D)_{\mu} (^{\widetilde{A}}D)_{\nu}(r^{-1}\widetilde{\uppsi})=\xi G_{\widetilde{A}}+F_{\widetilde{A},\xi},
\end{equation}
with
\begin{align}
\label{eq:uppsihattilde1}
\widetilde{\uppsi}:=&\:e^{i\q r_+^{-1}r_*-i\q\int_{0}^{r_*}\rho_+(r_*')\,dr_*'-i\q \int_{r_{\sharp}}^{r}r'^{-1}\widetilde{\h}(r')\,dr'} \uppsi,\\
\label{eq:Ghattilde1}
G_{\widetilde{A}}=&\:e^{i\q r_+^{-1}r_*-i\q \int_{0}^{r_*}\rho_+(r_*')\,dr_*'-i\q\int_{r_{\sharp}}^{r}r'^{-1}\widetilde{\h}(r')\,dr'} G_{\widehat{A}},\\
\label{eq:Fhattilde1}
F_{\widetilde{A},\xi}=&\:e^{i\q r_+^{-1}r_*-i\q\int_{0}^{r_*}\rho_+(r_*')\,dr_*'-i\q\int_{r_{\sharp}}^{r}r'^{-1}\widetilde{\h}(r')\,dr'} F_{\widehat{A}, \xi}.
\end{align}

The following alternative expressions will also be useful
\begin{align}
\label{eq:uppsihattilde2}
\widetilde{\uppsi}=&\:e^{iqr_c^{-1}r_*+iq\int_{0}^{r_*}\rho_c(r_*')\,dr_*'-iq\int_{2}^{r}r'^{-1}\Omega^{-2}(r')\h(r')\,dr'}  \uppsi,\\
\label{eq:Ghattilde2}
G_{\widetilde{A}}=&\:e^{iqr_c^{-1}r_*+iq\int_{0}^{r_*}\rho_c(r_*')\,dr_*'-iq\int_{2}^{r}r'^{-1}\Omega^{-2}(r')\h(r')\,dr'}   G_{\widehat{A}},\\
\label{eq:Fhattilde2}
F_{\widetilde{A},\xi}=&\:e^{iqr_c^{-1}r_*+iq\int_{0}^{r_*}\rho_c(r_*')\,dr_*'-iq\int_{2}^{r}r'^{-1}\Omega^{-2}(r')\h(r')\,dr'}  F_{\widehat{A}, \xi}.
\end{align}
\begin{definition}
\label{def:suffint}
A function $f: \R_{t}\times [M,\infty)_r\to \C$ is \emph{sufficiently integrable} if: for any $r\in (r_+,r_c)$ and any $k\in \N_0$:
\begin{align}
\label{eq:suffint1}
\int_{\R} |Y_*^kf|^2(t,r)\,dt<\infty,\\
\label{eq:suffint2}
\int_{-\infty}^{0}\int_{\R_t} |D_{\underline{L}}f|^2(t,r(r_*))\, dtdr_*+\int_{0}^{\infty}\int_{\R_t} |D_{L} f|^2(t,r_*(r))\, dtdr_*<\infty.
\end{align}
and
\begin{align}
\label{eq:suffint3}
\int_{\R} |Y_*^k\left((g^{-1}_{M,Q,\Lambda})^{\mu \nu}(^{\widehat{A}}D)_{\mu}(^{\widehat{A}}D)_{\nu}(r^{-1}f)\right)|^2(t,r)\,dt<\infty.
\end{align}
\end{definition}

\begin{definition}
	A solution $\psi$ to \eqref{eq:CSF} with $A=\widehat{A}$ is called \emph{future integrable} if for all $\ell\in \N_0$ and $m\in \Z$ with $|m|\leq \ell$, $\uppsi_{\ell m}$ is sufficiently integrable.
\end{definition}

\begin{proposition}
Let $\mathfrak{F}: L^2(\R_t)\to L^2(\R_{\omega})$ denote the Fourier transform. Let 
\begin{equation*}
\homega:=\omega-\q r_c^{-1}.
\end{equation*}
Assume that $\uppsi_{\ell m}$ is sufficiently integrable. Then for all $k\in \N_0$ and $r_*\in \R$,\\ $u_{\ell m}^{(k)}(\cdot,r_*):=(\mathfrak{F}(Y_*^k\widetilde{\uppsi}))_{\ell m}(\cdot,r(r_*))\in L^2(\R_{\omega})$ is well-defined and $u_{\ell m}(\cdot,r_*)$ satisfies the following equation in $L^2(\R_{\omega})$:
\begin{equation}
\label{eq:radialODE}
u_{\ell m}''(\cdot,r_*)+(\homega^2-V_{\homega  \ell})u_{\ell m}(\cdot,r_*)=H_{ \ell m}(\cdot,r_*),
\end{equation}
with
\begin{equation*}
H_{\ell m}(r_*)=r\Omega^{2}(\mathfrak{F}(\widehat{G}_{\widetilde{A}}(\cdot,r(r_*))))_{\ell m}+r\Omega^{2}(\mathfrak{F}(F_{\widetilde{A},\xi}(\cdot,r(r_*))))_{\ell m}
\end{equation*}
and $V_{\homega  \ell}:[r_+,r_c)\to \R$ is defined as follows:
\begin{equation*}
V_{\homega   \ell}(r):=\ell(\ell+1)\Omega^2r^{-2}-{\q}^2\rho_c^2(r)+2{\q} \homega \rho_c(r)+\Omega^2\left[r^{-1}\frac{d\Omega^2 }{dr}+\frac{2\Lambda}{3}\right] .
\end{equation*}
\end{proposition}
\begin{proof}
We apply \eqref{eq:maineqradfieldtilde} to obtain:
\begin{multline*}
r(g^{-1}_{M,Q,\Lambda})^{\mu \nu}(^{\widetilde{A}}D)_{\mu} (^{\widetilde{A}}D)_{\nu}(r^{-1}\widetilde{\uppsi}_{\ell m})=-K_c^2\widetilde{\uppsi}_{\ell m}+ Y_*^2\widetilde{\uppsi}_{\ell m}+r^{-2}(\slashed{\Delta}_{\s^2}\widetilde{\uppsi})_{\ell m}+
\q^2\rho_c^{2}\widetilde{\uppsi}\\
-\left[\frac{d}{dr}(\Omega^2)r^{-1}+\frac{2\Lambda}{3}\right]\Omega^2\widetilde{\uppsi}_{\ell m}-2i \q r^{-1}K_c\widetilde{\uppsi}_{\ell m}.
\end{multline*}
To obtain \eqref{eq:radialODE}, we first replace $K_c$ with $-i\homega \mathbf{1}$ and using that \eqref{eq:suffint1} in the definition of sufficient integrability allows us to take a Fourier transform in $t$.
\end{proof}
Let:
\begin{equation*}
\tomega:=\omega-\q r_+^{-1}.
\end{equation*}

It will be convenient to define the following shifted potential $\widetilde{V}_{\tomega \ell}: [r_+,r_c)_r\to \infty$:
\begin{equation*}
\widetilde{V}_{\tomega \ell}(r)-\tomega^2:=V_{\homega \ell}(r)-\homega^2.
\end{equation*}

We can then express:
\begin{equation*}
\widetilde{V}_{\tomega \ell}(r)=\ell(\ell+1)\Omega^2r^{-2}-{\q}^2\rho_+^2(r)-2{\q}\tomega \rho_+(r)+\Omega^2\left[r^{-1}\frac{d\Omega^2}{dr}+\frac{2\Lambda}{3}\right].
\end{equation*}

Given a solution $u_{\ell m}$ to \eqref{eq:radialODE}, we define the auxiliary functions:
\begin{align}
\label{eq:defv}
 v_{\ell m}:=&\:e^{i\tomega r_*+i \q \int_{0}^{r_*}\rho_+(r_*')\,dr_*'}u_{\ell m},\\
 \label{eq:defw}
 w_{\ell m}:=&\:e^{-i\homega r_*+i \q \int_{0}^{r_*}\rho_c(r_*')\,dr_*'}u_{\ell m}.
\end{align}

The functions $v_{\ell m}$ and $w_{\ell m}$ may be thought of as Fourier transforms of gauge transformations of $\uppsi$ with respect to the time coordinates $\tau=t+r_*$ and $\tau=t-r_*$, respectively, where $v_{\ell m}$ corresponds to the gauge choice $A=-\frac{Q}{r}dv $ and $w_{\ell m}$ corresponds to the gauge choice $A=-\frac{Q}{r}du$, cf. \eqref{eq:uppsihattilde1} with $\widetilde{\h}\equiv 0$ and  \eqref{eq:uppsihattilde2} with $\h\equiv 0$.

Observe that by Lemma \ref{lm:couchtorr}, the following equality holds when $\kappa_c=\kappa_+=0$:
\begin{equation}
\label{eq:ctorrpot}
\widetilde{V}_{\tomega \ell}(r(r_*))=V_{-\homega \ell}(r(-r_*)).
\end{equation}

In the lemma below, we establish an $L^2$-version of the appropriate boundary conditions satisfied by solutions to \eqref{eq:radialODE}.

\begin{lemma}
\label{lm:boundcond}
Assume that $\uppsi_{\ell m}$ is sufficiently integrable.  Then there exist sequences $\{r_i\}$ and $\{\tilde{r}_i\}$ with $r_i\to r_c$ and $\tilde{r}_i\to r_+$ as $i\to \infty$, such that:
\begin{align*}
\lim_{i\to \infty} ||w'_{\ell m}(\cdot,r_*(r_i))||_{L^2(\R_{\omega})}\to\: & 0,\\
\lim_{i\to \infty} ||v'_{\ell m}(\cdot,r_*(\tilde{r}_i))||_{L^2(\R_{\omega})}\to \:& 0
\end{align*}
as $i\to \infty$. In particular, $v'_{\ell m}(r_*)\to 0$ as $r_*\to -\infty$ and $w'_{\ell m}(r_*)\to 0$ as $r_*\to \infty$ for almost every $\omega\in \R$.
\end{lemma}
\begin{proof}
We will suppress the subscript $\ell m$ in the notation below. By the assumption \eqref{eq:suffint2}, together with the Plancherel identity, we have that for $\widetilde{\h}=0$ $D_{\underline{L}}\uppsi=\underline{L}\uppsi=(T-Y_*)\uppsi$, there exists a constant $C(\uppsi)>0$ such that:
\begin{multline*}
	C(\uppsi)>\frac{1}{4}\int_{-\infty}^0\int_{\R_t} |(T-Y_*)\uppsi|^2\,dt dr_*=\frac{1}{4}\int_{-\infty}^0\int_{\R_t} |\partial_t \widetilde{\uppsi}+i\q (r_+^{-1}-\rho_+(r_*))\widetilde{\uppsi}-Y_*\uppsi|^2\,dtdr_*\\
	=\frac{1}{4}\int_{-\infty}^0 \int_{\R_{\omega}}|v'|^2\,d\omega dr_*.
\end{multline*}
Consider the sequence $\{-2^i\}$. By the mean-value theorem, there exists a sequence $\{(r_*)_i\}$ with $-2^{i+1}\leq (r_*)_i\leq- 2^{i}$, such that
\begin{equation*}
	\int_{\R_{\omega}}|v'|^2(\omega,(r_*)_i)\,d\omega\leq C(\uppsi) 2^{-i-1}\to 0
\end{equation*}
as $i\to \infty$. Since $v'((r_*)_i,\cdot)\in L^2(\R_{\omega})$ converges to 0 , it must follow that  $v'((r_*)_i,\omega)\to 0$ pointwise for almost all $\omega\in \R$.

The convergence properties of $w_{\ell m}'$ follow analogously by considering $\h\equiv 0$ so that $D_{L}\uppsi=L\uppsi$. 
\end{proof}

Since we will derive $L^2_{\omega}$-estimates for \eqref{eq:radialODE}, we can consider, without loss of generality solutions to \eqref{eq:radialODE} that are valid outside of a zero measure subset of $\R_{\omega}$. The proposition below demonstrates that by omitting a zero measure subset of frequencies, we can in fact restrict to smooth solutions to \eqref{eq:radialODE}.
\begin{proposition}
Assume that $\uppsi_{\ell m}$ is sufficiently integrable, such that $\uppsi_{\ell m}$ is a solution to \eqref{eq:CSFtimecutoff}. Then, for any $\ell\in \N_0$ and $m\in \Z$ with $|m|\leq \ell$, $H_{ m \ell}(\omega,\cdot)$ is smooth for almost all $\omega\in \R$ and the corresponding solution $u_{\ell m}(\omega,\cdot)$ to \eqref{eq:radialODE} is smooth and satisfies the boundary conditions:
\begin{align}
\label{eq:inhomodebc1}
\lim_{r_*\to \infty} w'_{\ell m}(\omega,r_*)\to\: & 0,\\
\label{eq:inhomodebc2}
\lim_{r_*\to -\infty} v'_{\ell m}(\omega, r_*)\to \:& 0
\end{align}
\end{proposition}
\begin{proof}
	By \eqref{eq:suffint3}, we have that $u_{\ell m}^{(k)}\in L^2(\R_{\omega})$ are well-defined for all $k\in \N_0$ and there exists constants $D_k(\uppsi)>0$ such that $||u_{\ell m}||_{H^k((r_*-\epsilon,r_*+\epsilon))L^2(\R_{\omega})}\leq D_k(\uppsi)$.
	
	By Fubini's theorem and a Sobolev embedding, there therefore exist zero-Lesbesgue-measure subsets $U_k\subset \R_{\omega}$, such that $u_{\ell m}(\omega,\cdot)\in C^k((r_*-\epsilon,r_*+\epsilon))$ for $\omega\in U_k^c$. Note that $U:=\bigcup_{k\in \N_0}U_k$ also has zero measure and $u_{\ell m}(\omega,\cdot)\in C^{\infty}((r_*-\epsilon,r_*+\epsilon))$ for all $\omega\in U^c$. The same argument also applies to $H_{\ell m}$ to conclude that $H_{\ell m}(\omega,\cdot)\in C^{\infty}((r_*-\epsilon,r_*+\epsilon))$. Since $r_*\in \R$ was arbitrary and $\R_{r_*}$ can be covered by a countable number of sets of the form $(r_*-\epsilon,r_*+\epsilon)$, we conclude smoothness of $u_{\ell m}(\omega,\cdot)$ and $H_{\ell m}(\omega,\cdot)$ for almost every $\omega\in \R$.
	
	Finally, \eqref{eq:inhomodebc1} and \eqref{eq:inhomodebc2} follow (after removing a further zero-measure subset) from Lemma \ref{lm:boundcond}.
\end{proof}
\subsection{Properties of the potential $V_{\homega \ell}$}
The analysis of solutions to \eqref{eq:radialODE} depends crucially on the precise form of the potential $V_{\homega \ell}$. \textbf{For the sake of notational convenience, we will set $r_+=1$ the remainder of \S \ref{sec:iledfreq}!}

In the case $\kappa_c=\kappa_+=0$, the potentials ${V}_{\homega \ell}(r)$ and $\widetilde{V}_{\tomega \ell}$ are symmetric under the transformations $r_*\mapsto -r_*$ and $\homega\mapsto -\tomega$ by \eqref{eq:ctorrpot}. More generally, $\widetilde{V}_{\tomega \ell}$ has the same qualitative behaviour when we consider the transformation $(r_*,\homega, \kappa_c,\kappa_+)\mapsto (-r_*,-\tomega, \kappa_+,\kappa_c)$.

We will denote:
\begin{multline*}
V^{\sharp}_{\omega \ell}(r):=V_{\tomega \ell}(r)+{\q}^2(r^{-1}-r_c^{-1})^2-\Omega^2\left[r^{-1}\frac{d\Omega^2}{dr}+\frac{2\Lambda}{3}\right]=\ell(\ell+1)\Omega^2r^{-2}+2{\q} \homega (r^{-1}-r_c^{-1}).
\end{multline*}

The $V^{\sharp}_{\omega \ell}$ part of the potential $V_{\tomega \ell}$ is relevant when studying the large $|\omega|$ or large $\ell$ behaviour. In the lemma below, we determine the shape of $V^{\sharp}_{\omega \ell}$.
\begin{lemma}
\label{lm:generalcriticalpointsV}
In the interval $[1,r_c)$ function $V^{\sharp}_{\omega \ell}$ has at most two critical points.
\begin{enumerate}
\item Suppose that $\q\homega\neq 0$. In the case of exactly two critical points, there is a local minimum at $r=r_{\rm min}$ and a local maximum at $r=r_{\rm max}$, with $r_{\rm min}< r_{\rm max}$ if $\q\homega>0$ and $r_{\rm min}>r_{\rm max}$ if $\q\homega<0$.
\item Suppose that $\q\homega\neq 0$. In the case of exactly one critical point, the critical point must be a global maximum.
\item In the absence of critical points, $V^{\sharp}_{\omega \ell}$ is decreasing if $\q\homega>0$ and increasing if $\q\homega<0$.
 \end{enumerate}
\end{lemma}
\begin{proof}
Suppose that $\q\homega= 0$. Then $V^{\sharp}_{\omega \ell}(r)=\ell(\ell+1)\Omega^2r^{-2}$ and $\Omega^2r^{-2}=r^{-2}-2Mr^{-3}+Q^2r^{-4}$ satisfies $\frac{d}{dr}(\Omega^2 r^{-2})=0$ if and only if $r=r_{\sharp}=\frac{3M}{2}+\frac{1}{2}\sqrt{9M^2-8Q^2}$, which corresponds to a maximum.

Suppose that $\q\homega\neq 0$. We consider the expressions:
\begin{align*}
r^5\frac{dV^{\sharp}_{\omega \ell}}{dr}=&\: -2 \q \homega r^3-2\ell(\ell+1)(r^2-3Mr+2Q^2),\\
\frac{d}{dr}\left(r^5\frac{dV^{\sharp}_{\omega \ell}}{dr}\right)(r)=&\:-6{\q}\homega r^2-2\ell(\ell+1)(2r-3M).
\end{align*}
 The critical points of $r^5\frac{dV^{\sharp}_{\omega \ell}}{dr}$ are then located at:
\begin{equation*}
r_{1,2}=\frac{\ell(\ell+1)\pm \sqrt{\ell^2(\ell+1)^2+9M\ell(\ell+1) {\q} \homega}}{-3{\q}\homega}
\end{equation*}

Suppose that ${\q}\homega>0$. Then only one of the values $r_1$ and $r_2$ can be non-negative, so there can be at most one critical point of $r^5\frac{dV^{\sharp}_{\homega \ell}}{dr}$ in the interval $[1,r_c)$. As a result, $\frac{dV^{\sharp}_{\homega \ell}}{dr}$ can have at most two zeroes, so $V^{\sharp}_{\homega \ell}$ can have at most two critical points.

By extending the domain of the polynomial $r^5\frac{dV^{\sharp}_{\homega \ell}}{dr}$ to $r\in \R$, we note that when $\q\homega>0$, $(r^5\frac{dV^{\sharp}_{\homega \ell}}{dr})(r)\to \pm \infty$ as $r\to \mp \infty$, so the critical point with the largest $r$-value must be a local maximum of $r^5\frac{dV^{\sharp}_{\homega \ell}}{dr}$. If there is another critical point, it has to be a local minimum, since a maximum or saddle point would contradict $\frac{d}{dr}(r^5\frac{dV^{\sharp}_{\homega \ell}}{dr})$ having at most one zero. We conclude that $V^{\sharp}_{\homega \ell}$ either has no critical points and is strictly decreasing or a maximum at $r=r_{\rm max}>1$. In the latter case, $V^{\sharp}$ can have an additional critical point at $r=r_{\rm min}<r_{\rm max}$, which must be a minimum.

Now suppose that ${\q}\homega< 0$. Then it is straightforward to show that:
\begin{equation*}
r^5\frac{dV^{\sharp}_{\omega \ell}}{dr}(1)=\ell(\ell+1)\frac{d\Omega^2}{dr}(1)-2\q\homega\geq -2\q\homega>0.
\end{equation*}
Since $r^5\frac{dV^{\sharp}_{\omega \ell}}{dr}$ is a polynomial of third degree, it can have at most three zeroes. Note that $(r^5\frac{dV^{\sharp}_{\omega \ell}}{dr})(r)\to \pm \infty$ as $r\to \pm \infty$ when $q\homega>0$ and $r^5\frac{dV^{\sharp}_{\omega \ell}}{dr}(1)>0$, there must be one zero located in $r<r_+$ and $r^5\frac{dV^{\sharp}_{\omega \ell}}{dr}$ can therefore have at most two zeroes in $[r_+,r_c)$. This implies that $V^{\sharp}_{\omega \ell}$ can have at most two critical points.

The critical point with the largest $r$-value, if it exists, must be a local minimum since $(r^5\frac{dV^{\sharp}_{\omega \ell}}{dr})(r)\to \infty$ as $r\to \infty$ when $\q\homega<0$. The other critical point must be a maximum, since the presence of a another minimum or saddle point contradicts the fact that $V^{\sharp}_{\omega \ell}$ has at most two critical points.
\end{proof}

When $\ell$ is sufficiently large and $\homega^2\lesssim \ell(\ell+1)$, we can give a more precise characterization of the nature and locations of the critical points of  $V_{\homega \ell}$; see also Figure \ref{fig:introelldom}.
\begin{lemma}
\label{lm:criticalpointsVhighl}
Let $\eta>0$ and $\beta>0$. Then there exists a $L_0\in \N$ suitably large depending on $\q$, $\beta$ and $\eta$, such that for all $\ell\geq L_0+1$ and $\homega^2\leq \beta^{-1} \ell(\ell+1)$, $V_{\homega \ell}(r)$ has a unique maximum at $r=r_{\rm max}$, a global maximum, with
\begin{equation*}
|r_{\rm max}-r_{\sharp}|\leq \frac{\eta}{2}.
\end{equation*}
Furthermore, 
\begin{enumerate}
\item Suppose that $\q\homega<0$. For suitably small $\kappa_c$ (depending on $L_0$ and $\beta$), $V_{\homega \ell}(r)$ has exactly one additional critical point: a local minimum located at $r=r_{\rm min,\infty}>r_{\rm max}$, with $\rho_c(r_{\rm min,\infty})<\eta $ for $L_0$ suitably large.
\item Suppose that $\q\homega\geq 0$ and suppose additionally that $\q\tomega<0$.  Then $V_{\homega \ell}(r)$ has no additional critical points in $[r_+,r_c)$.
\item If $\q\homega\geq 0$ and $\q\tomega\geq 0$, then for $\kappa_+$ suitably small, $V_{\homega \ell}(r)$ has exactly one additional critical point: a local minimum located at $r=r_{\rm min,+}<r_{\rm max}$, with $\rho_+(r_{\rm min,+})<\eta$ for $L_0$ suitably large and $r_{\rm min}=r_+$ if $q\tomega=0$ and $\kappa_+=0$.
\end{enumerate}
\end{lemma}
\begin{proof}
Note first of all that $|\homega|\leq \beta^{-\frac{1}{2}} \sqrt{\ell(\ell+1)}\ll \ell(\ell+1)$ for $L_0$ suitably large. Since 
\begin{equation*}
r^5\frac{dV_{\homega \ell}}{dr}(r)=2\ell(\ell+1)(r^2-3r+2Q^2)+\q^2O(r^2)-\q\homega O(r^{2})+O(r),
\end{equation*}
we must have that for $\eta>0$ suitably small, $L_0$ sufficiently large, $\frac{dV_{\homega \ell}}{dr}$ must change sign in the interval $(r_{\sharp}-\frac{\eta}{2},r_{\sharp}+\frac{\eta}{2})$, where $r_{\sharp}$ is the largest root of the polynomial $r^2-3r+2Q^2$, so $V_{\homega \ell}$ must have a critical point in $(r_{\sharp}-\frac{\eta}{2},r_{\sharp}+\frac{\eta}{2})$. Furthermore, $\frac{d^2V_{\homega \ell}}{dr^2}<0$ in $(r_{\sharp}-\frac{\eta}{2},r_{\sharp}+\frac{\eta}{2})$, so that critical point must be a local maximum.

We have that
\begin{equation*}
V_{\homega \ell}(1)=-\left[{\q}^2 (1-r_c^{-1})-2{\q} \homega\right](1-r_c^{-1}) .
\end{equation*}
By taking $L_0$ sufficiently large, we moreover have that $V_{\homega \ell}(r_{\rm max})\geq \frac{1}{16}\ell(\ell+1)$, so for $L_0$ sufficiently large, $V_{\homega \ell}(r_{\rm max})>V_{\homega \ell}(1)$, so $r_{\rm max}$ is in fact always a global maximum.

When $\q\homega<0$, we have that $V_{\homega \ell}$ is strictly increasing for suitably large $L_0$ and $r$, and it is decreasing for $r>r_{\rm max}$, with $r-r_{\rm max}$ suitably small. For $\kappa_c\geq 0$ suitably small, it must therefore have a minimum at some $r_c>r_{\rm min,\infty}>r_{\rm max}$ and $\rho_c(r_{\rm min,\infty})>0$ can be made arbitrarily small by taking $L_0$ suitably large.

Furthermore, for suitably large $L_0$,  $V_{\homega \ell}$ can be approximated by  $V_{\homega \ell}^{\sharp}$, so we can apply Lemma \ref{lm:generalcriticalpointsV} to conclude that there are no further critical points. This concludes property 1..

Let $\q\omega\geq 0$. When $\kappa_+\neq 0$, we use the relative largeness of $\ell(\ell+1)$ to obtain $\frac{dV_{\homega \ell}}{dr}>0$ for $r<r_{\rm max}$. When $\kappa_+=0$, we have that for $\q\tomega<0$:
\begin{equation*}
\frac{dV_{\homega \ell}}{dr}(1)=\frac{d\widetilde{V}_{\tomega \ell}}{dr}(1)=-2\q\tomega>0,
\end{equation*}
so we can apply the relative largeness of $\ell(\ell+1)$ to obtain also in this case $\frac{dV_{\homega \ell}}{dr}>0$ for $r<r_{\rm max}$. Furthermore, in both the $\kappa_+>0$ and $\kappa_+=0$ cases, we have that $\frac{dV_{\homega \ell}}{dr}<0$ for $r>r_{\rm max}$, by using the relative largeness of $\ell(\ell+1)$, so we can conclude that $V_{\homega \ell}$ has no additional critical points and Property 2. holds.

When $\kappa_+=0$ and $\q\tomega>0$, we have that $\frac{dV_{\homega \ell}}{dr}(r_+)<0$, so $V_{\homega \ell}$ must have a minimum at $1<r_{\rm min,+}<r_{\rm max}$. Furthermore, by applying the relative largeness of $\ell(\ell+1)$, we must have that $r_{\rm min,+}$ approaches $r_+=1$ when $L_0\to \infty$. In the case $\q\tomega=0$, $\frac{dV_{\homega \ell}}{dr}(r_+)=0$ and $\frac{d^2V_{\homega \ell}}{dr^2}(r_+)>0$, so there is a minimum at $r=r_+$ and there are no further critical points in $(r_+,r_c)$ for suitably large $\ell$. We then conclude Property 3..
\end{proof}

Define the following modified potential function:
\begin{align*}
\mathcal{V}_{\omega \ell}(r):=&\:-\homega^2+V_{\homega\ell}(r)+\Omega^3\frac{d^2\Omega}{dr^2}(r).
\end{align*}
\begin{lemma}
\label{lm:mainasympmathcalpot}
The potential function $\mathcal{V}_{\omega \ell}$ satisfies the estimates:
	\begin{align}
\label{eq:boundfreqpotevent}
	\mathcal{V}_{\omega \ell}(r(\rho_+))=&-(\tomega^2+\kappa_+^2)+\left[-2\q \tomega+\kappa_+\left(2\ell(\ell+1)+4\kappa_++O(\kappa_c^2)\right)\right]\rho_+\\ \nonumber
	+&\:\left[\ell(\ell+1)-\q^2-r_+\kappa_+\left(6\ell(\ell+1)+r_+\kappa_+)+(\ell+1)^2O(\kappa_c^2)\right)\right]\rho_+^2\\ \nonumber
	&-2r_+\left[\ell(\ell+1)+3\ell(\ell+1)r_+\kappa_++2r_+^2\kappa_+^2+(\ell+1)^2O(\kappa_c^2)\right]\rho_+^3\\ \nonumber
	+&\:r_+^2\left[\ell(\ell+1)-r_+\kappa_+(2\ell(\ell+1)+r_+\kappa_+)+(\ell+1)^2O(\kappa_c^2)\right]\rho_+^4,\\
			\label{eq:boundfreqpotcosmo}
		\mathcal{V}_{\omega \ell}(r(\rho_c))=&-(\homega^2+\kappa_c^2)+\left[2\q\homega +\kappa_c\left(2\ell(\ell+1)+4\kappa_c+O(\kappa_c^2) \right)\right]\rho_c	\\ \nonumber
&\:+\left[\ell(\ell+1)-\q^2-6\kappa_c r_+\ell(\ell+1)+\kappa_+O(\kappa_c)+O(\kappa_c^2)\right]\rho_c^2+O^{\ell}_{\infty}(\rho_c^3).
	\end{align}
\end{lemma}
\begin{proof}
We have that:
\begin{equation*}
	\mathcal{V}_{\omega \ell}(r)=-\tomega^2-2{\q} \tomega \rho_++\ell(\ell+1)\Omega^2r^{-2}-{\q}^2\rho_+^2+\Omega^2\left[r^{-1}\frac{d\Omega^2 }{dr}+\Omega\frac{d^2\Omega}{dr^2}+\frac{2}{3}\Lambda\right].
\end{equation*}	
Furthermore,
\begin{equation*}
	r^{-1}\frac{d\Omega^2 }{dr}+\Omega\frac{d^2\Omega}{dr^2}=r^{-1}\frac{d\Omega^2 }{dr}+\frac{1}{2}\Omega\frac{d}{dr}(\Omega^{-1}\frac{d\Omega^2}{dr})=r^{-1}\frac{d\Omega^2 }{dr}-\frac{1}{4}\Omega^{-2}\left(\frac{d\Omega^2 }{dr}\right)^2+\frac{1}{2}\frac{d^2\Omega^2 }{dr^2}.
\end{equation*}
We now apply Lemma \ref{lm:metricest} and ignore higher order terms in $\kappa_c$ to obtain \eqref{eq:boundfreqpotevent}. 

The proof of \eqref{eq:boundfreqpotcosmo} proceeds analogously, but is a little easier, as we do not keep precise track of $\kappa_c$.
\end{proof}

\subsection{Microlocal energy currents}
\label{sec:microlocenergy}
In the remainder of \S \ref{sec:iledfreq}, we will suppress the $\omega$-dependence appearing in the functions in \eqref{eq:radialODE}, i.e.\ we denote $u(r)=u_{m \ell}(\omega,r)$, $H=H_{m\ell}(\omega,r)$ and $V(r)=V_{\homega \ell }(r)$. We will analogously denote  $v(r)=v_{m \ell}(\omega,r)$,  $w(r)=w_{m \ell}(\omega,r)$ and $\widetilde{V}(r)=\widetilde{V}_{\tomega \ell }(r)$. 

We will derive $L^2$-type estimates for $u$ in terms of $L^2$-type estimates for $H$ by considering the following \emph{microlocal energy currents}, which are frequency-space analogues of physical-space energy currents arising from vector field multipliers, that can be used to prove energy estimates to solutions to wave equations:

For smooth functions $h,y,f,g_+,g_{\infty}, \chi: (r_+,r_c)\to \R$, we define:
\begin{align*}
j_1^h[u]:=&\:h \re(u' \overline{u})-\frac{1}{2}h'|u|^2,\\
j_2^y[u]:=&\: y(|u'|^2+(\homega^2-V)|u|^2)=y(|u'|^2+(\tomega^2-\tV)|u|^2)\\
j_3^f[u]:=&f(|u'|^2+(\homega^2-V)|u|^2)+f' \re(u' \overline{u})-\frac{1}{2}f''|u|^2,\\
\chi j^T[u]:=&\:-\homega \chi \re(i u'\overline{u}),\\
\chi j^K[u]:=&\:-\tomega \chi \re(i u'\overline{u}),\\
j^{g_+}_{+}[v]:=&\:g_+|v'|^2-g_+\ell(\ell+1)\Omega^2r^{-2}|v|^2,\\
j^{g_{\infty}}_{\infty}[w]:=&\:g_{\infty}|w'|^2-g_{\infty}\ell(\ell+1)\Omega^2r^{-2}|w|^2,
\end{align*}
where we will either take $\chi\equiv 1$, $\chi=\chi_K$ (in the case $\chi j^K[u]$) or $\chi=\chi_T$ (in the case $\chi j^T[u]$).

Note that the functions $h,y,f,\chi, g_+,g_{\infty}$ will be frequency-independent. However, since we will be making different choices of functions for different frequency regimes, the resulting estimates could not be carried out directly in physical space, without a frequency decomposition. See however Appendix \ref{sec:purelyphysied}, where analogues of the above energy currents are applied in physical space, without a (time-)frequency decomposition, to obtain integrated energy estimates modulo zeroth order terms, or restricted to large angular frequencies.
\begin{lemma}
\label{lm:bulkterms}
The following identities hold:
\begin{align*}
\left(j_1^h[u]\right)'=&\:h(|u'|^2+(V-\homega^2)|u|^2)-\frac{1}{2}h''|u|^2+h \re(u\overline{H}),\\
\left(j_2^y[u]\right)'=&\:y'(|u'|^2+(\homega^2-V)|u|^2)-yV'|u|^2+2y\re(u' \overline{H})\\
=&\:y'(|u'|^2+(\tomega^2-\tV)|u|^2)-y\tV'|u|^2+2y\re(u' \overline{H}),\\
\left(j_3^f[u]\right)'=&\:2f'|u'|^2-fV'|u|^2-\frac{1}{2}f'''|u|^2+\re(f'u \overline{H}+2fu' \overline{H}),\\
\left(\chi j^T[u]\right)'=&\:-\homega \chi'  \re(\overline{i u'}u) +\homega \chi \re(iu \overline{ H} ),\\
\left(\chi j^K[u]\right)'=&\:-\tomega \chi' \re(\overline{iu'}u)+\tomega \chi \re(iu \overline{ H} ),\\
\left(j^{g_+}_{+}[v]\right)'=&\:g_{+}'|v' |^2-\ell(\ell+1)(g_{+}\Omega^2r^{-2})'|v|^2-2\q g_{+}\Omega^2r^{-2}\re(i\overline{ v}v')\\
&+2g_{+}\Omega^2\left[r^{-1}\frac{d\Omega^2 }{dr}+\frac{2\Lambda}{3}\right]\re(\overline{ v}v')-\re(2g_{+} v'e^{-i \tomega r_*-i\q\int_{0}^{r_*}\rho_+(r_*')\,dr_*'}\overline{H}),\\
\left(j^{g_{\infty}}_{\infty}[w]\right)'=&\:g_{\infty}'|w'|^2-\ell(\ell+1)(g_{\infty}\Omega^2r^{-2})'|w|^2-2\q g_{\infty}\Omega^2r^{-2}\re(i\overline{ w}w')\\
&+2g_{\infty}\Omega^2\left[r^{-1}\frac{d\Omega^2 }{dr}+\frac{2\Lambda}{3}\right]\re(\overline{ w}w')-\re(2g_{\infty} w'e^{i \homega r_*-i \q \int_{0}^{r_*}\rho_c(r_*')\,dr_*'}\overline{H}).
\end{align*}
\end{lemma}
\begin{proof}
The expressions for the derivatives of the currents $j_1^h$, $j_2^y$, $j_3^f$, $\chi j^T$ and $\chi j^K$ follow straightforwardly by applying \eqref{eq:radialODE}. 

To obtain $\left(j^{g_{\infty}}_{\infty}[w]\right)'$, we use that:
\begin{multline*}
e^{-i\homega r_*+i \q \int_{0}^{r_*}\rho_c(r_*')\,dr_*'}u''=e^{-i\homega r_*+i \q \int_{0}^{r_*}\rho_c(r_*')\,dr_*'}(e^{i\homega r_*-i \q \int_{0}^{r_*}\rho_c(r_*')\,dr_*'}w)''\\
=w''+2i(\homega -\q \rho_c(r_*))w'+\left[i\q r^{-2}\Omega^2-(\homega -\q \rho_c(r_*))^2\right]w.
\end{multline*}

Hence, $w$ satisfies the following ODE:
\begin{equation*}
w'' +2i(\homega -\q \rho_c(r_*))w'- [V-2 \q \homega \rho_c(r_*)+\q^2 \rho_c^2(r_*)-i\q r^{-2}\Omega^2]w=e^{-i\homega r_*+i \q \int_{0}^{r_*}\rho_c(r_*')\,dr_*'}H.
\end{equation*}
We then observe that
\begin{equation*}
V-2 \q \homega \rho_c(r_*)+\q^2 \rho_c^2(r_*)-i\q r^{-2}\Omega^2=\ell(\ell+1)\Omega^2r^{-2}+\Omega^2\left[r^{-1}\frac{d\Omega^2 }{dr}+\frac{2\Lambda}{3}\right]-i \q r^{-2}\Omega^2.
\end{equation*}
The expression for $\left(j^{g_{\infty}}_{\infty}[w]\right)'$ then follows by a straightforward application of the Leibniz rule.

Analogously, we obtain:
\begin{equation*}
v'' +2i(-\tomega -\q \rho_+(r_*))v'- [\widetilde{V}+2 \q \tomega \rho_+(r_*)+\q^2 \rho_+^2(r_*)-i\q r^{-2}\Omega^2]v=e^{+i \tomega r_*+i \q \int_{0}^{r_*}\rho_+(r_*')\,dr_*'}H
\end{equation*}
and we observe that
\begin{equation*}
\widetilde{V}+2 \q \tomega\rho_+(r_*)+\q^2 \rho^2_+(r_*)-i\q r^{-2}\Omega^2=\ell(\ell+1)\Omega^2 r^{-2}+\Omega^2\left[r^{-1}\frac{d\Omega^2 }{dr}+\frac{2\Lambda}{3}\right]-i \q r^{-2}\Omega^2.
\end{equation*}
The expression for $\left(j^{g_{+}}_{\infty}[v]\right)'$ then follows by a straightforward integration of the Leibniz rule.
\end{proof}
 
 \begin{lemma}
 \label{lm:superradiance}
 	We have that
 	\begin{align*}
 		j^T[u](\infty)=&\:\homega^2 |u|^2(\infty),\\
 		j^T[u](-\infty)=&\:-\homega \tomega |u|^2(-\infty),\\
 		j^K[u](\infty)=&\:\homega \tomega |u|^2(\infty),\\
 		j^K[u](-\infty)=&\:-\tomega^2|u|^2(-\infty).
 	\end{align*}
 	In particular, $j^T[u](\infty)-j^T[u](-\infty)$ or $j^K[u](\infty)-j^K[u](-\infty)$ are non-negative definite if and only if
 	\begin{equation*}
 		\homega \tomega>0.
 	\end{equation*}
 \end{lemma}
\begin{proof}
By \eqref{eq:inhomodebc1} and \eqref{eq:inhomodebc2}, we have that
\begin{align*}
	-\homega\lim_{r_*\to \infty}\re(i u'\overline{u})(r_*)=&\:\homega^2|u|^2(\infty),\\
	-\homega\lim_{r_*\to -\infty}\re(i u'\overline{u})(r_*)=&-\homega \tomega |u|^2(-\infty),
\end{align*}
which gives the stated limits for $j^T[u]$. The limits for $j^K[u]$ follow analogously.
\end{proof}
 \begin{definition}
  	We refer to frequencies satisfying $\homega \tomega<0$ as \emph{superradiant frequencies} and frequencies satisfying $\homega \tomega>0$ as \emph{non-superradiant frequencies}.
 \end{definition}

\subsection{$(\Omega^{-1}r)^{p}$-weighted energy estimates in frequency space}
In this section, we derive $L^2$-estimates with $(\Omega^{-1}r)^{p}$-weights in the regions $\{r\geq R_1\}$ and $\{r\leq r_1\}$, respectively, with $r_1\geq r_+$ and $R_1\leq r_c$ chosen, such that $\rho_c(R_1)$ and $\rho_+(r_1)$ are suitably small. 

\begin{proposition}
\label{prop:rweightestfreq}
Let $p\in (0,2)$. For $1<R_1<r_c$ with $\rho_c(R_1)$ suitably small, there exist a constant $c,C=c(p)>0$ and a numerical constant $C>0$, such that:
\begin{multline}
\label{eq:rpestwithbadterm}
c \int_{r_*(R_1)}^{\infty} \left[(\rho_c+\kappa_c)(r^{-1}\Omega)^{-p}|w'|^2+\ell(\ell+1)(\rho_c+\kappa_c)(r^{-1}\Omega)^{2-p}|w|^2\right]\,dr_*\\
\leq C p^{-1} \q^2 \int_{r_*\left(\frac{R_1}{2-r_c^{-1}R_1}\right)}^{\infty} \rho_c(r^{-1}\Omega)^{2-p}|w|^2\,dr_*\\
+C p^{-1}\ell(\ell+1)\int_{r_*\left(\frac{R_1}{2-r_c^{-1}R_1}\right)}^{r_*(R_1)}(\rho_c+\kappa_c)(r^{-1}\Omega)^{2-p}|w|^2\,dr_*-\int_{\R}\re(2 \chi_{R_1} (r^{-1}\Omega)^{-p}w'e^{i \homega r_*-i \q \int_{0}^{r_*}\rho_c(r_*')\,dr_*'}\overline{H})\,dr_*.
\end{multline}
Furthermore, if $\ell\geq L$, with $L\in \N_0$ suitably large, depending on $\q$, then:
\begin{multline}
\label{eq:rpestwobadterm}
c \int_{r_*(R_1)}^{\infty} \left[(\rho_c+\kappa_c) (r^{-1}\Omega)^{-p}|w'|^2+\ell(\ell+1)(\rho_c+\kappa_c)(r^{-1}\Omega)^{2-p}|w|^2\right]\,dr_*\\
\leq C p^{-1}\ell(\ell+1)\int_{r_*\left(\frac{R_1}{2-r_c^{-1}R_1}\right)}^{r_*(R_1)}(\rho_c+\kappa_c)(r^{-1}\Omega)^{2-p}|w|^2\,dr_*-\int_{\R}\re(2 \chi_{R_1} (r^{-1}\Omega)^{-p}w'e^{i \homega r_*-i \q \int_{0}^{r_*}\rho_c(r_*)\,dr_*'}\overline{H})\,dr_*.
\end{multline}
\end{proposition}
\begin{proof}
Let $g_{\infty}=\chi_{R_1}(r^{-1}\Omega)^{-p}=\chi_{R_1}(r^{-2}\Omega^2)^{-\frac{p}{2}}$ with $0<p<2$. By Lemma \ref{lm:metricest}, it follows that:
\begin{align*}
r^{-2}\Omega^2=& (\rho_c+2\kappa_c) \rho_c+\kappa_c O_{\infty}(\rho_c^2)+O_{\infty}(\rho_c^3),\\
\frac{d}{dr}(r^{-2}\Omega^2)=&\: -2r^{-2}(\rho_c+\kappa_c)+\kappa_c r^{-2} O_{\infty}(\rho_c)+r^{-2}O_{\infty}(\rho_c^2).
\end{align*}
Furthermore, we can estimate:
\begin{multline*}
2\Omega^2\left[r^{-1}\frac{d\Omega^2 }{dr}+\frac{2\Lambda}{3}\right]=2r^{-2}\Omega^2\left[ r\frac{d(1-2Mr^{-1}+Q^2r^{-2}-\frac{\Lambda}{3}r^2) }{dr}+\frac{2\Lambda}{3}r^2\right]=4r^{-2}\Omega^2r^{-1}(M-Q^2r^{-1})\\
=r^{-2}\Omega^2 (O_{\infty}(\rho_c)+\kappa_c O_{\infty}(\rho_c^0)).
\end{multline*}

Then we integrate $\left(j^{g_{\infty}}_{\infty}[w]\right)'$ and apply Lemma \ref{lm:bulkterms}, using that the boundary conditions on $u$ at $r_*=\infty$ imply that $w'=\Omega^2O_{\omega,\ell}(r^{-2})$, so the boundary terms $j^{g_{\infty}}_{\infty}[w](\pm \infty)$ vanish. We obtain:
\begin{multline}
\label{eq:mainrpid}
\int_{\R}\left[\left(\rho_c(p+O(\rho_c))+\kappa_c(p+O(\rho_c))\right)\chi_{R_1}+\chi'_{R_1}\right](r^{-1}\Omega)^{-p}|w'|^2\\
+(2-p)\ell(\ell+1)\left(\rho_c(1+O(\rho_c))+\kappa_c(1+O(\rho_c))\right)\chi_{R_1}(r^{-1}\Omega)^{2-p}|w|^2\,dr_*\\
=2\q \int_{\R}\chi_{R_1}(r^{-1}\Omega^2)^{2-p}\re(i\overline{ w}w')\,dr_*-\int_{\R}\chi_{R_1}(r^{-2}\Omega^2)^{2-p} (O_{\infty}(\rho_c)+\kappa_c O_{\infty}(\rho_c^0))\re(i\overline{ w}w')\\
+\ell(\ell+1)\int_{\R}\chi_{R_1}'(r^{-1}\Omega)^{2-p}|w|^2\,dr_*+\int_{\R}\re(2\chi_{R_1}(r^{-1}\Omega)^{-p}w'e^{i \homega r_*-i \q \int_{0}^{r_*}\rho_c(r_*')\,dr_*'}\overline{H})\,dr_*.
\end{multline}
We apply Young's inequality to obtain the following: there exists a constant $C>0$, such that for $\mu\in (0,2)$:
\begin{multline}
\label{eq:youngrpest}
2(|\q|+C(\rho_c+\kappa_c))(r^{-1}\Omega)^{2-p}|\re(i\overline{ w}w')|\leq \frac{p\mu}{2} (\rho_c+2\kappa_c)(r^{-1}\Omega)^{-p}|w'|^2+\frac{2}{p\mu }(\q^2\rho_c+\frac{r^{-2}\Omega^2}{\rho_++2\kappa_c})(r^{-1}\Omega)^{2-p}|w|^2\\
\leq \frac{p\mu}{2} (\rho_c+2\kappa_c)(r^{-1}\Omega)^{-p}|w'|^2+\frac{2}{p\mu }(|\q|+C(\rho_c+\kappa_c))^2\rho_c(1+O(\kappa_c)+O(\rho_c))(r^{-1}\Omega)^{2-p}|w|^2.
\end{multline}
Note that for fixed $p\in (0,2)$ and $\q\in \R$, there exists an $L\in \N$ such that for $\ell\geq L$, we can absorb the second term on the very RHS of \eqref{eq:youngrpest} into $(2-p)\ell(\ell+1)\rho_c(r^{-1}\Omega)^{2-p}|w|^2$. 

For general $\ell$, we will make use of a Hardy-type inequality to absorb integral of the terms on the very RHS of \eqref{eq:youngrpest} that are non-vanishing when $\q=0$.
First, observe that:
\begin{multline*}
\int_{\R} \rho_c(\rho_c+\kappa_c)^2(r^{-1}\Omega)^{2-p}\chi_{R_1}|w|^2\,dr_*\leq C\int_{\R}  \frac{d(r^{-2}\Omega^2)}{d\rho_c}(r^{-2}\Omega^2)^{2-\frac{p}{2}}\chi_{R_1}|w|^2\,dr_*\\
=C\int_{\R}  -\frac{d(r^{-2}\Omega^2)}{dr}(r^{-2}\Omega^2)^{1-\frac{p}{2}}\chi_{R_1}|w|^2\,dr= \frac{C}{2-\frac{p}{2}}\int_1^{r_c} -\frac{d}{dr}\left((\Omega^2r^{-2})^{2-\frac{p}{2}}\right)\chi_{R_1}|w|^2\,dr.
\end{multline*}
We integrate by parts and estimate:
\begin{multline*}
\int_1^{r_c} -\frac{d}{dr}\left((\Omega^2r^{-2})^{2-\frac{p}{2}}\right)\chi_{R_1}|w|^2\,dr=\int_1^{r_c} (\Omega^2r^{-2})^{2-\frac{p}{2}}\frac{d}{dr}\left(\chi_{R_1}|w|^2\right)\,dr\\
\leq 2\int_{\R} (r^{-1}\Omega )^{4-p}\chi_{R_1}|w||w'|\,dr_*+\int_{\R} (r^{-1}\Omega )^{4-p}\chi_{R_1}'|w|^2\,dr_*\\
\leq \frac{p(2-\frac{p}{2})}{8C}\int_{\R} (\rho_c+\kappa_c)(r^{-1}\Omega)^{-p}\chi_{R_1}|w'|^2\,dr_*+\frac{8C}{p(2-\frac{p}{2})}\int_{\R} (\rho_c+\kappa_c)^{-1}(r^{-1}\Omega)^{8-p}\chi_{R_1}|w|^2\,dr_*\\
+\int_{\R}(r^{-1}\Omega )^{4-p}\chi_{R_1}'|w|^2\,dr_*.
\end{multline*}
The second term on the very right-hand side above can be absorbed into the left-hand side for $p\in (0,4)$ and sufficiently small $\rho_c(R_1)$ and $\mu$, so we conclude that:
\begin{equation*}
\int_{\R} \rho_c(\rho_c+\kappa_c)^2(r^{-1}\Omega)^{2-p}\chi_{R_1}|w|^2\,dr_*\leq \frac{p}{8}\int_{\R}r^{p-1}(\Omega^2)^{-\frac{p}{2}}\chi_{R_1}|w'|^2\,dr_*+C\int_{\R} (r^{-1}\Omega )^{4-p}\chi_{R_1}'|w|^2\,dr_*.
\end{equation*}
Finally, the first term on the RHS above can be absorbed into the LHS of \eqref{eq:mainrpid}.
\end{proof}

\begin{proposition}
\label{prop:rmin1pweightestfreq}
Let $p\in (0,2)$. For $1<r_1<r_c$ with $\rho_+(r_1)$ suitably small, there exist a constant $c,C=c(p)>0$ and a numerical constant $C>0$, such that:
\begin{multline}
\label{eq:rmin1pestwithbadterm}
c \int_{-\infty}^{r_*(r_1)} \left[(\rho_++\kappa_+)(r^{-1}\Omega)^{-p}|v'|^2+\ell(\ell+1)(\rho_++\kappa_+)(r^{-1}\Omega)^{2-p}|v|^2\right]\,dr_*\\
\leq C p^{-1} \q^2 \int_{-\infty}^{r_*(\frac{r_i}{2-r_+^{-1}r_1})} \rho_+(r^{-1}\Omega)^{2-p}|v|^2\,dr_*\\
+C p^{-1}\ell(\ell+1)\int_{r_*(r_1)}^{r_*(\frac{r_1}{2-r_+^{-1}r_1})}(\rho_++\kappa_+)(r^{-1}\Omega)^{2-p}|v|^2\,dr_*-\int_{\R}\re(2 \chi_{r_1} (r^{-1}\Omega)^{-p}w'e^{-i \tomega r_*-i \q \int_{0}^{r_*}\rho_+(r_*')\,dr_*'}\overline{H})\,dr_*.
\end{multline}
Furthermore, if $\ell\geq L$, with $L\in \N_0$ suitably large, depending on $\q$, then:
\begin{multline}
\label{eq:rmin1estwobadterm}
c \int_{-\infty}^{r_*(r_1)} \left[(\rho_++\kappa_+)(r^{-1}\Omega)^{-p}|v'|^2+\ell(\ell+1)(\rho_++\kappa_+)(r^{-1}\Omega)^{2-p}|v|^2\right]\,dr_*\\
\leq C p^{-1}\ell(\ell+1)\int_{r_*(r_1)}^{r_*(\frac{r_1}{2-r_+^{-1}r_1})}(\rho_++\kappa_+)(r^{-1}\Omega)^{2-p}|v|^2\,dr_*-\int_{\R}\re(2 \chi_{r_1} (r^{-1}\Omega)^{-p}w'e^{-i \tomega r_*+i q \int_{0}^{r_*}\rho_+(r_*')\,dr_*'}\overline{H})\,dr_*.
\end{multline}
\end{proposition}
\begin{proof}
We repeat the arguments in the proof of Proposition \ref{eq:rpestwithbadterm}, with $\rho_+$ taking on the role of $\rho_c$, $v$ taking on the role of $w$ and with  $g_{+}=\chi_{r_1}(r^{-1}\Omega)^{-p}$.
\end{proof}

\subsection{Frequency domains}
\label{sec:freqdom}
We will divide the frequency domain $\R_{\omega}\times (\N_0)_{\ell}$ into sub-domains, where we will apply different choices of microlocal energy currents from \S \ref{sec:microlocenergy}.
Let $L_0\in \N_0$ and $\alpha,\beta,\gamma\in \R_+$ be dimensionless constants. We will later choose $L_0$ to be sufficiently large and $\alpha,\beta,\gamma\in \R_+$ to be suitably small. Define:
\begin{align*}
\mathcal{F}_{\sharp,{\rm angular}}:=&\:\left \{(\omega,\ell)\in \R\times \N_0\,|\, \ell\geq L_0+1, r_+^2\homega^2\leq \alpha \ell(\ell+1) \right\},\\
\mathcal{F}_{\sharp,{\rm trap}}:=&\:\left \{(\omega,\ell)\in \R\times \N_0\,|\, \ell\geq L_0+1, \alpha \ell(\ell+1)< r_+^2\homega^2\leq \beta^{-1}\ell(\ell+1)\right\},\\
\mathcal{F}_{\sharp,\rm{time}}:=&\:\left \{(\omega,\ell)\in \R\times \N_0\,|\, \ell\in \N_0, r_+^2\homega^2>\beta^{-1}(1+\ell(\ell+1)) \right\},\\
\mathcal{F}_{\flat,+}:=&\:\left\{(\omega,\ell)\in \R\times \N_0\,|\, \ell\leq L_0,\: r_+^2\tomega^2\leq \gamma,\, \gamma<r_+^2\homega^2 \leq \beta^{-1}(1+L_0(L_0+1)) \right\},\\
\mathcal{F}_{\flat,\sim }:=&\:\left\{(\omega,\ell)\in \R\times \N_0\,|\, \ell\leq L_0,\: r_+^2\tomega^2> \gamma,\, \gamma<r_+^2\homega^2 \leq \beta^{-1}(1+L_0(L_0+1)) \right\},\\
\mathcal{F}_{\flat,\infty}:=&\:\left\{(\omega,\ell)\in \R\times \N_0\,|\, \ell\leq L_0,\: r_+^2\homega^2\leq \gamma \right\}.
\end{align*}
We moreover denote $\mathcal{F}_{\sharp}:=\mathcal{F}_{\sharp,{\rm angular}}\cup \mathcal{F}_{\sharp,{\rm trap}}\cup \mathcal{F}_{\sharp, {\rm time}}$ and $\mathcal{F}_{\flat}:=\mathcal{F}_{\flat,+}\cup \mathcal{F}_{\flat,\sim } \cup \mathcal{F}_{\flat,\infty}$.

Let $\q_{1}>0$. We will restrict to $|\q|<\q_{1}$. Let $\q_0>0$. We will further restrict to $\q_0<|\q|<\q_{1}$ when considering $\mathcal{F}_{\flat}$. We will allow any uniform constant in the estimates in the sections below to depend on $\q_{1}$ and we will specify when allow the constants to depend also on $\q_0$.

The frequency domain $\mathcal{F}_{\sharp}$ can be considered a high-frequency domain. Superradiant frequencies in $\mathcal{F}_{\sharp}$ are entirely contained $\mathcal{F}_{\sharp,{\rm angular}}$, because $\homega \tomega>0$ in $\mathcal{F}_{\sharp,{\rm trap}}$ and $\mathcal{F}_{\sharp, {\rm time}}$.

The frequency domain $\mathcal{F}_{\flat}$ consists of bounded superradiant and non-superradiant frequencies.

In the high-frequency domain $\mathcal{F}_{\sharp}$, we will only apply microlocal energy currents to derive the desired estimates for $u$. In the bounded frequency domain $\mathcal{F}_{\flat}$, we will instead appeal to Wronskian estimates and pointwise estimates of homogeneous solutions to \eqref{eq:radialODE}. In particular, the analysis of $\mathcal{F}_{\flat,\sim }$ is closely connected to the validity of mode stability on the real axis, away from $\homega=0$ and $\tomega=0$, for \eqref{eq:CSF}.

We will treat each of the above frequency domains independently in the sections below.

\section{Integrated estimates in frequency space: large frequencies}
\label{sec:iedlargefreq}
In this section, we consider the large frequency domain $\mathcal{F}_{\sharp}$. We treat relatively large $\omega$-dominated frequencies $\mathcal{F}_{\sharp,{\rm time}}$ in \S \ref{sec:largeomega}, relatively large $\ell$-dominated frequencies $\mathcal{F}_{\sharp, {\rm angular}}$ in \S \ref{sec:largel} and frequencies where $\omega$ and $\ell$ are large and $\omega^2$ is comparable to $\ell(\ell+1)$ $\mathcal{F}_{\sharp,{\rm trap}}$ in \S \ref{sec:trapped}.

\textbf{In the remainder of \S \ref{sec:iedlargefreq}, we will assume that $r_+=1$, for notational convenience.}

\subsection{$\mathcal{F}_{\sharp,{\rm time}}$: $\omega$-dominated frequencies}
\label{sec:largeomega}
In the frequency domain $\mathcal{F}_{\sharp,{\rm time}}$, we can exploit that $\homega^2-V$ is globally very large and therefore the estimates are not sensitive to the precise form of the potential $V$.
\begin{proposition}
\label{prop:highfreqomega}
Let $E\geq 2$ and $1\leq  s<2$. Then, for suitably small $\beta>0$, there exists a constant $c>0$, such that for $y: [r_+,r_c)_r\to \R$, with $y(r)=(1-r_+r_c^{-1})^{-2+s}(\rho_+(r)^{2-s}-\rho_c(r)^{2-s})$:
\begin{multline}
\label{eq:highfreqomega}
c\int_{\R}r^{-2}\Omega^{2}\left(\rho_c^{1-s}+\rho_+^{1-s}\right) (|u'|^2+(1+\ell(\ell+1)+\homega^2) |u|^2)\,dr_*=-\int_{\R}2y\re(u' \overline{H})\,dr_*\\
+E\int_{\R}\tomega \chi_{K} \re(i u \overline{H} )\,dr_*+E\int_{\R}\homega \chi_{T}\re(i u \overline{H} )\,dr_*
\end{multline}
\end{proposition}
\begin{proof}
Consider the microlocal energy current $j_2^y[u]$. Note that
\begin{equation*}
y'(r)=(2-s)(1-r_+r_c^{-1})^{-2+s}r^{-2}\Omega^{2}\left(\rho_c(r)^{1-s}+\rho_+(r)^{1-s}\right).
\end{equation*}

Integrating $(j_2^y[u])'$ then gives:
\begin{equation*}
\int_{\R}y'(|u'|^2+(\homega^2-V)|u|^2)-yV'|u|^2\,dr_*=2(\tomega^2|u|^2(-\infty)+\homega^2|u|^2(\infty))-\int_{\R}y\re(u' \overline{H})\,dr_*.
\end{equation*}
Using that $\homega^2\geq \beta^{-1 }(1+\ell(\ell+1))$ and $1\leq s<2$, we can take $\beta$ suitably small to ensure that there exists a constant $c>0$ such that:
\begin{equation*}
y'(|u'|^2+(\homega^2-V)|u|^2)-yV'|u|^2\geq c (2-s)r^{-2}\Omega^{2}\left(\rho_c^{1-s}+\rho_+^{1-s}\right)\left[|u'|^2+(1+\ell(\ell+1)+\homega^2) |u|^2\right].
\end{equation*}
We therefore obtain
\begin{multline*}
c\int_{\R} (2-s)r^{-2}\Omega^{2}\left(\rho_c^{1-s}+\rho_+^{1-s}\right)(|u'|^2+(1+\ell(\ell+1)+\homega^2) |u|^2)\,dr_*\\
=2(\tomega^2|u|^2(-\infty)+\homega^2|u|^2(\infty))-\int_{\R}y\re(u' \overline{H})\,dr_*.
\end{multline*}
We can integrate $E (\chi_Tj^T[u])'$ or $E (\chi_Kj^K[u])'$ with $E\geq 2$ and $\beta$ sufficiently large to estimate the boundary terms on the right-hand side and absorb the bulk terms involving $\chi'_T$ and $\chi_K'$ to conclude \eqref{eq:highfreqomega}.
\end{proof}

From Proposition \ref{prop:highfreqomega}, it follows moreover that we extend the domain of validity of the $r$-weighted and $(r-1)^{-1}$-weighted estimates from  that are valid for large $\ell$ to the frequency domain $\mathcal{F}_{\sharp,{\rm time}}$ (which includes bounded $\ell$).
\begin{corollary}
\label{cor:highfreqomegawrp}
Let $E\geq 2$, $1\leq s<2$ and $p\leq s$. Then, for $\beta$ suitably small, there exists a constant $c>0$, such that for $y: [r_+,r_c)_r\to \R$, with $y(r)=(1-r_+r_c^{-1})^{-2+s}(\rho_+(r)^{2-s}-\rho_c(r)^{2-s})$:
\begin{multline}
\label{eq:highfreqomegawrp}
c\int_{\R}r^{-2}\Omega^{2}\left(\rho_c^{1-s}+\rho_+^{1-s}\right)(|u'|^2+(1+\ell(\ell+1)+\homega^2) |u|^2)\,dr_*\\
+c \int_{-\infty}^{r_*(r_1)}\left[(\rho_++\kappa_+)(r^{-1}\Omega)^{-p} |v'|^2+(1+\ell(\ell+1))(\rho_++\kappa_+)(r^{-1}\Omega)^{2-p} |v|^2\right]\,dr_*\\
+c\int_{r_*(R_1)}^{\infty} \left[(\rho_c+\kappa_c)(r^{-1}\Omega)^{-p} |w'|^2+\ell(\ell+1)(\rho_c+\kappa_c)(r^{-1}\Omega)^{2-p} |w|^2\right]\,dr_*\\
\leq -2\int_{\R}y\re(u' \overline{H})\,dr_*+E\int_{\R}\tomega \chi_{K} \re(i u \overline{H})\,dr_*+E\int_{\R}\homega \chi_{T}\re(i u\overline{H})\,dr_*\\
- \int_{\R}\re(2 \chi_{r_1}(r^{-1}\Omega)^{-p} v' e^{-i \tomega r_*-i\q\int_{0}^{r_*}\rho_+(r_*')\,dr_*'}\overline{H}) \\
- \int_{\R}\re(2 \chi_{R_1} (r^{-1}\Omega)^{-p}w'  e^{i \homega r_*-i \q \int_{0}^{r_*}\rho_c(r_*')\,dr_*'}\overline{H})\,dr_*.
\end{multline}
\end{corollary}
\begin{proof}
Let $s=p$ and combine \eqref{eq:highfreqomega} with \eqref{eq:rpestwithbadterm} and \eqref{eq:rmin1pestwithbadterm} and use that $\homega^2\gg \ell(\ell+1)$ to absorb the terms appearing with a factor $q^2$ and $\chi'_{R_1}$ or $\chi'_{r_1}$ on the right-hand side of \eqref{eq:rpestwithbadterm} and \eqref{eq:rmin1pestwithbadterm}.
\end{proof}

\subsection{$\mathcal{F}_{\sharp,{\rm angular}}$: large $\ell$, superradiant frequencies}
\label{sec:largel}
This frequency regime contains the superradiant frequencies, so we cannot appeal to $j^T[u]$ or $j^K[u]$ currents to directly control boundary terms. Instead, we will exploit that $V-\omega^2$ is very large near the maximum of the potential to absorb error terms generated by the \emph{localized} currents: $\chi_{T}j^T[u]$ and $\chi_{K}j^K[u]$. This can be thought of as a quantitative realization of the fact that \emph{superradiant frequencies are not trapped}.

We will moreover need to couple estimates in $\mathcal{F}_{\sharp,{\rm angular}}$ with the $(\Omega^{-1}r)^{p}$-weighted estimates from Propositions \ref{prop:rweightestfreq} and \ref{prop:rmin1pweightestfreq}, which are valid for sufficiently large $\ell$. This is a characteristic feature of the fact that we wish to include the extremal cases $\kappa_+=0$ and $\kappa_c=0$ in our estimates, and it would be absent if we restricted to $\kappa_+>0$ and $\kappa_c>0$.\footnote{In the setting of the neutral wave equation on Kerr in a near-extremal setting, the coupling also appears in \cite{gaj23}, but it does not appear in the purely sub-extremal analysis of \cite{part3}.}

\begin{figure}[h!]
	\begin{center}
\includegraphics[scale=0.6]{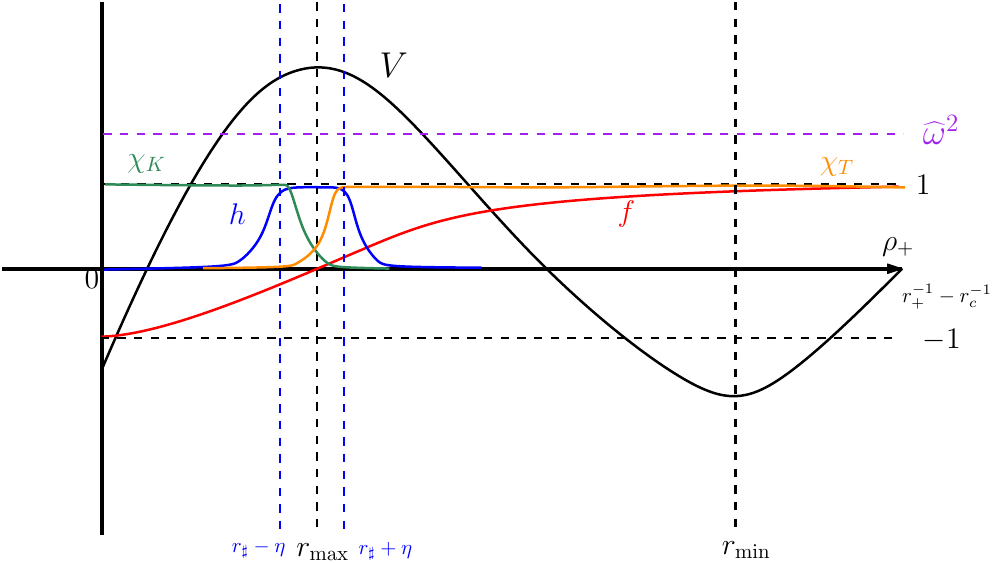}
\end{center}
\caption{A sketch of the graph of the potential $V$ as a function of $\rho_+$ in the frequency regime $\mathcal{F}_{\sharp,{\rm angular}}$ with $\q\homega<0$ and $\q\tomega<0$, $\kappa_c$ suitably small and sketches of the graphs of the functions $f$, $h$, $\chi_T$ and $\chi_K$ that appear in the currents $j_2^h[u]$, $j_3^f[u]$, $\chi_Tj^T[u]$ and $\chi_Kj^K[u]$.}
	\label{fig:fhKTcurrents}
\end{figure}

\begin{proposition}
\label{prop:srhighfreqest}
Let $E\geq 2$, $\delta,\eta>0$ and $1\leq p<2$. Then there exist functions $f,h: \R_{r_*}\to \R$ such that $f(r_*)=-1+(r^{-1}\Omega)^{\delta}(r(r_*))$ for $r_*\leq (r_3)_*$ and $f(r_*)=1-(r^{-1}\Omega)^{\delta}(r(r_*))$ for $r_*\geq (R_3)_*$, and $h$ is compactly supported in $(r_*(r_{\sharp}-2\eta), r_*(r_{\sharp}+2\eta))$ and there exists a constant $c>0$, such that for $L_0$ suitably large and $\alpha>0$ suitably small:
  \begin{multline}
  \label{eq:srhighfreqest}
 c \int_{(r_2)_*}^{(R_2)_*} |u'|^2+\left[\homega^2+\ell(\ell+1)+1\right]|u|^2\,dr_*\\
+c  \int_{-\infty}^{r_*(r_1)} \Big[(\rho_++\kappa_+)(r^{-1}\Omega)^{-p} |v'|^2+(1+\ell(\ell+1)+\homega^2)(\rho_++\kappa_+)(r^{-1}\Omega)^{2-p} |v|^2\\
+(\rho_++\kappa_+)(r^{-1}\Omega)^{\delta}|u'|^2\Big]\,dr_*\\
+c  \int_{r_*(R_1)}^{\infty} \Big[(\rho_c+\kappa_c)(r^{-1}\Omega)^{-p} |w'|^2+(1+\ell(\ell+1)+\homega^2)(\rho_c+\kappa_c)(r^{-1}\Omega)^{2-p} |w|^2\\
+(\rho_c+\kappa_c)(r^{-1}\Omega)^{\delta}|u'|^2\Big]\,dr_* \leq -\int_{\R}\Re((f'+h)u \overline{H}+2fu' \overline{H})\,dr_*\\
- \int_{\R}\re(2 \chi_{r_1} (r^{-1}\Omega)^{-p} v' e^{-i \tomega r_*+i\q\int_{0}^{r_*}\rho_+(r_*')\,dr_*'}\overline{H}) - \int_{\R}\re(2 \chi_{R_1} (r^{-1}\Omega)^{-p} w'  e^{i \homega r_*-i \q \int_{0}^{r_*}\rho_c(r_*')\,dr_*'}\overline{H})\,dr_*,\\
+E\int_{\R}\tomega \chi_{K} \re(i u \overline{H})\,dr_*+E\int_{\R}\homega \chi_{T}\re(i u \overline{H} )\,dr_*.
 \end{multline}
\end{proposition}
\begin{proof}
\textbf{$j_3^f[u]$-current estimates:}\\
By taking $\eta>0$ arbitrarily small, $L_0$ appropriately large and applying Lemma \ref{lm:criticalpointsVhighl}, we obtain the existence of a maximum of the potential $V$ at some $r_+<r_{\rm max}<r_c$, with $|r_{\rm max}-r_{\sharp}|\leq \frac{1}{2}\eta$.

Let $r_+<r_3<r_2<r_1<R_1<R_2<R_3<r_c$. We define the function $f: \R_{r_*}\to \R$ as follows: 
let $a>0$, $s>0$ and let $f(r_*)=-1+(r^{-2}\Omega^2)^{\frac{\delta}{2}}(r(r_*))$ for $r_*\leq r_*(r_3)$, $f(r_*)=1-(r^{-2}\Omega^2)^{\frac{\delta}{2}}(r(r_*))$ for $r_*\geq r_*(R_3)$ and $f(r_*)=\frac{2 }{\pi}\arctan (a(r_*-r_*(r_{\rm max})))$, for $r_*(r_2)\leq r\leq r_*(R_2)$. Then, for $r_2\leq r\leq R_2$:
\begin{align*}
f'(r_*)=&\:\frac{2a}{\pi (1+a^2(r_*-r_*(r_{\rm max}))^2)}\\
f''(r_*)=&\: -\frac{4a^3((r_*-r_*(r_{\rm max}))}{\pi (1+a^2(r_*-r_*(r_{\rm max}))^2)^2},\\
f'''(r_*)=&\: -\frac{4a^3(1-3a^2(r_*-r_*(r_{\rm max}))^2)}{\pi (1+a^2(r_*-r_*(r_{\rm max}))^2)^3},
\end{align*}
so $-f'''>0$ iff $|r_*-r_*(r_{\rm max})|< \frac{1}{\sqrt{3}}$.

For $r\leq r_3$ and $r\geq R_3$ we have that: 
\begin{align*}
f'=&\:\frac{\delta}{2}\Omega^2(r^{-2}\Omega^2)^{\frac{\delta}{2}-1}\frac{d}{dr}(r^{-2}\Omega^2)\geq \frac{\delta}{2}(r^{-2}\Omega^2)^{\frac{\delta}{2}}(\rho_++\kappa_+)\quad (r\leq r_3),\\
f'=&\:-\frac{\delta}{2}\Omega^2(r^{-2}\Omega^2)^{\frac{\delta}{2}-1}\frac{d}{dr}(r^{-2}\Omega^2)\geq \frac{\delta}{2}(r^{-2}\Omega^2)^{\frac{\delta}{2}}(\rho_c+\kappa_c)\quad (r\geq R_3).
\end{align*}

We can choose $f$ in $(r_3,r_2)\cup (R_2,R_3)$ such that $f$ is a $C^3$ function, $f'>0$ globally, and there exists a constant $C>0$ such that globally:
\begin{align*}
|f'''(r_*)|\leq &\: C |r_*|^{-3}.
\end{align*}
 See Figure \ref{fig:fhKTcurrents} for an illustration of the graph of $f$.
 
We then integrate $(j_3^f[u])'$ to obtain:
\begin{multline*}
\int_{\R}2f'|u'|^2-fV'|u|^2-\frac{1}{2}f'''|u|^2\,dr_*=\homega^2|u|^2(\infty)+|u'|^2(\infty)+\tomega^2|u|^2(-\infty)+|u'|^2(-\infty).\\
 -\int_{\R}\re(f'u \overline{H}+2fu' \overline{H})\,dr_*
\end{multline*}
Note that the boundary conditions on $u$ at $r_*=\pm \infty$ imply that:
\begin{equation*}
\homega^2|u|^2(\infty)+|u'|^2(\infty)+\tomega^2|u|^2(-\infty)+|u'|^2(-\infty)=2(\homega^2|u|^2(\infty)+\tomega^2|u|^2(-\infty)).
\end{equation*}

From the proof of Lemma  \ref{lm:criticalpointsVhighl} and the behaviour of $f$ near $r=r_{\rm max}$, it follows moreover that for $\ell\geq L_0$ with $L_0$ suitably large and $|q\homega|\ll \ell(\ell+1)$, which, for $L_0$ suitably large, is guaranteed by $\homega^2\leq \beta^{-1}\ell(\ell+1)$, there exist constants $b,B>0$, such that
\begin{equation*}
-fV'\geq b (\homega^2+\ell(\ell+1))(\rho_++\kappa_+)(\rho_c+\kappa_c) r^{-2}\Omega^2(1-r_{\max}r^{-1})^2+Bq\tomega r^{-2}\Omega^2\mathbf{1}_{r\leq r_2}-Bq\homega r^{-2}\Omega^2\mathbf{1}_{r\geq R_2}.
\end{equation*}
If $\q\homega>0$ or $\q\tomega<0$, the right-hand side above will be negative for sufficiently small $\rho_c$ and $\kappa_c$ or $\rho_+$ and $\kappa_+$, respectively. This reflects the possible existence of a minimum of $V$ for $\kappa_c$ suitably small (depending on $\ell$) at a sufficiently large value of $r$ when $q\homega>0$ and at a small value of $r-1$ when $\q\tomega<0$ and $\kappa_+$ is suitably small, corresponding to points 1. and 3. of Lemma  \ref{lm:criticalpointsVhighl}, respectively. See also Figure \ref{fig:fhKTcurrents} in the case $\q\homega<0$.

Since $-f'''$ has a good sign near $r=r_{\max}$ and decays sufficiently fast for $\rho_c$ or $\rho_+$ suitably small, we can in fact estimate:
\begin{equation*}
-fV'-\frac{1}{2}f'''\geq b(\rho_++\kappa_+)(\rho_c+\kappa_c)r^{-2}\Omega^2[( \ell(\ell+1)+\homega^2)(1-r_{\max}r^{-1})^2+1]+Bq\tomega r^{-2}\Omega^2\mathbf{1}_{r\leq r_2}-Bq\homega r^{-2}\Omega^2\mathbf{1}_{r\geq R_2}.
\end{equation*}
Hence, there exist a constants $c>0$, depending on $r_2,R_2$, such that:
\begin{multline*}
c\int_{\R}(r^{-2}\Omega^2)^{\frac{\delta}{2}}(\rho_++\kappa_+)(\rho_c+\kappa_c)|u'|^2\,dr_*+c \int_{r_*(r_2)}^{r_*(R_2)} \left[(\omega^2+\ell(\ell+1))(1-r_{\max}r^{-1})^2+1\right]|u|^2\,dr_*\\
=2(\homega^2|u|^2(\infty)+\tomega^2|u|^2(-\infty))- B\int_{-\infty}^{r_*(r_2)}|\q\tomega|\Omega^2r^{-2}|u|^2\,dr_*-B\int_{r_*(R_2)}^{\infty}|\q\omega|\Omega^2r^{-2}|u|^2\,dr_*\\
-\int_{\R}\re(f'u \overline{H}+2fu' \overline{H})\,dr_*.
\end{multline*}
It remains to control the boundary terms and the integrals over $(-\infty,r_*(r_2))$ and $(r_*(R_2),\infty)$ on the right-hand side. 
We control the latter by coupling the above $j_3^f[u]$-current estimate with the $(r^{-1}\Omega)^{-p}$-weighted estimates \eqref{eq:rpestwobadterm} and \eqref{eq:rmin1estwobadterm} from Propositions \ref{prop:rweightestfreq} and \ref{prop:rmin1pweightestfreq}, respectively, with $p\geq 1$. Here, we use again that $|q\homega|\ll \ell(\ell+1)$ and $|q\tomega|\ll \ell(\ell+1)$. We obtain:
\begin{multline}
\label{eq:fcurrentandrweightedest}
c\int_{\R}(r^{-2}\Omega^2)^{\frac{\delta}{2}}(\rho_++\kappa_+)(\rho_c+\kappa_c)|u'|^2\,dr_*+c \int_{r_*(r_2)}^{r_*(R_2)}|u'|^2+ \left[(\homega^2+\ell(\ell+1))(1-r_{\max}r^{-1})^2+1\right]|u|^2\,dr_*\\
+c \int_{-\infty}^{r_*(r_1)} \left[(\rho_++\kappa_+)(r^{-1}\Omega)^{-p}|v'|^2+\ell(\ell+1)(\rho_++\kappa_+)(r^{-1}\Omega)^{2-p}|v|^2\right]\,dr_*\\
+c \int_{r_*(R_1)}^{\infty} \left[(\rho_c+\kappa_c)(r^{-1}\Omega)^{-p}|w'|^2+\ell(\ell+1)(\rho_c+\kappa_c)(r^{-1}\Omega)^{2-p}|w|^2\right]\,dr_*=2(\homega^2|u|^2(\infty)+\tomega^2|u|^2(-\infty))\\
-\int_{\R}\re(2 \chi_{r_1} (r^{-1}\Omega)^{-p}v'e^{-i \tomega r_*-i\q\int_{0}^{r_*}\rho_+(r_*')\,dr_*'}\overline{H})\,dr_*\\
-\int_{\R}\re(2 \chi_{R_1}(r^{-1}\Omega)^{-p}w'e^{i \homega r_*-i \q \int_{0}^{r_*}\rho_c(r_*')\,dr_*'}\overline{H})\,dr_*-\int_{\R}\re(f'u \overline{H}+2fu' \overline{H})\,dr_*.
\end{multline}
Now we turn to the boundary terms $2(\homega^2|u|^2(\infty)+\tomega^2|u|^2(-\infty))$. Due to the fact that $\mathcal{F}_{\sharp,{\rm angular}}$ includes superradiant frequencies, we cannot simply control the boundary terms with the microlocal energy currents $j^T[u]$ or $j^K[u]$; see Lemma \ref{lm:superradiance}. Instead, we consider \emph{localized} currents $E \chi_{T}j^T[u]$ or $E \chi_{K}j^K[u]$, with $E\geq 2$ a constant, which will introduce additional error terms in the support of $\chi'_{K}$ and $\chi'_{T}$ that will subsequently need to be controlled via a $j_1^h[u]$-current supported near the maximum of $V$; see Figure \ref{fig:fhKTcurrents} for an illustration.

\textbf{$j_2^h[u]$-current estimates:}\\
We construct the function $h: [r_+,r_c)\to [0,1]$ as follows: let $h$ be a smooth function such that $h(r)=1$ for $r\in (r_{\sharp}-\eta,r_{\sharp}+\eta)$ and $h(r)=0$ for $r\leq r_{\sharp}-2\eta$ or $r\geq r_{\sharp}+2\eta$, with $\eta>0$ suitably small. We can moreover choose $h$ such that there exists a numerical constant $C>0$, such that:
\begin{equation*}
|h''(r)|\leq C \eta^{-2}.
\end{equation*}
By integrating $(j_2^h[u])'$ and applying Lemma \ref{lm:bulkterms}, we then obtain:
\begin{equation}
\label{eq:hcurrent}
\int_{\R}h(|u'|^2+(V-\omega^2)|u|^2)-\frac{1}{2}h''|u|^2\,dr_*=-\int_{\R}h \re(u\overline{H})\,dr_*.
\end{equation}

Using that $\ell\geq L_0$ with $L_0\gg 1$ and $\omega^2\ll \ell(\ell+1)$ (i.e.\ taking $\alpha$ suitably small), we conclude that for $r\in [r_{\sharp}-2\eta,r_{\sharp}+2\eta]$, we can estimate:
\begin{equation}
\label{eq:keypropforhcurrent}
V(r)-\omega^2 \geq \frac{1}{8}(r^{-2}\Omega^2)(r_{\sharp})\ell(\ell+1)
\end{equation}
and hence, for $r\in (r_{\sharp}-\eta,r_{\sharp}+\eta)$:
\begin{equation*}
h(V(r)-\homega^2) \geq \frac{1}{8}(r^{-2}\Omega^2)(r_{\sharp})\ell(\ell+1).
\end{equation*}
By combining \eqref{eq:hcurrent} with \eqref{eq:fcurrentandrweightedest}, we can absorb the term $-\frac{1}{2}h'''$ into the left-hand side of \eqref{eq:fcurrentandrweightedest}, using that, by construction, $r_{\rm max}\notin \textnormal{supp}\, h''$ and taking $L_0$ sufficiently large, depending on $c$ and $\eta$.

We then obtain:
\begin{multline}
\label{eq:fhcurrentandrweightedest}
c \int_{(r_2)_*}^{(R_2)_*}|u'|^2+(\homega^2+\ell(\ell+1)+1)|u|^2\,dr_*+c\int_{\R}(r^{-2}\Omega^2)^{\frac{\delta}{2}}(\rho_++\kappa_+)(\rho_c+\kappa_c)|u'|^2\,dr_*\\
+c \int_{-\infty}^{r_*(r_1)} \left[(\rho_++\kappa_+)(r^{-1}\Omega)^{-p}|v'|^2+\ell(\ell+1)(\rho_++\kappa_+)(r^{-1}\Omega)^{2-p}|v|^2\right]\,dr_*\\
+c \int_{r_*(R_1)}^{\infty} \left[(\rho_c+\kappa_c)(r^{-1}\Omega)^{-p}|w'|^2+\ell(\ell+1)(\rho_c+\kappa_c)(r^{-1}\Omega)^{2-p}|w|^2\right]\,dr_*=2(\homega^2|u|^2(\infty)+\tomega^2|u|^2(-\infty))\\
-\int_{\R}\re(2 \chi_{r_1} (r^{-1}\Omega)^{-p}v'e^{-i \tomega r_*-i\q\int_{0}^{r_*}\rho_+(r_*')\,dr_*'}\overline{H})\,dr_*-\int_{\R}\re(2 \chi_{R_1}(r^{-1}\Omega)^{-p}w'e^{i \homega r_*-i \q \int_{0}^{r_*}\rho_c(r_*')\,dr_*'}\overline{H})\,dr_*\\
-\int_{\R}\re((f'+h)u \overline{H}+2fu' \overline{H})\,dr_*.
\end{multline}

\textbf{$j^T[u]$- and $j^K[u]$-current estimates:}\\
Finally, we integrate $E (\chi_{T}j^T[u])'$ or $E (\chi_{K}j^K[u])'$ and use the defining properties of the cut-off functions $\chi_K$ and $\chi_T$ from \S \ref{sec:cutoffs} to obtain:
\begin{multline}
\label{eq:localTKcurrentest}
E(\omega^2|u|^2(\infty)+\tomega^2|u|^2(-\infty))=-E\int_{r_*(r_{\sharp}-\eta)}^{r_*(r_{\sharp}+\eta)}(\tomega \chi_K'+\homega \chi_T')  \re(\overline{i u'}u)\,dr_*+E\int_{\R}(\tomega \chi_K+\homega \chi_T) \re(iu \overline{H})\,dr_*
\end{multline}
We can estimate
\begin{align*}
\left|\tomega \chi_K'  \re(\overline{i u'}u)\right|\leq &\: \frac{c}{2}|u'|^2+2C^2\eta^{-2}c^{-1}\tomega^2|u|^2,\\
\left|\homega \chi_T'  \re(\overline{i u'}u)\right|\leq &\: \frac{c}{2}|u'|^2+2C^2\eta^{-2}c^{-1}\homega^2|u|^2
\end{align*}
and combine \eqref{eq:localTKcurrentest} with \eqref{eq:fhcurrentandrweightedest}, taking $\alpha$ suitably small and $L_0$ suitably large in order to conclude \eqref{eq:srhighfreqest}.
\end{proof}

\subsection{$\mathcal{F}_{\sharp,{\rm trap}}$: trapped frequencies}
\label{sec:trapped}
In the frequency domain $\mathcal{F}_{\sharp,{\rm trap}}$, we proceed analogously to $\mathcal{F}_{\sharp,{\rm angular}}$, but we use that the corresponding frequencies are \underline{non}-superradiant to integrate $E (j^T[u])'$ or $E (j^K[u])'$ without the need to use cut-off functions $\chi_K$ and $\chi_T$ and absorb error terms caused by $\chi_K'\neq 0$ and $\chi'_T\neq 0$.

When combining the estimates in the different frequency regimes, however, it will be convenient to consider $E \chi_{r_1}(j^K[u])'+E (1-\chi_{r_1})(j^T[u])'$, with $\chi_{r_1}$ a cut-off function with derivative supported close to $r=r_+$.
\begin{proposition}
\label{prop:trappedfreqest}
Let $E\geq 2$, $s>0$ and $1\leq p<2$. Then there exists a function $f: \R_{r_*}\to \R$ such that $f(r_*)=-1+(r^{-1}\Omega)^{s}(r(r_*))$ for $r_*\leq (r_3)_*$ and $f(r_*)=1-(r^{-1}\Omega)^{s}(r(r_*))$ for $r_*\geq (R_3)_*$, and there exists a constant $c>0$, such that for $L_0$ suitably large and $\alpha>0$ suitably small:
  \begin{multline}
  \label{eq:trappedfreqest}
 c \int_{(r_2)_*}^{(R_2)_*} |u'|^2+\left[(1-r_{\rm max}r^{-1})^2(\homega^2+\ell(\ell+1))+1\right]|u|^2\,dr_*\\
 +c  \int_{-\infty}^{r_*(r_1)} \Big[(\rho_++\kappa_+)(r^{-1}\Omega)^{-p} |v'|^2+(1+\ell(\ell+1)+\homega^2)(\rho_++\kappa_+)(r^{-1}\Omega)^{2-p} |v|^2\\
+(\rho_++\kappa_+)(r^{-1}\Omega)^{\delta}|u'|^2\Big]\,dr_*\\
+c  \int_{r_*(R_1)}^{\infty} \Big[(\rho_c+\kappa_c)(r^{-1}\Omega)^{-p} |w'|^2+(1+\ell(\ell+1)+\homega^2)(\rho_c+\kappa_c)(r^{-1}\Omega)^{2-p} |w|^2\\
+(\rho_c+\kappa_c)(r^{-1}\Omega)^{\delta}|u'|^2\Big]\,dr_* \leq -\int_{\R}\Re(f'u \overline{H}+2fu' \overline{H})\,dr_*\\
- \int_{\R}\re(2 \chi_{r_1} (r^{-1}\Omega)^{-p} v' e^{-i \tomega r_*+i\q\int_{0}^{r_*}\rho_+(r_*')\,dr_*'}\overline{H}) - \int_{\R}\re(2 \chi_{R_1} (r^{-1}\Omega)^{-p} w'  e^{i \homega r_*-i \q \int_{0}^{r_*}\rho_c(r_*')\,dr_*'}\overline{H})\,dr_*,\\
+E\int_{\R}\tomega \re(i \chi_{r_1}u\overline{H} )\,dr_*+E\int_{\R}\homega \re(i (1-\chi_{r_1})u \overline{H} )\,dr_*.
 \end{multline}
\end{proposition}
\begin{proof}
We first proceed exactly as in the proof of Proposition \ref{prop:srhighfreqest}, by coupling $j_3^f[u]$-current, $j_+^{g_+}[v]$-current and $j_{\infty}^{g_{\infty}}[w]$-current estimates to obtain \eqref{eq:fcurrentandrweightedest}.

Since $\omega \tomega>0$ in $\mathcal{F}_{\sharp,{\rm trap}}$, we can then immediately integrate $E (j^T[u])'$ or $E (j^K[u])'$ with $E\geq 2$ and apply Lemma \ref{lm:superradiance} to estimate the boundary terms on the right-hand side. However, for the sake of later convenience when passing to physical space, we nevertheless add cut-off functions: $E (\chi_{r_1}j^K[u])'$ and $E ((1-\chi_{r_1})j^T[u])'$. We then have to control the following additional term:
\begin{equation*}
	E\int_R (\homega-\tomega) \chi_{r_1}' \re(\overline{iu'}u)\,dr_*.
\end{equation*}
We then estimate
\begin{equation}
\label{eq:errortermtrappedfreq}
	2|(\homega-\tomega)  \re(\overline{iu'}u)|\leq \eta |u'|^2+\eta^{-1}|\homega-\tomega|^2|u|^2.
\end{equation}
We then take $\eta>0$ suitably small and use that $|\homega-\tomega|=|q|(1-r_c^{-1})$ is bounded, $\chi'_{r_1}$ is supported away from $r=r_{\max}$ and $\tomega^2$ is large to absorb both terms on the right-hand side of \eqref{eq:errortermtrappedfreq} and conclude \eqref{eq:trappedfreqest}.

Observe however that, in contrast with the proof of Proposition \ref{prop:srhighfreqest}, we cannot remove the factor $(1-r_{\max}r^{-1})^2$ appearing in \eqref{eq:fcurrentandrweightedest}, since that relied on \eqref{eq:keypropforhcurrent}, which used the assumption $\ell(\ell+1)\gg \homega^2$, which is not valid in $\mathcal{F}_{\sharp,{\rm trap}}$ where $\ell(\ell+1)\sim \homega^2$.
\end{proof}

\section{Integrated estimates in frequency space: bounded frequencies}
\label{sec:intestboundfreq}
In this section, we will restrict to $(\omega,\ell)\in \mathcal{F}_{\flat}$. An important role will be played by \emph{homogeneous solutions} $u_+$ and $u_{\infty}$ to \eqref{eq:radialODE} (with $H\equiv 0$), satisfying the following boundary conditions:
\begin{equation*}
\lim_{r_*\to -\infty}e^{i\tomega r_*+i \q \int_{0}^{r_*}\rho_+(r_*')\,dr_*'}u_+(r_*)=\lim_{r_*\to \infty}e^{-i\homega r_*+i \q \int_{0}^{r_*}\rho_c(r_*')\,dr_*'}u_{\infty}(r_*)=1.
\end{equation*}
We will refer to $u_+$ as the \emph{event-horizon-normalized} homogeneous solution and $u_{\infty}$ as the \emph{infinity-normalized} homogeneous solution (when $\kappa_c=0$).

We define the corresponding Wronskian as follows:
\begin{equation*}
\mathfrak{W}(\homega,\ell):=\mathcal{W}(u_+,u_{\infty})=u_{+}(r_*)\frac{du_{\infty}}{dr_*}(r_*)-u_{\infty}(r_*)\frac{du_{+}}{dr_*}(r_*).
\end{equation*}
Since \eqref{eq:radialODE} does not have any first-order terms, $\mathfrak{W}$ is independent of $r_*$.

We will restrict to $|\q|\geq q_0$, where $\q_0>0$ can be taken arbitrarily small. \textbf{In the remainder of \S \ref{sec:intestboundfreq}, we will allow any constant denoted $C$ and $c$, as well as the constants implicit in the big-O notation, to depend on $L_0$, $\beta$, $\q_0$ and $\q_{1}$.} If there is a further dependence on other constants, we will denote this in the subscript.

We will apply Green's formula \eqref{eq:greenformula} in \S\ref{sec:greenhm} and \S \ref{sec:greeninhm} to estimate solutions $u$ to \eqref{eq:radialODE} that satisfy the boundary conditions \eqref{eq:inhomodebc1} and \eqref{eq:inhomodebc2} by homogeneous solutions $u_+$ and $u_{\infty}$. An important ingredient for these estimates is a uniform estimate for $|\mathfrak{W}|^{-1}$.

The structure of the present section is as follows:
\begin{itemize}
	\item In \S \ref{sec:homest} we approximate general solutions $u$ to \eqref{eq:radialODE} with $H\equiv 0$ by appropriate special functions by partioning the $r$-interval $(r_+,r_c)$ into subintervals. These approximations are uniformly valid as $\tomega,\homega\to 0$.
	\item In \S \ref{sec:matchasymptomega} and \S \ref{sec:matchasymphomega}, we derive upper-bound estimates for the coefficients appearing in the approximations from the previous step in the case that $u=u_+$ or $u=u_{\infty}$. This results in uniform estimates for  $u=u_+$ and $u_{\infty}$. In the cases $(\homega,\ell)\in \mathcal{F}_{\flat,+}\cup \mathcal{F}_{\flat,\sim}$ and $(\homega,\ell)\in \mathcal{F}_{\flat,\infty}\cup \mathcal{F}_{\flat,\sim}$, we moreover derive lower bounds for these coefficients. These lower bounds form a key ingredients in the Wronskian estimates. The key technique for deriving lower bound is a mathematically rigorous realization of the \emph{matched asymptotics method} employed in black hole spacetimes; see also the discussion in \S \ref{intro:sketchstep2}. This method will be discussed more generally in upcoming work.
	\item In \S \ref{sec:wronskianestimates}, we derive estimates for $|\mathfrak{W}|^{-1}$ that are uniform and do not degenerate as $\kappa_+,\kappa_c,\tomega,\homega\to 0$ by combining the coefficient estimates from the previous step for $(\homega,\ell)\in \mathcal{F}_{\flat,\infty}\cup \mathcal{F}_{\flat,+}$ with a mode stability argument for $(\homega,\ell)\in \mathcal{F}_{\flat,\sim}$.
	\item Finally, we use Green's formula \eqref{eq:greenformula} in \S \ref{sec:greenhm} and \S \ref{sec:greeninhm} and apply the homogeneous and Wronskian estimates from the previous step to derive $(r^{-1}\Omega)^{-p}$-weighted $L^2$-estimates for solutions $u$ to \eqref{eq:radialODE} that satisfy the boundary conditions \eqref{eq:inhomodebc1} and \eqref{eq:inhomodebc2}
\end{itemize}
\textbf{In the remainder of this \S \ref{sec:intestboundfreq}, we fix $r_+=1$.}

\subsection{Asymptotics of homogeneous solutions}
\label{sec:homest}
We will first derive uniform estimates for homogeneous solutions $u_+$ and $u_{\infty}$ to \eqref{eq:radialODE} in the frequency range $\mathcal{F}_{\flat}$ that remain valid as $|\homega|\to 0$ or $|\tomega|\to 0$. 

It will be convenient to employ the Couch--Torrence radial coordinates $s_+$ and $s_c$, introduced in \S \ref{sec:couchtorrencerad}. Via the Lemma \ref{lm:U+Uceqs} below, we will reformulate \eqref{eq:radialODE} with $H\equiv 0$ to bring it in a form amenable to a uniform analysis in $\kappa_+,\kappa_c$.
\begin{lemma}
\label{lm:U+Uceqs}
Set $r_+=1$. Let
\begin{align*}
U_+(s_+):=&\:(s_++1)(r^{-1}\Omega u)(r_*(s_+)),\\
U_c(s_c):=&\: (s_c+1)(r^{-1}\Omega u)(r_*(s_c)). 
\end{align*}
Then $u$ satisfies  \eqref{eq:radialODE} with $H\equiv 0$ and if and only if:
\begin{align}
\label{eq:maineqUevent}
\frac{d^2U_+}{ds_+^2}=&\:\rho_+^{4}(r^{-2}\Omega^2)^{-2}{\mathcal{V}}U_+, \\
\label{eq:maineqUcosmo}
\frac{d^2U_c}{ds_c^2}=&\:\rho_c^{4}(r^{-2}\Omega^2)^{-2}\mathcal{V}U_c,
\end{align}
with $\mathcal{V}$ introduced in Lemma \ref{lm:mainasympmathcalpot}.

We can write:
\begin{align}
\label{eq:tildemathcalwithfactor}
\rho_+^{4}(r^{-2}\Omega^2)^{-2}{\mathcal{V}}=&(1+2\kappa_+s_+)^{-2}\left[-(\tomega_+^2+\kappa_+^2)+2\nu_+ s_+^{-1}+\left(\mu_+^2-\frac{1}{4}\right)s_+^{-2}+(|\homega|+\kappa_c)O_{\infty}(s_+^{-3})\right],\\
\label{eq:mathcalVwithfactor}
\rho_c^4(r^{-2}\Omega^2)^{-2}\mathcal{V}=&(1+2\kappa_c s_c)^{-2}\left[-(\homega^2_c+\kappa_c^2)+2\nu_c s_c^{-1}+\left(\mu_c^2-\frac{1}{4}\right)s_c^{-2}+(|\tomega|+\kappa_c+\kappa_+)O_{\infty}(s_c^{-3})\right],
\end{align}
 with the constants $\homega_+,\tomega_+,\mu_+,\mu_c,\nu_+,\nu_c$ defined above satisfying:
\begin{align*}
\homega^2_++\kappa_+^2=&\:\tomega^2+\kappa_+^2+O(\kappa_c^2),\\
\homega^2_c+\kappa_c^2=&\:\homega^2+O(\kappa_c),\\
\mu_{+}=&\: \sqrt{\ell(\ell+1)+\frac{1}{4}-\q^2-6\tomega \homega+O(\kappa_c^2)},\\
\mu_c=&\: \sqrt{\ell(\ell+1)+\frac{1}{4}-\q^2-6\tomega \homega+ O(\kappa_c)+O(|\homega|)},\\
\nu_+=&\: -\tomega(\q+2\tomega)+\kappa_+\ell(\ell+1)+O(\kappa_c^2),\\
	\nu_c=&\:\homega(\q-2\homega)  +\kappa_c\ell(\ell+1)+O(\kappa_c)+O(|\homega|).
\end{align*}

In particular, when $\kappa_c=\kappa_+=0$:
\begin{align}
\label{eq:tildemathcalVext}
\rho_+^4(r^{-2}\Omega^2)^{-2}{\mathcal{V}}=&-\tomega^2+2\tomega (\q-2\homega) s_+^{-1}+(\ell(\ell+1)-\q^2-6\tomega \homega) s_+^{-2}\\ \nonumber
&+2 \homega(\q-2\homega) s_+^{-3}-\homega^2 s_+^{-4},\\
\label{eq:mathcalVext}
\rho_c^4(r^{-2}\Omega^2)^{-2}\mathcal{V}=&-\homega^2-2 \homega(\q+2\tomega) s_c^{-1}+\left(\ell(\ell+1)-\q^2-6\omega \tomega\right)s_c^{-2}\\\nonumber
&+2\tomega(\q-2\homega)s_c^{-3}-\tomega^2s_c^{-4}.
\end{align}
\end{lemma}
\begin{proof}
The function $u$ satisfies \eqref{eq:radialODE} with $H\equiv 0$ if and only if
\begin{equation*}
\frac{d^2(\Omega u)}{dr^2}+\Omega^{-4}\left(\homega^2-V-\Omega^3\frac{d^2\Omega}{dr^2}\right)(\Omega u)=0,
\end{equation*}
or equivalently,
\begin{equation*}
\frac{d^2(\Omega u)}{dr^2}+\Omega^{-4}\left(\tomega^2-\widetilde{V}-\Omega^3\frac{d^2\Omega}{dr^2}\right)(\Omega u)=0.
\end{equation*}

Note that
\begin{align*}
\frac{d^2(\Omega u)}{dr^2}=&\:r^{-3}\frac{d^2(r^{-1}\Omega u)}{d\rho_c^2},\\
\frac{d^2(\Omega u)}{dr^2}=&\:=r^{-3}\frac{d^2(r^{-1}\Omega u)}{d\rho_+^2}.
\end{align*}
Hence,
\begin{align*}
\frac{d^2(r^{-1}\Omega u)}{d\rho_c^2}=&\:(r^{-2}\Omega^2)^{-2}\left(-\homega^2+V+\Omega^3\frac{d^2\Omega}{dr^2}\right)(r^{-1}\Omega u),\\
\frac{d^2(r^{-1}\Omega u)}{d\rho_+^2}=&\:(r^{-2}\Omega^2)^{-2}\left(-\tomega^2+\widetilde{V}+\Omega^3\frac{d^2\Omega}{dr^2}\right)(r^{-1}\Omega u).
\end{align*}

Finally, note that
\begin{align*}
	\frac{d^2}{ds_c^2}((s_c+1) r^{-1}\Omega u)=&\:\rho_c^3\frac{d^2(r^{-1}\Omega u)}{d\rho_c^2},\\
	\frac{d^2}{ds_+^2}((s_++1) r^{-1}\Omega u)=&\:\rho_+^3\frac{d^2(r^{-1}\Omega u)}{d\rho_+^2},
\end{align*}
to conclude \eqref{eq:maineqUevent} and \eqref{eq:maineqUcosmo}.

We apply Lemma \ref{lm:metricest} to write:
\begin{multline*}
	\rho_+^4 (\Omega^2 r^{-2})^{-2}=\frac{\rho_+^4}{(2\kappa_+ \rho_++(1-6\kappa_+ +O(\kappa_c^2))\rho_+^2+(-2+6\kappa_+ +O(\kappa_c^2))\rho_+^3+\left(1-2\kappa_+ +O(\kappa_c^2)\right)\rho_+^4)^2}\\
	=\frac{1}{(1-6\kappa_++O(\kappa_c^2)+(-2+6\kappa_+ +O(\kappa_c^2))\rho_++\left(1-2\kappa_+ +O(\kappa_c^2)\right)\rho_+^2)^2}\\
	=(1+2\kappa_+ s_+)^{-2}(1+s^{-1})^4(1+\kappa_c^2 O_{\infty}(s^0)).
\end{multline*}
After applying \eqref{eq:boundfreqpotevent}, we then obtain \eqref{eq:tildemathcalwithfactor}.

We obtain in a similar manner:
\begin{equation*}
\rho_c^4 (\Omega^2 r^{-2})^{-2}=	(1+2\kappa_c s_c)^{-2}(1+O_{\infty}(\kappa_c)+O_{\infty}(\rho_c)).
\end{equation*}
 We can similarly derive \eqref{eq:mathcalVwithfactor} after applying \eqref{eq:boundfreqpotcosmo}.

When $\kappa_c=\kappa_+=0$, we have that $\rho_c^4(r^{-2}\Omega^2)^{-2}=\Omega^{-4}$ and $\rho_+^4(r^{-2}\Omega^2)^{-2}=r^4$ and the expressions simplify greatly.
\end{proof}

In order to obtain uniform estimates for solutions $u$ to \eqref{eq:radialODE} with $H\equiv 0$ as $|\tomega|\to 0$ or $|\homega|\to 0$, we divide the $r$-domain $(1,r_c)$ into four subintervals, where we will approximate $u$ in different ways. 

Let $\digamma_0>0$ be a constant, which we will later take to be suitably large. Consider the regions $I,II,III,IV\subset (1,r_c)$, which are defined as follows:
\begin{itemize}
\item \textbf{I}: 
\begin{align*}
	\frac{1}{2} \digamma_0^{\frac{1}{16}}|\tomega|^{-1}\leq s_+<&\:\infty & (|\tomega| \kappa_+^{-1}>& \digamma_0),\\
	\frac{1}{2}\digamma_0^{\frac{1}{16}}\kappa_+^{-\frac{1}{4}}(|\homega|+\kappa_c)^{\frac{1}{2}}\leq s_+<&\:\infty & (|\tomega| \kappa_+^{-1}\leq& \digamma_0).
\end{align*}
\item \textbf{II}: 
\begin{align*}
	\frac{1}{2}\digamma_0^{-\frac{1}{3}}(\kappa_c+|\homega|)^{\frac{1}{4}} |\tomega|^{-\frac{1}{4}}\leq s_+\leq &\:2\digamma_0^{\frac{3}{8}}|\tomega|^{-1}& (\kappa_+>&\:0\quad \textnormal{and}\quad |\tomega| \kappa_+^{-1}> \digamma_0),\\
	\frac{1}{2}\digamma_0^{-\frac{1}{3}}(\kappa_c+|\homega|)^{\frac{1}{4}} |\tomega|^{-\frac{1}{4}}\leq s_+ <&\:\infty & (\kappa_+=&\:0).
\end{align*}
\item \textbf{III}: 
\begin{align*}
	\frac{1}{2} \digamma_0^{-\frac{1}{3}}(\kappa_c+|\tomega|)^{\frac{1}{4}}|\homega|^{-\frac{1}{4}}\leq s_c\leq &\:2\digamma_0^{\frac{3}{8}}|\homega|^{-1}& (\kappa_c>&\:0\quad \textnormal{and}\quad |\homega| \kappa_c^{-1}> \digamma_0),\\
	\frac{1}{2} \digamma_0^{-\frac{1}{3}}|\tomega|^{\frac{1}{4}}|\homega|^{-\frac{1}{4}}\leq s_c <&\:\infty & (\kappa_c=&\:0).
\end{align*}
\item \textbf{IV}: 
\begin{align*}
	\frac{1}{2} \digamma_0^{\frac{1}{16}}|\homega|^{-1}\leq s_c<&\:\infty & (|\homega| \kappa_c^{-1}>& \digamma_0),\\
	\frac{1}{2}\digamma_0^{\frac{1}{16}}\leq s_c<&\:\infty & (|\homega| \kappa_c^{-1}\leq& \digamma_0).
\end{align*}
\end{itemize}

The lemma below shows in particular that Regions I--IV cover the full $r$-interval $(r_+,r_c)$ for $\kappa_c$ and $\kappa_+$ suitably small and we identify intervals contained in two intersecting regions.
\begin{lemma}
\label{lm:setincl}
	The following set inclusions hold: for $\digamma_0$ suitably large and $\kappa_c$ suitably small, depending on $\gamma$ and $q_0$, the following set inclusions hold: 	\begin{align}
	\label{eq:setincl1}
		\left\{\frac{1}{2}\digamma_0^{\frac{1}{4}}|\tomega|^{-1}\leq s_+\leq 2\digamma_0^{\frac{1}{4}}|\tomega|^{-1}\right\}\subset&\: I\cap II,\\
			\label{eq:setincl2}
		\left\{\frac{3}{4}\digamma_0^{-\frac{1}{3}}|\homega|^{-\frac{1}{4}}(\kappa_++\kappa_c+|\tomega|)^{\frac{1}{4}}\leq s_c\leq \frac{4}{3}\digamma_0^{-\frac{1}{3}}|\homega|^{-\frac{1}{4}}(\kappa_++\kappa_c+|\tomega|)^{\frac{1}{4}}\right\}\subset&\: I\cap III\quad &(|\tomega|\leq\digamma_0 \kappa_+),\\
		\label{eq:setincl3}
		\left\{\frac{3}{4}\digamma_0^{-\frac{1}{3}}|\homega|^{-\frac{1}{4}}(\kappa_++\kappa_c+|\tomega|)^{\frac{1}{4}}\leq s_c\leq \frac{4}{3}\digamma_0^{-\frac{1}{3}}|\homega|^{-\frac{1}{4}}(\kappa_++\kappa_c+|\tomega|)^{\frac{1}{4}}\right\}\subset&\: II\cap III\quad &(|\tomega|>\digamma_0 \kappa_+),\\
		\label{eq:setincl4}
		\left\{\frac{1}{2} \digamma_0^{-\frac{1}{4}}|\homega|^{-1}\leq s_c\leq 2 \digamma_0^{-\frac{1}{4}}|\homega|^{-1}\right\}\subset&\: III\cap IV,\\
				\label{eq:setincl5}
		\left\{\frac{3}{4}\digamma_0^{-\frac{1}{3}}|\tomega|^{-\frac{1}{4}}(\kappa_c+|\homega|)^{\frac{1}{4}}\leq s_+\leq \frac{4}{3}\digamma_0^{-\frac{1}{3}}|\tomega|^{-\frac{1}{4}}(\kappa_c+|\homega|)^{\frac{1}{4}}\right\}\subset&\: IV\cap II\quad &(|\homega|\leq\digamma_0 \kappa_c),\\
		\label{eq:setincl6}
		\left\{\frac{3}{4}\digamma_0^{-\frac{1}{3}}|\tomega|^{-\frac{1}{4}}(\kappa_c+|\homega|)^{\frac{1}{4}}\leq s_+\leq \frac{4}{3}\digamma_0^{-\frac{1}{3}}|\tomega|^{-\frac{1}{4}}(\kappa_c+|\homega|)^{\frac{1}{4}}\right\}\subset&\: III\cap II\quad &(|\homega|>\digamma_0 \kappa_c).
	\end{align}
\end{lemma}
\begin{proof}
	The inclusions \eqref{eq:setincl1} and \eqref{eq:setincl4} follow immediately. To derive \eqref{eq:setincl2}, we apply \eqref{eq:s+intermssc} with $r_+=1$ and obtain for $\kappa_c$ suitably small and $\digamma_0$ suitably large:
	\begin{multline}
	\label{eq:auxsetconv}
		\left\{\frac{3}{4}\digamma_0^{-\frac{1}{3}}|\homega|^{-\frac{1}{4}}(\kappa_++\kappa_c+|\tomega|)^{\frac{1}{4}}\leq s_c\leq \frac{4}{3}\digamma_0^{-\frac{1}{3}}|\homega|^{-\frac{1}{4}}(\kappa_++\kappa_c+|\tomega|)^{\frac{1}{4}}\right\}\\=\left\{\frac{1+r_c^{-2}(1-\frac{4}{3}\digamma_0^{-\frac{1}{3}}|\homega|^{-\frac{1}{4}}(\kappa_++\kappa_c+|\tomega|)^{\frac{1}{4}})}{\frac{4}{3}(1-r_c^{-1})\digamma_0^{-\frac{1}{3}}|\homega|^{-\frac{1}{4}}(\kappa_++\kappa_c+|\tomega|)^{\frac{1}{4}}-r_c^{-1}}\leq s_+\leq \frac{1+r_c^{-2}(1-\frac{3}{4}\digamma_0^{-\frac{1}{2}}|\homega|^{-\frac{1}{4}}(\kappa_++\kappa_c+|\tomega|)^{\frac{1}{4}})}{\frac{3}{4}\digamma_0^{-\frac{1}{3}}(1-r_c^{-1})|\homega|^{-\frac{1}{4}}(\kappa_++\kappa_c+|\tomega|)^{\frac{1}{4}}-r_c^{-1}}\right\}\\
		\subseteq \left\{\digamma_0^{\frac{1}{3}}|\homega|^{\frac{1}{4}}(\kappa_++\kappa_c+|\tomega|)^{-\frac{1}{4}}\leq s_+<\frac{5}{4}\digamma_0^{\frac{1}{3}}|\homega|^{\frac{1}{4}}(\kappa_++\kappa_c+|\tomega|)^{-\frac{1}{4}}\right\}
	\end{multline}
	For $|\tomega|\leq \digamma_0 \kappa_+$, with $\digamma_0$ suitably large and $\kappa_c\leq \kappa_+$, the set on the right-hand side of the inclusion in \eqref{eq:auxsetconv} is contained in I, because:
	\begin{equation*}
		\digamma_0^{\frac{1}{3}}|\homega|^{\frac{1}{4}}(\kappa_++\kappa_c+|\tomega|)^{-\frac{1}{4}}\geq \frac{1}{2}\digamma_0^{\frac{1}{12}}|\homega|^{\frac{1}{4}}\kappa_+^{-\frac{1}{4}}\geq \digamma_0^{\frac{1}{16}}(|\homega|+|\kappa_c|)^{\frac{1}{2}}\kappa_+^{-\frac{1}{4}},
	\end{equation*}
	so \eqref{eq:setincl2} holds.
	
	For $|\tomega|> \digamma_0 \kappa_+$ and $|\homega|> \digamma_0 \kappa_c$ or $\kappa_c=0$, with $\digamma_0$ suitably large and $\kappa_c$ suitably small, we moreover have that:
	\begin{equation*}
		\digamma_0^{\frac{1}{3}}|\homega|^{\frac{1}{4}}(\kappa_c+|\tomega|)^{-\frac{1}{4}}\geq \frac{3}{4}\digamma_0^{\frac{1}{3}}|\homega|^{\frac{1}{4}}|\tomega|^{-\frac{1}{4}}\geq \frac{1}{2}\digamma_0^{-\frac{1}{2}}(\kappa_c+|\homega|)^{\frac{1}{4}} |\tomega|^{-\frac{1}{4}}
	\end{equation*}
	and when $\kappa_c>0$, also:
	\begin{equation*}
		\frac{5}{4}\digamma_0^{\frac{1}{3}}|\homega|^{\frac{1}{4}}(\kappa_++\kappa_c+|\tomega|)^{-\frac{1}{4}}\leq \frac{5}{4}\digamma_0^{\frac{1}{3}}|\homega|^{\frac{1}{4}}|\tomega|^{-\frac{1}{4}}\leq 2 F_0^{\frac{3}{8}}|\tomega|^{-1},
	\end{equation*}
	so the set on the right-hand side of the inclusion in \eqref{eq:auxsetconv} is also contained in II. This concludes \eqref{eq:setincl3}. 
	
	The inclusions \eqref{eq:setincl5} and \eqref{eq:setincl6} follow analogously with the roles of $s_+$ and $s_c$ and $|\tomega|$ and $|\homega|$ interchanged.
	\end{proof}

\subsubsection{Region I}
\label{sec:homestRegionI}
We first approximate solutions $u$ to \eqref{eq:radialODE} with $H\equiv 0$ by purely oscillating exponentials via a standard WKB approximation.
\begin{proposition}
\label{prop:regionIeleq1}
 Let $u$ be a solution to \eqref{eq:radialODE} with $H\equiv 0$. Then, there exist constants $A_1,A_2\in \C$ and a suitably small $\delta>0$, such that for $s_+>\frac{1}{2}\delta^{-1} |\tomega|^{-1}$, we can express:
\begin{equation*}
u(r_*)=A_1e^{-i\tomega r_*-i \q \int_{0}^{r_*}\rho_+(r_*')\,dr_*'}(1+\varepsilon_1(r_*))+A_2e^{i\tomega r_*+i \q \int_{0}^{r_*}\rho_+(r_*')\,dr_*'}(1+\varepsilon_2(r_*)),
\end{equation*}
and there exists a uniform constant $C>0$ (independent of $\delta$), such that:
\begin{align*}
|\varepsilon_i|(r_*)\leq &\: C\tomega^{-2}r^{-2}\Omega^2(|\tomega|+|\kappa_+|)+C|\tomega|^{-1}s_+^{-1},\\
\left|\frac{d\varepsilon_i}{dr_*}\right|(r_*)\leq &\: Cr^{-2}\Omega^2|\tomega|^{-1}.
\end{align*}
In particular, $|\varepsilon_i|(r_*)\leq C\delta$ if $|\tomega|> \digamma_0 \kappa_+$ and $\digamma_0>1$.
\end{proposition}
\begin{proof}
Without loss of generality, we will take $\delta>0$ to be suitably small. Let $\zeta(r_*):=\int_{0}^{r_*}\sqrt{\tomega^2-\widetilde{V}(r_*')}\,dr_*'$. Then
	$\zeta':=\frac{d\zeta}{dr_*}=\sqrt{\tomega^2-\widetilde{V}(r_*)}$ and $\zeta''=\frac{d^2\zeta}{dr_*^2}$ and we obtain:
	\begin{equation}
	\label{eq:expapprox}
	\frac{d^2}{d\zeta^2}((\zeta')^{\frac{1}{2}}u)=-(1-\vartheta)(\zeta')^{\frac{1}{2}}u,
	\end{equation}
with
\begin{equation*}
	\vartheta(\zeta):=\frac{1}{2}(\zeta')^{-\frac{3}{2}}\frac{d}{dr_*}(\zeta'' \zeta'^{-\frac{3}{2}})=-2(\zeta')^{-\frac{3}{2}}\frac{d^2}{dr_*^2}(\zeta'^{-\frac{1}{2}})
\end{equation*}
Let $W_1(\zeta)=e^{-i\sign(\tomega) \zeta}$ and $W_2(\zeta)=e^{+i\sign(\tomega) \zeta}$. Then $W_i$ are solutions to \eqref{eq:expapprox} with $\vartheta\equiv 0$. Note that $\mathcal{W}(W_1,W_2)=2\sign(\tomega)$.

We can then apply \cite{olv97}[Theorem 6.2.2] to obtain:
\begin{equation*}
	u(r_*(\zeta))=\widetilde{A}_1(\zeta')^{-\frac{1}{2}}e^{-i\sign(\tomega) \zeta}(1+\tilde{\varepsilon}_1(\zeta))+\widetilde{A}_2(\zeta')^{-\frac{1}{2}}e^{+i\sign(\tomega) \zeta}(1+\tilde{\varepsilon}_2(\zeta)),
\end{equation*}
with
\begin{equation*}
|\tilde{\varepsilon}_i|(\zeta)\leq  e^{\frac{1}{2}\int^{\infty}_{\zeta(r_*)}|\vartheta|(\eta)\,d\eta}-1,\quad \frac{1}{2}\left|\frac{d\tilde{\varepsilon}_i}{d\zeta}\right|(r_*)\leq e^{\frac{1}{2}\int^{\infty}_{\zeta(r_*)}|\vartheta|(\eta)\,d\eta}-1.
\end{equation*}

Note that for $s_+\geq \frac{1}{2}\delta^{-1}|\tomega|^{-1}$ and $\delta>0$ suitably small, we can estimate $\zeta'\geq \frac{1}{2}|\tomega|$. We can therefore estimate:
\begin{multline*}
\int^{\infty}_{\zeta(r_*)}|\vartheta|(\eta)\,d\eta=\int^{\infty}_{r_*}|\zeta'||\vartheta|(r_*')\,dr_*'\leq \int_{1}^{r(r_*)}|\zeta'|^{-\frac{1}{2}}\left|\frac{d}{dr}((\zeta'^{-\frac{1}{2}})')\right|\,dr'\\
\leq C\tomega^{-2}r^{-2}\Omega^2(|\tomega|+|\kappa_+|+(r(r_*)-1))\leq C\tomega^{-2}r^{-2}\Omega^2(|\tomega|+|\kappa_+|).
\end{multline*}

We conclude that
\begin{equation*}
|\tilde{\varepsilon}_i|(\zeta)\leq  C\tomega^{-2}r^{-2}\Omega^2(|\tomega|+|\kappa_+|),\\
\left|\frac{d\tilde{\varepsilon}_i}{d\zeta}\right|(r_*)\leq C\tomega^{-2}r^{-2}\Omega^2(|\tomega|+|\kappa_+|) .
\end{equation*}
To conclude the estimates in the proposition, we note that
\begin{multline*}
(\zeta')^{-\frac{1}{2}}(r_*)e^{\mp i\sign(\tomega) \zeta(r_*)}=|\tomega|^{-\frac{1}{2}}(1-q\tomega^{-1}s_+^{-1}+\kappa_+\tomega^{-2}O(s_+^{-1})+ O(|\tomega|^{-2} s_+^{-2}))\\
\times e^{\mp i\tomega r_*\pm i q \int_{0}^{r_*}\rho_+(r_*')\,dr_*'+\int_0^{r_*}\tomega^{-2}(s_+)(r_*')^{-2}\,dr_*'}
\end{multline*}
and we rescale $\tilde{A}_i$ and $\tilde{\varepsilon_i}$ appropriately to obtain $A_i$ and $\varepsilon_i$.
\end{proof}
The estimates in Proposition \ref{prop:regionIeleq1} cover Region I, in the case $|\tomega| \kappa_+^{-1}> \digamma_0$ if $\delta\leq \digamma_0^{-\frac{1}{16}}$.

\subsubsection{Region I: $|\tomega| \kappa_+^{-1}\leq  \digamma_0$}
When $|\tomega| \kappa_+^{-1}\leq  \digamma_0$, we can improve the approximation in \S \ref{sec:homestRegionI} by considering hypergeometric functions, which approximate $U_+$ well in a larger $r$-interval.

	Define
	\begin{equation*}
	\zeta(s_+):=2\kappa_+ s_+ (1+2\kappa_+ s_+)^{-1}=1-(1+2\kappa_+ s_+)^{-1}.
	\end{equation*}
	Note that $\zeta(s_+)\in (\frac{2\kappa_+}{r_c-1+2\kappa_+ },1)$.

\begin{proposition}
\label{prop:regionIhypgeom}
	Assume that $\kappa_c\leq \kappa_+$ and that $|\tomega| \kappa_+^{-1}\leq  \digamma_0$. Let $U_+$ be a solution to \eqref{eq:maineqUevent}. Let $\delta> \kappa_c+|\tomega \homega|$. Then there exist constants $B_1,B_2\in \C$ and a uniform constant $C>0$, such that for $s_+\geq \delta^{-1} \kappa_+^{-\frac{1}{4}}(\kappa_c+|\homega|)^{\frac{3}{4}}$:
	\begin{multline*}
	(2\kappa_+ s_+)^{-\frac{1}{2}}U_+(s)=B_1(2\kappa_+ s_+)^{-\frac{1}{2}}F_{\underline{\sigma},\frac{1}{2}\beta_{\ell}}(2\kappa_+s_+)(1+\varepsilon_1(s_+))\\+B_2 (2\kappa_+ s_+)^{-\frac{1}{2}}\left(G_{\underline{\sigma},\frac{1}{2}\beta_{\ell}}(2\kappa_+s_+)+\varepsilon_2(s_+)\right),
\end{multline*}
with 
\begin{itemize}
\item 
$F_{\underline{\sigma},\frac{1}{2}\beta_{\ell}}$ and $G_{\underline{\sigma},\frac{1}{2}\beta_{\ell}}$ proportional to hypergeometric functions and are precisely defined in Lemma \ref{prop:hypgeomfundsoln},
\item $\underline{\sigma}=(\sigma_a,\sigma_b)$ satisfying
\begin{align*}
	\sigma_a=&-\q+O(\kappa_c)+\kappa_+^{-1}|\tomega|O(|\homega|),\\
	\sigma_b=&\: \q+\kappa_+^{-1}\tomega +\kappa_+^{-1}|\tomega|O(|\homega|),
	\end{align*}
	\item $\varepsilon_1$ and $\varepsilon_2$ satisfying the estimates:
\begin{align*}
\left|\varepsilon_1\right|(s_+)\leq&\: C\delta,\\
\left|\varepsilon_2\right|(s_+)\leq&\: C \delta |F_{\underline{\sigma},\frac{1}{2}\beta_{\ell}}|(2\kappa_+s_+),\\
\left|\frac{d \varepsilon_1}{ds_+}\right|(s_+)\leq&\: C\delta \frac{1}{s_+(1+2\kappa_+s_+)}\\
\left|\frac{d \varepsilon_2}{ds_+}\right|(s_+)\leq&\: C\delta \frac{1}{s_+(1+2\kappa_+s_+)}|F_{\underline{\sigma},\frac{1}{2}\beta_{\ell}}|(2\kappa_+s_+).
\end{align*}
\end{itemize}
\end{proposition}
\begin{proof}
Rescale $s_+$ as follows: $\tilde{s}:=2\kappa_+s_+$. Then \eqref{eq:maineqUevent} is equivalent to \eqref{eq:hypgeomalterr}, if we set:
\begin{equation*}
	w:=\kappa_+^{-1}\tomega_+=\kappa_+^{-1}\tomega+O(\kappa_c), \quad \mu:=\frac{1}{2}\beta_{\ell}, \quad v=\kappa_+^{-1}\nu_+
\end{equation*}
and take
\begin{equation*}
\vartheta(\tilde{s})=\kappa_+ \kappa_c \frac{1}{(1+\tilde{s})^2}O_{\infty}(\tilde{s}^{-3})+\kappa_+|\homega| \frac{1}{(1+\tilde{s})^2}O_{\infty}(\tilde{s}^{-3})+\left(\mu^2_+-\frac{1}{4}\beta_{\ell}^2\right) \frac{1}{(1+\tilde{s})^2}O_{\infty}(\tilde{s}^{-2}).
\end{equation*}

Take \begin{align*}
	\sigma_a=&-\frac{\tomega_+}{2\kappa_+}+\sign\left(-\q+\frac{\tomega_+}{2\kappa_+}\right) \sqrt{\frac{1}{4}-\mu_+^2+\kappa_+^{-1}\nu_++\frac{1}{4}\kappa_+^{-2}\tomega_+^2},\\
	\sigma_b=&-\frac{\tomega_+}{2\kappa_+}-\sign\left(-\q+\frac{\tomega_+}{2\kappa_+}\right)  \sqrt{\frac{1}{4}-\mu_+^2+\kappa_+^{-1}\nu_++\frac{1}{4}\kappa_+^{-2}\tomega_+^2}.
\end{align*}
and use that for $\kappa_c\leq\kappa_+$:
\begin{align*}
\mu_{+}^2-\frac{1}{4}=&\:\ell(\ell+1)-\q^2+O(|\homega||\tomega|)+O(\kappa_c^2),\\
\kappa_+^{-1}\nu_+=&\:  -\frac{\tomega}{\kappa_+}(\q+2\tomega)+\ell(\ell+1)+O(\kappa_c^2)\\
=&\:\frac{\tomega_+}{\kappa_+}\q+\ell(\ell+1)-2\kappa_+^{-1}\tomega \homega +O(\kappa_c)
\end{align*}
to obtain
\begin{align*}
	\sigma_a=&-\frac{\tomega_+}{2\kappa_+}+\sign\left(-\q+\frac{\tomega_+}{2\kappa_+}\right) \sqrt{\left(\frac{\tomega_+}{2\kappa_+}-\q\right)^2+O(\kappa_c)+O(|\homega|)},\\
	=&\:q+O(\kappa_c)+\kappa_+^{-1}|\tomega| O(|\homega|)\\
	\sigma_b=&-\frac{\tomega_+}{2\kappa_+}-\sign\left(-\q+\frac{\tomega_+}{2\kappa_+}\right)  \sqrt{\left(\frac{\tomega_+}{2\kappa_+}-\q\right)^2+O(\kappa_c)+\frac{|\tomega|}{\kappa_+^{-1}}O(|\homega|)}\\
=& -q+\kappa_+^{-1}\tomega +O(\kappa_c)+\kappa_+^{-1}|\tomega|O(|\homega|).
\end{align*}
Then we can estimate:
\begin{multline*}
\int_{\tilde{s}}^{\infty}\eta \log^2(2+\eta)|\vartheta|(\eta)\,d\eta
	\leq C\kappa_+(\kappa_c+|\homega|) \log^2(2+\tilde{s})\tilde{s}^{-1}(\tilde{s}+1)^{-2}+C\left(\mu_+^2-\frac{1}{4}\beta_{\ell}^2\right)\log^2(2+\tilde{s})(\tilde{s}+1)^{-2}\\
	\leq  C(\kappa_c+|\homega|) \log^2(2+2\kappa_+ s_+)s_+^{-1}(1+2\kappa_+s_+)^{-2}+C (\kappa_c^2+|\homega||\tomega|)\log^2(2+\kappa_+ s_+)(1+2\kappa_+s_+)^{-2}
\end{multline*}
and similarly
\begin{multline*}
	\log(1+\zeta^{-1}(\tilde{s}))\int_{\tilde{s}}^{\infty}\log^2(2+\eta)\eta \log^2(2+\eta)|\vartheta|(\eta)\,d\eta\leq  C\log\left(2+\frac{1}{2\kappa_+ s_+}\right)(1+2\kappa_+s_+)^{-2}\\
	\times  \log^2(2+2\kappa_+ s_+)\left[(\kappa_c+|\homega|) s_+^{-1}+\kappa_c^2+ |\homega||\tomega|\right].
\end{multline*}
Suppose that $2\kappa_+s_+\geq 1$. Then
\begin{equation*}
	\log(1+\zeta^{-1}(\tilde{s}))\int_{\tilde{s}}^{\infty}\log^2(2+\eta)\eta \log^2(2+\eta)|\vartheta|(\eta)\,d\eta\leq  C\left[(\kappa_c+|\homega|) s_+^{-1}+\kappa_c^2+ |\homega||\tomega|\right]\leq C\delta.
\end{equation*}

Suppose that $\delta^{-1}\kappa_+^{-1} (\kappa_+ \homega|+\kappa_+\kappa_c)^{\frac{3}{4}}\leq s_+\leq (2\kappa_+)^{-1}$. Then
\begin{equation*}
	\log(1+\zeta^{-1}(\tilde{s}))\int_{\tilde{s}}^{\infty}\eta \log^2(2+\eta)|\vartheta|(\eta)\,d\eta\leq C\delta+C\delta \log(2+\delta (\kappa_c+|\tomega \homega|)^{-1})(\kappa_c+|\tomega \homega|)^{\frac{1}{4}}\leq C \delta.
\end{equation*}

We can now apply Proposition \ref{prop:hypgeomerror} to obtain the desired estimates.
\end{proof}

Proposition \ref{prop:regionIhypgeom} covers Region I, in the case that $\kappa_c\leq \kappa_+$ and $|\tomega| \kappa_+^{-1}\leq  \digamma_0$, provided that $\delta\geq 2 \digamma_0^{-\frac{1}{16}}(\kappa_c+|\homega|)^{\frac{1}{4}}$.

\subsubsection{Region II: $|\tomega| \kappa_+^{-1}> \digamma_0$}
In this section, we derive approximations that are valid when $|\tomega| \kappa_+^{-1}> \digamma_0$, and in particular, when $\kappa_+=0$. Note that boundedness of $|\tomega|$ implies smallness of $\kappa_+$ in this regime: $\kappa_+\leq C \digamma_0^{-1}$.
\begin{corollary}
\label{cor:regioniiwhitt}
	Assume that $\kappa_c\leq \kappa_+$ and that $|\tomega| \kappa_+^{-1}> \digamma_0$. Then $U_+$ is a solution to \eqref{eq:maineqUevent} if and only if
	\begin{equation}
	\label{eq:whittU+}
		\frac{d^2U_+}{d\xi^2}=\left[-\frac{1}{4}-\sigma \xi^{-1}+\left(\frac{1}{4}\beta_{\ell}^2-\frac{1}{4}\right)\xi^{-2}+\vartheta(\xi)\right] U_+,
	\end{equation}
	with
	\begin{align*}
	\xi(s_+)=&-2w_+ s_+,\quad \textnormal{where}\\
	w_+ =&\:\tomega(1+O(\digamma_0^{-1})),\\
	\sigma=&\: -\q-2\tomega+O(\digamma_0^{-1}),\\
		\vartheta(\xi)=&\: \left(\kappa_+O(|\homega|)+O(\kappa_c)\right)\xi^{-2}+w_+ (|\homega|+\kappa_c)O(\xi^{-3})+w_+^{-1}\kappa_+ O_{\infty}(\xi).
	\end{align*}
	Furthermore, for $\kappa_c+\digamma_0^{-1}<\delta$, there exist constants $B_1,B_2\in \C$ and a uniform constant $C$, such that we can write:
\begin{equation*}
U_+(s_+)=B_1W_{-i\sigma,\frac{1}{2}\beta_{\ell}}(i\xi(s_+))(1+\varepsilon_1(s_+))+B_2(M_{-i\sigma,\frac{1}{2}\beta_{\ell}}(i \xi(s_+))+\varepsilon_2(s_+)),
\end{equation*}
with $\varepsilon_i(\xi)$, with $i\in \{1,2\}$, satisfying the following bounds for $s_+\in (\delta^{-1}|\tomega|^{-\frac{1}{8}}(|\homega|+\kappa_c),\frac{1}{|\tomega|}\xi_{\kappa_+})$, with $\xi_{\kappa_+}=\delta^{\frac{1}{2}} (\kappa_+^{-1}|\tomega|)^{\frac{1}{2}}$ if $\kappa_+>0$ and $\xi_{\kappa_+}=\infty$ if $\kappa_+=0$:
\begin{align*}
\left|\varepsilon_1\right|(s_+)\leq&\: C\delta,\\
\left|\varepsilon_2\right|(s_+)\leq&\: C\delta |W_{\sigma,\frac{1}{2}\beta_{\ell}}(i\xi(s_+))|,\\
\left|\frac{d \varepsilon_1}{ds_+}\right|(s_+)\leq&\: C \delta |\tomega|,\\
\left|\frac{d \varepsilon_2}{ds_+}\right|(s_+)\leq&\: C \delta |\tomega| |W_{\sigma,\frac{1}{2}\beta_{\ell}}(i\xi(s_+))|.
\end{align*}
For $\kappa_+=0$ and $0<s_+<|\tomega|^{-1}$, we can moreover estimate for $i\in \{1,2\}$:
\begin{align*}
\left|\varepsilon_i\right|(s_+)\leq&\: C |\tomega|^{-1} s_+^{-2},\\
\left|\frac{d \varepsilon_i}{ds_+}\right|(s_+)\leq&\: C s_+^{-2}.
\end{align*}
\end{corollary}
\begin{proof}
	We assume that $s_+\leq \kappa_+^{-1}$ and expand:
	\begin{equation*}
		(1+2\kappa_+ s_+)^{-2}=1-4 (\kappa_+s_+)+12 (\kappa_+ s_+)^2+O_{\infty}((\kappa_+ s_+)^3)
	\end{equation*}
	and apply the above expansion together with \eqref{eq:maineqUevent} and \eqref{eq:tildemathcalwithfactor} to obtain:
	\begin{multline*}
		\frac{d^2U_+}{ds_+^2}=\Bigg[-w_+^2+2 
		\left(\nu_+ -4\kappa_+ \left(\mu_+^2-\frac{1}{4}\right)+ \kappa_+ O(|\homega|)+\kappa_+O(\kappa_c)\right)s_+^{-1}\\
		+\left(\mu_+^2-\frac{1}{4}+\kappa_+ O(|\homega|)+\kappa_+O(\kappa_c)\right)s_+^{-2}+(|\homega|+\kappa_c)O_{\infty}(s_+^{-3})+\kappa_+O_{\infty}(s_+)\Bigg]U_+,
	\end{multline*}
	with $w_+=\tomega+O(\kappa_+)$. Changing variables from $s_+$ to $\xi$ then results in \eqref{eq:whittU+}.
	
	In order to apply Corollary \ref{cor:errorwhitt}, we need to estimate:
	\begin{equation*}
		 \int^{\xi_0}_{|\xi|} \log^2(2+\eta^{-2})(\eta^{-2}+1)^{-\frac{1}{2}}|\vartheta(\eta)|\,d|\eta|.
	\end{equation*}
	
	Let $\xi_{\kappa_+}=\delta^{\frac{1}{2}} (\kappa_+^{-1}|\tomega|)^{\frac{1}{2}}$ if $\kappa_+>0$ and $\xi_{\kappa_+}=\infty$ if $\kappa_+=0$ and let $|\xi|>\delta^{-1}|\tomega|^{\frac{7}{8}}(|\homega|+\kappa_c)$. Then, for suitably large $\digamma_0$, we have that $|\xi|<|\xi_{\kappa_+}|$, and we can estimate:
	\begin{multline*}
		  \int^{\xi_{\kappa_+}}_{|\xi|} \log^2(2+\eta^{-2})(\eta^{-2}+1)^{-\frac{1}{2}}|\vartheta(\eta)|\,d|\eta| \leq \int^{\xi_{\kappa_+}}_{|\xi|} \log^2(2+y^{-2})(y^{-2}+1)^{-\frac{1}{2}} \\
\times ( \kappa_+ y^{-2}+ |w_+|(|\tomega|+\kappa_c) y^{-3} + \kappa_+ |w_+|^{-1} y) \,dy\\
		 \leq C\kappa_+ \log^3(2+\xi^{-2})+ C \log^2(2+\xi^{-2})|w_+|(|\homega|+\kappa_c) \xi^{-1}+ C(\kappa_+ |w_+|^{-1}) \xi_{\kappa_+}^2  \leq C\delta+C\digamma_0^{-1}.
	\end{multline*}
	Now we apply Corollary \ref{cor:errorwhitt} to obtain the desired estimates.
	
	In the case $\kappa_+=0$, we obtain for $|\xi|>1$ :
	\begin{equation*}
		 \int^{\xi_0}_{|\xi|} \log^2(2+\eta^{-2})(\eta^{-2}+1)^{-\frac{1}{2}}|\vartheta(\eta)|\,d|\eta|\leq C |w_+|\xi^{-2}.
	\end{equation*}
\end{proof}

Corollary \ref{cor:regioniiwhitt} provides estimates in Region II in the case that $\kappa_c\leq \kappa_+$, $|\tomega| \kappa_+^{-1}>\digamma_0$ for $\delta>2|\tomega|^{\frac{1}{8}}(|\homega|+\kappa_c)^{\frac{3}{4}}F_0^{\frac{1}{3}}$ and $\delta^{\frac{1}{2}}\kappa_+^{-1/2}|\tomega|^{1/2}>2\digamma_0^{\frac{3}{8}}$, the latter which is guaranteed if $\delta>4\digamma_0^{-\frac{1}{4}}$. Note that for $\delta>0$ to be small, where therefore need either $|\tomega|$ or $|\homega|+\kappa_c$ to be sufficiently small (depending $\digamma_0$, which is assumed to be large)! 

\subsubsection{Region III: $ |\homega| \kappa_c^{-1}> \digamma_0$}
In Region III, we interchange the roles of $s_+$ and $s_c$, $\tomega$ and $-\homega$, and $\kappa_+$ and $\kappa_c$ to obtain approximations for $U_c$. In this case, there is no need for the assumption $\kappa_+\leq \kappa_c$.

\begin{corollary}
\label{cor:regioniiiwhitt}
Assume that $|\homega| \kappa_c^{-1}> \digamma_0$ if $\kappa_c>0$. Let $U_c$ be a solution to \eqref{eq:maineqUcosmo} if and only if
	\begin{equation}
	\label{eq:whittUc}
		\frac{d^2U_c}{d\xi^2}=\left[-\frac{1}{4}-\sigma \xi^{-1}+\left(\frac{1}{4}\beta_{\ell}^2-\frac{1}{4}\right)\xi^{-2}+\vartheta(\xi)\right] U_c,
	\end{equation}
	with
	\begin{align*}
	\xi(s_c):=&-2w_c s_c,\quad \textnormal{where}\\
	w_c=&\:\homega(1+O(\digamma_0^{-1})),\\
	\sigma:=&-\q+2\tomega+O(\digamma_0^{-1}),\\
		\vartheta(\xi)=&\: \tomega (O(\kappa_c)+O(|\homega|))\xi^{-2}+w_c(|\tomega|+\kappa_++\kappa_c) O(\xi^{-3})+w_c^{-1}\kappa_c O_{\infty}(\xi).
	\end{align*}
	Furthermore, for $\kappa_c+\digamma_0^{-1}<\delta$, there exist constants $D_1,D_2\in \C$ and a uniform constant $C$, such that we can write:
\begin{equation*}
U_c(s_c)=D_1W_{-i\sigma,\frac{1}{2}\beta_{\ell}}(i\xi(s_c))(1+\varepsilon_1(s_c))+D_2(M_{-i\sigma,\frac{1}{2}\beta_{\ell}}(i \xi(s_c))+\varepsilon_2(s_c)),
\end{equation*}
with $\varepsilon_i(\xi)$, with $i\in \{1,2\}$, satisfying the following bounds for $s_c\in (\delta^{-1}|\homega|^{-\frac{1}{8}}(|\tomega|+\kappa_++\kappa_c),|\homega|^{-1}\xi_{\kappa_c})$, with $\xi_{\kappa_c}=\delta^{\frac{1}{2}}  (\kappa_c^{-1}|\homega|)^{\frac{1}{2}}$ if $\kappa_c>0$ and $\xi_{\kappa_c}=\infty$ if $\kappa_c=0$:
\begin{align*}
\left|\varepsilon_1\right|(s_c)\leq&\: C\delta,\\
\left|\varepsilon_2\right|(s_c)\leq&\: C\delta |W_{\sigma,\frac{1}{2}\beta_{\ell}}(i\xi(s_+))|,\\
\left|\frac{d \varepsilon_1}{ds_c}\right|(s_c)\leq&\: C \delta |\tomega|,\\
\left|\frac{d \varepsilon_2}{ds_c}\right|(s_c)\leq&\: C \delta |\tomega| |W_{\sigma,\frac{1}{2}\beta_{\ell}}(i\xi(s_c))|.
\end{align*}
For $\kappa_c=0$ and $0<s_c<|\homega|^{-1}$, we can moreover estimate for $i\in \{1,2\}$:
\begin{align*}
\left|\varepsilon_i\right|(s_c)\leq&\: C |\homega|^{-1} s_c^{-2},\\
\left|\frac{d \varepsilon_i}{ds_c}\right|(s_c)\leq&\: C s_c^{-2}.
\end{align*}
\end{corollary}
\begin{proof}
We repeat the proof of Corollary \ref{cor:regioniiiwhitt} with $s_c$ taking on the role of $s_+$.
\end{proof}
Corollary \ref{cor:regioniiiwhitt} provides estimates in Region III in the case that: $|\homega| \kappa_c^{-1}>\digamma_0$ for $\delta>2|\homega|^{\frac{1}{8}}(|\tomega|+\kappa_++\kappa_c)^{\frac{3}{4}}F_0^{\frac{1}{3}}$ and $\delta^{\frac{1}{2}}\kappa_c^{-1/2}|\homega|^{1/2}>2\digamma_0^{\frac{3}{8}}$, the latter which is guaranteed if $\delta>4\digamma_0^{-\frac{1}{4}}$. Note that in this case, the above conditions are compatible with smallness of $\delta$ if $|\homega|$ is small or if $|\tomega|+\kappa_+$ is small. 

\subsubsection{Region IV}
\label{sec:regionIV}
In this region, we approximate $u$ with exponentials, as in Region I.
\begin{proposition}
\label{prop:regionIVeleq1}
 Let $u$ be a solution to \eqref{eq:radialODE} with $H\equiv 0$. Then, there exist constants $A_1,A_2\in \C$ and a suitably small $\delta>\kappa_c$, such that for $s_c>\frac{1}{2}\delta^{-1} |\homega|^{-1}$, we can express:
 \begin{equation*}
u(r_*)=E_1e^{i\homega r_*-i \q \int_{0}^{r_*}\rho_c(r_*')\,dr_*'}(1+\varepsilon_1(r_*))+E_2e^{-i\homega r_*+i \q \int_{0}^{r_*}\rho_c(r_*')\,dr_*'}(1+\varepsilon_2(r_*)),
\end{equation*}
and there exists a uniform constant $C>0$ (independent of $\delta$), such that:
\begin{align*}
|\varepsilon_i|(r_*)\leq &\: C\homega^{-2}r^{-2}\Omega^2(|\homega|+|\kappa_c|)+C|\homega|^{-1}s_c^{-1},\\
\left|\frac{d\varepsilon_i}{dr_*}\right|(r_*)\leq &\: Cr^{-2}\Omega^2|\tomega|^{-1}.
\end{align*}
In particular, $|\varepsilon_i|(r_*)\leq C\delta$ if $|\homega|> \digamma_0 \kappa_c$ and $\digamma_0>1$.
\end{proposition}
\begin{proof}
	We repeat the proof of Proposition \ref{prop:regionIeleq1}, with $s_+$ replaced by $s_c$ and $-\homega$ taking the role of $\tomega$.
\end{proof}
The estimates in Proposition \ref{prop:regionIVeleq1} cover Region I, in the case $|\homega| \kappa_c^{-1}> \digamma_0$ if $\delta\leq \digamma_0^{-\frac{1}{16}}$.

\subsubsection{Region IV: $|\homega| \kappa_c^{-1}\leq  \digamma_0$}
An analogue of Proposition \ref{prop:regionIhypgeom} is valid when $s_c$ takes on the role of $s_+$, $\kappa_c$ takes on the role of $\kappa_+$, $\mu_+$ takes on the role of $\mu_c$ and $-\homega$ takes on the role of $\tomega$. In this case, there is no need for the assumption $\kappa_+\leq \kappa_c$.
\begin{proposition}
\label{prop:regionIhypgeomhomega}
	Assume that $|\homega| \kappa_c^{-1}\leq  \digamma_0$. Let $U_c$ be a solution to \eqref{eq:maineqUcosmo}. Then there exist constants $D_1,D_2\in \C$ and a uniform constant $C>0$, such that for $\delta>\kappa_c$ and $s_c>\delta^{-1}$:
	\begin{equation*}
	(2\kappa_c s_c)^{-\frac{1}{2}}U_c(s_c)=D_1(2\kappa_c s_c)^{-\frac{1}{2}}F_{\underline{\sigma},\frac{1}{2}\beta_{\ell}}(2\kappa_c s_c)(1+\varepsilon_1(s_c))+D_2\left((2\kappa_c s_c)^{-\frac{1}{2}}G_{\underline{\sigma},\frac{1}{2}\beta_{\ell}}(2\kappa_c s_c)+\varepsilon_2(s_c)\right),
\end{equation*}
with $F_{\underline{\sigma},\frac{1}{2}\beta_{\ell}}$ and $G_{\underline{\sigma},\frac{1}{2}\beta_{\ell}}$ defined in Lemma \ref{prop:hypgeomfundsoln}, $\underline{\sigma}=(\sigma_a,\sigma_b)$ satisfying
\begin{align*}
	\sigma_a=&\:q+O(\kappa_c)\\
	\sigma_b=&-q-\kappa_+^{-1}\homega +O(\kappa_c)
\end{align*}
and
\begin{align*}
\left|\varepsilon_1\right|(s_c)\leq&\: C\delta,\\
\left|\varepsilon_2\right|(s_c)\leq&\: C \delta |F_{\underline{\sigma},\frac{1}{2}\beta_{\ell}}|(2\kappa_c s_c),\\
\left|\frac{d \varepsilon_1}{ds_c}\right|(s_c)\leq&\: C\delta s_c^{-1}(s_c+1)^{-1}\\
\left|\frac{d \varepsilon_2}{ds_c}\right|(s_c)\leq&\: C\delta s_c^{-1}(s_c+1)^{-1} |F_{\underline{\sigma},\frac{1}{2}\beta_{\ell}}|(2\kappa_c s_c).
\end{align*}
\end{proposition}
\begin{proof}
In this case $\tilde{s}=2\kappa_cs_c$ and:
\begin{equation*}
\vartheta(\tilde{s})=(\kappa_c+|\homega|) \frac{1}{(1+\tilde{s})^2}O_{\infty}(\tilde{s}^{-2})+ \frac{\kappa_c}{(1+\tilde{s})^2}(|\tomega|+\kappa_++\kappa_c)O_{\infty}(\tilde{s}^{-3}).
\end{equation*}	
We then obtain:
\begin{equation*}
	\int_{\tilde{s}}^{\infty}y\log^2(2+y) |\vartheta(y)|\,dy\leq C\kappa_c\log^2(2+\tilde{s})(1+\tilde{s}^{-1})(1+\tilde{s})^{-2}\leq C \log^2(2+2\kappa_c s_c)(\kappa_c+s_c^{-1})(1+2\kappa_c s_c)^{-2}.
\end{equation*}
The rest of the proof proceeds as the proof of Proposition \ref{prop:regionIhypgeom}.
\end{proof}

Proposition \ref{prop:regionIhypgeomhomega} covers Region IV, in the case that $|\homega| \kappa_c^{-1}\leq  \digamma_0$, provided that $\delta\geq 2 \digamma_0^{-\frac{1}{16}}$.

\subsection{Matching asymptotic estimates: small $|\tomega|$}
\label{sec:matchasymptomega}
We will consider the following \emph{hierarchy of smallness of constants}:
\begin{equation*}
	\kappa_c\ll \digamma_0^{-1}\ll \delta\ll 1,
\end{equation*}
with moreover $\kappa_+\geq \kappa_c$. In order for all the constants $\delta$ in \S \ref{sec:homest} to be small, we need either $|\homega|$ or $|\tomega|+\kappa_+$ to be small. Since we are only assuming that $|\tomega|$ is small in the present section, this means that we need the additional assumption $\kappa_+\ll 1$.
\begin{proposition}
\label{prop:boundfreqcoeffesttomega1}
Let $(\homega,\ell)\in \mathcal{F}_{\flat,+}\cup \mathcal{F}_{\flat,\sim}$. Then the event-horizon-normalized solution $u_+$ to \eqref{eq:radialODE} with $H\equiv 0$ satisfies the estimates in \S \S\ref{sec:homestRegionI}--\S \S \ref{sec:regionIV}, with coefficients $A_i$, $B_i$, $D_i$, $E_i$ that satisfy the following upper bounds: there exists a constant $C=C(\q_0,\q_{1},L_0,\beta)>0$ such that
\begin{align}
	\label{eq:upboundAhor}
	|A_1|=&\:1,\quad A_2=0,\\
	\label{eq:upboundBhor}
	|B_1|\leq &\: C,\quad B_2=0,\\
	\label{eq:upboundD1hor}
	|D_1|\leq &\: C(\kappa_++|\tomega|)^{\frac{1}{2}+\frac{1}{2}\re \beta_{\ell}},\\
	\label{eq:upboundD2hor}
	 |D_2|\leq &\: C(\kappa_++|\tomega|)^{\frac{1}{2}-\frac{1}{2}\re \beta_{\ell}}(1+\delta_{\beta_{\ell}0}\log(1+(|\kappa_+|+|\tomega|)^{-1})),\\
	 \label{eq:upboundEhor}
	|E_i|\leq &\: C(\kappa_++|\tomega|)^{\frac{1}{2}-\frac{1}{2}\re \beta_{\ell}}(1+\delta_{\beta_{\ell}0}\log(1+(|\kappa_+|+|\tomega|)^{-1}))\quad i\in\{1,2\}.
\end{align}
Furthermore, for $(\omega,\ell)\in \mathcal{F}_{\flat,+}$, $\kappa_c\leq \kappa_+$ and both $\kappa_+$ and $\gamma$ suitably small, there exists a constant\\ $C=C(q_0,q_{\rm max},L_0,\gamma)>0$ such that
\begin{equation}
 \label{eq:lowboundEhor}
	|E_2|\geq \frac{1}{C}(\kappa_++|\tomega|)^{\frac{1}{2}-\re \beta_{\ell}}(1+\delta_{\beta_{\ell}0}\log(1+(|\kappa_+|+|\tomega|)^{-1})).
\end{equation}
\end{proposition}
\begin{proof}
Recall the boundary condition:
\begin{equation*}
	\lim_{r_*\to \infty }e^{i\tomega r_*-i q \int_{0}^{r_*}\rho_+(r_*')\,dr_*'}u_+(r_*)=1.
\end{equation*}

\textbf{Estimates for $A_i$: $|\tomega| \kappa_+^{-1}>  \digamma_0$}\\
By Proposition \ref{prop:regionIeleq1} and the boundary conditions on $u_+$, we immediately obtain $A_1=1$ and $A_2=0$.

\textbf{Estimates for $B_i$: $|\tomega| \kappa_+^{-1}\leq  \digamma_0$}\\

	Let $|\tomega| \kappa_+^{-1}\leq  \digamma_0$. Then we apply Proposition \ref{prop:regionIhypgeom} to obtain for $s_+\geq \delta^{-1} \kappa_+^{-\frac{1}{4}}(\kappa_c+|\homega|)^{\frac{3}{4}}$:
	\begin{multline*}
	(2\kappa_+ s_+)^{-\frac{1}{2}}(s_++1)(r^{-2}\Omega^2)^{\frac{1}{2}}u_+(r_*(s_+))=B_1(2\kappa_+(s_++1))^{-\frac{1}{2}}F_{\underline{\sigma},\mu}(2\kappa_+(s_++1))(1+\varepsilon_1(s_+))\\+B_2\left((2\kappa_+(s_++1))^{-\frac{1}{2}}G_{\underline{\sigma},\mu}(2\kappa_+(s_++1))+\varepsilon_2(s_+)\right),
\end{multline*}
with $\mu=\frac{1}{2}\beta_{\ell}$.

Note that by Lemma \ref{lm:metricest}:
\begin{equation*}
	(2\kappa_+ s_+)^{-\frac{1}{2}}(s_++1)(r^{-2}\Omega^2)^{\frac{1}{2}}=1+\kappa_+^{-1}O(s_+^{-1}).
\end{equation*}

Then, by \eqref{eq:tortoiseevent3} and \eqref{eq:asymphypgeomsmallarg1}:
\begin{multline*}
(2\kappa_+ s_+)^{-\frac{1}{2}}F_{\underline{\sigma},\mu}(2\kappa_+s_+)=(2\kappa_+s_+)^{-\frac{1}{2}}(1-\zeta(\tilde{s}))^{\frac{1}{2}(-1-i\frac{\tomega}{\kappa_+})}+\kappa_+^{-1}O(s_+^{-1})\\
=(2\kappa_+s_+)^{i\frac{\tomega}{2\kappa_+}}+\kappa_+^{-1}O(s_+^{-1})\\
=e^{i\frac{\tomega}{2\kappa_+}\log(2\kappa_+ s_+)}+\kappa_+^{-1}O(s_+^{-1})\\
=e^{-i\tomega r_*+i\frac{\tomega}{\kappa_+} O_{\infty}((\kappa_+^{-1} s_+^{-1})^0)}+\kappa_+^{-1}O(s_+^{-1}).
\end{multline*}

Since $G_{\underline{\sigma},\frac{1}{2}\beta_{\ell}}$ can be written as a linear combination of $F_{\underline{\sigma},\frac{1}{2}\beta_{\ell}}$ and $F_{-\underline{\sigma},\frac{1}{2}\beta_{\ell}}$ and $(2\kappa_+s_+)^{-\frac{1}{2}}F_{-\underline{\sigma},\frac{1}{2}\beta_{\ell}}$ features the factor $e^{+i\tomega r_*}$, we can conclude can conclude that $B_2=0$ and $|B_1|=1$.

\textbf{Estimates for $B_i$: $|\tomega| \kappa_+^{-1}>  \digamma_0$}\\
In order to obtain estimates for $B_i$, we will consider the intersection between Region I and Region II, where both Proposition \ref{prop:regionIeleq1} and Corollary \ref{cor:regioniiwhitt} apply. Let $\frac{1}{2}\digamma_0^{-\frac{1}{4}}|\tomega|^{-1}<s_+<2\digamma_0^{-\frac{1}{4}}|\tomega|^{-1}$. Write $\mu=\frac{1}{2}\beta_{\ell}$. Then by \eqref{eq:tortoiseevent3} and \eqref{eq:whittasymp1}:
\begin{multline*}
W_{ -i\sigma, \mu}(i \xi(s_+))= e^{\sign(\xi(s_+))\frac{\pi}{2}\sigma }e^{-  i\frac{\xi(s_+)}{2}} |\xi|^{ -i\sigma}(1+O(|\xi|^{-1}))\\
=e^{-\sign(\tomega)\frac{\pi}{2}\sigma }e^{+  iw s_+^{-1} }|2w|^{ -i\sigma}2s_+^{ -i\sigma}(1+|\tomega|O(|s_+|^{-1}))\\
=e^{-\sign(\tomega)\frac{\pi}{2}\sigma }e^{-  i\tomega r_*(1+O(\digamma_0^{-1})+O_{\infty}(\kappa_+ s_+)) }|2w|^{ -i\sigma}2s_+^{ -i\sigma}(1+|\tomega|O(|s_+|^{-1})).
\end{multline*}
Since $M_{-i\sigma,\mu}$ can be written as a linear combination of $W_{-i \sigma,\mu}$ and $W_{ +i\sigma,\mu}$, it will feature the factor $e^{+i\tomega r_*}$, so we can conclude that $B_2=0$ and $|B_1|=e^{\sign(\tomega)\frac{\pi}{2}\sigma }|A_1|=e^{\sign(\tomega)\frac{\pi}{2}\sigma }$.

\textbf{Estimates for $D_i$: $|\tomega| \kappa_+^{-1}\leq  \digamma_0$}\\
We will estimate $D_i$ in terms of $B_i$ by exploiting the intersection of Regions I and III in the case $|\tomega| \kappa_+^{-1}\leq  \digamma_0$. This step involves matching asymptotic expressions of hypergeometric functions to asymptotic expressions of Whittaker functions.

Let $\frac{3}{4}\digamma_0^{-\frac{1}{3}}|\homega|^{-\frac{1}{4}}(\kappa_++\kappa_c+|\tomega|)^{\frac{1}{4}}\leq s_c\leq \frac{4}{3}\digamma_0^{-\frac{1}{3}}|\homega|^{-\frac{1}{4}}(\kappa_++\kappa_c+|\tomega|)^{\frac{1}{4}}$, which is contained in $I\cap III$ by \eqref{eq:setincl2}. Then $2\kappa_+ s_+\ll 1$, since $|\homega|+\kappa_++\kappa_c\ll 1$ and $s_c\ll 1$. We can therefore apply Proposition \ref{prop:regionIhypgeom}, \eqref{eq:relhypgeomfunct1}, \eqref{eq:relhypgeomfunct2} and \eqref{eq:asymphypgeomsmallarg1} to obtain for $\mu=2\beta_{\ell}\neq 0$:
\begin{multline*}
	U_+(s_+)=B_1F_{\underline{\sigma},\mu}(2\kappa_+(s_++1))(1+\delta O(s_+^0))\\
	=B_1\Gamma\left(1+i(\sigma_a+\sigma_b)\right)\Bigg[\frac{\Gamma(-2\mu)}{\Gamma(\frac{1}{2}-\mu+i\sigma_a)\Gamma(\frac{1}{2}-\mu+i\sigma_b)}G_{\underline{\sigma},\mu}(2\kappa_+(s_++1))\\
	+\frac{\Gamma(2\mu)}{\Gamma(\frac{1}{2}+\mu+i\sigma_a)\Gamma(\frac{1}{2}+\mu+i\sigma_b)}G_{\underline{\sigma},-\mu}(2\kappa_+(s_++1))(1+\delta O_{\infty}(s_+^0))\Bigg]\\
	=B_1\Gamma\left(1+i(\sigma_a+\sigma_b)\right)\Bigg[\frac{\Gamma(-2\mu)}{\Gamma(\frac{1}{2}-\mu+i\sigma_a)\Gamma(\frac{1}{2}-\mu+i\sigma_b)}(2\kappa_+)^{\frac{1}{2}+\mu}s_+^{\frac{1}{2}+\mu}\\
	+\frac{\Gamma(2\mu)}{\Gamma(\frac{1}{2}+\mu+i\sigma_a)\Gamma(\frac{1}{2}+\mu+i\sigma_b)}(2\kappa_+)^{\frac{1}{2}-\mu}s_+^{\frac{1}{2}-\mu}(1+\delta O_{\infty}(s_+^0))\Bigg]\\
	=B_1\Gamma\left(1+i(\sigma_a+\sigma_b)\right)\Bigg[\frac{\Gamma(-2\mu)}{\Gamma(\frac{1}{2}-\mu+i\sigma_a)\Gamma(\frac{1}{2}-\mu+i\sigma_b)}(2\kappa_+)^{\frac{1}{2}+\mu}s_c^{-\frac{1}{2}-\mu}\\
	+\frac{\Gamma(2\mu)}{\Gamma(\frac{1}{2}+\mu+i\sigma_a)\Gamma(\frac{1}{2}+\mu+i\sigma_b)}(2\kappa_+)^{\frac{1}{2}-\mu}s_c^{-\frac{1}{2}+\mu}(1+\delta O_{\infty}(s_+^0))\Bigg],
\end{multline*}
where we used that $s_c= s_+^{-1}(1+\delta O_{\infty}(s_+^0))$ to arrive at the final step.

We can express the above expressions for $U_+(s_+)$ in terms of Whittaker functions $M_{-i\sigma,\mu}(i\xi(s_c))$ and $M_{-i\sigma,-\mu}(i\xi(s_c))$ and use the smallness of $|\xi(s_c)|=|2w_c s_c|$ together with \eqref{eq:whittasymp3} to obtain:
\begin{multline}
\label{eq:matchhypwhitmunonzero}
	U_c(s_c)=\frac{s_c+1}{s_++1}U_+(s_+(s_c))=B_1\Gamma\left(1+i(\sigma_a+\sigma_b)\right)\\
	\times\Bigg[\frac{\Gamma(-2\mu)}{\Gamma(\frac{1}{2}-\mu+i\sigma_a)\Gamma(\frac{1}{2}-\mu+i\sigma_b)}(2\kappa_+)^{\frac{1}{2}+\mu}|2\homega|^{-\frac{1}{2}+\mu}e^{\sign \homega \frac{i\pi}{2}(\frac{1}{2}-\mu)}M_{-i\sigma,-\mu}(-2iw_cs_c)\\
	+\frac{\Gamma(2\mu)}{\Gamma(\frac{1}{2}+\mu+i\sigma_a)\Gamma(\frac{1}{2}+\mu+i\sigma_b)}(2\kappa_+)^{\frac{1}{2}-\mu}|2\homega|^{-\frac{1}{2}-\mu}e^{\sign \homega \frac{i\pi}{2}(\frac{1}{2}+\mu)}M_{-i\sigma,\mu}(-2iw_cs_c)(1+\delta O_{\infty}(s_+^0))\Bigg].
\end{multline}
We apply \eqref{eq:whittasymp4}--\eqref{eq:whittasymp6} again to express $M_{-i\sigma,-\mu}$ in terms of $W_{-i\sigma,\mu}$ and $M_{-i\sigma,\mu}$ and obtain for $\mu\neq 0$:
\begin{multline*}
	U_c(s_c)=B_1\Gamma\left(1+i(\sigma_a+\sigma_b)\right)\frac{\Gamma(-2\mu)}{\Gamma(2\mu)}\frac{\Gamma\left(\frac{1}{2}+\mu+i\sigma\right)}{\Gamma\left(\frac{1}{2}+\mu+i\sigma_a\right)\Gamma\left(\frac{1}{2}+\mu+i\sigma_b\right)}\\
	\times (2\kappa_+)^{\frac{1}{2}+\mu}|2\homega|^{-\frac{1}{2}+\mu}e^{\sign(\homega)\frac{i\pi}{2}(\frac{1}{2}-\mu)}W_{\sigma,\mu}(-2iw_cs_c)(1+\delta O_{\infty}(s_+^0))\\
	+B_1\Gamma\left(1+i(\sigma_a+\sigma_b)\right)\frac{\Gamma(2\mu)}{\Gamma\left(\frac{1}{2}+\mu+i\sigma_a\right)\Gamma\left(\frac{1}{2}+\mu+i\sigma_b\right)}(2\kappa_+)^{\frac{1}{2}-\mu}|2\homega|^{-\frac{1}{2}-\mu}e^{\sign(\homega)\frac{i\pi}{2}(\frac{1}{2}+\mu)}\\
	\times \Bigg[1-(2\kappa_+)^{2\mu}|2\homega|^{2\mu}e^{-\sign(\homega)i\pi\mu)}\frac{\Gamma^2(-2\mu)}{\Gamma^2(2\mu)}\frac{\Gamma\left(\frac{1}{2}+\mu+i\sigma\right)}{\Gamma\left(\frac{1}{2}-\mu+i\sigma\right)}\frac{\Gamma\left(\frac{1}{2}+\mu+i\sigma_a\right)}{\Gamma\left(\frac{1}{2}-\mu+i\sigma_a\right)}\frac{\Gamma\left(\frac{1}{2}+\mu+i\sigma_b\right)}{\Gamma\left(\frac{1}{2}-\mu+i\sigma_b\right)}\Bigg]\\
	\times M_{\sigma,\mu}(-2iw_c s_c)(1+\delta O_{\infty}(s_+^0)).
	\end{multline*}
	
	When $\mu=0$, we instead apply \eqref{eq:relhypgeomfunct3} to obtain:
\begin{multline*}
	U_+(s_+)=B_1F_{\underline{\sigma},0}(2\kappa_+(s_++1))(1+\delta O(s_+^0))\\
	=-B_1\frac{\Gamma\left(1+i(\sigma_a+\sigma_b)\right)}{\Gamma(\frac{1}{2}+i\sigma_a)\Gamma(\frac{1}{2}+i\sigma_b)}G_{0,\sigma}(2\kappa_+(s_++1))\Bigg[\log (2\kappa_+(s_++1))(1+\delta O_{\infty}(s_+^0))\\
	+\left(2\gamma_{\rm Euler}+\frac{\Gamma'(\frac{1}{2}+i\sigma_a)}{\Gamma(\frac{1}{2}+i\sigma_a)}+\frac{\Gamma'(\frac{1}{2}+i\sigma_b)}{\Gamma(\frac{1}{2}+i\sigma_b)}\right)\Bigg]\\
	=-B_1\frac{\Gamma\left(1+i(\sigma_a+\sigma_b)\right)}{\Gamma(\frac{1}{2}+i\sigma_a)\Gamma(\frac{1}{2}+i\sigma_b)}(2\kappa_+)^{\frac{1}{2}}s_c^{-\frac{1}{2}}\Bigg[\log (2\kappa_+s_c^{-1})(1+\delta O_{\infty}(s_+^0))\\
	+\left(2\gamma_{\rm Euler}+\frac{\Gamma'(\frac{1}{2}+i\sigma_a)}{\Gamma(\frac{1}{2}+i\sigma_a)}+\frac{\Gamma'(\frac{1}{2}+i\sigma_b)}{\Gamma(\frac{1}{2}+i\sigma_b)}\right)\Bigg]
\end{multline*}
Therefore:
\begin{multline*}
	U_c(s_c)=\frac{s_c+1}{s_++1}U_+(s_+(s_c))=B_1\frac{\Gamma\left(1+i(\sigma_a+\sigma_b)\right)}{\Gamma(\frac{1}{2}+i\sigma_a)\Gamma(\frac{1}{2}+i\sigma_b)}(2\kappa_+)^{\frac{1}{2}}s_c^{\frac{1}{2}}\Bigg[(\log (s_c)+\log (2\kappa_+)^{-1}(1+\delta O_{\infty}(s_+^0))\\
	-\left(2\gamma_{\rm Euler}+\frac{\Gamma'(\frac{1}{2}+i\sigma_a)}{\Gamma(\frac{1}{2}+i\sigma_a)}+\frac{\Gamma'(\frac{1}{2}+i\sigma_b)}{\Gamma(\frac{1}{2}+i\sigma_b)}\right)\Bigg].
\end{multline*}
On the other hand, we can apply \eqref{eq:whittasymp3} and \eqref{eq:whittasymp7} to obtain:
\begin{align*}
M_{-i\sigma,0}(-2iw_cs_c)=&\:e^{-\sign(\homega)i\frac{\pi}{4}}|2\homega|^{\frac{1}{2}}s_c^{\frac{1}{2}}(1+\delta O(s_c^0)),\\
	W_{-i\sigma,0}(-2iw_cs_c)=&-\frac{1}{\Gamma(\frac{1}{2}+i\sigma)}e^{-\sign(\homega)i\frac{\pi}{4}}|2\homega|^{\frac{1}{2}}s_c^{\frac{1}{2}}\left[\log(-2i\homega s_c)+\frac{\Gamma'(\frac{1}{2}-i\sigma)}{\Gamma(\frac{1}{2}-i\sigma)}+2\gamma_{\rm Euler}+\delta O(s_c^0)\right]\\
	=&-\frac{1}{\Gamma(\frac{1}{2}+i\sigma)}e^{-\sign(\homega)i\frac{\pi}{4}}|2\homega|^{\frac{1}{2}}s_c^{\frac{1}{2}}\Bigg[\log(s_c)+\log|2\homega|-\sign(\homega)\frac{i\pi}{2}+\frac{\Gamma'(\frac{1}{2}-i\sigma)}{\Gamma(\frac{1}{2}-i\sigma)}\\
	+2\gamma_{\rm Euler}+\delta O(s_c^0)\Bigg].
\end{align*}
Hence,
\begin{multline*}
	U_c(s_c)=-B_1\frac{\Gamma\left(1+i(\sigma_a+\sigma_b)\right)}{\Gamma(\frac{1}{2}+i\sigma_a)\Gamma(\frac{1}{2}+i\sigma_b)}e^{+\sign(\homega)i\frac{\pi}{4}}|2\homega|^{-\frac{1}{2}}(2\kappa_+)^{\frac{1}{2}}\Gamma\left(\frac{1}{2}+i\sigma\right)(1+\delta O(s_c^0))W_{-i\sigma,0}(-2iw_cs_c)\\
	+B_1\frac{\Gamma\left(1+i(\sigma_a+\sigma_b)\right)}{\Gamma(\frac{1}{2}+i\sigma_a)\Gamma(\frac{1}{2}+i\sigma_b)}e^{+\sign(\homega)i\frac{\pi}{4}}|2\homega|^{-\frac{1}{2}}(2\kappa_+)^{\frac{1}{2}}\\
	\times\Bigg[\log (2\kappa_+)^{-1}(1+\delta O(s_c^0))-\log|2\homega|+\sign(\homega)\frac{i\pi}{2}-4\gamma_{\rm Euler}-\frac{\Gamma'(\frac{1}{2}+i\sigma)}{\Gamma(\frac{1}{2}+i\sigma)}-\frac{\Gamma'(\frac{1}{2}+i\sigma_a)}{\Gamma(\frac{1}{2}+i\sigma_a)}-\frac{\Gamma'(\frac{1}{2}+i\sigma_b)}{\Gamma(\frac{1}{2}+i\sigma_b)}\Bigg]\\
\times M_{-i\sigma,0}(-2iws_c)
\end{multline*}

Since Corollary \ref{cor:regioniiiwhitt} also applies for the $s_+$-values under consideration, we conclude that for $\mu\neq 0$:
\begin{align*}
	D_1=&\:B_1\Gamma\left(1+i(\sigma_a+\sigma_b)\right)\frac{\Gamma(-2\mu)}{\Gamma(2\mu)}\frac{\Gamma\left(\frac{1}{2}+\mu+i\sigma\right)}{\Gamma\left(\frac{1}{2}+\mu+i\sigma_a\right)\Gamma\left(\frac{1}{2}+\mu+i\sigma_b\right)}\\
	\times &\:(2\kappa_+)^{\frac{1}{2}+\mu}|2\homega|^{-\frac{1}{2}+\mu}e^{\sign(\homega)\frac{i\pi}{2}(\frac{1}{2}-\mu)}(1+O(\delta)),\\
	D_2=&\:B_1\Gamma\left(1+i(\sigma_a+\sigma_b)\right)\frac{\Gamma(2\mu)}{\Gamma\left(\frac{1}{2}+\mu+i\sigma_a\right)\Gamma\left(\frac{1}{2}+\mu+i\sigma_b\right)}(2\kappa_+)^{\frac{1}{2}-\mu}|2\homega|^{-\frac{1}{2}-\mu}e^{\sign(\homega)\frac{i\pi}{2}(\frac{1}{2}+\mu)}\\
	\times &\:\Bigg[1-(2\kappa_+)^{2\mu}|2\homega|^{2\mu}e^{-\sign(\homega)i\pi\mu}\frac{\Gamma^2(-2\mu)}{\Gamma^2(2\mu)}\frac{\Gamma\left(\frac{1}{2}+\mu+i\sigma\right)}{\Gamma\left(\frac{1}{2}-\mu+i\sigma\right)}\frac{\Gamma\left(\frac{1}{2}+\mu+i\sigma_a\right)}{\Gamma\left(\frac{1}{2}-\mu+i\sigma_a\right)}\frac{\Gamma\left(\frac{1}{2}+\mu+i\sigma_b\right)}{\Gamma\left(\frac{1}{2}-\mu+i\sigma_b\right)}\\
	+&\:O(\delta)\Bigg]\quad (\mu\neq 0),\\
	D_2=&\:B_1\frac{\Gamma\left(1+i(\sigma_a+\sigma_b)\right)}{\Gamma(\frac{1}{2}+i\sigma_a)\Gamma(\frac{1}{2}+i\sigma_b)}e^{+\sign(\homega)i\frac{\pi}{4}}|2\homega|^{-\frac{1}{2}}(2\kappa_+)^{\frac{1}{2}}\\
	\times&\:\Bigg[\log (2\kappa_+)^{-1}(1+\delta O(s_c^0))-\log|2\homega|+\sign(\homega)\frac{i\pi}{2}\\
	&-4\gamma_{\rm Euler}-\frac{\Gamma'(\frac{1}{2}+i\sigma)}{\Gamma(\frac{1}{2}+i\sigma)}-\frac{\Gamma'(\frac{1}{2}+i\sigma_a)}{\Gamma(\frac{1}{2}+i\sigma_a)}-\frac{\Gamma'(\frac{1}{2}+i\sigma_b)}{\Gamma(\frac{1}{2}+i\sigma_b)}+O(\delta)\Bigg]\quad (\mu= 0),
\end{align*}
where in the expression for $D_1$, we used that $\lim_{\mu\to 0}\frac{\Gamma(-2\mu)}{\Gamma(2\mu)}=-1$.

Since all the $\Gamma$-factors in the expressions for $D_1$ and $D_2$ can be bounded uniformly when $|\tomega| \kappa_+^{-1}\leq  \digamma_0$, we obtain $|D_1|\leq C |\kappa_+|^{\frac{1}{2}+\re \mu}$ and $|D_2|\leq C |\kappa_+|^{\frac{1}{2}-\re \mu}$ when $\mu\neq 0$ and $|D_2|\leq C |\kappa_+|^{\frac{1}{2}}\log |\kappa_+|^{-1}$ when $\mu=0$. 

When $\mu\notin i(0,\infty)$, we also obtain a lower bound: $|D_2|\geq \frac{1}{C} |\kappa_+|^{\frac{1}{2}-\re \mu}$ when $\mu\neq 0$ and $|D_2|\geq \frac{1}{C} |\kappa_+|^{\frac{1}{2}}\log \kappa_+^{-1}$ when $\mu=0$.

To obtain a lower bound for $|D_2|$ when $\mu\in i(0,\infty)$, we need to keep more careful track of the ratios of $\Gamma$-factors. Let:
\begin{multline*}
	g(\tomega,q,\mu,\kappa_+):=(2\kappa_+)^{2\mu}|2\homega|^{2\mu}e^{-\sign(\homega)i\pi\mu}\\
	\times \frac{\Gamma^2(-2\mu)}{\Gamma^2(2\mu)}\frac{\Gamma\left(\frac{1}{2}+\mu+i\sigma\right)}{\Gamma\left(\frac{1}{2}-\mu+i\sigma\right)}\frac{\Gamma\left(\frac{1}{2}+\mu+i\sigma_a\right)}{\Gamma\left(\frac{1}{2}-\mu+i\sigma_a\right)}\frac{\Gamma\left(\frac{1}{2}+\mu+i\sigma_b\right)}{\Gamma\left(\frac{1}{2}-\mu+i\sigma_b\right)}.
\end{multline*}
Then we conclude that $|D_2|\geq \frac{1}{C} |\kappa_+|^{\frac{1}{2}}$ if $|1-g|\geq b$ for some constant $b>0$.

Note that for $\mu\in i(0,\infty)$:
\begin{equation*}
	|g|(\tomega,q,\mu,\kappa_+)=e^{\sign(\homega)\pi \im \mu}\left(\frac{\cosh(\pi(\mu-\sigma))}{\cosh(\pi(\mu+\sigma))}\frac{\cosh(\pi(\mu-\sigma_a))}{\cosh(\pi(\mu+\sigma_a))}\frac{\cosh(\pi(\mu-\sigma_b))}{\cosh(\pi(\mu+\sigma_b))}\right)^{\frac{1}{2}}.
\end{equation*}
We will use that: $\homega=\q+O(\tomega)$, $\sigma=\q+O(\kappa_c)$, $\sigma_a=-\q+O(\kappa_c)+\kappa_+^{-1}|\tomega|O(|\homega|)$ and $\sigma_b=\q+\kappa_+^{-1}\tomega+O(\kappa_c)+\kappa_+^{-1}|\tomega|O(|\homega|)$ to estimate further:
\begin{equation*}
	|g|(\tomega,q,\mu,\kappa_+)=e^{-\sign(q)\pi \im \mu}\left(\frac{\cosh(\pi(\mu-q))}{\cosh(\pi(\mu+q))}\right)\left(\frac{\cosh(\pi(\mu-\sigma_b))}{\cosh(\pi(\mu+\sigma_b))}\right)^{\frac{1}{2}}+O(\kappa_c)+O(\homega).
\end{equation*}
By Lemma \ref{lm:coshratios}, we have that :
\begin{equation*}
	e^{\sign(\q)\pi \im \mu}\left(\frac{\cosh(\pi(\mu+\q))}{\cosh(\pi(\mu-\q))}\right)\in \begin{cases}
 (e^{-3\pi\im \mu},e^{-\pi\im \mu})\quad \q<0,\\
 (e^{\pi\im \mu},	e^{3\pi\im \mu})\quad \q>0.
 \end{cases}
 \end{equation*}
Furthermore, $\left(\frac{\cosh(\pi(\mu-\sigma_b))}{\cosh(\pi(\mu+\sigma_b))}\right)^{\frac{1}{2}}\in (e^{-\pi \im \mu},e^{\pi \im \mu})$, so $|g|$ must be bounded away from $1$ and we obtain the desired estimate for $|1-g|$, with a constant $B$ that depends on $\im \mu$ and $\digamma_0$.

The non-vanishing of the leading-order term in $D_2$ is closely related to the properties of the stationary $(\homega=0)$ solutions in Reissner--Nordstr\"om; see \cite{gaj26b}[Proposition 3.1].

\textbf{Estimates for $D_i$: $|\tomega| \kappa_+^{-1}> \digamma_0$}\\
We will estimate $D_i$ in terms of $B_i$ by exploiting the intersection of Regions II and III in the case $|\tomega| \kappa_+^{-1}> \digamma_0$. This step involves matching asymptotic expressions of Whittaker functions to asymptotic expressions of different Whittaker functions.

Let $\frac{3}{4}\digamma_0^{-\frac{1}{3}}|\homega|^{-\frac{1}{4}}(\kappa_++\kappa_c+|\tomega|)^{\frac{1}{4}}\leq s_c\leq \frac{4}{3}\digamma_0^{-\frac{1}{3}}|\homega|^{-\frac{1}{4}}(\kappa_++\kappa_c+|\tomega|)^{\frac{1}{4}}$, which is a subset of $II\cap III$ by \eqref{eq:setincl3}. Then we apply Corollary \ref{cor:regioniiwhitt} together with \eqref{eq:whittasymp4}--\eqref{eq:whittasymp6} to obtain for $\mu\neq 0$:
\begin{multline*}
	U_+(s_+)=B_1W_{-i\sigma,\mu}(-2i w_+s_+)(1+\delta O(s_+^0))=B_1\frac{\Gamma(-2\mu)}{\Gamma\left(\frac{1}{2}-\mu+i\sigma\right)}(1+\delta O(s_+^0))M_{-i\sigma,\mu}(-2i w_+s_+)\\
	+B_1\frac{\Gamma(2\mu)}{\Gamma\left(\frac{1}{2}+\mu+i\sigma\right)}(1+\delta O(s_+^0))M_{-i\sigma,-\mu}(-2i w_+s_+)\\
	=B_1\frac{\Gamma(-2\mu)}{\Gamma\left(\frac{1}{2}-\mu+i\sigma\right)}(1+\delta O(s_+^0))e^{-\sign(\tomega)\frac{i\pi}{2}(\frac{1}{2}+\mu)}|2\tomega s_+|^{\frac{1}{2}+\mu}\\
	+B_1\frac{\Gamma(2\mu)}{\Gamma\left(\frac{1}{2}+\mu+i\sigma\right)}(1+\delta O(s_+^0))e^{-\sign(\tomega)\frac{i\pi}{2}(\frac{1}{2}-\mu)}|2\tomega s_+|^{\frac{1}{2}-\mu}\\
	=B_1\frac{\Gamma(-2\mu)}{\Gamma\left(\frac{1}{2}-\mu+i\sigma\right)}(1+\delta O(s_+^0))e^{-\sign(\tomega)\frac{i\pi}{2}(\frac{1}{2}+\mu)}|2\tomega|^{\frac{1}{2}+\mu}s_c^{-\frac{1}{2}-\mu}\\
	+B_1\frac{\Gamma(2\mu)}{\Gamma\left(\frac{1}{2}+\mu+i\sigma\right)}(1+\delta O(s_+^0))e^{-\sign(\tomega)\frac{i\pi}{2}(\frac{1}{2}-\mu)}|2\tomega |^{\frac{1}{2}-\mu}s_c^{-\frac{1}{2}+\mu}.
\end{multline*}
We can express the above expressions for $U_+(s_+)$ in terms of Whittaker functions $M_{-i\sigma,\mu}(i\xi(s_c))$ and $M_{-i\sigma,-\mu}(i\xi(s_c))$ and use the smallness of $|\xi(s_c)|$, together with the expansions in Lemma \ref{lm:propwhitt} to obtain:
\begin{multline}
\label{eq:matchwhitwhitmunonzero}
	U_c(s_c)=\frac{s_c+1}{s_++1}U_+(s_+(s_c))\\
=B_1\frac{\Gamma(-2\mu)}{\Gamma\left(\frac{1}{2}-\mu+i\sigma\right)}(1+\delta O(s_+^0)e^{(\sign(\homega)-\sign(\tomega))\frac{i\pi}{4}}e^{-(\sign(\homega)+\sign(\tomega))\frac{i\pi}{2}\mu}|2\tomega|^{\frac{1}{2}+\mu}|2\homega|^{-\frac{1}{2}+\mu}M_{-i\sigma,-\mu}(-2i w_cs_c)\\
	+B_1\frac{\Gamma(2\mu)}{\Gamma\left(\frac{1}{2}+\mu+i\sigma\right)}(1+\delta O(s_+^0))e^{(\sign(\homega)-\sign(\tomega))\frac{i\pi}{4}}e^{+(\sign(\homega)+\sign(\tomega))\frac{i\pi}{2}\mu}|2\tomega|^{\frac{1}{2}-\mu}|2\homega|^{-\frac{1}{2}-\mu}M_{-i\sigma,\mu}(-2i w_cs_c).
\end{multline}
Now observe that \eqref{eq:matchwhitwhitmunonzero} has the same form as \eqref{eq:matchhypwhitmunonzero}, with the following substitutions:
\begin{align*}
	B_1\Gamma\left(1+i(\sigma_a+\sigma_b)\right)\mapsto&\: B_1,\\
	e^{\sign(\homega)\frac{i\pi}{4}}\mapsto &\:e^{(\sign(\homega)-\sign(\tomega))\frac{i\pi}{4}},\\
	e^{\pm \sign(\homega)\frac{i\pi}{2}\mu} \mapsto &\: e^{\pm (\sign(\homega)-\sign(\tomega))\frac{i\pi}{2}\mu},\\
	|2\kappa_+|^{\frac{1}{2}\pm \mu}\mapsto&\: |2\tomega|^{\frac{1}{2}\pm \mu},\\
	\Gamma\left(\frac{1}{2}-\mu+i\sigma_b\right)\mapsto &\: \Gamma\left(\frac{1}{2}-\mu+i\sigma\right),\\
	\Gamma\left(\frac{1}{2}-\mu+i\sigma_a\right)\mapsto &\: 1.
\end{align*}

Now let $\mu=0$. Then we apply instead \eqref{eq:whittasymp7} to obtain:
\begin{multline*}
	U_+(s_+)=B_1W_{-i\sigma,0}(-2i w_+s_+)(1+\delta O(s_+^0))\\
	= -\frac{B_1}{\Gamma(\frac{1}{2}+i\sigma)}M_{ -i\sigma,0}( -2i w_+s_+)\left(\log (-2\tomega s_+)(1+\delta O(s_+^0))+\frac{\Gamma'(\frac{1}{2}+i\sigma)}{\Gamma(\frac{1}{2}+i\sigma)}+2\gamma_{\rm Euler}\right)\\
	= \frac{B_1}{\Gamma(\frac{1}{2}+i\sigma)}e^{-\sign(\tomega)\frac{i\pi}{4}}|2\tomega|^{\frac{1}{2}}s_c^{-\frac{1}{2}}\left(-\log (-(2\tomega)^{-1})+\log(s_c)(1+\delta O(s_+^0))-\frac{\Gamma'(\frac{1}{2}+i\sigma)}{\Gamma(\frac{1}{2}+i\sigma)}-2\gamma_{\rm Euler}\right)
\end{multline*}
Hence, 
\begin{multline*}
	U_c(s_c)=-B_1e^{(\sign(\homega)-\sign(\tomega))i\frac{\pi}{4}}|2\homega|^{-\frac{1}{2}}|2\tomega|^{\frac{1}{2}}(1+\delta O(s_+^0))W_{-i\sigma,0}(-2i w_cs_c)\\
-B_1e^{(\sign(\homega)-\sign(\tomega))i\frac{\pi}{4}}\frac{|2\homega|^{-\frac{1}{2}}|2\tomega|^{\frac{1}{2}}}{\Gamma(\frac{1}{2}+i\sigma)}\left[\log(-(2\tomega)^{-1})-\log(2\homega)+2\frac{\Gamma'(\frac{1}{2}+i\sigma)}{\Gamma(\frac{1}{2}+i\sigma)}+4\gamma_{\rm Euler} +\delta O(s_+^0)\right]\\
\times M_{-i\sigma,0}(-2i w_cs_c).
\end{multline*}
Repeating the matched asymptotics argument in the hypergeometric case for $\mu\neq 0$ with the above substitutions and applying the above expression for $U_c(s_c)$ in the $\mu\neq 0$ case, we are left with:
\begin{align*}
	D_1=&\:B_1\frac{\Gamma(-2\mu)}{\Gamma(2\mu)}|2\tomega|^{\frac{1}{2}+\mu}|2\homega|^{-\frac{1}{2}+\mu}e^{(\sign(\homega)-\sign(\tomega))i\frac{\pi}{4}} e^{- (\sign(\homega)+\sign(\tomega))\frac{i\pi}{2}\mu}(1+ O(\delta)),\\
		D_2=&\:B_1\frac{\Gamma(2\mu)}{\Gamma\left(\frac{1}{2}+\mu+i\sigma\right)}|2\tomega|^{\frac{1}{2}-\mu}|2\homega|^{-\frac{1}{2}-\mu}e^{(\sign(\homega)-\sign(\tomega))i\frac{\pi}{4}} e^{(\sign(\homega)+\sign(\tomega))\frac{i\pi}{2}\mu}\\
	\times &\:\Bigg[1-|2\tomega|^{2\mu}|2\homega|^{2\mu}e^{-(\sign(\homega)+\sign(\tomega))i\pi\mu}\frac{\Gamma^2(-2\mu)}{\Gamma^2(2\mu)}\frac{\Gamma^2\left(\frac{1}{2}+\mu+i\sigma\right)}{\Gamma^2\left(\frac{1}{2}-\mu+i\sigma\right)}+O(\delta)\Bigg]\quad (\mu\neq 0),\\
	D_2=&\:-B_1e^{(\sign(\homega)-\sign(\tomega))i\frac{\pi}{4}}\frac{|2\homega|^{-\frac{1}{2}}|2\tomega|^{\frac{1}{2}}}{\Gamma(\frac{1}{2}+i\sigma)}\Bigg[\log(-(2\tomega)^{-1})-\log(2\homega) \\
	&\:+2\frac{\Gamma'(\frac{1}{2}+i\sigma)}{\Gamma(\frac{1}{2}+i\sigma)}+4\gamma_{\rm Euler}+O(\delta)\Bigg]\quad (\mu= 0),
\end{align*}
where we used that $\lim_{\mu\to 0}\frac{\Gamma(-2\mu)}{\Gamma(2\mu)}=-1$.

Since all the gamma factors in the expressions for $D_1$ and $D_2$ can be bounded uniformly, we obtain $|D_1|\leq C |\tomega|^{\frac{1}{2}+\re \mu}$ and $|D_2|\leq C |\tomega|^{\frac{1}{2}-\re \mu}$ when $\mu\neq 0$ and $|D_2|\leq C |\tomega|^{\frac{1}{2}}\log |\tomega|^{-1}$ when $\mu=0$. 

When $\mu\notin i(0,\infty)$, we also obtain a lower bound: $|D_2|\geq \frac{1}{C} |\tomega|^{\frac{1}{2}-\re \mu}$ when $\mu\neq 0$ and $|D_2|\geq \frac{1}{C} |\tomega|^{\frac{1}{2}}\log |\tomega|^{-1}$ when $\mu=0$.

To obtain a lower bound for $|D_2|$ when $\mu\in i(0,\infty)$, we need to keep more careful track of the ratios of gamma factors, as in the case of matched asymptotics with hypergeometric functions.

Let $\mu\in i(0,\infty)$ and define in this case:
\begin{equation*}
g(\tomega,\sigma,\mu):=	|2\tomega|^{2\mu}|2\homega|^{2\mu}e^{-(\sign(\homega)+\sign(\tomega))i\pi\mu}\frac{\Gamma^2(-2\mu)}{\Gamma^2(2\mu)}\frac{\Gamma^2\left(\frac{1}{2}+\mu+i\sigma\right)}{\Gamma^2\left(\frac{1}{2}-\mu+i\sigma\right)}.
\end{equation*}
Then, by the fact that: $\homega=-q+O(\tomega)$ and $\sigma=q+O(\kappa_c)$, we obtain:
\begin{multline*}
	|g|(\tomega,q,\mu,\kappa_+)=e^{(\sign(\homega)+\sign(\tomega))\pi \im \mu}\frac{\cosh(\pi(\mu-\sigma))}{\cosh(\pi(\mu+\sigma))}\\
	=e^{(\sign(q)+\sign(\tomega))\pi \im \mu}\frac{\cosh(\pi(\mu+q))}{\cosh(\pi(\mu-q))}+O(|\tomega|).
\end{multline*}
 Furthermore, by Lemma \ref{lm:coshratios}:
 \begin{equation*}
 	e^{(\sign(q)+\sign(\tomega))\pi \im \mu}\frac{\cosh(\pi(\mu+q))}{\cosh(\pi(\mu-q))}\in \begin{cases}
 		(e^{(1+\sign(\tomega))\pi \im \mu},e^{(2+\sign(\tomega))\pi \im \mu})\quad (q>0),\\
 		(e^{-(2+\sign(\tomega))\pi \im \mu},e^{-(1+\sign(\tomega))\pi \im \mu})\quad (q<0)
 	\end{cases}
 \end{equation*}
We conclude that $|1-g|\geq b$ for some constant $b >0$ that depends on $q$, $\mu$ and $\tomega$.

\textbf{Estimates for $E_i$:}
We will match the Whittaker asymptotics in Region III to the complex exponentials in Region IV. 

By Lemma \ref{lm:propwhitt}, we have that:
\begin{multline*}
	u(r_*)=(s_c+1)^{-1}(r^{-2}\Omega^2)^{-\frac{1}{2}}U_c(s_c(r_*))\\
	=(s_c+1)^{-1}(r^{-2}\Omega^2)^{-\frac{1}{2}}[D_1W_{-i\sigma,\mu}(-2i w_cs_c)+D_2M_{-i\sigma,\mu}(-2i w_cs_c)]\\
	=(s_c+1)^{-1}(r^{-2}\Omega^2)^{-\frac{1}{2}}\Bigg[\left(D_1+D_2\frac{\Gamma(1+2\mu)}{\Gamma(\frac{1}{2}+\mu-i\sigma)}e^{\sign(\homega)\pi(\sigma-i\mu-\frac{1}{2})}\right)W_{-i\sigma,\mu}(-2i w_cs_c)\\
	+D_2\frac{\Gamma(1+2\mu)}{\Gamma(\frac{1}{2}+i\sigma)}e^{\sign(\homega)}W_{+i\sigma,\mu}(-2i w_cs_c)\Bigg].
\end{multline*}

Let $\frac{\frac{1}{2}(1+r_c^{-2})\kappa_c^{-\frac{1}{4}}+r_c^{-1}}{(1-r_c^{-1})+\frac{1}{2}r_c^{-1}\kappa_c^{-\frac{1}{4}}}\leq s_c\leq \frac{2(1+r_c^{-2})\kappa_c^{-\frac{1}{4}}+r_c^{-1}}{(1-r_c^{-1})+2r_c^{-1}\kappa_c^{-\frac{1}{4}}}$, which corresponds to a subset of $III\cap IV$ by \eqref{eq:setincl4}. Then $s_c\gg 1$, $|\xi(s_c)|\gg 1$ and $\kappa_c s_c\ll 1$ and by \eqref{eq:tortoisecosmo2} and by Lemma \ref{lm:metricest}:
\begin{align*}
	(s_c+1)^2 r^{-2}\Omega^2=&\:1+O(\kappa_c^{\frac{3}{4}})+O(s_c^{-1}),\\
	r_*(s_c)=&\: s_c(1+\kappa_c^{\frac{3}{4}}O(s_c^0))+\log s_c O(s_c^0)
\end{align*}
so we can apply the large $|\xi|$ Whittaker asymptotics in Lemma \ref{lm:propwhitt} to write:
\begin{multline*}
	u(r_*)
	=\left(D_1+D_2\frac{\Gamma(1+2\mu)}{\Gamma(\frac{1}{2}+\mu-i\sigma)}e^{\sign(\homega)\pi(\sigma-i\mu-\frac{1}{2})}\right)e^{i\homega r_*+i q \int_{0}^{r_*}\rho_c(r_*')\,dr_*'}(1+\delta O(s_c^0))\\
	+D_2\frac{\Gamma(1+2\mu)}{\Gamma(\frac{1}{2}+i\sigma)}e^{\sign(\homega)}e^{-i\homega r_*-i q \int_{0}^{r_*}\rho_c(r_*')\,dr_*'}(1+\delta O(s_c^0)).
\end{multline*}
We conclude that Proposition \ref{prop:regionIVeleq1} holds with:
\begin{align*}
	E_1=&\:(D_1+D_2\frac{\Gamma(1+2\mu)}{\Gamma(\frac{1}{2}+\mu-i\sigma)}e^{\sign(\homega)\pi(\sigma-i\mu-\frac{1}{2})}+O(\delta),\\
	E_2=&\:D_2\frac{\Gamma(1+2\mu)}{\Gamma(\frac{1}{2}+i\sigma)}e^{\sign(\homega)}(1+O(\delta)).
\end{align*}
The upper and lower bounds for $E_1$ and $E_2$ therefore follow immediately from the upper and lower bounds for $D_1$ and $D_2$ above.
\end{proof}

\begin{proposition}
\label{prop:boundfreqcoeffesttomega2}
Let $(\omega,\ell)\in \mathcal{F}_{\flat,+}\cup \mathcal{F}_{\flat,\sim}$. Then the infinity-normalized solution $u_{\infty}$ to \eqref{eq:radialODE} with $H\equiv 0$ satisfies the estimates in \S \S\ref{sec:homestRegionI}--\S \S \ref{sec:regionIV} with coefficients $A_i$, $B_i$, $D_i$, $E_i$ that satisfy the following upper bounds: there exists a constant $C=C(\q_0,\q_{1},L_0,\beta)>0$ such that
\begin{align}
	\label{eq:upboundEinf}
	|E_1|=&\:1,\quad E_2=0,\\
	\label{eq:upboundDinf}
	|D_1|\leq &\: C,\quad D_2=0,\\
	\label{eq:upboundB1inf}
	|B_1|\leq &\: C(\kappa_++|\tomega|)^{-\frac{1}{2}+\frac{1}{2}\re \beta_{\ell}},\\
	\label{eq:upboundB2inf}
	 |B_2|\leq &\: C(\kappa_++|\tomega|)^{-\frac{1}{2}-\frac{1}{2}\re \beta_{\ell}}(1+\delta_{\beta_{\ell}0}\log(1+(|\kappa_+|+|\tomega|)^{-1})),\\
	 \label{eq:upboundAinf}
	|A_i|\leq &\: C|\tomega|^{-1}(\kappa_++|\tomega|)^{\frac{1}{2}-\re \beta_{\ell}}(1+\delta_{\beta_{\ell}0}\log(1+(|\kappa_+|+|\tomega|)^{-1}))\quad i\in\{1,2\}.
\end{align}
\end{proposition}
\begin{proof}
	We proceed as in the proof of the upper bound estimates in Proposition \ref{prop:boundfreqcoeffesttomega1}, but we start at $r=r_c$. In addition, we apply need to apply the uniform estimates on hypergeometric functions from Corollary \ref{cor:esthypgeomG} to arrive at estimates for $A_1$ and $A_2$ in terms of $B_1$ and $B_2$.
\end{proof}

\subsection{Matching asymptotic estimates: small $|\homega|$}
\label{sec:matchasymphomega}
We will consider the following hierarchy of smallness:
\begin{equation*}
	\kappa_c\ll \digamma_0^{-1}\ll \delta\ll 1.
\end{equation*}
We will also assume that $\kappa_+\geq \kappa_c$. Note that we do not need to impose a smallness condition on $\kappa_+$ when $|\homega|$ is small to guarantee smallness of the constants $\delta$ in  \S \ref{sec:homest}.

We can state analogues of the propositions in \S\ref{sec:matchasymphomega} in the case, including lower bounds when $\kappa_c$ and $|\homega|$ are suitably small. We can simply repeat the arguments in \S\ref{sec:matchasymphomega} with the following exchanges of variables and constants:
\begin{align*}
	\tomega &\leftrightarrow -\homega,\\
	\kappa_c &\leftrightarrow \kappa_+,\\
	s_c &\leftrightarrow s_+,\\
	r_*& \leftrightarrow -r_*.
\end{align*}

\begin{proposition}
\label{prop:boundfreqcoeffesthomega1}
Let $(\omega,\ell)\in \mathcal{F}_{\flat,\infty}\cup \mathcal{F}_{\flat,\sim}$. Then the infinity-normalized solution $u_{\infty}$ to \eqref{eq:radialODE} with $H\equiv 0$ satisfies the estimates in \S \S\ref{sec:homestRegionI}--\S \S \ref{sec:regionIV}, and with coefficients $A_i$ in Region I, $B_i$ in Region I and II, $D_i$ in Region III, $E_i$ in Region IV that satisfy the following upper bounds: there exists a constant $C=C(\q_0,\q_{1},L_0,\beta)>0$ such that
\begin{align}
	\label{eq:upboundEinf2}
	|E_1|=&\:1,\quad E_2=0,\\
	\label{eq:upboundDinf2}
	|D_1|\leq &\: C,\quad D_2=0,\\
	\label{eq:upboundB1inf2}
	|B_1|\leq &\: C(\kappa_c+|\homega|)^{\frac{1}{2}+\re \beta_{\ell}},\\
	\label{eq:upboundB2inf2}
	 |B_2|\leq &\: C(\kappa_c+|\homega|)^{\frac{1}{2}-\re \beta_{\ell}}(1+\delta_{\beta_{\ell}0}\log(1+(|\kappa_c|+|\homega|)^{-1})),\\
	 \label{eq:upboundAinf2}
	|A_i|\leq &\: C(\kappa_c+|\homega|)^{\frac{1}{2}-\re \beta_{\ell}}(1+\delta_{\beta_{\ell}0}\log(1+(|\kappa_c|+|\homega|)^{-1}))\quad i\in\{1,2\}.
\end{align}
Furthermore, for $(\omega,\ell)\in \mathcal{F}_{\flat,\infty}$ and $\kappa_c$ and $\gamma$ suitably small, there exists a constant \\$C=C(\q_0,\q_{1},L_0,\gamma)>0$ such that
\begin{equation}
 \label{eq:lowboundAinf2}
	|A_2|\geq \frac{1}{C}(\kappa_c+|\homega|)^{\frac{1}{2}-\re \beta_{\ell}}(1+\delta_{\beta_{\ell}0}\log(1+(|\kappa_c|+|\homega|)^{-1})).
\end{equation}
\end{proposition}

\begin{proposition}
\label{prop:boundfreqcoeffesthomega2}
Let $(\omega,\ell)\in \mathcal{F}_{\flat,+}\cup \mathcal{F}_{\flat,\sim}$. Then the event-horizon-normalized solution $u_{+}$ to \eqref{eq:radialODE} with $H\equiv 0$ satisfies the estimates in \S \S\ref{sec:homestRegionI}--\S \S \ref{sec:regionIV} with coefficients $A_i$, $B_i$, $D_i$, $E_i$ that satisfy the following upper bounds: there exists a constant $C=C(q_0,q_{\rm max},L_0,\beta)>0$ such that
\begin{align}
	\label{eq:upboundAhor2}
	|A_1|=&\:1,\quad A_2=0,\\
	\label{eq:upboundBhor2}
	|B_1|\leq &\: C,\quad B_2=0,\\
	\label{eq:upboundD1hor2}
	|D_1|\leq &\: C(\kappa_c+|\homega|)^{-\frac{1}{2}+\re \beta_{\ell}},\\
	\label{eq:upboundD2hor2}
	 |D_2|\leq &\: C(\kappa_c+|\homega|)^{-\frac{1}{2}-\re \beta_{\ell}}(1+\delta_{\beta_{\ell}0}\log(1+(|\kappa_c|+|\homega|)^{-1})),\\
	 \label{eq:upboundEhor2}
	|E_i|\leq &\: C|\homega|^{-1}(\kappa_c+|\homega|)^{\frac{1}{2}-\re \beta_{\ell}}(1+\delta_{\beta_{\ell}0}\log(1+(|\kappa_c|+|\homega|)^{-1}))\quad i\in\{1,2\}.
\end{align}
\end{proposition}

\subsection{Wronskian estimates}
\label{sec:wronskianestimates}
In this section, we will combine the lower-bound estimates in Propositions \ref{prop:boundfreqcoeffesttomega1}, \ref{prop:boundfreqcoeffesttomega2}, \ref{prop:boundfreqcoeffesthomega1} and \ref{prop:boundfreqcoeffesthomega2} together with appropriate Wronskian estimates in the frequency regime $\mathcal{F}_{\flat,\sim }$ that follow from an application of the integral transformations introduced in \cite{costa20} in the extremal Kerr setting.

In the lemma below, we show that for non-superradiant frequencies we can immediately obtain an estimate for $|\mathfrak{W}|^{-1}$ by considering $j^T[u_{\infty}]$.
\begin{lemma}
\label{lm:nonsuperradwronsk}
Let $(\homega,\ell)\in \mathcal{F}_{\flat,\sim }$ and assume that $\homega \tomega>0$. Then:
\begin{equation}
\label{eq:simplewronskext}
|\mathfrak{W}|^{-1}\geq \frac{1}{\sqrt{\gamma}}.
\end{equation}
\end{lemma}
\begin{proof}
We can write:
\begin{align}
\label{eq:basicexpode}
u_{+}(r_*)=&\:{\alpha}_{1,+}u_{\infty}(r_*)+{\alpha}_{2,+}\overline{u}_{\infty}(r_*).
\end{align}
Hence,
\begin{multline}
\label{eq:wronskcoeff1}
-\mathfrak{W}=\mathcal{W}(u_{\infty},u_+)\\
={\alpha}_{1,+}u_{\infty}(\infty)\frac{du_{\infty}}{dr_*}(\infty)+{\alpha}_{2,+}u_{\infty}(\infty)\frac{d\overline{u_{\infty}}}{dr_*}(\infty)-({\alpha}_{1,+}u_{\infty}(\infty)+{\alpha}_{2,+}\overline{u}_{\infty}(\infty))\frac{du_{\infty}}{dr_*}(\infty)\\
={\alpha}_{2,+}\left(u_{\infty}\frac{d\overline{u_{\infty}}}{dr_*}-\overline{u}_+\frac{du_{\infty}}{dr_*}\right)(\infty)=-2i\homega {\alpha}_{2,+}.
\end{multline}
Furthermore, by considering $j^K[u_{+}]:=-\tomega  \re(i u_{+}'\overline{u}_{+})$ and using that $\frac{d}{dr_*}j^K[u_{+}]=0$, we obtain:
\begin{multline*}
\tomega^2=\tomega^2|u_{+}|^2(-\infty)=\tomega   \re(i u_{+}'\overline{u}_{+})(-\infty)=\homega   \re(i u_{+}'\overline{u}_{+})(\infty)\\
=-\homega \tomega  \re(({\alpha}_{1,+}-{\alpha}_{2,+})(\overline{{\alpha}_{1,+}}+\overline{{\alpha}_{2,+}}))=\homega\tomega(|{\alpha}_{2,+}|^2-|{\alpha}_{1,+}|^2)\\
\leq \homega\tomega|{\alpha}_{2,+}|^2.
\end{multline*}
Hence,
\begin{equation*}
|\mathfrak{W}|^2=4\homega^2| {\alpha}_{2,+}|^2\geq \tomega \homega\geq \gamma. \qedhere
\end{equation*}
\end{proof}

We now consider superradiant frequencies in $\mathcal{F}_{\flat,\sim }$, which satisfy $\homega \tomega<0$. We first restrict to the case $\kappa_+=\kappa_c=0$. We obtain an upper bound estimate for $|\mathfrak{W}|^{-1}$ by applying the arguments in \cite{costa20} with some slight modifications.
\begin{theorem}
\label{thm:modestab}
Let $\kappa_c=\kappa_+=0$ and $(\omega,\ell)\in\mathcal{F}_{\flat,\sim }$. Assume that $\homega \tomega<0$. Then there exists a constant $K=K(\gamma,L_0)>0$ such that:
\begin{equation}
\label{eq:ritawronskext}
|\mathfrak{W}|^{-1}\leq  K.
\end{equation}
\end{theorem}
\begin{proof}
In the case of \eqref{eq:radialODE} with $\Lambda=0$ and $|Q|=M=1$, the ODE for $R=r^{-1}u$ can be written as follows:
\begin{equation}
\label{eqconfluentheun}
\frac{d}{dr}\left((r-1)^2\frac{d}{dr}R\right)-\left[\frac{(\beta-\alpha (r-1)-\gamma (r-1)^2)^2}{(r-1)^2}+L\right]R=r \Omega^{-2}H,
\end{equation}
with
\begin{align*}
\alpha=&-i(\omega+\tomega)=-2i\left(\omega-\frac{\q}{2}\right),\\
\beta=&\:i\tomega,\\
\gamma=&-i\omega,\\
L=&\:\ell(\ell+1).
\end{align*}
Note that $\gamma$ here is unrelated to the $\gamma$ appearing in the definition of $\mathcal{F}_{\flat,\sim }$. Equation \eqref{eqconfluentheun} is an inhomogeneous doubly-confluent Heun equation; see for example \cite{NIST:DLMF}[\S 31.12].

In the case of the radial ODE on extremal Kerr with $M=1$, on the other hand, $R$ also satisfies the form \eqref{eqconfluentheun}, but with parameters:
\begin{align*}
\alpha_{\rm Kerr}=&-2i\omega,\\
\beta_{\rm Kerr}=&\:2i\tomega,\\
\gamma_{\rm Kerr}=&-i\omega,\\
L_{\rm Kerr}=&\:\lambda+\omega^2-2m\omega,
\end{align*}
with $\tomega=\omega-\frac{m}{2}$; see \cite{costa20}[Eq. (3.2)]. With the translations $a\leftrightarrow Q$ and $(\alpha_{\rm Kerr},\beta_{\rm Kerr}, L_{\rm Kerr}) \leftrightarrow (\alpha, \beta, L)$, the confluent Heun operators are therefore are the same, so we can directly apply methods from the extremal Kerr setting.

It then follows from \cite{costa20}[Proposition 3.1] that
\begin{equation*}
\mathbb{U}(x):=\lim_{y\to 0}(x^2+2)^{\frac{1}{2}}(x-2)^{\alpha}\int_1^{\infty}e^{-2\gamma (x+iy-1)(r-1)}(r-1)^{\alpha}e^{\beta (r-1)^{-1}}e^{-\gamma r}R(r)\,dr
\end{equation*}
is well-defined as a limit with respect to $L^2_x((2,\infty))$ and is a smooth solution to the ODE
\begin{equation*}
\frac{d^2\mathbb{U}}{dx_*^2}+\mathbb{V}\mathbb{U}=\frac{(x-1)(x-2)}{x^2+2}\mathbb{H},
\end{equation*}
with $\frac{dx_*}{dx}=\frac{x^2+2}{(x-1)(x-2)}$,
\begin{equation*}
\mathbb{H}(x)=\lim_{y\to 0}(x^2+2)^{\frac{1}{2}}(x-2)^{\alpha}\int_1^{\infty}e^{-2\gamma (x+iy-1)(r-1)}(r-1)^{\alpha}e^{\beta (r-1)^{-1}}e^{-\gamma r}(r\Omega^{-2}H)\,dr
\end{equation*}
and
\begin{multline*}
\mathbb{V}(x)=(x^2+2)^{-4}(x-1)\Big[4\beta \gamma (x^2+2)^2(x-2)(x-1)-\alpha^2(x^2+2)^2(x-1)\\
-L (x^2+2)^2(x-2)-(x-2)(3x^3+2x(x-6)+4)\Big].
\end{multline*}
Note that $\mathbb{V}$ is real-valued and that in the extremal Kerr case with $M=1$, one obtains the same expression for $\mathbb{V}$, but with $\alpha,\beta,\gamma, L$ replaced by $\alpha_{\rm Kerr},\beta_{\rm Kerr},\gamma_{\rm Kerr},L_{\rm Kerr}$. Filling in the above expressions for $\alpha,\beta,\gamma, L$, we obtain:
\begin{align*}
\mathbb{V}(\infty)=&\:\omega^2,\\
\mathbb{V}(2)=&\:\frac{1}{9}\left(\omega-\frac{\q}{2}\right)^2.
\end{align*}
Furthermore, for $\omega\tomega<0$, \cite{costa20}[Proposition 3.1] gives the following asymptotic behaviour:
\begin{align*}
x^{\frac{1}{4}}e^{8\sqrt{-\omega \tomega}x^{\frac{1}{2}}}\left[\frac{d\mathbb{U}}{dx_*}-4i\sqrt{-\omega \tomega} x^{-\frac{1}{2}} \mathbb{U}\right](x_*)=&\: O(x^{-\frac{1}{2}})\quad\textnormal{as $x_*\to \infty$},\\
\left[\frac{d\mathbb{U}}{dx_*}+\frac{i}{3}\left(\omega+\frac{q}{2}\right)\mathbb{U}\right](x_*)\to &\:0\quad\textnormal{as $x_*\to -\infty$},\\
|e^{8\sqrt{-\omega \tomega}x^{\frac{1}{2}}}x^{-\frac{1}{4}}\mathbb{U}|^2(\infty)=&\: \frac{\pi}{4|\omega|}\left|\frac{\tomega}{\omega}\right|^{\frac{1}{2}}|u(1)|^2.
\end{align*}

Defining $ \mathbb{J}^T[\mathbb{U}]:=(\omega+\frac{q}{2})\re(i \frac{d\mathbb{U}}{dx_*}\overline{\mathbb{U}})$, and using that $\mathbb{V}$ is real, we obtain 
\begin{align*}
\frac{d}{dx_*}\left( \mathbb{J}^T[\mathbb{U}]\right)= \left(\omega+\frac{q}{2}\right)\frac{(x-1)(x-2)}{x^2+a^2}\re(\overline{i \mathbb{H}}\mathbb{U}).
\end{align*}
By applying the above boundary conditions on $\mathbb{U}$, we moreover obtain for $\omega \tomega <0$:
\begin{align*}
\mathbb{J}^T[\mathbb{U}](\infty)=&\:0,\\
\mathbb{J}^T[\mathbb{U}](-\infty)=&\:-\frac{1}{2}\left(\left|\frac{d\mathbb{U}}{dx_*}\right|^2(-\infty)+\frac{1}{3}\left(\omega+\frac{q}{2}\right)^2|\mathbb{U}|^2(-\infty)\right)=-\frac{2}{9}\left(\omega-\frac{\q}{2}\right)^2|\mathbb{U}|^2(-\infty).
\end{align*}
In contrast with the currents $j^T[u]$ for solutions $u$ to \eqref{eq:radialODE} in the superradiant frequency regime, see Lemma \ref{lm:superradiance}, we have positivity of $\mathbb{J}^T[\mathbb{U}](\infty)-\mathbb{J}^T[\mathbb{U}](-\infty)$.

Furthermore,
\begin{equation*}
|\mathbb{J}^T[\mathbb{U}](2)|+|\mathbb{J}^T[\mathbb{U}](\infty)|\leq\left|\left(\omega+\frac{q}{2}\right)\int_{\R_{x_*}}\frac{(x-1)(x-2)}{x^2+a^2}\re(\overline{i \mathbb{H}} \mathbb{U})\,dx_*\right|
\end{equation*}

In view of the above identities and estimates for $\mathbb{J}^T[\mathbb{U}]$ the proof of \cite{costa20}[Proposition 5.4] applies without modification, and the corresponding result can be used to obtain a proof of the Wronskian estimate in \cite{costa20}[Theorem 5.1] with minor modifications. We then conclude \eqref{eq:ritawronskext}.
\end{proof}

\begin{lemma}
\label{lm:estubyexthom}
Let $(\omega,\ell)\in\mathcal{F}_{\flat,\sim }$. Let $\delta>0$. Let $u$ be a solution to \eqref{eq:radialODE} with $H\equiv 0$. Let $u_{+,0}$ and $u_{\infty,0}$ be the event-horizon-normalized and infinity-normalized solutions in the case $\kappa_c=\kappa_+=0$. 
Let $\epsilon>0$. Then, for $\kappa_+,\kappa_c\geq 0$, there exist $a_+,a_{\infty}\in \C$ and uniform constants $C>0$ and $\delta>0$, such that for $r\in (r_++\epsilon, \frac{1}{r_c^{-1}+\epsilon})$, we can write:
\begin{equation}
\label{eq:estubyexthom}
	u(r_*(r))=a_+ (u_{+,0}(r_*(r))+\varepsilon_1(r))+a_{\infty}(u_{\infty}(r_*(r))+\varepsilon_2(r)),
\end{equation}
with $\varepsilon_i(r)\leq C(\kappa_++\kappa_c)K \epsilon^{-3}$, where $K$ is the constant appearing in \eqref{eq:ritawronskext}.
\end{lemma}
\begin{proof}
As in the proof of Lemma \ref{lm:U+Uceqs}, we we can write:
	\begin{equation*}
\frac{d^2(\Omega u)}{dr^2}+\Omega^{-4}\left(\homega^2-V-\Omega^3\frac{d^2\Omega}{dr^2}\right)(\Omega u)=0.
\end{equation*}
We can split up the right-hand side above in the following way:
\begin{equation*}
\frac{d^2(\Omega u)}{dr^2}+\left[\Omega^{-4}\left(\homega^2-V-\Omega^3\frac{d^2\Omega}{dr^2}\right)\right]_{\kappa_+=\kappa_c=0}\Omega u=\vartheta \Omega u.
\end{equation*}
Then, $|\vartheta(r)|\leq C(\kappa_++\kappa_c)(\rho_+^{-4}+\rho_c^{-4})$.

	We will now apply Proposition \ref{prop:genoderrorest} with $W_1(r):=((1-r^{-1})u_{+,0})(r)$ and $W_2=((1-r^{-1}) u_{\infty,0})(r)$. Note first that by Theorem \ref{thm:modestab}:
	\begin{equation*}
		|\mathcal{W}(W_1,W_2)|=|\mathfrak{W}|\geq \frac{1}{K}.
	\end{equation*}
	Furthermore, by applying the upper bound estimates in Propositions \ref{prop:boundfreqcoeffesttomega1}, \ref{prop:boundfreqcoeffesttomega2}, \ref{prop:boundfreqcoeffesthomega1} and \ref{prop:boundfreqcoeffesthomega2}, we conclude moreover that $|W_i|\leq C(1-r^{-1})$ and $|\frac{dW_i}{dr}|\leq C(1-r^{-1})^{-1}$, so we can take $P_0=Q=P_1\equiv 1$ and $f\equiv 1$ in Proposition \ref{prop:genoderrorest}.
	
	Finally, observe that
	\begin{equation*}
		|\mathfrak{W}|^{-1}\int_{r_++\epsilon}^{r_c-\epsilon^{-1}}|\vartheta|(r')\,dr'\leq C(\kappa_++\kappa_c)|\mathfrak{W}|^{-1}\epsilon^{-3}.
	\end{equation*}
	We conclude that
	\begin{equation*}
	\Omega u(r_*(r))=a_+ \left[((1-r^{-1})u_{+,0})(r)+\tilde{\varepsilon}_1(r)\right]+a_{\infty}\left[((1-r^{-1}) u_{\infty,0})(r)+\tilde{\varepsilon}_2(r)\right],
\end{equation*}
with $|\tilde{\varepsilon}_i|(r)\leq C \epsilon^{-3}(\kappa_++\kappa_c)(1-r^{-1})$. Since in the interval $(r_++\epsilon, \frac{1}{r_c^{-1}+\epsilon})$, we have:
\begin{equation*}
	\left|(1-r^{-1})-\Omega(r)\right|\leq C (\kappa_++\kappa_c),
\end{equation*}
we obtain \eqref{eq:estubyexthom} after defining $\varepsilon_1(r):=\Omega^{-1}(r)\tilde{\varepsilon}_1(r)+\left(\Omega^{-1}(1-r^{-1})-1\right)u_{+,0}(r_*(r))$ and $\varepsilon_2(r):=\Omega^{-1}(r)\tilde{\varepsilon}_1(r)+\left(\Omega^{-1}(1-r^{-1})-1\right)u_{\infty,0}(r_*(r))$.
\end{proof}

\begin{proposition}
\label{prop:wronskestFsim}
	Let $(\omega,\ell)\in\mathcal{F}_{\flat,\sim }$. Assume that $\homega \tomega<0$. Let $\kappa_c\leq \kappa_+\leq \kappa_1$. For suitably small $\kappa_1>0$, there exists a constant $K=K(\gamma,L_0,\kappa_1)>0$ such that:
\begin{equation}
\label{eq:wronskextnearext}
|\mathfrak{W}|^{-1}\leq  2K.
\end{equation}
\end{proposition}
\begin{proof}
	Consider $u_+$ and let $\delta>0$. Then by using that $|\tomega|>\gamma$, we can apply
Proposition \ref{prop:regionIeleq1} to obtain for $\epsilon=\gamma \delta$, $\kappa_+,\kappa_c>0$ suitably small and $r-r_+\leq 2\epsilon$:
\begin{equation*}
u_+(r_*(r))=e^{-i\tomega r_*+i q \int_{r(0)}^{r}\rho_+(r')\Omega^{-2}(r')\,dr'}(1+\delta O ((r-r_+)^0)).
\end{equation*}
In this case, we can combine the above expression with Lemma \ref{lm:estubyexthom} and take $\kappa_++\kappa_c$ suitably small to conclude that $\alpha_2=0$ and $\alpha_1=1$.

Hence, we obtain for $ r_++\epsilon\leq r\leq \frac{1}{r_c^{-1}+\epsilon}$:
\begin{equation*}
u_+(r_*)=u_{+,0}(r_*(r))+\delta O ((r-r_+)^0)=\left[\alpha_{1,+,0} u_{\infty,0}(r_*(r))+\alpha_{2,+,0} \overline{u_{\infty,0}}(r_*(r))\right]+\delta O ((r-r_+)^0),
\end{equation*}
for some complex constants $\alpha_{1,+,0}$ and $\alpha_{2,+,0}$; see also \eqref{eq:basicexpode}. 
Now, we use that $|\homega|>\gamma$ to apply Proposition \ref{prop:regionIVeleq1} to obtain for $\epsilon=\gamma \delta$, $\kappa_+,\kappa_c>0$ suitably small and $r_++\epsilon\leq r\leq \frac{1}{r_c^{-1}+\epsilon}$:
\begin{align*}
	u_{\infty}(r_*(r))=&\:(1+\delta O (s_c^0))u_{\infty,0}(r_*(r)),\\
	\overline{u}_{\infty}(r_*(r))=&\:(1+\delta O (s_c^0))\overline{u}_{\infty,0}(r_*(r)).
\end{align*}
Hence, we can express globally:
\begin{equation*}
	u_+(r_*)=\alpha_{1,+,0}(1+O(\delta))u_{\infty}(r_*(r))+\alpha_{2,+,0}(1+O(\delta))\overline{u_{\infty}}(r_*(r)).
\end{equation*}

By the identity \eqref{eq:wronskcoeff1}, we can therefore conclude that $\mathfrak{W}=-2i\homega \alpha_{2,+,0}(1+O(\delta))$. Hence, $|\mathfrak{W}|^{-1}\leq K(1+ O(\delta))\leq 2K$, for suitably small $\delta>0$.

\end{proof}

If $\kappa_+>\kappa_1$, we cannot infer \eqref{eq:wronskextnearext} immediately. Instead, it follows from the additional assumption of the following condition:
\begin{condition}[Quantitative mode stability away from extremality]
\label{cond:quantmodestab}
Let $(\omega,\ell)\in\mathcal{F}_{\flat,\sim }$. Let $\kappa_+> \kappa_1$ and $\kappa_=0$, with $\kappa_1$ the constant from Proposition \ref{prop:wronskestFsim}. Then there exists a constant $\tilde{K}=\tilde{K}(\gamma,L_0,\kappa_1)>0$ such that:
\begin{equation}
\label{eq:wronskextawayext}
|\mathfrak{W}|^{-1}\leq  \tilde{K}.
\end{equation}
\end{condition}

We now conclude both an upper bound and a lower bound on the Wronskian $\mathfrak{W}$, valid for all bounded frequencies $(\omega,\ell)\in \mathcal{F}_{\flat}$.

\begin{corollary}
\label{cor:fullwronskest}
Let $(\omega,\ell)\in \mathcal{F}_{\flat}$ and assume that $|\q|<\q_{1}$. There exists a suitably small $\kappa_1>0$ and a constant $K=K(L_0,\gamma,\beta,\q_{1},\kappa_1)>0$, such that for all $\kappa_c\leq\kappa_+\leq \kappa_1$:
\begin{multline}
\label{eq:genwronskest}
	\frac{1}{K}(|\tomega|+\kappa_+)^{-\frac{1}{2}+\frac{1}{2} \re \beta_{\ell}}(|\homega|+\kappa_c)^{-\frac{1}{2}+\frac{1}{2} \re \beta_{\ell}}\frac{1}{1+\delta_{\beta_{\ell}0}\log((1+|\tomega| +\kappa_+)^{-1}(1+|\homega| +\kappa_c)^{-1})}\\
	\leq |\mathfrak{W}|^{-1}\leq K(|\tomega|+\kappa_+)^{-\frac{1}{2}+\frac{1}{2} \re \beta_{\ell}}(|\homega|+\kappa_c)^{-\frac{1}{2}+\frac{1}{2} \re \beta_{\ell}}\frac{1}{1+\delta_{\beta_{\ell}0}\log((1+|\tomega| +\kappa_+)^{-1}(1+|\homega| +\kappa_c)^{-1})}.
	\end{multline}
	If Condition \ref{cond:quantmodestab} holds, then the assumption $\kappa_+\leq \kappa_1$ can be dropped.
\end{corollary}
\begin{proof}
If $(\omega,\ell)\in\mathcal{F}_{\flat,\sim }$, then \eqref{eq:genwronskest} follows directly from Proposition \ref{prop:wronskestFsim}.

If $(\omega,\ell)\in\mathcal{F}_{\flat,+}$, then we write
\begin{equation*}
	u_{+}(r_*)={\alpha}_{1,+}u_{\infty}(r_*)+{\alpha}_{2,+}\overline{u}_{\infty}(r_*).
\end{equation*}
and we apply \eqref{eq:wronskcoeff1} to obtain
\begin{equation*}
	|\mathfrak{W}|^{-1}=\frac{1}{2|{\alpha}_{2,+}| |\tomega|}.
\end{equation*}
Using that $|\alpha_{2,+}|=|E_2|$ and $|\homega|\geq ||\q|-|\tomega||\geq |\q_0|-\gamma>0$ for suitably small $\gamma$, we can apply \eqref{eq:lowboundEhor} to conclude  \eqref{eq:genwronskest}.

If $(\omega,\ell)\in\mathcal{F}_{\flat,\infty}$, we write 
\begin{equation*}
	u_{\infty}(r_*)={\alpha}_{1,\infty}u_{+}(r_*)+{\alpha}_{2,\infty}\overline{u}_{+}(r_*).
\end{equation*}
and we have that:
\begin{equation*}
	|\mathfrak{W}|^{-1}=\frac{1}{2|{\alpha}_{2,\infty}| |\homega|}.
\end{equation*}
We then use that $|\alpha_{2,\infty}|=|A_2|$ and $|\homega|\geq |q_0|-\gamma>0$ and we apply \eqref{eq:lowboundAinf2} to conclude  \eqref{eq:genwronskest}.
\end{proof}

\subsection{Green's formula estimates: homogeneous part}
\label{sec:greenhm}
For suitably rapidly decaying $H(r_*)$ as $r_*\to \pm \infty$, we can apply Green's formula to express:
\begin{equation}
\label{eq:greenformula}
u(r_*)=\frac{u_{\infty}(r_*)}{\mathfrak{W}}\int_{-\infty}^{r_*}u_+(r_*')H(r_*')\,dr_*'+\frac{u_{+}(r_*)}{\mathfrak{W}}\int_{r_*}^{\infty}u_{\infty}(r_*')H(r_*')\,dr_*'.
\end{equation}
The validity of the above formula in an $L^2_{r,\omega}$ sense will be justified by first assuming smoothness and compact support in $r_*$ of $H$ to derive estimates and then applying a standard density argument.

In this section, we will derive $L^{\infty}$-estimates for the following products that relate to \eqref{eq:greenformula}:
\begin{align*}
|\mathfrak{W}|^{-1}|u_{+}|(r_*)|u_{\infty}|(r_*'),\\
|\mathfrak{W}|^{-1}|v_+'|(r_*)|u_{\infty}|(r_*'),\\
|\mathfrak{W}|^{-1}|u_+|(r_*)|v_{\infty}'|(r_*'),\\
|\mathfrak{W}|^{-1}|u_+|(r_*)|w_{\infty}'|(r_*'),\\
|\mathfrak{W}|^{-1}|w_+'|(r_*)|u_{\infty}|(r_*').
\end{align*}
with different values of $r_*,r_*'\in \R$, where we define in analogy with \eqref{eq:defv} and \eqref{eq:defw}:
\begin{align*}
v_{\square }=&\:e^{i \tomega r_*+i\q \int_{0}^{r_*} \rho_+(r_*')\,dr_*'}u_{\square},\\
 w_{\square}=&\:e^{-i\homega r_*+i \q \int_{0}^{r_*} \rho_c(r_*')\,dr_*'}u_{\square},
\end{align*}
where $\square\in\{+,\infty\}$.\newpage

\begin{proposition}
\label{prop:mainhom}
Let $(\omega,\ell)\in \mathcal{F}_{\flat}$.
\begin{enumerate}
\item Let $0\leq \rho_+\leq \rho_+(R)$ with $ 1<R<r_c$. Then, there exists a $\kappa_1>0$ and a constant $C=C(\q_0,\q_{1},L_0,\beta,\gamma,R)>0$, such that for all $\kappa_c\leq  \kappa_1$ and either $\kappa_1\leq \kappa_1$ or $\kappa_1>\kappa_1$ with the additional assumption of Condition \ref{cond:quantmodestab}:
\begin{align}
\label{eq:globestuplusext1}
|u_+(r_*(\rho_+))|\leq &\: C(1+(|\tomega|+\kappa_+)^{-1}\rho_+)^{-\frac{1}{2}+\frac{1}{2}\re \beta_{\ell}}(1+\delta_{\beta_{\ell}0}\log(1+(|\tomega|+\kappa_+)^{-1}\rho_+)),\\
\label{eq:globestdvplusext1}
\left|\frac{dv_+}{dr_*}(r_*(\rho_+))\right|\leq &\:C(|\tomega|+\kappa_+)^{-1}\Omega^2(1+(|\tomega|+\kappa_+)^{-1}\rho_+)^{-\frac{3}{2}+\frac{1}{2}\re \beta_{\ell}}\\ \nonumber
\times &(1+\delta_{\beta_{\ell}0}\log(1+(|\tomega|+\kappa_+)^{-1}\rho_+),\\
\label{eq:globestuinfext1}
|\mathfrak{W}|^{-1}|u_{\infty}(r_*(\rho_+))|&+(|\tomega|+\rho_+)^{-1}|\mathfrak{W}|^{-1}\left|\frac{dv_{\infty}}{dr_*}(r_*(\rho_+))\right|\\ \nonumber
\leq  &\: C |\tomega|^{-1}(1+|\tomega|^{-1}\rho_+ )^{-1}(1+(|\tomega|+\kappa_+)^{-1}\rho_+ )^{\frac{1}{2}-\frac{1}{2}\re \beta_{\ell}}.
\end{align}
\item Let $0\leq \rho_c\leq \rho_c(R)$ with $ 1<R<r_c$. Then there exists a constant \\$C=C(\q_0,\q_{1},L_0,\beta,\gamma,R)>0$, such that:
\begin{align}
\label{eq:globestuplusext2}
|u_{\infty}(r_*(\rho_c))|\leq &\: C(1+(|\homega|+\kappa_c)^{-1}\rho_c)^{-\frac{1}{2}+\frac{1}{2}\re \beta_{\ell}}(1+\delta_{\beta_{\ell}0}\log(1+(|\homega|+\kappa_c)^{-1}\rho_c)),\\
\label{eq:globestdvplusext2}
\left|\frac{dw_{\infty}}{dr_*}(r_*(\rho_c))\right|\leq &\:C(|\homega|+\kappa_c)^{-1}\Omega^2(1+(|\homega|+\kappa_c)^{-1}\rho_c)^{-\frac{3}{2}+\frac{1}{2}\re \beta_{\ell}}\\ \nonumber
\times &(1+\delta_{\beta_{\ell}0}\log(1+(|\homega|+\kappa_c)^{-1}\rho_c)),\\
\label{eq:globestuinfext2}
|\mathfrak{W}|^{-1}|u_{+}(r_*(\rho_c))|&+(|\homega|+\rho_c)^{-1}|\mathfrak{W}|^{-1}\left|\frac{dw_{+}}{dr_*}(r_*(\rho_c))\right|\\ \nonumber
\leq C |\homega|^{-1}&\:(1+|\homega|^{-1}\rho_c )^{-1}(1+(|\homega|+\kappa_c)^{-1}\rho_c )^{\frac{1}{2}-\frac{1}{2}\re \beta_{\ell}}.
\end{align}
\end{enumerate}
\end{proposition}
\begin{proof}
We consider separately Regions I, II, III and IV and the upper bound estimates on the coefficients $A_i,B_i,D_i$ that are derived in \S \ref{sec:matchasymptomega} and \S \ref{sec:matchasymphomega}. We combine these with the expressions of $u_+$ and $u_{\infty}$ in terms of special functions in \S \ref{sec:homest} and estimates on these special functions that follow from the expansions in Lemma \ref{lm:propwhitt} and Lemma \ref{prop:hypgeomfundsoln}.
\end{proof}

\subsection{Green's formula estimates: inhomogeneous part}
\label{sec:greeninhm}
In this section, we will always assume that $\kappa_+\geq \kappa_c$ and $\kappa_c\leq \kappa_1$, with $\kappa_1>0$ the constant appearing in Proposition \ref{prop:mainhom}. The goal is to prove the following weighted $L^2$-estimate for solutions $u$ to \eqref{eq:radialODE}:
\begin{theorem}
\label{thm:boundfreqest}
Let $\epsilon>0$ and $\max\{0,1-\re\beta_{\ell}\}<p<\min\{2,1+\re\beta_{\ell}\}+\frac{\epsilon}{2}$.  Then there exists a constant $C=C(q_0,q_{\rm max},\gamma, \beta, L_0,p,\epsilon)>0$, such that
\begin{align}
\label{eq:mainboundfreqesthor}
\sum_{\ell\in \N_0, \ell\leq L_0}&\sum_{m\in \Z, |m|\leq \ell}\int_{\R_{\omega}\cap \mathcal{F}_{\flat}}\int_1^{r_c}\rho_+^{1+\epsilon}\rho_c^{1+\epsilon}r^{-2}(r^{-1}\Omega)^{-p}|u|^2\,dr\\ \nonumber
\leq &\:C\int_{0}^{\infty}\left[\int_{\Sigma_{\tau}} (r^{-1}\Omega)^{-p}\left(\Omega^2|rF_{\xi,\widetilde{A}}|^2+\rho_+^{1-\epsilon}\xi^2|rG_{\widetilde{A}}|^2\right)\,d\sigma dr \right]\,d\tau\\ \nonumber
+&\:C\int_{0}^{\infty}\left[\int_{\Sigma_{\tau}} (r^{-1}\Omega)^{-p}\left(\Omega^2|rF_{\xi,\widetilde{A}}|^2+\rho_c^{1-\epsilon}\xi^2|rG_{\widetilde{A}}|^2\right)\,d\sigma dr \right]\,d\tau.
\end{align}
\end{theorem}
We introduce:
\begin{align*}
\widetilde{H}_+(r_*):=&\:e^{-i\q \int_{0}^{r_*} \rho_+(r_*')\,dr_*'}H(r_*),\\
\widetilde{H}_{\infty}(r_*):=&\:e^{-i\q \int_{0}^{r_*} \rho_c(r_*')\,dr_*'}H(r_*).
\end{align*}

We will first rewrite the Green's formula \eqref{eq:greenformula} in terms of $v_+$, $v_{\infty}$, $w_{+}$,$w_{\infty}$ and $\widetilde{H}$. 
\begin{lemma}
Let $u$ be a solution to \eqref{eq:radialODE} and let $p,p',R_*\in \R$ and assume that $H\in C_{c}^{\infty}(\R_*)$. Then we can express for $r_*\leq R_*$:
\begin{multline}
\label{eq:uintermsofv}
u(r_*)=\frac{u_{\infty}(r_*)((r^{-1}\Omega)^{\frac{p}{2}}v_+)(r_*)}{\mathfrak{W}}\int_{-\infty}^{R_*}e^{-i\tomega y}(r^{-1}\Omega)^{-\frac{p}{2}}(y)\widetilde{H}_+(y)dy\\
+\frac{u_{+}(r_*)(( r^{-1}\Omega)^{\frac{p'}{2}}w_{\infty})(r_*)}{\mathfrak{W}}\int_{R_*}^{\infty}e^{i\homega y}(r^{-1}\Omega)^{-\frac{p'}{2}}(y)\widetilde{H}_{\infty}(y)dy\\
+\int_{-\infty}^{R_*}f_+(r_*',r_*,\tomega)g_{+,1}(r_*',\tomega)\,dr_*'+\int_{R_*}^{\infty}f_+(r_*',r_*,\tomega)g_{+,2}(r_*',\tomega)\,dr_*',
\end{multline}
where for $r_*\leq R_*$:
\begin{align*}
f_+(r_*',r_*,\tomega)=&\:\begin{cases}-\frac{u_{\infty}(r_*)((r^{-1}\Omega)^{\frac{p}{2}}v_+)'(r_*')}{\mathfrak{W}}\quad &r_*'\leq r_*, \: r_*<R_*,\\
\frac{u_{+}(r_*)((r^{-1}\Omega)^{\frac{p}{2}}v_{\infty})'(r_*')}{\mathfrak{W}}\quad &r_*'> r_*,\: r_*<R_*\\
\frac{u_{+}(r_*)((r^{-1}\Omega)^{\frac{p'}{2}}w_{\infty})'(r_*')}{\mathfrak{W}}\quad &r_*\geq R_*
\end{cases},\\
g_{+,1}(r_*',\tomega)=&\:\int_{-\infty}^{r_*'}e^{-i\tomega y}(r^{-1}\Omega)^{-\frac{p}{2}}(y)\widetilde{H}_+(y)dy,\\
g_{+,2}(r_*',\tomega)=&\:\int_{r_*'}^{\infty}e^{i\homega y}(r^{-1}\Omega)^{-\frac{p'}{2}}(y)\widetilde{H}_{\infty}(y)dy.
\end{align*}
For $r_*\geq R_*$, we can express:
\begin{multline}
\label{eq:uintermsofw}
u(r_*)=\frac{u_{+}(r_*)((r^{-1}\Omega)^{\frac{p'}{2}}w_{\infty})(r_*)}{\mathfrak{W}}\int_{R_*}^{\infty}e^{i\homega y}(r^{-1}\Omega)^{-\frac{p'}{2}}(y)\widetilde{H}_{\infty}(y)dy\\
+\frac{u_{\infty}(r_*)((r^{-1}\Omega)^{\frac{p}{2}}v_+)(r_*)}{\mathfrak{W}}\int_{-\infty}^{R_*}e^{-i\tomega y}(r^{-1}\Omega)^{-\frac{p}{2}}\widetilde{H}_+(y)dy\\
+\int_{R_*}^{\infty}f_{\infty}(r_*',r_*,\homega)g_{\infty,1}(r_*',\homega)\,dr_*'+\int_{-\infty}^{R_*}f_{\infty}(r_*',r_*,\homega)g_{\infty,2}(\homega)\,dr_*',
\end{multline}
where for $r_*\leq R_*$:
\begin{align*}
f_{\infty}(r_*',r_*,\homega)=&\:\begin{cases}-\frac{u_{+}(r_*)((r^{-1}\Omega)^{\frac{p'}{2}}w_{\infty})'(r_*')}{\mathfrak{W}}\quad &r_*'\geq r_*,\, r_*> R_*,\\
\frac{u_{\infty}(r_*)((r^{-1}\Omega)^{\frac{p'}{2}}w_+)'(r_*')}{\mathfrak{W}}\quad &r_*'< r_*,\, r_*> R_*,\\
\frac{u_{\infty}(r_*)((r^{-1}\Omega)^{\frac{p}{2}}v_+)'(r_*')}{\mathfrak{W}}\quad &r_*\leq R_*
\end{cases},\\
g_{\infty,1}(r_*',\homega)=&\:\int_{r_*'}^{\infty}e^{i\homega y}(r^{-1}\Omega)^{-\frac{p'}{2}}(y)\widetilde{H}_{\infty}(y)dy,\\
g_{\infty,2}(r_*',\homega)=&\:\int_{-\infty}^{r_*'}e^{-i\tomega y}(r^{-1}\Omega)^{-\frac{p}{2}}(y)\widetilde{H}_+(y)dy.
\end{align*}
\end{lemma}
\begin{proof}
Let $p,p',R_*\in \R$. We introduce the notation: $v_{+,p}:=(r^{-1}\Omega)^{\frac{p}{2}}v_+$, $v_{\infty,p}:=(r^{-1}\Omega)^{\frac{p}{2}}v_{\infty}$, $w_{\infty,p'}=(r^{-1}\Omega)^{\frac{p'}{2}}w_{\infty}$, $\widetilde{H}_{+,p}=(r^{-1}\Omega)^{-\frac{p}{2}}\widetilde{H}_+$ and $\widetilde{H}_{\infty,p'}=(r^{-1}\Omega)^{-\frac{p'}{2}}\widetilde{H}_{\infty}$.  Then we can write for $r_*\leq R_*$:
\begin{multline*}
u(r_*)=\frac{u_{\infty}(r_*)}{\mathfrak{W}}\int_{-\infty}^{r_*}v_{+,p}(r_*')e^{-i\tomega r_*'}\widetilde{H}_{+,p}(r_*')\,dr_*'+\frac{u_{+}(r_*)}{\mathfrak{W}}\int_{r_*}^{R_*}v_{\infty,p}(r_*')e^{-i\tomega r_*'}\widetilde{H}_{+,p}(r_*')\,dr_*'\\
+\frac{u_{+}(r_*)}{\mathfrak{W}}\int_{R_*}^{\infty}w_{\infty,p'}(r_*')e^{i\homega r_*'}\widetilde{H}_{\infty,p'}(r_*')\,dr_*'\\
=\frac{u_{\infty}(r_*)}{\mathfrak{W}}\int_{-\infty}^{r_*}v_{+,p}(r_*')\frac{d}{dr_*'}\left[\int_{-\infty}^{r_*'}e^{-i\tomega y}\widetilde{H}_{+,p}(y)dy\right]\,dr_*'-\frac{u_{+}(r_*)}{\mathfrak{W}}\int_{r_*}^{R_*}v_{\infty,p}(r_*')\frac{d}{dr_*'}\left[\int_{r_*'}^{R_*}e^{-i\tomega y}\widetilde{H}_{+,p}(y)dy\right]\,dr_*'\\
-\frac{u_{+}(r_*)}{\mathfrak{W}}\int_{R_*}^{\infty}w_{\infty,p'}(r_*')\frac{d}{dr_*'}\left[\int_{r_*'}^{\infty}e^{i\homega y}\widetilde{H}_{\infty,p'}(y)dy\right]\,dr_*'\\
=\frac{u_{\infty}(r_*)v_{+,p}(r_*)}{\mathfrak{W}}\int_{-\infty}^{R_*}e^{-i\tomega y}\widetilde{H}_{+,p}(y)dy+\frac{u_{+}(r_*)w_{\infty,p'}(r_*)}{\mathfrak{W}}\int_{R_*}^{\infty}e^{i\homega y}\widetilde{H}_{\infty,p'}(y)dy\\
-\frac{u_{\infty}(r_*)}{\mathfrak{W}}\int_{-\infty}^{r_*}v_{+,p}'(r_*')\left[\int_{-\infty}^{r_*'}e^{-i\tomega y}\widetilde{H}_{+,p}(y)dy\right]\,dr_*'+\frac{u_{+}(r_*)}{\mathfrak{W}}\int_{r_*}^{R_*}v'_{\infty,p}(r_*')\left[\int_{r_*'}^{R_*}e^{-i\tomega y}\widetilde{H}_{+,p}(y)dy\right]\,dr_*'\\
+\frac{u_{+}(r_*)}{\mathfrak{W}}\int_{R_*}^{\infty}w'_{\infty,p'}(r_*')\left[\int_{r_*'}^{\infty}e^{i\homega y}\widetilde{H}_{\infty,p'}(y)dy\right]\,dr_*'\\
=\frac{u_{\infty}(r_*)v_{+,p}(r_*)}{\mathfrak{W}}\int_{-\infty}^{R_*}e^{-i\tomega y}\widetilde{H}_{+,p}(y)dy+\frac{u_{+}(r_*)w_{\infty,p'}(r_*)}{\mathfrak{W}}\int_{R_*}^{\infty}e^{i\homega y}\widetilde{H}_{\infty,p'}(y)dy\\
+\int_{-\infty}^{R_*}f_+(r_*',r_*,\tomega)g_{+,1}(r_*',\tomega)\,dr_*'+\int_{R_*}^{\infty}f_+(r_*',r_*,\tomega)g_{+,2}(r_*',\tomega)\,dr_*'.
\end{multline*}
This concludes \eqref{eq:uintermsofv}. To obtain \eqref{eq:uintermsofw}, we repeat the above argument with $v_{+}$ replaced by $w_{\infty}$, $v_{\infty}$ replaced by $w_{+}$, $w_{\infty}$ replaced by $v_+$ and the roles of $\homega$ and $-\tomega$ interchanged.
\end{proof}

In order to obtain suitable weighted $L^2$-estimates for $u$, satisfying the identities in \eqref{eq:uintermsofv} and \eqref{eq:uintermsofw}, we will need the following lemma:
\begin{lemma}
\label{lm:auxest}
Let $\epsilon\geq 0$, $\delta>0$, $p,p'\in \R$ and $R_*\in \R$, $\tilde{R}>1$. Then there exists a constant $C=C(\delta,R_*)>0$ such that for all suitably regular and decaying functions $f: (1,\infty)_r\times \R_y\times \R_{\omega}\to \C$ and $g: \R_y\times \R_{\omega}\to \C$:
\begin{align}
\label{eq.fgest1}
\int_{\R}&\:\int_1^{\tilde{R}}\rho_+^{1+\epsilon}(r)(r^{-1}\Omega)^{-p}(r)r^{-2}\left|\int_{-\infty}^{R_*} f(r,y,\omega)g(y,\omega)\,dy\right|^2\,drd\omega\\ \nonumber
\leq &\:C\left[ \sup_{\omega\in \R} \int_{1}^{\tilde{R}}\rho_+^{1+\epsilon}(r)(r^{-1}\Omega)^{-p}(r)r^{-2}\int_{-\infty}^{R_*} (\rho_+(y)+2\kappa_+)^{-1}\rho_+^{-\delta}(y)|f|^2(r,y,\omega)\,dydr\right]\\ \nonumber
 \times &\: \left[ \sup_{-\infty<y\leq r_*(\tilde{R})}\int_{\R} |g|^2(y,\omega)\,d\omega\right],\\
\label{eq.fgest2}
\int_{\R}&\:\int_1^{\tilde{R}}\rho_+^{1+\epsilon}(r)(r^{-1}\Omega)^{-p}(r)r^{-2}\left|\int_{R_*}^{\infty} f(r,y,\omega)g(y,\omega)\,dy\right|^2\,drd\omega\\ \nonumber
\leq &\:C\left[\sup_{\omega\in \R}\int_1^{\tilde{R}}\rho_+^{1+\epsilon}(r)(r^{-1}\Omega)^{2-p}(r)\int_{R_*}^{\infty}(\rho_c(y)+2\kappa_c)^{-1}\rho_c^{-\delta}(y) |f|^2(r,y,\omega)\,dy dr\right]\\ \nonumber
\times &\:  \left[\sup_{R_*\leq y<\infty} \int_{\R} |g|^2(y,\omega)\,d\omega \right].
\end{align}
We can also estimate:
\begin{align}
\label{eq.fgest3}
\int_{\R}&\:\int_{\tilde{R}}^{r_c}\rho_c^{1+\epsilon}(r)(r^{-1}\Omega)^{-p}(r)r^{-2}\left|\int_{R_*}^{\infty} f(r,y,\omega)g(y,\omega)\,dy\right|^2\,drd\omega\\ \nonumber
\leq&\: C\left[ \sup_{\omega\in \R} \int_{\tilde{R}}^{\infty}\rho_c^{1+\epsilon}(r)(r^{-1}\Omega)^{-p}(r)r^{-2}\int_{R_*}^{\infty} (\rho_c(y)+2\kappa_c)^{-1}\rho_c^{-\delta}(y)|f|^2(r,y,\omega)\,dydr\right] \\ \nonumber
\times &\: \left[ \sup_{r_*(\tilde{R})<y<\infty }\int_{\R} |g|^2(y,\omega)\,d\omega\right],\\
\label{eq.fgest4}
\int_{\R}&\:\int_{\tilde{R}}^{\infty}\rho_c^{1+\epsilon}(r)(r^{-1}\Omega)^{-p}(r)r^{-2}\left|\int_{-\infty}^{R_*} f(r,y,\omega)g(y,\omega)\,dy\right|^2\,drd\omega\\ \nonumber
\lesssim_{\delta} &\:\left[\sup_{\omega\in \R}\int_{\tilde{R}}^{\infty}\rho_c^{1+\epsilon}(r)(r^{-1}\Omega)^{-p}(r)r^{-2}\int_{-\infty}^{R_*}(\rho_+(y)+2\kappa_+)^{-1}\rho_+^{-\delta}(y) |f|^2(r,y,\omega)\,dy dr\right]\\\nonumber
\times &\: \left[\sup_{-\infty<y<R_* } \int_{\R} |g|^2(y,\omega)\,d\omega \right].
\end{align}
\end{lemma}
\begin{proof}
We first apply Cauchy--Schwarz in the $y$-direction to obtain:
\begin{multline*}
\int_{\R_{\omega}}\int_1^{\tilde{R}}\rho_+^{1+\epsilon}(r^{-1}\Omega)^{-p}r^{-2}\left|\int_{-\infty}^{R_*} f(r,y,\omega)g(y,\omega)\,dy\right|^2\,drd\omega\\
\leq \int_{\R_{\omega}}\left[\int_1^{\tilde{R}}\rho_+^{1+\epsilon}(r^{-1}\Omega)^{-p}r^{-2}\int_{-\infty}^{R_*} (\rho_+(y)+2\kappa_+)^{-1}\rho_+^{-\delta}(y)|f|^2(r,y,\omega)\,dydr\right] \\
\times \left[\int_{-\infty}^{R_*} (\rho_+(y)+2\kappa_+)\rho_+^{\delta}(y)|g|^2(y,\omega)\,dy\right]\,d\omega.
\end{multline*}
We now take estimate the RHS above further taking the supremum in $\omega$ of the terms in square brackets involving $|f|^2$ and interchanging the order of integration in the integrals of $|g|^2$:
\begin{multline*}
\int_{\R_{\omega}}\left[\int_1^{\tilde{R}}\rho_+^{1+\epsilon}(r^{-1}\Omega)^{-p}r^{-2}\int_{-\infty}^{R_*}(\rho_+(y)+2\kappa_+)^{-1}\rho_+^{-\delta}(y)|f|^2(r,y,\omega)\,dydr\right]\\
 \times \left[\int_{-\infty}^{R_*} (\rho_+(y)+2\kappa_+)\rho_+(y)^{\delta}|g|^2(y,\omega)\,dy\right]\,d\omega\\
\leq \sup_{\omega\in \R}\left[\int_1^{\tilde{R}}\rho_+^{1+\epsilon}(r^{-1}\Omega)^{-p}r^{-2}\int_{-\infty}^{R_*} (\rho_+(y)+2\kappa_+)^{-1}\rho_+^{-\delta}(y)|f|^2(r,y,\omega)\,dydr\right]\\
\times \int_{-\infty}^{R_*} \int_{\R_{\omega}}(\rho_+(y)+2\kappa_+)\rho_+(y)^{\delta}|g|^2(y,\omega)\,d\omega dy\\
=\sup_{\omega\in \R}\left[\int_1^{\tilde{R}}\rho_+^{1+\epsilon}(r^{-1}\Omega)^{-p}r^{-2}\int_{-\infty}^{R_*} (\rho_+(y)+2\kappa_+)^{-1}\rho_+^{-\delta}(y)|f|^2(r,y,\omega)\,dydr\right]\\
\times \int_{1}^{r(R_*)}\Omega^{-2}(r)(\rho_+(r)+2\kappa_+)\rho_+^{\delta}(r)\,dr'\cdot \sup_{-\infty<y<R_*}\int_{\R_{\omega}}|g|^2(y,\omega)\,d\omega.
\end{multline*}
We conclude that \eqref{eq.fgest1} holds by using that the integral of $\Omega^{-2}(r)(\rho_+(r)+2\kappa_+)\rho_+(r)^{\delta}$ is finite for $\delta>0$. The estimate \eqref{eq.fgest2} follows by repeating the above argument, but with $(\rho_+(y)+2\kappa_+)\rho_+(y)^{\delta}$ replaced by $(\rho_c(y)+2\kappa_c)\rho_c(y)^{\delta}$ and adjusting the integration domain of $y$.

We also immediately obtain \eqref{eq.fgest3} and \eqref{eq.fgest4} by interchanging the roles of $\rho_+$ and $\rho_c$ everywhere above.
\end{proof}

In the lemma below, we apply Plancherel's theorem to estimate integrals of the inhomogeneity $H$ in $\R_{\omega}$ in terms of integrals in $\tau$ of $F_{\xi, \widetilde{A}}$ and $G_{\widetilde{A}}$, which appear on the right-hand side of \eqref{eq:CSFtimecutoff}.
\begin{lemma}
\label{lm:estHbyF}
Let $p,p'\in [0,\infty)$ and $R_*\in \R$. Then, there exists a constant $C=C(p,p',R_*)>0$ such that:
\begin{multline}
\label{eq:mainestH}
\sup_{-\infty< r_*<R_*}\sum_{\ell\in \N_0}\sum_{m\in \Z, |m|\leq \ell}\int_{\R_{\omega}}\left|\int_{-\infty}^{r_*}(r^{-1}\Omega)^{-\frac{p}{2}}(y)e^{- i\tomega y}\widetilde{H}_{+}(y,\omega)\,dy\right|^2\,d\omega\\
\leq C\int_{0}^{\infty}\left[\int_{\Sigma_{\tau}\cap \{r\leq r(R_*)\}} (r^{-1}\Omega)^{-p}\h^{-2}\left[\Omega^2|rF_{\xi,\widetilde{A}}|^2+\rho_+^{1-\delta}\xi^2r^{-2}|r^3G_{\widetilde{A}}|^2\right]\,d\sigma dr \right]\,d\tau.
\end{multline}
Furthermore, 
\begin{multline}
\label{eq:mainestHinf}
\sup_{R_*\leq  r_*<\infty}\sum_{\ell\in \N_0}\sum_{m\in \Z}\int_{\R_{\omega}}\left|\int_{r_*}^{\infty}e^{+ i\homega y}(r^{-1}\Omega)^{-\frac{p'}{2}}\widetilde{H}_{\infty}(y)\,dy\right|^2\,d\omega\\
\leq C\int_{0}^{\infty}\left[\int_{\Sigma_{\tau}\cap \{r\geq r(R_*)\}} (r^{-1}\Omega)^{-p'}\widetilde{\h}^{-2}\left[\Omega^2|rF_{\xi,\widetilde{A}}|^2+\rho_c^{1-\delta}\xi^2r^{-2}|r^3G_{\widetilde{A}}|^2\right]\,d\sigma dr \right]\,d\tau.
\end{multline}
\end{lemma}
\begin{proof}
Let $\widetilde{H}_{F}$ denote the part of $\widetilde{H}_+$ coming from $F_{\xi,\widetilde{A}}$ and let  $\widetilde{H}_{G}$ denote the part of $\widetilde{H}_+$ coming from $G_{\widetilde{A}}$, so $\widetilde{H}_+=\widetilde{H}_{F}+\widetilde{H}_{G}$. We will first assume that $F_{\xi,\tilde{A}}$ and $\xi G_{\tilde{A}}$ are compactly supported in $r_*$. After establishing, \eqref{eq:mainestH} and \eqref{eq:mainestHinf}, the general case then follows from a standard density argument.

By applying an inverse Fourier transform, we obtain:
\begin{multline*}
\int_{\R_{\omega}}\left|\int_{-\infty}^{r_*}(r^{-1}\Omega)^{-\frac{p}{2}}(y)e^{- i\tomega y}\widetilde{H}_{F}(y,\omega)\,dy\right|^2\,d\omega\\
=\frac{1}{2\pi}\int_{\R_{\omega}}\left|\int_{\R_t}\int_{-\infty}^{r_*}e(r^{-1}\Omega)^{-\frac{p}{2}}(y)e^{- i\tomega y+i\omega t}r\Omega^2 (F_{\xi,\widetilde{A}})_{\ell m}(y,t)\,dy dt\right|^2\,d\omega.
\end{multline*}
Recall that $\tau(t,r_*)=t+r_*-\int_1^{r(r_*)}\widetilde{\h}(r')\,dr'$ and consider $u(t,r_*):=t-r_*$. Then
\begin{equation*}
d\tau\wedge du=\left(\widetilde{\h}\Omega^2-2\right)dt\wedge dr_*=-\h dt\wedge dr_*,
\end{equation*}
so
\begin{multline*}
\left|\int_{\R_t}\int_{-\infty}^{r_*}(r^{-1}\Omega)^{-\frac{p}{2}}(y)e^{- i\tomega y+i\omega t}r\Omega^2 (F_{\xi,\widetilde{A}})_{\ell m}(y,t)\,dy dt\right|^2\\
=\left|\int_{0}^{\infty}\int_{u(\tau,r_*)}^{\infty}(r^{-1}\Omega)^{-\frac{p}{2}}(\tau,u)e^{ i\q r_*(\tau,u)+i\omega u}\h^{-1}r\Omega^2 (F_{\xi,\widetilde{A}})_{\ell m}(\tau,u)\,dud\tau \right|^2\\
=\left|\int_{0}^{\infty}\int_{\R_u}e^{i\omega u}f(\tau,u)\,dud\tau \right|^2\leq \int_{0}^{\infty}(1+\tau)^{2}\left|\int_{\R_u}e^{-i\omega u}f(\tau,u)\,du\right|^2\,d\tau.
\end{multline*}
with
\begin{equation*}
f(\tau,u):=\mathbf{1}_{u\geq u(\tau,r_*)}\left[ (r^{-1}\Omega)^{-\frac{p}{2}}e^{ i\q r_*+i\omega u}\h^{-1}r\Omega^2 (F_{\xi,\widetilde{A}})_{\ell m}\right](\tau,u).
\end{equation*}
By Placherel's theorem, we have that
\begin{equation*}
\int_{\R_{\omega}}\left|\int_{\R_u}e^{i\omega u}f(\tau,u)\,du\right|^2\,d\omega=2\pi\int_{\R_u}|f|^2(\tau,u)\,du=2\pi\int_{1}^{r(r_*)}\Omega^{-2}|f|^2(\tau,u(r,\tau))\,dr.
\end{equation*}
Putting everything together, we obtain:
\begin{multline*}
\int_{\R_{\omega}}\left|\int_{-\infty}^{r_*}(r^{-1}\Omega)^{-\frac{p}{2}}(y)e^{- i\tomega y}\widetilde{H}_{+}(y,\omega)\,dy\right|^2\,d\omega\leq C\int_{0}^{\infty}\int_1^{r(r_*)}(1+\tau)^{2}\Omega^{-2}r|f|^2(\tau,u(r,\tau))\,drd\tau\\
\leq C \int_{0}^{\infty}\int_1^{r(r_*)}(1+\tau)^{2}(r^{-1}\Omega)^{-p} \Omega^2\h^{-2}\left|r(F_{\xi,\widetilde{A}})_{\ell m}\right|^2\,drd\tau.
\end{multline*}
We conclude \eqref{eq:mainestH} by using that $F_{\xi,\widetilde{A}}$ is supported in $\{0\leq \tau\leq 1\}$. We can repeat the above argument with a coordinate change to $(\tau,v)$, $v=t+r_*$ and with the roles of $\rho_+$ and $\rho_c$ interchanged, to derive \eqref{eq:mainestHinf}.

We estimate the contribution of $\widetilde{H}_G$ by simply using Cauchy--Schwarz and Plancherel's theorem together: for $r_*\leq R_*$:
\begin{multline*}
\int_{\R_{\omega}}\left|\int_{-\infty}^{r_*}(r^{-1}\Omega)^{-\frac{p}{2}}(y)e^{- i\tomega y}(\widetilde{H}_G)_{\ell m}(y,\omega)\,dy\right|^2\,d\omega\leq 
C\int_{1}^{r}\rho_+^{-1+\delta}\,r^{-2}dr\\
\times  \int_{\R_{\omega}} \int_{-\infty}^{R_*}(r^{-1}\Omega)^{-p}r^2\rho_+^{1-\delta}|(\widetilde{H}_G)_{\ell m}|^2\,dr_*d\omega\\
\leq C_{\delta,R_*}\int_0^{\infty}\int_1^{r(R_*)}(r^{-1}\Omega)^{-p}r^{-2}\rho_+^{1-\delta}|r^3G_{\xi,\widetilde{A}}|^2\,d\sigma d\rho_+ d\tau.
\end{multline*}
We conclude\eqref{eq:mainestH} by using that $d\rho_+=r^{-2}dr$. The estimates for \eqref{eq:mainestHinf} proceed analogously, with $\rho_c$ taking on the role of $\rho_+$.
\end{proof}
\begin{proof}[Proof of Theorem \ref{thm:boundfreqest}]
We apply the formula \eqref{eq:uintermsofv}, together with Lemmas \ref{lm:auxest} and \ref{lm:estHbyF} to obtain:
\begin{multline*}
\sum_{\ell\in \N_0, \ell\leq L_0}\sum_{m\in \Z, |m|\leq \ell}\int_{\R_{\omega}\cap \mathcal{F}_{\flat}}\int_1^{\tilde{R}}\rho_+^{1+\epsilon}(r^{-1}\Omega)^{-p}|u|^2\,r^{-2}dr\lesssim\\
 \left[\sup_{\omega\in \R, \ell\in \N_0, \ell\leq L_0, |m|\leq \ell }\sum_{i=1}^4 U_i\right]\Bigg[\int_{0}^{\infty}\int_{\Sigma_{\tau}} (r^{-1}\Omega)^{-p}\left[\Omega^{2}|rF_{\xi,\widetilde{A}}|^2+\rho_+^{1-\delta}r^{-2}\xi^2|r^3G_{\widetilde{A}}|^2\right]\,d\sigma dr \\
 +\int_{\Sigma_{\tau}} (r^{-1}\Omega)^{-p'}\left[\Omega^{2}|rF_{\xi,\widetilde{A}}|^2+\rho_c^{1-\delta}r^{-2}\xi^2|r^3G_{\widetilde{A}}|^2\right]\,d\sigma dr \,d\tau\Bigg],
\end{multline*}
where $\delta>0$ can be taken arbitrarily small and
\begin{align*}
U_1:=&\:\int_1^{\tilde{R}}\rho_+^{1+\epsilon}\frac{|u_{\infty}u_+|^2}{|\mathfrak{W}|^2}(r_*(r))(1+(r^{-2}\Omega^2)^{p'-p})\,r^{-2}dr,\\
U_2:=&\:  \int_{1}^{\tilde{R}}\rho_+^{1+\epsilon}(r^{-1}\Omega)^{-p}r^{-2}{|\mathfrak{W}|^{-2}|u_{\infty}|^2(r_*(r))\int_{-\infty}^{r_*(r)}(\rho_+(y)+2\kappa_+)^{-1}\rho_+^{-\delta}(y)|((r^{-1}\Omega)^{\frac{p}{2}}v_+)'(y)|^2}\,dydr,\\
U_3:=&\:  \int_{1}^{\tilde{R}}\rho_+^{1+\epsilon}(r^{-1}\Omega)^{-p}|u_{+}|^2(r_*(r))r^{-2}\int_{r_*(r)}^{R_*} (\rho_+(y)+2\kappa_+)^{-1}\rho_+^{-\delta}(y)|\mathfrak{W}|^{-2}|((r^{-1}\Omega)^{\frac{p}{2}}v_{\infty})'(y)|^2\,dydr,\\
U_4:=&\:  \int_{1}^{\tilde{R}}\rho_+^{1+\epsilon}(r^{-1}\Omega)^{-p}|\mathfrak{W}|^{-2} |u_{+}|^2(r_*(r))r^{-2}\int_{R_{*}}^{\infty} (\rho_c(y)+2\kappa_c)^{-1}\rho_c^{-\delta}(y)|((r^{-1}\Omega)^{\frac{p'}{2}}w_{\infty})'(y)|^2\,dydr.
\end{align*}

We will first estimate $U_1$. Using that $p<2+\epsilon$ and applying \eqref{eq:globestuplusext1} and \eqref{eq:globestuinfext1}, we obtain for $\epsilon>0$ and $p<p'+\epsilon$:
\begin{multline*}
U_1\lesssim \int_0^{\rho_+(\tilde{R})}\rho_+^{1+\epsilon}|\tomega|^{-2}(1+|\tomega|^{-1}\rho_+ )^{-2}(1+\log^2(1+(|\tomega|+\kappa_+)^{-1}\rho_+))(1+(r^{-2}\Omega^2)^{p'-p})\,d\rho_+\\
\lesssim \int_0^{|\tomega|}\tomega^{-2}\rho_+^{1+\epsilon}(1+\rho_+^{p'-p})\,dr+ \int_{|\tomega|}^{\rho_+(\tilde{R})}\rho_+^{-1+\epsilon}(1+\log^2 \rho_+)(1+\rho_+^{p'-p})\,dr\lesssim |\tomega|^{\epsilon}+|\tomega|^{\epsilon+p'-p}+1.
\end{multline*}

We now consider $U_2$. We first apply \eqref{eq:globestuplusext1} and \eqref{eq:globestdvplusext1} to obtain for $p>\delta$ and $\rho_+\leq |\tomega|+\kappa_+$:
\begin{multline*}
\int_{-\infty}^{r_*(r)}(\rho_+(y)+2\kappa_+)^{-1}\rho_+^{-\delta}(y)|((r^{-1}\Omega)^{\frac{p}{2}}v_+)'(y)|^2\,dy\\
\lesssim \int_{1}^{r}\rho_+^{1-\delta}(r')\left|\frac{d}{dr}\left((r^{-1}\Omega)^{\frac{p}{2}}v_+\right)\right|^2(r_*(r'))\,dr'\\
\lesssim (|\tomega|+\kappa_+)^{-2}\int_{0}^{\rho_+(r)}(\rho_++2\kappa_+)^{\frac{p}{2}}\rho_+^{\frac{p}{2}+1-\delta}(1+(|\tomega|+\kappa_+)^{-1}\rho_+)^{-3+\re \beta_{\ell}} (1+\log^2(1+(|\tomega|+\kappa_+)^{-1}\rho_+))\,d\rho_+\\
\lesssim (|\tomega|+\kappa_+)^{-2}\int_{0}^{\rho_+(r)}(\rho_++2\kappa_+)^{\frac{p}{2}}\rho_+^{\frac{p}{2}+1-\delta}\,d\rho_+\lesssim  (|\tomega|+\kappa_+)^{-2}(\rho_+(r)+\kappa_+)^{\frac{p}{2}}\rho_+^{\frac{p}{2}+2-\delta}.
\end{multline*}
For $\rho_+> |\tomega|+\kappa_+$ and $p>\delta$, we obtain instead:
\begin{multline*}
\int_{-\infty}^{r_*(r)}(\rho_+(y)+2\kappa_+)^{-1}\rho_+^{-\delta}(y)|((r^{-1}\Omega)^{\frac{p}{2}}v_+)'(y)|^2\,dy
\lesssim (|\tomega|+\kappa_+)^{-2}\int_{0}^{|\tomega|+\kappa_+}(\rho_++2\kappa_+)^{\frac{p}{2}}\rho_+^{\frac{p}{2}+1-\delta}\,d\rho_+\\
+(|\tomega|+\kappa_+)^{1-\re \beta_{\ell}}\int_{|\tomega|+\kappa_+}^{\rho_+(r)}(\rho_++2\kappa_+)^{\frac{p}{2}}\rho_+^{\frac{p}{2}-2+\re\beta_{\ell}-\delta}(1+\log^2 \rho_+)\,d\rho_+\\
\lesssim \begin{cases}
(|\tomega|+\kappa_+)^{p-\delta-\delta'}\quad p\leq 1+\delta-\re\beta_{\ell},\\
(|\tomega|+\kappa_+)^{1-\re \beta_{\ell}}\rho_+^{p-1+\re\beta_{\ell}-\delta-\delta'}(r)\quad p> 1+\delta-\re\beta_{\ell}.
\end{cases}
\end{multline*}
for some sufficiently small $\delta'>0$ (when $\re \beta_{\ell}\neq 0$, we can in fact set $\delta'=0$).

We can now estimate $U_2$ by applying \eqref{eq:globestuinfext1} to obtain for $\epsilon>\delta+\delta'$ and $\max\{\delta,1+\delta-\re\beta_{\ell}\}<p$:
\begin{multline*}
U_2\lesssim  \tomega^{-2}\int_{0}^{|\tomega|}\rho_+^{1+\epsilon}(r^{-1}\Omega)^{-p} (|\tomega|+\kappa_+)^{-2}(\rho_+(r)+\kappa_+)^{\frac{p}{2}}\rho_+^{\frac{p}{2}+2-\delta}\,d\rho_+\\
+ \int_{|\tomega|}^{\rho_+(\tilde{R})}\rho_+^{\frac{p}{2}-1+\epsilon-\delta-\delta'}(\rho_++\kappa_+)^{-\frac{p}{2}}\,d\rho_+\lesssim |\tomega|^{2+\epsilon-\delta}(|\tomega|+\kappa_+)^{-2}+\rho_+(\tilde{R})^{\epsilon-\delta-\delta'}\\
\lesssim |\tomega|^{\epsilon-\delta-\delta'}+1.
\end{multline*}

Suppose now that $p>\delta$ and $p< 1+\delta-\re\beta_{\ell}$. Then we can estimate instead:
\begin{multline*}
U_2\lesssim  \tomega^{-2}(|\tomega|+\kappa_+)^{p-\delta-\delta'-1+\re \beta_{\ell}}\int_{0}^{|\tomega|}(|\tomega|+\kappa_+)^{-2}\rho_+^{4+\epsilon+\delta'-p-\re \beta_{\ell}}\,d\rho_+\\
+ (|\tomega|+\kappa_+)^{p-\delta-\delta'-1+\re \beta_{\ell}}\int_{|\tomega|}^{\rho_+(\tilde{R})}\rho_+^{-\frac{p}{2}-\re \beta_{\ell}+\epsilon}(\rho_++\kappa_+)^{-\frac{p}{2}}\,d\rho_+
\lesssim 
|\tomega|^{\epsilon-\delta-\delta'}+1.
\end{multline*}
To obtain boundedness of $U_2$, we therefore need restrict to $\epsilon\geq \delta+\delta'$ and $p>\delta$.

Consider $U_3$. Note first that we can apply \eqref{eq:globestuinfext1} to obtain for $\rho_+(r)\geq |\tomega|+\kappa_+$:
\begin{multline*}
\int_{r_*(r)}^{R_*} (\rho_+(y)+2\kappa_+)^{-1}\rho_+^{-\delta}(y)|\mathfrak{W}|^{-2}|((r^{-1}\Omega)^{\frac{p}{2}}v_{\infty})'(y)|^2\,dy\\
\lesssim (|\tomega|+\kappa_+)^{\re\beta_{\ell}-1}\int_{\rho_+(r)}^{\rho_+(R)} \rho_+^{\frac{p}{2}-\re\beta_{\ell}-\delta}(\rho_++2\kappa_+)^{\frac{p}{2}-2} \,d\rho_+\\
\lesssim(|\tomega|+\kappa_+)^{\re\beta_{\ell}-1}\int_{\rho_+(r)}^{\rho_+(R)} \max\{(\rho_++\kappa_+)^{p-2-\re\beta_{\ell}-\delta},\rho_+^{p-2-\re\beta_{\ell}-\delta}\} \,d\rho_+\\
\lesssim \begin{cases}(|\tomega|+\kappa_+)^{\re\beta_{\ell}-1}(\rho_+(r)+\kappa_+)^{p-1-\re\beta_{\ell}-\delta} \quad p< 1+\delta+\re\beta_{\ell}\\
(|\tomega|+\kappa_+)^{\re\beta_{\ell}-1}\quad p> 1+\delta+\re\beta_{\ell},\\
(|\tomega|+\kappa_+)^{\re\beta_{\ell}-1}\log \rho_+^{-1}(r)\quad p= 1+\delta+\re\beta_{\ell}.
\end{cases}
\end{multline*}
We will first suppose that $p< 1+\delta+\re\beta_{\ell}$. Then we moreover obtain for $\rho_+(r)\leq |\tomega|+\kappa_+$ and $p<2+\delta$:
\begin{multline*}
	\int_{r_*(r)}^{R_*} (\rho_+(y)+2\kappa_+)^{-1}\rho_+^{-\delta}(y)|\mathfrak{W}|^{-2}|((r^{-1}\Omega)^{\frac{p}{2}}v_{\infty})'(y)|^2\,dy\\
=\int_{|\tomega|+\kappa_+}^{\rho_+(R)} (\rho_++2\kappa_+)^{-1}\rho_+^{-\delta}|\mathfrak{W}|^{-2}|((r^{-1}\Omega)^{\frac{p}{2}}v_{\infty})'(r_*(\rho_+))|^2r^2\Omega^{-2}\,d\rho_+\\
+\int_{\rho_+(r)}^{|\tomega|+\kappa_+}(\rho_++2\kappa_+)^{-2}\rho_+^{-1-\delta}|\mathfrak{W}|^{-2}|((r^{-1}\Omega)^{\frac{p}{2}}v_{\infty})'(r_*(\rho_+))|^2r^2\Omega^{-2}\,d\rho_+\\
\lesssim (|\tomega|+\kappa_+)^{p-2-\delta}+\int_{\rho_+(r)}^{|\tomega|+\kappa_+}(\rho_++2\kappa_+)^{\frac{p}{2}-2}\rho_+^{\frac{p}{2}-1-\delta}\,d\rho_+\\
\lesssim(|\tomega|+\kappa_+)^{p-2-\delta}+(\rho_+(r)+2\kappa_+)^{\frac{p}{2}-1}\rho_+(r)^{\frac{p}{2}-1-\delta}\leq (\rho_+(r)+2\kappa_+)^{\frac{p}{2}-1}\rho_+(r)^{\frac{p}{2}-1-\delta}.
\end{multline*}
If $p>1+\delta+\re \beta_{\ell}$ and $p<2+\delta$, we similarly obtain:
\begin{equation*}
	\int_{r_*(r)}^{R_*} (\rho_+(y)+2\kappa_+)^{-1}\rho_+^{-\delta}(y)|\mathfrak{W}|^{-2}|((r^{-1}\Omega)^{\frac{p}{2}}v_{\infty})'(y)|^2\,dy\lesssim (|\tomega|+\kappa_+)^{\re \beta_{\ell}-1}+(\rho_+(r)+2\kappa_+)^{\frac{p}{2}-1}\rho_+(r)^{\frac{p}{2}-1-\delta}
\end{equation*}

Hence, for $p< 2+\delta$ and $p<1+\delta+\re\beta_{\ell}$, we apply  \eqref{eq:globestuplusext1} to obtain for $\epsilon>\delta$:
\begin{multline*}
U_3\lesssim  \int_0^{|\tomega|+\kappa_+}\rho_+^{1+\epsilon}(\rho_++2\kappa_+)^{\frac{p}{2}-1}\rho_+^{\frac{p}{2}-1-\delta}(r^{-1}\Omega)^{-p}|u_{+}|^2(r_*(\rho_+))\,d\rho_+\\
+ \int_{|\tomega|+\kappa_+}^{\rho_+(\widetilde{R})}(|\tomega|+\kappa_+)^{\re\beta_{\ell}-1}\rho_+^{1+\epsilon-\frac{p}{2}}(\rho_+(r)+\kappa_+)^{\frac{p}{2}-1-\re\beta_{\ell}-\delta}|u_{+}|^2(r_*(\rho_+))\,d\rho_+\\
\lesssim  (|\tomega|+\kappa_+)^{\epsilon-\delta}+  \int_{|\tomega|+\kappa_+}^{\rho_+(\widetilde{R})}\rho_+^{\re \beta_{\ell}+\epsilon-\frac{p}{2}}(\rho_+(r)+\kappa_+)^{\frac{p}{2}-1-\re\beta_{\ell}-\delta}(1+\delta_{\beta_{\ell}0}\log^2\rho_+)\,d\rho_+\\
\lesssim (|\tomega|+\kappa_+)^{\epsilon-\delta-\delta'},
\end{multline*}
with $\delta'=0$ if $\beta_{\ell}\neq 0$.

For $p>1+\delta+\re\beta_{\ell}$ and $p<\min\{1+\epsilon+\re \beta_{\ell},2+\delta\}$, we obtain:
\begin{multline*}
U_3\lesssim  \int_0^{|\tomega|+\kappa_+}\rho_+^{1+\epsilon}(\rho_++2\kappa_+)^{\frac{p}{2}-1}\rho_+^{\frac{p}{2}-1-\delta}(r^{-1}\Omega)^{-p}|u_{+}|^2(r_*(\rho_+))\,d\rho_+\\
+ \int_{|\tomega|+\kappa_+}^{\rho_+(\widetilde{R})}(|\tomega|+\kappa_+)^{\re\beta_{\ell}-1}\rho_+^{1+\epsilon-\frac{p}{2}}(\rho_+(r)+\kappa_+)^{-\frac{p}{2}}|u_{+}|^2(r_*(\rho_+))\,d\rho_+\\
+(|\tomega|+\kappa_+)^{\re \beta_{\ell}-1}\int_0^{|\tomega|+\kappa_+}\rho_+^{1+\epsilon}(r^{-1}\Omega)^{-p}|u_{+}|^2(r_*(\rho_+))\,d\rho_+\\
\lesssim  (|\tomega|+\kappa_+)^{\epsilon-\delta}+  \int_{|\tomega|+\kappa_+}^{\rho_+(\widetilde{R})}\rho_+^{\re \beta_{\ell}+\epsilon-\frac{p}{2}}(\rho_+(r)+\kappa_+)^{-\frac{p}{2}}(1+\delta_{\beta_{\ell}0}\log^2\rho_+)\,d\rho_+\\
+(|\tomega|+\kappa_+)^{\re \beta_{\ell}-1}\int_0^{|\tomega|+\kappa_+}\rho_+^{1-\frac{p}{2}+\epsilon}(\rho_+(r)+\kappa_+)^{-\frac{p}{2}}\,d\rho_+\\
\lesssim (|\tomega|+\kappa_+)^{\epsilon-\delta}+(|\tomega|+\kappa_+)^{-p+1+\re\beta_{\ell}+\epsilon-\delta'},
\end{multline*}
with $\delta'=0$ if $\beta_{\ell}\neq 0$ and $\delta'>0$ and $p<1+\epsilon-\delta'+\re \beta_{\ell}$ if $\beta_{\ell}=0$.

We conclude that $U_3$ is uniformly bounded if $p<2+\delta$, $p<1+\epsilon+\re \beta_{\ell}-\delta'$, $\epsilon\geq \delta+\delta'$ and $\epsilon>\delta$. In other words, for $\epsilon>0$ and $p<\max\{1+\epsilon+\re \beta_{\ell},2+\epsilon\}$, there exist $\delta,\delta'>0$, with $\epsilon-\delta>0$ and $\delta'>0$ suitably small, such that $U_3$ is uniformly bounded.

Finally, we turn to $U_4$. We proceed as in the estimates for $U_2$ (in the case $\rho_+\geq |\tomega|+\kappa_+$), but with the roles of $\rho_+$ and $\rho_c$ and $|\tomega|$ and $|\homega|$ interchanged. By \eqref{eq:globestuplusext2} and \eqref{eq:globestdvplusext2} have that for $p'>\delta$:
\begin{multline*}
	\int_{R_{*}}^{\infty} (\rho_c(y)+2\kappa_c)^{-1}\rho_c^{-\delta}(y)|((r^{-1}\Omega)^{\frac{p'}{2}}w_{\infty})'(y)|^2\,dy\\
	\lesssim \begin{cases}(|\homega|+\kappa_c)^{p'-\delta-\delta'}\quad p'\leq 1-\re\beta_{\ell}+\delta,\\
	(|\homega|+\kappa_c)^{1-\re \beta_{\ell}}\rho_c^{p'-1+\re\beta_{\ell}-\delta-\delta'}(R)\lesssim (|\homega|+\kappa_c)^{1-\re \beta_{\ell}}\quad p'>1-\re\beta_{\ell}+\delta,
 \end{cases}
\end{multline*}
We will first suppose that  $p'>\delta$, $p'> 1+\delta'-\re\beta_{\ell}$, $p\leq 2+\epsilon$ and $p<1+\re \beta_{\ell}+\epsilon$ and apply \eqref{eq:globestuplusext1} and \eqref{eq:globestuinfext2} to obtain:
\begin{multline*}
	U_4\lesssim (|\homega|+\kappa_c)^{1-\re\beta_{\ell}}\int_0^{|\tomega|}\rho_+^{1+\epsilon}(r^{-1}\Omega)^{-p}|\mathfrak{W}|^{-2} |u_{+}|^2(r_*(\rho_+))\,d\rho_+\\
	+(|\homega|+\kappa_c)^{1-\re\beta_{\ell}}\int_{|\tomega|}^{\rho_+(\tilde{R})}\rho_+^{1+\epsilon}(r^{-1}\Omega)^{-p}|\mathfrak{W}|^{-2} |u_{+}|^2(r_*(\rho_+))\,d\rho_+\\
	\lesssim (|\tomega|+\kappa_+)^{\re\beta_{\ell}-1}\int_0^{|\tomega|}\rho_+^{1-\frac{p}{2}+\epsilon}(\rho_++2\kappa_+)^{-\frac{p}{2}}\,d\rho_++\int_{|\tomega|}^{\rho_+(\tilde{R})}\rho_+^{\re\beta_{\ell}-\frac{p}{2}+\epsilon}(\rho_++2\kappa_+)^{-\frac{p}{2}}\,d\rho_+\\
\lesssim (|\tomega|+\kappa_+)^{\re\beta_{\ell}-1}|\tomega|^{2+\epsilon-p}+\rho_+(\tilde{R})^{1+\epsilon+\re \beta_{\ell}-p}.
\end{multline*}
Now suppose that  $p'>\delta'$, $p'\leq 1+\delta'-\re \beta_{\ell}$, $p\leq 2+\epsilon$ and $p<1+\re \beta_{\ell}+\epsilon$. Then we can repeat the above computation with an additional factor $(|\homega|+\kappa_c)^{p'-1+\re\beta_{\ell}-\delta-\delta''}$ to obtain:
\begin{equation*}
	U_4\lesssim(|\homega|+\kappa_c)^{p'-1+\re\beta_{\ell}-\delta-\delta''}( (|\tomega|+\kappa_+)^{\re\beta_{\ell}-1}|\tomega|^{2+\epsilon-p}+\rho_+(\tilde{R})^{1+\epsilon+\re \beta_{\ell}-p})
\end{equation*}
Uniform boundedness of $U_4$ therefore only follows if $p'>\delta'$, $p'>1+\delta'-\re\beta_{\ell}$, $p\leq 2+\epsilon$ and $p<1+\re \beta_{\ell}+\epsilon$, with $\delta'>0$ arbitrarily small.

Combining the above estimates on $U_1$--$U_4$, we obtain the following restrictions on $p,p'$: $\delta<p<\delta+\min\{2,1+\re\beta_{\ell}\}$, $p'>\max\{0,1-\re \beta_{\ell}+\tilde{\delta},p-\epsilon\}$, where $\epsilon-\delta>0$ can be made arbitrarily small.  We take $\delta=\frac{\epsilon}{2}$ and restrict:
\begin{equation*}
\min\{0,1-\re\beta_{\ell}\}<p<\min\{2,1+\re\beta_{\ell}\}+\frac{\epsilon}{2}.
\end{equation*}
Then take $p'=p$ to guarantee $p'>\max\{1-\re \beta_{\ell},p-\epsilon\}$ and conclude \eqref{eq:mainboundfreqesthor}.

The estimate \eqref{eq:mainestHinf} follows in an entirely analogous manner, by interchanging the roles of $\rho_+$ and $\rho_c$ and $\homega$ and $-\tomega$.
\end{proof}

\section{Integrated energy estimates in frequency space: combining all frequency regimes}
\label{sec:iedcomb}
We can split $u=u_{\ell\leq L_0}+u_{\ell>L_0}$ and observe that $u_{\ell\leq L_0}$ and $u_{\ell>L_0}$ each independently satisfy \eqref{eq:radialODE} with $H_{m \ell}$ replaced by $\mathbf{1}_{\ell\leq L_0}H_{m \ell}$ and $\mathbf{1}_{\ell> L_0}H_{m \ell}$, respectively. This is not necessary, but highlights features that are particular to large angular frequencies or bounded angular frequencies.

In the remainder of $\S \ref{sec:iedcomb}$, we will fix $r_+=1$. We will moreover assume that $u$ arises from a sufficiently integrable $\uppsi$.
\begin{proposition}
\label{prop:iedhighlfreq}
Let $E\geq 2$, $\delta>0$, $1\leq p<2$ and $\ell>L_0$ and write $u=u_{\ell> L_0}$. Then there exist constants $C,c=C(r_1,r_2,r_3,R_1,R_2,R_3,L_0,\q_{1},\alpha,\beta,\delta), c(r_1,r_2,r_3,R_1,R_2,R_3,L_0,\q_{1},\alpha,\beta,\delta)>0$, such that:
\begin{multline}
  \label{eq:highangfreq}
 c \int_{(r_2)_*}^{(R_2)_*} |u'|^2+\left[(\mathbf{1}_{\mathcal{F}_{\sharp,{\rm angular}}\cup \mathcal{F}_{\sharp, {\rm time}}}+(1-r_{\rm max}r^{-1})^2)(\homega^2+\ell(\ell+1))+1\right]|u|^2\,dr_*\\
+c  \int_{-\infty}^{r_*(r_1)} \Big[(\rho_++\kappa_+)(r^{-1}\Omega)^{-p} |v'|^2+(1+\ell(\ell+1)+\homega^2)(\rho_++\kappa_+)(r^{-1}\Omega)^{2-p} |v|^2\\
+(\rho_++\kappa_+)(r^{-1}\Omega)^{\delta}(|u'|^2+\tomega^2|u|^2)\Big]\,dr_*\\
+c  \int_{r_*(R_1)}^{\infty} \Big[(\rho_c+\kappa_c)(r^{-1}\Omega)^{-p} |w'|^2+(1+\ell(\ell+1)+\homega^2)(\rho_c+\kappa_c)(r^{-1}\Omega)^{2-p} |w|^2\\
+(\rho_c+\kappa_c)(r^{-1}\Omega)^{\delta}(|u'|^2+\homega^2|u|^2)\Big]\,dr_*
\leq +\int_{-\infty}^{(r_3)*}\re(2u' \overline{H})\,dr_*-\int_{(R_3)_*}^{\infty}\re(2u' \overline{H})\,dr_*\\
- \int_{\R}\re(2 \chi_{r_1} (r^{-1}\Omega)^{-p} v' e^{-i \tomega r_*-i\q\int_{0}^{r_*}\rho_+(r_*')\,dr_*'}\overline{H}) - \int_{\R}\re(2 \chi_{R_1} (r^{-1}\Omega)^{-p}w'  e^{i \homega r_*-i \q \int_{0}^{r_*}\rho_c(r_*')\,dr_*'}\overline{H})\,dr_*\\
+E\int_{\R}\tomega \chi_{r_1} \re(i u \overline{H})\,dr_*+E\int_{\R}\homega (1-\chi_{r_1})\re(i u \overline{H})\,dr_*+ C\int_{(r_3)_*}^{(R_3)_*}(|u|+|u'|)\cdot |H|'\,dr_*+\int_{(r_2)_*}^{(R_2)_*}|H|^2\,dr_*\\
+C\int_{-\infty}^{(r_3)*}(r^{-1}\Omega)^{\delta}|u'| |H|\,dr_*+C\int_{(R_3)*}^{\infty}(r^{-1}\Omega)^{\delta}|u'| |H|\,dr_*.
 \end{multline}
\end{proposition}
\begin{proof}
The proposition concerns the frequency ranges $\mathcal{F}_{\sharp,{\rm angular}}$, $\mathcal{F}_{\sharp,{\rm trap}}$ and $\mathcal{F}_{\sharp, {\rm time}}$. We combine Propositions \ref{prop:srhighfreqest}, \ref{prop:trappedfreqest} and Corollary \ref{cor:highfreqomegawrp} and use that the leading-order behaviours of $y$ and $f$ agree when $r\downarrow 1$ and $r\uparrow r_c$ and that $|y+1|, |f+1|\leq C (r^{-1}\Omega)^{\delta}$ near $r=r_+$ and $|y-1|,|f-1|\leq C (r^{-1}\Omega)^{\delta}$ near $r=r_c$.

Since the cut-off functions appearing in front of $j^T[u]$ and $j^K[u]$ are not identical in each frequency regime, but they do agree close to $r=r_+$ and $r=r_c$, we moreover apply Young's inequality to obtain:
\begin{equation*}
\int_{(r_2)_*}^{(R_2)_*}\mathbf{1}_{\mathcal{F}_{\sharp,{\rm angular}}\cup \mathcal{F}_{\sharp, {\rm time}}}|\homega| |u|\cdot |H|\,dr_*\leq\int_{(r_2)_*}^{(R_2)_*}\epsilon^2\mathbf{1}_{\mathcal{F}_{\sharp,{\rm angular}}\cup \mathcal{F}_{\sharp, {\rm time}}}\homega^2 |u|^2+\frac{1}{4}\epsilon^{-2} |H|^2\,dr_*,
\end{equation*}
with $\mathbf{1}$ the indicator function, and we absorb the $\homega^2|u|^2$ term to the left-hand side, using that there is no degenerate factor $(1-r_{\rm max}r^{-1})^2$ when $(\omega,\ell)\in \mathcal{F}_{\sharp,{\rm angular}}\cup \mathcal{F}_{\sharp, {\rm time}}$.

In the cases $\mathcal{F}_{\sharp,{\rm angular}}$ and $\mathcal{F}_{\sharp,{\rm trap}}$, we need an additional step to control $(\rho_++\kappa_+)(r^{-1}\Omega)^{\delta}\tomega^2|u|^2$ and $(\rho_c+\kappa_c)(r^{-1}\Omega)^{\delta}\homega^2|u|^2$. Note first that:
\begin{equation*}
	\tomega^2|u|^2\leq 2|u'+i\tomega u|^2+|u'|^2\leq 2|v'|^2+ C\q^2\rho_+^2|u|^2+|u'|^2.
\end{equation*}
Hence,
\begin{multline*}
	(\rho_++\kappa_+)(r^{-1}\Omega)^{\delta}\tomega^2|u|^2\leq (r^{-1}\Omega)^{p-\delta}(\rho_++\kappa_+)(r^{-1}\Omega)^{-p}|v'|^2+2(\rho_++\kappa_+)(r^{-1}\Omega)^{\delta}|u'|^2\\
	+C\q^2(\rho_++\kappa_+)(r^{-1}\Omega)^{p-\delta} (r^{-1}\Omega)^{2-p}|u|^2.
\end{multline*}
Similarly,
\begin{multline*}
	(\rho_c+\kappa_c)(r^{-1}\Omega)^{\delta}\homega^2|u|^2\leq (r^{1}\Omega)^{p-\delta}(\rho_c+\kappa_c)(r^{-1}\Omega)^{-p}|w'|^2+2(\rho_c+\kappa_c)(r^{-1}\Omega)^{\delta}|u'|^2\\
	+C\q^2(\rho_c+\kappa_c)(r^{-1}\Omega)^{p-\delta} (r^{-1}\Omega)^{2-p}|u|^2.
\end{multline*}

Since the factor $(r^{-1}\Omega)^{p-\delta}$ can be made arbitrarily small for $r_1-1$ and $R_2^{-1}-r_c^{-1}$ suitably small, we can absorb the right-hand sides of the above inequalities into the left-hand sides of the estimates in Proposition  \ref{prop:trappedfreqest} and Corollary \ref{cor:highfreqomegawrp}.
\end{proof}

\begin{proposition}
\label{prop:iedlowlfreq}
Let $|\q|\geq \q_0>0$, $E\geq 2$, $\delta>0$ and $\ell\leq L_0$ and write $u=u_{\ell\leq L_0}$. Let $\epsilon>0$ and $\max\{0,1-\re\beta_{\ell}\}<p<\min\{2,1+\re\beta_{\ell}\}+\epsilon$. For suitably small $\kappa_1>0$ there exist constants\\ $C,c=C(r_1,r_2,r_3,R_1,R_2,R_3,L_0,\q_0,\q_{1},p,\kappa_1), c(r_1,r_2,r_3,R_1,R_2,R_3,L_0,\q_0,\q_{1},p,\kappa_1)>0$, such that for all $\kappa_c\leq \kappa_1$ and $\kappa_+\geq \kappa_c$ with either $\kappa_+\leq \kappa_1$ or $\kappa_+>\kappa_1$ with the additional assumption of Condition \ref{cond:quantmodestab}:
\begin{multline}
  \label{eq:boundangfreq}
 c \int_{\R_{\omega}\cap \mathcal{F}_{\sharp,{\rm time}}\cap \mathcal{F}_{\flat}}\int_{(r_2)_*}^{(R_2)_*} |u'|^2+(\homega^2+\ell(\ell+1)+1)|u|^2\,dr_*d\omega\\
+c   \int_{\R_{\omega}\cap \mathcal{F}_{\sharp,{\rm time}}\cap \mathcal{F}_{\flat}}\int_{-\infty}^{r_*(r_1)} \left[\rho_+^{2\epsilon}(\rho_++\kappa_+)(r^{-1}\Omega)^{-p} |v'|^2+(1+\ell(\ell+1))\rho_+^{2\epsilon}(\rho_++\kappa_+)(r^{-1}\Omega)^{2-p}\right] |v|^2\\
+\rho_+^{2\epsilon}(\rho_++\kappa_+)(r^{-1}\Omega)^{\delta}(|u'|^2+\tomega^2|u|^2)\Big]\,dr_* d\omega \\
+c   \int_{\R_{\omega}\cap \mathcal{F}_{\sharp,{\rm time}}\cap \mathcal{F}_{\flat}} \int_{r_*(R_1)}^{\infty} \Big[\rho_c^{2\epsilon}(\rho_c+\kappa_c)(r^{-1}\Omega)^{-p} |w'|^2+(1+\ell(\ell+1)+\homega^2)\rho_c^{2\epsilon}(\rho_c+\kappa_c)(r^{-1}\Omega)^{2-p} |w|^2\\
+\rho_c^{2\epsilon}(\rho_c+\kappa_c)(r^{-1}\Omega)^{\delta}(\homega^2|u|^2+|u'|^2)\Big]\,dr_* d\omega\\
\leq  \int_{\R_{\omega}\cap \mathcal{F}_{\sharp,{\rm time}}\cap \mathcal{F}_{\flat}}\int_{-\infty}^{(r_3)*}\re(2u' \overline{H})\,dr_*d\omega- \int_{\R_{\omega}\cap \mathcal{F}_{\sharp,{\rm time}}\cap \mathcal{F}_{\flat}}\int_{(R_3)_*}^{\infty}\re(2u' \overline{H})\,dr_*d\omega\\
-  \int_{\R_{\omega}\cap \mathcal{F}_{\sharp,{\rm time}}\cap \mathcal{F}_{\flat}}\int_{\R}\re(2 \chi_{r_1} (r^{-1}\Omega)^{-p} v' e^{-i \tomega r_*-i\q\int_{0}^{r_*}\rho_+(r_*')\,dr_*'}\overline{H}) '\,dr_*d\omega\\
-  \int_{\R_{\omega}\cap \mathcal{F}_{\sharp,{\rm time}}\cap \mathcal{F}_{\flat}}\int_{\R}\re(2 \chi_{R_1} (\Omega r^{-1})^pw'  e^{i \homega r_*-i \q \int_{0}^{r_*}\rho_c(r_*')\,dr_*'}\overline{H})\,dr_*d\omega\\
+E \int_{\R_{\omega}\cap \mathcal{F}_{\sharp,{\rm time}}\cap \mathcal{F}_{\flat}}\int_{\R}\tomega \chi_{K} \re(i u \overline{H})\,dr_*d\omega+E \int_{\R_{\omega}\cap \mathcal{F}_{\sharp,{\rm time}}\cap \mathcal{F}_{\flat}}\int_{\R}\homega \chi_{T}\re(i u\overline{H})\,dr_*d\omega\\
+ C \int_{\R_{\omega}\cap \mathcal{F}_{\sharp,{\rm time}}\cap \mathcal{F}_{\flat}}\int_{(r_3)_*}^{(R_3)_*}(|u|+|u'|)\cdot |H|\,dr_*d\omega +C \int_{\R_{\omega}\cap \mathcal{F}_{\sharp,{\rm time}}\cap \mathcal{F}_{\flat}}\int_{-\infty}^{(r_3)*}(r^{-1}\Omega)^{\delta} |u'| |H|\,dr_*d\omega\\
+C \int_{\R_{\omega}\cap \mathcal{F}_{\sharp,{\rm time}}\cap \mathcal{F}_{\flat}}\int_{(R_3)*}^{\infty}(r^{-1}\Omega)^{\delta}|u'| |H|\,dr_*d\omega\\
+C\int_{0}^{\infty}\left[\int_{\Sigma_{\tau}} (r^{-1}\Omega)^{-p}\left(\Omega^2|rF_{\xi,\widetilde{A}}|^2+\rho_+^{1-2\epsilon}\rho_c^{1-2\epsilon}\xi^2r^{-2}|r^3G_{\widetilde{A}}|^2\right)\,d\sigma dr \right]\,d\tau.
 \end{multline}
\end{proposition}
\begin{proof}
The proposition concerns the frequency ranges $\mathcal{F}_{\sharp,{\rm time}}$ and $\mathcal{F}_{\flat}$. In the case of $\mathcal{F}_{\sharp,{\rm time}}$, the estimate \eqref{eq:boundangfreq} follows directly from Corollary \ref{cor:highfreqomegawrp}.

We will consider the microlocal energy current $j_2^y[u]$ with $y(r)=(1-r_+r_c^{-1})^{-2+s}(\rho_+(r)^{2-s}-\rho_c(r)^{2-s})$.  Note that
\begin{equation*}
y'(r)=(2-s)r^{-2}\Omega^{2}(1-r_+r_c^{-1})^{-2+s}\left(\rho_c^{1-s}(r)+\rho_+^{1-s}(r)\right).
\end{equation*}

Integrating $(j_2^y[u])'$ then gives:
\begin{equation*}
\int_{\R}y'(|u'|^2+(\homega^2-V)|u|^2)-yV'|u|^2\,dr_*=2(\tomega^2|u|^2(-\infty)+\homega^2|u|^2(\infty))-\int_{\R}y\re(u' \overline{H})\,dr_*.
\end{equation*}

There exist constants $c,C>0$ such that for $s\geq 1$:
\begin{multline*}
y'(|u'|^2+(\homega^2-V)|u|^2)-yV'|u|^2\\
\geq c (2-s)r^{-2}\Omega^{2}\left(\rho_c^{1-s}+\rho_+^{1-s}\right)(|u'|^2+\homega^2\mathbf{1}_{r_*\geq 0}+\tomega^2\mathbf{1}_{r_*\leq 0}) |u|^2)\\
-Cr^{-2}\Omega^{2}(\rho_++\kappa_+)(\rho_c+\kappa_c)|u|^2.
\end{multline*}

Now let $s=2-\delta$. Since moreover $\ell\leq L_0$, we can apply Theorem \ref{thm:boundfreqest} where $\epsilon$ is replaced by $2\epsilon$, to obtain for $\max\{0,1-\re\beta_{\ell}\}+\frac{\epsilon}{2}<p<\min\{2,1+\re\beta_{\ell}\}+\frac{3\epsilon}{2}$:
\begin{multline*}
c\int_{\R} r^{-2}\Omega^{2}\left(\rho_c^{-1+\delta}+\rho_+^{-1+\delta}\right)(|u'|^2+\homega^2\mathbf{1}_{r_*\geq 0}+\tomega^2\mathbf{1}_{r_*\leq 0}) |u|^2)=-\int_{\R}y\re(u' \overline{H})\,dr_*\\
+C\int_{0}^{\infty}\left[\int_{\Sigma_{\tau}} (r^{-1}\Omega)^{-p-\epsilon}\Omega^2\left(|rF_{\xi,\widetilde{A}}|^2+\rho_+^{1-\epsilon}\rho_c^{1-\epsilon}(r^{-1}\Omega)^{-2}\xi^2|G_{\widetilde{A}}|^2\right)\,d\sigma dr \right]\,d\tau,
\end{multline*}
with $\mathbf{1}$ the indicator function.

The above estimate can then be combined with \eqref{eq:rpestwithbadterm}, \eqref{eq:rmin1pestwithbadterm} and Theorem \ref{thm:boundfreqest} (with $\epsilon$ replaced by $2\epsilon$) to control integrals of $|u|^2$, to obtain \eqref{eq:boundangfreq} with additionally the boundary term $2(\tomega^2|u|^2(-\infty)+\homega^2|u|^2(\infty))$ on the right-hand side. We then control the boundary terms by integrating $E (\chi_{T}j^T[u])'$ and  $E (\chi_{K}j^K[u])'$ and apply Theorem \ref{thm:boundfreqest} to absorb the non-trivial terms with a factor $\chi'_T$ and $\chi_K'$ and conclude \eqref{eq:boundangfreq}.
\end{proof}
\section{Integrated energy estimates: physical space analysis}
\label{sec:iedphy}
We will now prove Theorem \ref{thm:main}. The outline of this section is as follows:
\begin{itemize}
	\item We will first establish \eqref{eq:mainied} for future-integrable $\psi$ in \S \ref{sec:iedfutuint}.
	\item Then, we will drop the assumption of future integrability and \emph{derive} future-integrability of $\psi$ in the case $\kappa_+>0$ and $\kappa_c>0$ in \S \ref{sec:verifyfutureint} to conclude that \eqref{eq:mainied} holds for $\kappa_+>0$ and $\kappa_c>0$. 
\item Finally, we will show that \eqref{eq:mainied} also holds for $\kappa_+=0$ or $\kappa_c=0$ by making use of the fact that the constant $C$ in \eqref{eq:mainied} is uniform in $\kappa_+$ and $\kappa_c$. We then conclude the proof of Theorem \ref{thm:main}.
\end{itemize}

\subsection{Integrated energy estimates for future-integrable solutions}
\label{sec:iedfutuint}
We first apply Plancherel's theorem in combination with Propositions \ref{prop:iedhighlfreq} and \ref{prop:iedlowlfreq}.
\begin{proposition}
\label{prop:iedphysloword}
Let $\ell_0,L_0\in \N_0$ and assume that $\psi$ is a future-integrable solution to \eqref{eq:CSF} with $A=\widehat{A}=-Qr^{-1}d\tau$ that is supported on angular frequencies $\ell\geq \ell_0$. Let $\kappa_+\geq \kappa_c$ and $ \kappa_c\leq \kappa_1$, with $\kappa_1>0$. Assume moreover that $|\q|\geq \q_0>0$. Let $1<r_H<r_I<r_c$ and consider the constants $\delta, \epsilon>0$. Let $\max\{0,1-\re\beta_{\ell_0}\}<p<\min\{2,1+\re\beta_{\ell_0}\}+\epsilon$ if $\psi=\psi_{\leq L_0}$ and $1<p<\min\{2,1+\re\beta_{\ell_0}\}+\epsilon$ if $\psi=\psi_{> L_0}$.

Then, for $\kappa_1$ suitably small, there exist a constant $C=C(\epsilon,p, r_H,r_I,\h,\delta,\kappa_1)>0$, such that:
\begin{multline}
\label{eq:iedphysicalspace}
\int_{0}^{\infty}\int_{\Sigma_{\tau}\cap\{r_H\leq r\leq r_I\}} |Y_*\psi|^2+|\psi|^2+\upzeta(r)(|T\psi|^2+|\snabla_{\s^2}\psi|^2)\,d\sigma dr d\tau\\
+\int_{0}^{\infty}\int_{\Sigma_{\tau}\setminus\{r_H\leq r\leq r_I\}} \rho_+^{2\epsilon}\rho_c^{2\epsilon}(\rho_++\kappa_+)(\rho_c+\kappa_c)(r^{-1}\Omega)^{-p}\Omega^2|X\psi|^2\\
+(\rho_++\kappa_+)(\rho_c+\kappa_c)(r^{-1}\Omega)^{\delta}\Omega^{-2}(\mathbf{1}_{r\leq r_H}|K_+\psi|^2+\mathbf{1}_{r\geq r_I}|K_c\psi|^2)\\
+\rho_+^{2\epsilon}\rho_c^{2\epsilon}(\rho_++\kappa_+)(\rho_c+\kappa_c)(r^{-1}\Omega)^{2-p}\Omega^{-2}(|\psi|^2+|\snabla_{\s^2}\psi|^2)\,d\sigma drd\tau\\
\leq C\Bigg[\int_{\Sigma_0} \mathcal{E}_{p}[\psi]\,d\sigma dr\\
+\int_{0}^{\infty}\left[\int_{\Sigma_{\tau}} \max\{(r^{-1}\Omega)^{-p}\rho_+^{1-2\epsilon}\rho_c^{1-2\epsilon},\mathbf{1}_{0\leq \tau\leq 1}\}r^{-2}|r^3G_{\widehat{A}}|^2+\Omega^2(1-\upzeta)|T(rG_{\widehat{A}})|^2\,d\sigma dr \right]\,d\tau\Bigg].
\end{multline}
If $\psi=\psi_{\leq L_0}$, with $L_0\in \N_0$, then we can omit the factor $\upzeta$ in \eqref{eq:iedphysicalspace} at the expense of adding $L_0$-dependence in the constant $C$. If $\psi=\psi_{>L_0}$ for suitably large $L_0$, we can set $\epsilon=0$.
\end{proposition}
\begin{proof}
Recall from \eqref{eq:uppsihattilde1}--\eqref{eq:Fhattilde2} that we can express
\begin{align*}
\widetilde{\uppsi}=&\:e^{i\q r_+^{-1}r_*-i\q\int_{0}^{r_*}\rho_+(r_*')\,dr_*'-i\q\int_{r_{\sharp}}^{r}r'^{-1}\widetilde{\h}(r')\,dr'} \uppsi,\\
G_{\widetilde{A}}=&\:e^{i\q r_+^{-1}r_*-i\q\int_{0}^{r_*}\rho_+(r_*')\,dr_*'-i\q\int_{r_{\sharp}}^{r}r'^{-1}\widetilde{\h}(r')\,dr'} G_{\widehat{A}},\\
F_{\widetilde{A},\xi}=&\:e^{i\q r_+^{-1}r_*-i\q\int_{0}^{r_*}\rho_+(r_*')\,dr_*'-i\q\int_{r_{\sharp}}^{r}r'^{-1}\widetilde{\h}(r')\,dr'} F_{\widehat{A}, \xi},
\end{align*}
or, alternatively,
\begin{align*}
\widetilde{\uppsi}=&\:e^{i\q r_c^{-1}r_*+i\q\int_{0}^{r_*}\rho_c(r_*')\,dr_*'-i\q\int_{2}^{r}r'^{-1}\Omega^{-2}(r')\h(r')\,dr'}  \uppsi,\\
G_{\widetilde{A}}=&\:e^{i\q r_c^{-1}r_*+i\q\int_{0}^{r_*}\rho_c(r_*')\,dr_*'-i\q\int_{2}^{r}r'^{-1}\Omega^{-2}(r')\h(r')\,dr'}   G_{\widehat{A}},\\
F_{\widetilde{A},\xi}=&\:e^{i\q r_c^{-1}r_*+i\q\int_{0}^{r_*}\rho_c(r_*')\,dr_*'-i\q\int_{2}^{r}r'^{-1}\Omega^{-2}(r')\h(r')\,dr'}  F_{\widehat{A}, \xi}.
\end{align*}

We first restrict to $\psi=\psi_{ \ell}$ and omit the subscript $\ell$ and derive estimates that are uniform in $\ell$, so that we can sum over $\ell$ at the end of the argument.

 We use that the Fourier transform operator: $\mathfrak{F}: L^2(\R)\to L^2(\R)$ and its inverse $\mathfrak{F}^{-1}: L^2(\R)\to L^2(\R)$ are unitary operators and hence preserve inner products in $L^2(\R)$. We first apply \eqref{eq:highangfreq} with $\ell>L_0$ together with Plancherel's theorem to obtain:
\begin{multline}
\label{eq:Plancherel}
c\int_{0}^{\infty}\int_{\Sigma_{\tau}\cap\{r_H\leq r\leq r_I\}} |Y_*\widetilde{\uppsi}|^2+\upzeta(r)(|T\uppsi|^2+|\snabla_{\s^2}\uppsi|^2)+|\uppsi|^2\,d\sigma dr d\tau\\
+c\int_{0}^{\infty}\int_{\Sigma_{\tau}\cap\{r\leq r_H\}} (\rho_++\kappa_+)(r^{-1}\Omega)^{-p}\Omega^{-2}|\underline{L}\uppsi|^2\\
+(\rho_++\kappa_+)(r^{-1}\Omega)^{\delta}\Omega^{-2}(|K_+\uppsi|^2+|Y_*\widetilde{\uppsi}|^2)+(\rho_++\kappa_+)(r^{-1}\Omega)^{2-p}\Omega^{-2}(|\uppsi|^2+|\snabla_{\s^2}\uppsi|^2)\,d\sigma drd\tau\\
+c\int_{0}^{\infty}\int_{\Sigma_{\tau}\cap\{r\geq r_I\}} (\rho_c+\kappa_c)(r^{-1}\Omega)^{-p}\Omega^{-2}|L\uppsi|^2\\
+(\rho_c+\kappa_c)(r^{-1}\Omega)^{\delta}\Omega^{-2}(|K_c\uppsi|^2+|Y_*\widetilde{\uppsi}|^2)+(\rho_c+\kappa_c)(r^{-1}\Omega)^{2-p}\Omega^{-2}(|\uppsi|^2+|\snabla_{\s^2}\uppsi|^2)\,d\sigma drd\tau\\
\leq \overbrace{E\int_0^{\infty}\int_{\Sigma_{\tau}}\chi_{r_1} \re(K_+ \uppsi  \overline{r F_{\widehat{A},\xi}+r\xi G_{\widehat{A}}})\,d\sigma dr d\tau}^{=:I_1}+\overbrace{E \int_0^{\infty}\int_{\Sigma_{\tau}} (1-\chi_{r_1})\re(K_c\uppsi  \overline{r F_{\widehat{A},\xi}+r\xi G_{\widehat{A}}})\,d\sigma dr d\tau}^{=:I_2}\\
+\overbrace{\int_0^{\infty}\int_{\Sigma_{\tau}\cap\{r\leq r_3\}}\Re(2Y_*\widetilde{\uppsi} \overline{rF_{\widetilde{A},\xi}+r \xi G_{\widetilde{A}}})\,d\sigma dr d\tau}^{=:I_3}-\overbrace{\int_0^{\infty}\int_{\Sigma_{\tau}\cap\{r\geq R_3\}}\Re(2Y_*\widetilde{\uppsi} \overline{rF_{\widetilde{A},\xi}+r\xi G_{\widetilde{A}}})\,d\sigma dr d\tau}^{=:I_4}\\
-\overbrace{\int_{0}^{\infty}\int_{\Sigma_{\tau}\cap \{r\leq r_1\}}\re(2 \chi_{r_1} (r^{-1}\Omega)^{-p} \underline{L}(e^{iq\int_{r_{\sharp}}^{r}r'^{-1}\widetilde{\h}(r')\,dr'}\uppsi)\cdot e^{-iq\int_{r_{\sharp}}^{r}r'^{-1}\widetilde{\h}(r')\,dr'}\overline{rF_{\widehat{A},\xi}+rG_{\widehat{A}}})\,d\sigma dr d\tau}^{=:I_5} \\
-\overbrace{ \int_{0}^{\infty}\int_{\Sigma_{\tau}\cap\{r\geq R_1\}}\re(2 (1-\chi_{r_1}) (r^{-1}\Omega)^{-p} L(e^{-iq\int_{r_{\sharp}}^{r}r'^{-1}\Omega^{-2}(r')\h(r')\,dr'} \uppsi)\cdot e^{iq\int_{r_{\sharp}}^{r}r'^{-1}\Omega^{-2}(r')\h(r')\,dr'} \overline{rF_{\widehat{A},\xi}+rG_{\widehat{A}}})\,d\sigma dr d\tau}^{=:I_6}\\\
+C\int_{0}^{\infty}\int_{\Sigma_{\tau}\cap\{r_H\leq r\leq r_I\}} \frac{1}{\eta}(|rF_{\widehat{A},\xi}|^2+|rG_{\widehat{A}}|^2)+\eta (|\widetilde{\uppsi}|^2+|Y_*\widetilde{\uppsi}|^2)\,d\sigma dr d\tau\\
+C\int_{0}^{\infty}\int_{\Sigma_{\tau}\setminus\{r_H\leq r\leq r_I\}} \frac{1}{\eta} \rho_+ \rho_c(r^{-1}\Omega)^{-2+\delta}\Omega^2(|rF_{\widehat{A},\xi}|^2+|rG_{\widehat{A}}|^2)\\
+\eta (\rho_++\kappa_+)(\rho_c+\kappa_c)(r^{-1}\Omega)^{\delta}\Omega^{-2} |Y_*\widetilde{\uppsi}|^2\,d\sigma dr d\tau,
\end{multline}
where $\eta>0$ is a constant that we will take to be suitably small and we used that $r^{-2}\Omega\sim \rho_+(\rho_++2\kappa_+)$ near $r=r_+$ and  $r^{-2}\Omega\sim \rho_c(\rho_c+2\kappa_c)$ near $r=r_c$.

Using that $F_{\widehat{A},\xi}$ is compactly supported in $\tau$ and applying \eqref{eq:Ftimecutoff} together with the local-in-time energy estimates from Theorem \ref{thm:localenest} (with $p=0$), we estimate:
\begin{multline*}
	\int_{0}^{\infty}\int_{\Sigma_{\tau}\cap\{r_H\leq r\leq r_I\}} |rF_{\widehat{A},\xi}|^2d\sigma drd\tau\lesssim \int_{0}^{1}\int_{\Sigma_{\tau}\cap\{r_H\leq r\leq r_I\}} \mathcal{E}_0[\psi]\,d\sigma drd\tau\lesssim \int_{\Sigma_{0}} \mathcal{E}_0[\psi]\,d\sigma dr\\
	+ \int_{0}^{1}\int_{\Sigma_{\tau}}r^{-2}|r^3G_{\widehat{A}}|^2\,d\sigma drd\tau.
\end{multline*}
We also estimate using Theorem \ref{thm:localenest}:
\begin{multline*}
	\int_{0}^{\infty}\int_{\Sigma_{\tau}\setminus\{r_H\leq r\leq r_I\}}  \rho_+ \rho_c(r^{-1}\Omega)^{-2+\delta}\Omega^2|rF_{\widehat{A},\xi}|^2\,d\sigma dr d\tau\\
\lesssim \int_{0}^{1}\int_{\Sigma_{\tau}\setminus\{r_H\leq r\leq r_I\}} (r^{-1}\Omega)^{-1+\delta}\Omega^2|X\psi|^2+\mathcal{E}_0[\psi]\,d\sigma dr d\tau\lesssim \int_{0}^{1}\int_{\Sigma_{\tau}\setminus\{r_H\leq r\leq r_I\}} \mathcal{E}_{1-\delta}[\psi]\,d\sigma dr d\tau\\
\lesssim \int_{\Sigma_{0}} \mathcal{E}_{1-\delta}[\psi]\,d\sigma dr+\int_{0}^{1}\int_{\Sigma_{\tau}}r^{-2}|r^3G_{\widehat{A}}|^2\,d\sigma drd\tau.
\end{multline*}

Furthermore, for suitably small $\eta>0$, we can absorb the $|Y_*\widetilde{\uppsi}|^2$ term on the RHS of \eqref{eq:Plancherel} into the LHS.

In the remainder of the proof, we will estimate the integrals $I_1$--$I_6$.

\textbf{Estimating $I_2$ and $I_1$:}\\ 
We will estimate $I_2$. The estimates for $I_1$ proceed entirely analogously, with the roles of $\rho_+$ and $\rho_c$ interchanged.

By \eqref{eq:Ftimecutoff} and the fact that $\dot{\xi},\ddot{\xi}$ are supported in $\{0\leq \tau\leq 1\}$ and $(1-\chi_{r_1})$ vanishes when $r\leq r_1$, we can split and estimate:
\begin{multline*}
	|I_2|\leq C\left|\overbrace{\int_0^{\infty}\int_{\Sigma_{\tau}\cap \{r\geq r_1\}} 2\dot{\xi} (1-\chi_{r_1}) \re\left(K_c(\xi\psi)\cdot \overline{X\psi}\right)\,d\sigma dr d\tau}^{=:I_{2,a}}\right|\\
	+C\overbrace{\int_0^{1} \int_{\Sigma_{\tau}\cap \{r\geq  r_1\}}|K_c(\xi \psi)|\left[|\h| |T\psi|+\left(|\h|+\left|\frac{d\h}{dr}\right|\right)|\psi|\right]\,d\sigma dr d\tau}^{I_{2,b}}\\
	+C\overbrace{\int_0^{\infty} \int_{(\Sigma_{\tau}\cap \{r\geq r_1\})\setminus \{ r_{\sharp}-\eta\leq r\leq r_{\sharp}+\eta\}}|K_c(\xi \psi)|\cdot |rG_{\widehat{A}}|\,d\sigma dr d\tau}^{=:I_{2,c}}+\overbrace{\left| \int_0^{\infty}\int_{\Sigma_{\tau}\cap \{ r_{\sharp}-\eta\leq r\leq r_{\sharp}+\eta\}}(1-\chi_{r_1})\re(K_c\uppsi  \overline{r\xi G_{\widehat{A}}})\,d\sigma dr d\tau\right|}^{=:I_{2,d}}.
\end{multline*}

Notice first that we can estimate $I_{2,d}$ by writing $K_c\psi=\upzeta K_c\psi+(1-\upzeta)K_c\psi$ and integrating by parts in the $T$-direction to estimate the terms with a factor $1-\upzeta$, using that $\lim_{\tau\to \infty}|\uppsi|^2(\tau,r,\theta,\varphi)=0$, by the sufficient integrability assumption. We obtain:
We obtain:
\begin{multline*}
	I_{2,d}\leq \mu  \int_0^{\infty}\int_{\Sigma_{\tau}\cap \{ r_{\sharp}-\eta\leq r\leq r_{\sharp}+\eta\}}\xi(1-\upzeta)|\uppsi|^2+\xi\upzeta |K_c\psi|^2\,d\sigma dr d\tau\\
	+C \mu^{-1} \int_0^{\infty}\int_{\Sigma_{\tau}\cap \{ r_{\sharp}-\eta\leq r\leq r_{\sharp}+\eta\}}\xi \upzeta |rG_{\widehat{A}}|^2+\xi (1-\upzeta)|TG_{\widehat{A}}|^2\,d\sigma dr d\tau.
\end{multline*}

Now consider $I_{2,b}$. By Theorem \ref{thm:localenest} with $p=0$, we can estimate:
\begin{multline*}
	I_{2,b}\leq \int_0^1\int_{\Sigma_{\tau}\cap \{r\geq  r_1\}} r^{-2}\Omega^2(|T\psi|^2+|\psi|^2)+r^{-2}(\rho_c+\kappa_c)|\psi|^2\,d\sigma dr d\tau\\
	\lesssim \int_0^1\int_{\Sigma_{\tau}\cap \{r\geq  r_1\}} \mathcal{E}_0[\psi]\,d\sigma dr d\tau\lesssim\int_{\Sigma_{0}} \mathcal{E}_0[\psi]\,d\sigma dr+\int_{0}^{1}\int_{\Sigma_{\tau}}r^{-2}|r^3G_{\widehat{A}}|^2\,d\sigma drd\tau.
\end{multline*}
Consider $I_{2,c}$. We can estimate
\begin{multline*}
I_{2,c}\leq  \int_0^{\infty}\int_{(\Sigma_{\tau}\cap \{r\geq r_1\})\setminus \{ r_{\sharp}-\eta\leq r\leq r_{\sharp}+\eta\}} \mu\xi^2(r^{-1}\Omega)^{\delta}(\rho_c+\kappa_c) \Omega^{-2}|K_c\psi|^2+|\dot{\xi}| r^{-2}|\psi|^2\,d\sigma dr d\tau\\
+C\int_0^{\infty}\int_{(\Sigma_{\tau}\cap \{r\geq r_1\})\setminus \{ r_{\sharp}-\eta\leq r\leq r_{\sharp}+\eta\}} \mu^{-1}(r^{-1}\Omega)^{-\delta}\rho_c|rG_{\widehat{A}}|^2+|\dot\xi|r^2|rG_{\widehat{A}}|^2\,d\sigma dr d\tau\\
\leq \eta \int_0^{\infty}\int_{(\Sigma_{\tau}\cap \{r\geq r_1\})\setminus \{ r_{\sharp}-\eta\leq r\leq r_{\sharp}+\eta\}} \xi^2(r^{-2}\Omega^2)^{\delta}(\rho_c+\kappa_c) \Omega^{-2}|K_c\psi|^2\,d\sigma drd\tau+\int_0^1\int_{\Sigma_{\tau}\cap \{r\geq  r_H\}} \mathcal{E}_0[\psi]\,d\sigma drd\tau\\
+ C\eta^{-1} \int_0^{\infty}\int_{(\Sigma_{\tau}\cap \{r\geq r_1\})\setminus \{ r_{\sharp}-\eta\leq r\leq r_{\sharp}+\eta\}} (r^{-1}\Omega)^{-\delta}\rho_c|rG_{\widehat{A}}|^2\,d\sigma dr d\tau+C\int_0^1\int_{\Sigma_{\tau}\cap \{r\geq  r_H\}}|r^3G_{\widehat{A}}|^2\,d\sigma d\rho_cd\tau.
	\end{multline*}
The integral of $\mathcal{E}_0[\psi]$ can be estimated as above, by applying Theorem \ref{thm:localenest} with $p=0$.

We are left with estimating $I_{2,a}$. We first integrate by parts in the $X$-direction to obtain:
\begin{multline}
\label{eq:idI2a}
	I_{2,a}=-\int_0^{\infty}\int_{\Sigma_{\tau}\cap \{r\geq  r_1\}} 2\dot{\xi} \chi_T \re\left(XK_c(\xi\psi)\cdot \overline{\psi}\right)\,d\sigma dr d\tau\\
	-\int_0^{\infty}\int_{\Sigma_{\tau}\cap \{r\geq  r_1\}} 2\dot{\xi} \frac{d\chi_T}{dr} \re\left(K_c(\xi\psi)\cdot \overline{\psi}\right)\,d\sigma dr d\tau+\int_{\mathcal{C}^+}2\dot{\xi} \re\left(K_c(\xi\psi)\cdot \overline{\psi}\right)\,d\sigma d\tau.
\end{multline}
The second term on the RHS of \eqref{eq:idI2a} is supported in $\{r_{\sharp}-\eta\leq r\leq r_{\sharp}+\eta\}\cap\{0\leq \tau\leq 1$, so its norm can immediately using Theorem \ref{thm:localenest} with $p=0$. The third term on the RHS of \eqref{eq:idI2a} can be estimated by
\begin{equation*}
	\int_{\mathcal{C}^+}\dot{\xi}(|\psi|^2+|K_c\psi|^2)\,d\sigma d\tau,
\end{equation*}
which can be estimated via Theorem \ref{thm:localenest} with $p=0$ after applying a Hardy inequality, making use of the compact support of $\dot{\xi}$.

We estimate the third term on the RHS of \eqref{eq:idI2a} as follows:
\begin{multline*}
	\left|\int_0^{\infty}\int_{\Sigma_{\tau}\cap \{r\geq  r_1\}} 2\dot{\xi} \chi_T \re\left(XK_c(\xi\psi)\cdot \overline{\psi}\right)\,d\sigma dr d\tau\right|\leq \left|\int_0^{\infty}\int_{\Sigma_{\tau}\cap \{r\geq  r_1\}} T(\xi^2) \chi_T \re\left(XK_c\psi\cdot \overline{\psi}\right)\,d\sigma dr d\tau\right|\\
	+\left|\int_0^{\infty}\int_{\Sigma_{\tau}\cap \{r\geq  r_1\}} \dot{\xi}^2 \chi_T X(|\psi|^2)\,d\sigma dr d\tau\right|\\
	\leq \left|\int_0^{\infty}\int_{\Sigma_{\tau}\cap \{r\geq  r_1\}} T(\xi^2) \chi_T \re\left(XK_c\psi\cdot \overline{\psi}\right)\,d\sigma dr d\tau\right|+C\left|\int_{\mathcal{C}^+} \dot{\xi}^2 |\psi|^2\,d\sigma d\tau\right|\\
	+C\int_0^1\int_{\Sigma_{\tau}\cap \{ r_{\sharp}-\eta\leq r\leq r_{\sharp}+\eta\}}\mathcal{E}_0[\psi]\,d\sigma dr d\tau.
\end{multline*}
The last two terms on the very RHS can be estimated as above. We estimate the first term by applying \eqref{eq:maineqradfield} to obtain:
\begin{multline*}
	XK_c\psi=O_{\infty}(\rho_c^0)X(\Omega^2 X\psi)+O_{\infty}(\rho_c^0)r^{-2}\slashed{\Delta}_{\s^2}\psi+O_{\infty}(\rho_c^0)r^{-2}\psi+(\rho_c+\kappa_c)O_{\infty}(\rho_c)K_c^2\psi\\
+O_{\infty}(\rho_c)X\psi+O_{\infty}(\rho_c^0)rG_{\widehat{A}}.
\end{multline*}
Using the above identity and integrating by parts to deal with the second-order derivatives, we conclude that: 
\begin{equation*}
	|I_{2,a}|\lesssim\int_{\Sigma_{0}} \mathcal{E}_0[\psi]\,d\sigma dr +\int_0^1\int_{\Sigma_{\tau}\cap \{r\geq  r_H\}}|r^3G_{\widehat{A}}|^2\,d\sigma d\rho_c d\tau.
\end{equation*}
Putting the above together, we conclude that:
\begin{multline*}
	|I_{2}|\leq  \eta \int_0^{\infty}\int_{\Sigma_{\tau}\cap \{r\geq  r_1\}} \xi^2(r^{-2}\Omega^2)^{\delta}(\rho_c+\kappa_c) \Omega^{-2}|K_c\psi|^2\,d\sigma drd\tau+C\int_{\Sigma_{0}} \mathcal{E}_0[\psi]\,d\sigma dr \\
+C\eta^{-1} \int_0^{\infty}\int_{\Sigma_{\tau}} (r^{-2}\Omega^2)^{-\delta}\rho_c|rG_{\widehat{A}}|^2\,d\sigma dr d\tau	C\int_0^1\int_{\Sigma_{\tau}}|r^3G_{\widehat{A}}|^2\,d\sigma d\rho_c d\tau.
\end{multline*}
The first term on the RHS can be absorbed into the LHS of \eqref{eq:Plancherel} for suitably small $\eta$.

The integral $I_1$ can be estimated entirely analogously, with $\rho_+$ taking on the role of $\rho_c$.

\textbf{Estimating $I_3$ and $I_4$:}\\
We will estimate $I_4$. The estimate for $I_3$ will proceed analogously, with the roles of $\rho_c$ and $\rho_+$ interchanged. Note first that:
\begin{equation*}
	Y_*\widetilde{\uppsi} \overline{rF_{\widetilde{A},\xi}+\xi rG_{\widetilde{A}} }=K_c\uppsi\cdot \overline{rF_{\widehat{A},\xi}+\xi rG_{\widehat{A}} }-i\q (\rho_c-r^{-1}\h)\uppsi \overline{rF_{\widehat{A},\xi}+\xi rG_{\widehat{A}} }-2L\uppsi \overline{rF_{\widehat{A},\xi}+\xi rG_{\widehat{A}} }.
\end{equation*}
The contribution of the first term to $I_4$ can be estimated in the same way as $I_2$. We estimate the terms involving $rF_{\widehat{A},\xi}$ by applying Young's inequality and \eqref{eq:Ftimecutoff}:
\begin{multline*}
	\left|\int_{\R}\int_{\Sigma_{\tau}\cap\{r\geq R_3\}}\re\left((-i\q (\rho_c-r^{-1}\h)\uppsi -2L \uppsi)\overline{rF_{\widehat{A},\xi}}\right)\,d\sigma dr d\tau\right|\\
	\leq C\Bigg|\int_0^{\infty}\int_{\Sigma_{\tau}\cap\{r\geq R_3\}}\dot{\xi}\xi (\Omega^2|X\psi|+\rho_c|\psi|)\left(\rho_c|X\psi|+\rho_c^3|T\psi|+(\rho_c+\kappa_c)\rho_c|\psi| \right)\,d\sigma dr d\tau\Bigg|\\
\leq C\int_0^{1}\int_{\Sigma_{\tau}\cap\{r\geq R_3\}}\mathcal{E}_0[\psi]\,d\sigma drd\tau.
\end{multline*}
Furthermore,
\begin{multline*}
	\left|\int_{\R}\int_{\Sigma_{\tau}\cap\{r\geq R_3\}}\re\left((-i\q (\rho_c-r^{-1}\h)\uppsi -2L \uppsi) \overline{r\xi G_{\widehat{A}}}\right)\,d\sigma dr d\tau\right|\\
	\leq\int_0^{\infty}\int_{\Sigma_{\tau}\cap\{r\geq R_3\}} \eta (r^{-1}\Omega)^{2-p}\Omega^{-2} ((\rho_c+\kappa_c)|\uppsi|^2+(\rho_c+\kappa_c)^{-1}|L\uppsi|^2)+C\eta^{-1}\rho_c^2 (r^{-1}\Omega)^{p-2}|rG_{\widehat{A}}|^2\,d\sigma drd\tau\\
	\leq \int_0^{\infty}\int_{\Sigma_{\tau}\cap\{r\geq R_3\}} \eta (r^{-1}\Omega)^{2-p}\Omega^{-2} (\rho_c+\kappa_c)|\uppsi|^2+\eta \rho_c (r^{-1}\Omega)^{-p}\Omega^{-2}|L\uppsi|^2+C\eta^{-1}\rho_c^2 (r^{-1}\Omega)^{p-2}|rG_{\widehat{A}}|^2\,d\sigma drd\tau.
\end{multline*}
We conclude that
\begin{multline*}
	|I_4|\leq  \eta \int_0^{\infty}\int_{\Sigma_{\tau}\cap\{r\geq R_3\}}  (r^{-1}\Omega)^{2-p}\Omega^{-2} (\rho_c+\kappa_c)|\uppsi|^2+ \rho_c (r^{-1}\Omega)^{-p}\Omega^{-2}|L\uppsi|^2\\
	+C\int_{\Sigma_{0}}\mathcal{E}_0[\psi]\,d\sigma dr+C\eta^{-1}\int_0^{\infty}\int_{\Sigma_{\tau}\cap\{r\geq R_3\}}\rho_c^2 (r^{-1}\Omega)^{p-2}|rG_{\widehat{A}}|^2\,d\sigma dr d\tau.
\end{multline*}
The first integral on the RHS can be absorbed into the LHS of \eqref{eq:Plancherel} for suitably small $\eta$.

The estimates for $|I_3|$ follow analogously.
 
\textbf{Estimating $I_5$ and $I_6$:}
Consider $I_6$. The estimates for $I_5$ follow analogously. We write:
\begin{multline*}
(r^{-1}\Omega)^{-p} L(e^{i\q\int_{r_{\sharp}}^{r}r'^{-1}\Omega^{-2}(r')\h(r')\,dr'} \uppsi)\cdot e^{-i\q\int_{r_{\sharp}}^{r}r'^{-1}\Omega^{-2}(r')\h(r')\,dr'} \overline{rF_{\widehat{A},\xi}+rG_{\widehat{A}}}=(r^{-1}\Omega)^{-p} L\uppsi \cdot\overline{rF_{\widehat{A},\xi}+rG_{\widehat{A}}}\\
+i\q r^{-1}\h (r^{-1}\Omega)^{-p}\uppsi\cdot \overline{rF_{\widehat{A},\xi}+rG_{\widehat{A}}}.
\end{multline*}
We estimate the terms involving $rF_{\widehat{A},\xi}$ by applying Young's inequality and \eqref{eq:Ftimecutoff}:
\begin{multline*}
	\left|\int_{\R}\int_{\Sigma_{\tau}\cap\{r\geq R_1\}}(r^{-1}\Omega)^{-p}\re\left((iq r^{-1}\h\uppsi -L \uppsi)\overline{rF_{\widehat{A},\xi}}\right)\,d\sigma dr d\tau\right|\\
	\leq C\Bigg|\int_0^{\infty}\int_{\Sigma_{\tau}\cap\{r\geq R_3\}}\dot{\xi}\xi (r^{-1}\Omega)^{-p}(\Omega^2|X\psi|+r^{-1}r^{-2}\Omega^2|\psi|)\left(\rho_c|X\psi|+\rho_c^3|T\psi|+r^{-2}\Omega^2|\psi| \right)\,d\sigma dr d\tau\Bigg|\\
\leq C\int_0^{1}\int_{\Sigma_{\tau}\cap\{r\geq R_1\}}\rho_c\mathcal{E}_p[\psi]\,d\sigma drd\tau.
\end{multline*}

Furthermore,
\begin{multline*}
	\left|\int_{\R}\int_{\Sigma_{\tau}\cap\{r\geq R_1\}}(r^{-1}\Omega)^{-p}\re\left((-i\q r^{-1}\h\uppsi -L \uppsi)\overline{\xi rG_{\widehat{A}}}\right)\,d\sigma dr d\tau\right|\\
	\leq\int_0^{\infty}\int_{\Sigma_{\tau}\cap\{r\geq R_3\}} \mu (r^{-1}\Omega)^{-p} (\Omega^{-2}(\rho_c+\kappa_c)^3(r^{-1}\Omega)^4|\uppsi|^2+(\rho_c+\kappa_c)\Omega^{-2}|L\uppsi|^2)\\
	+C\mu^{-1}(\rho_c+\kappa_c)^{-1}\Omega^2(r^{-1}\Omega)^{-p}|rG_{\widehat{A}}|^2\,d\sigma drd\tau\\
	\leq \int_0^{\infty}\int_{\Sigma_{\tau}\cap\{r\geq R_3\}} \mu (r^{-1}\Omega)^{4-p}\Omega^{-2} (\rho_c+\kappa_c)^3|\uppsi|^2+\mu (\rho_c+\kappa_c) (r^{-1}\Omega)^{-p}\Omega^{-2}|L\uppsi|^2\\
	+C\mu^{-1}\rho_c(r^{-1}\Omega)^{-p}r^{2}|rG_{\widehat{A}}|^2\,d\sigma drd\tau.
\end{multline*}
Therefore,
\begin{multline*}
	|I_6|\leq \mu \int_0^{\infty}\int_{\Sigma_{\tau}\cap\{r\geq R_3\}}  (r^{-1}\Omega)^{4-p}\Omega^{-2} (\rho_c+\kappa_c)^3|\uppsi|^2+ (\rho_c+\kappa_c) (r^{-1}\Omega)^{-p}\Omega^{-2}|L\uppsi|^2\,d\sigma dr d\tau\\
+C\int_{\Sigma_0}\mathcal{E}_{p-1}[\psi]\,d\sigma dr d\tau+C\mu^{-1}\int_0^{\infty}\int_{\Sigma_{\tau}\cap\{r\geq R_3\}}\eta^{-1}\rho_c(r^{-1}\Omega)^{-p}r^{2}|rG_{\widehat{A}}|^2\,d\sigma drd\tau.
\end{multline*}
The estimate for $|I_5|$ proceeds analogously

We conclude that \eqref{eq:iedphysicalspace} holds for $\psi_{> L_0}$.

The case $\psi_{\leq L_0}$ proceeds entirely analogously, starting from Proposition \ref{prop:iedlowlfreq} if we additionally assume that $|q|>q_0$, with the key difference being the range of $p$, $\max\{1-\re \beta_{\ell}-\epsilon,\epsilon\}<  p<\min\{1+\re \beta_{\ell},2\}$, and the addition of the following term on the RHS of the final estimate:
\begin{equation*}
C\int_{0}^{\infty}\left[\int_{\Sigma_{\tau}} (r^{-1}\Omega)^{-p-\epsilon}\left(\Omega^2|rF_{\xi,\widetilde{A}}|^2+\rho_+^{1-\epsilon}\rho_c^{1-\epsilon}\xi^2r^{-2}|r^3G_{\widetilde{A}}|^2\right)\,d\sigma dr \right]\,d\tau.
\end{equation*}
We estimate via Theorem \ref{thm:localenest}:
\begin{equation*}
	\int_{0}^{\infty}\int_{\Sigma_{\tau}} (r^{-1}\Omega)^{-p-\epsilon}\Omega^2|rF_{\xi,\widetilde{A}}|^2\,d\sigma drd\tau\lesssim \int_0^1 \mathcal{E}_{p+\epsilon}[\psi]\,d\sigma dr d\tau\lesssim \int_{\Sigma_0}  \mathcal{E}_{p+\epsilon}[\psi]\,d\sigma dr. \qedhere
\end{equation*}
Furthermore, there is no need for $\upzeta$ when considering $\psi_{\leq L_0}$.
\end{proof}

\begin{corollary}
\label{cor:iedphysicalspaceho}
Let $\ell_0,N\in \N_0$, $\epsilon>0$ and $\max\{0,1-\re\beta_{\ell_0}\}<p<\min\{2,1+\re\beta_{\ell_0}\}+\epsilon$. Then there exist a constant $C=C(\epsilon,p, r_H,r_I,\h,N)>0$, such that for $\psi=\psi_{\geq \ell_0}$:
\begin{multline}
\label{eq:iedphysicalspaceho}
\sum_{k_1+k_2+k_3\leq N}\int_{0}^{\infty}\int_{\Sigma_{\tau}\cap\{r_H\leq r\leq r_I\}} |\snabla_{\s^2}^{k_1}Y_*^{1+k_2}T^{k_3}\psi|^2+|\snabla_{\s^2}^{k_1}Y_*^{k_2}T^{k_3}\psi|^2\\
+\upzeta(r)(|\snabla_{\s^2}^{k_1}Y_*^{k_2}T^{k_3+1}\psi|^2+|\snabla_{\s^2}^{k_1+1}Y_*^{k_2}T^{k_3}\psi|^2)\,d\sigma dr d\tau\,d\sigma dr d\tau\\
\leq C \sum_{k\leq N}\int_{\Sigma_0} \mathcal{E}_{p}[T^{k}\psi]\,d\sigma dr+C\int_{0}^{\infty}\Bigg[\int_{\Sigma_{\tau}} \max\{(r^{-1}\Omega)^{-p}\rho_+^{1-2\epsilon}\rho_c^{1-2\epsilon},\mathbf{1}_{0\leq \tau\leq 1}\}r^{-2}|r^3T^kG_{\widehat{A}}|^2\\
+\Omega^2(1-\upzeta)|T^{k+1}(rG_{\widehat{A}})|^2\,d\sigma dr \Bigg]\,d\tau.
\end{multline}
If $\psi=\psi_{\leq L_0}$, with $L_0\in \N_0$, then we can omit the factor $\upzeta$ in \eqref{eq:iedphysicalspaceho} at the expense of adding $L_0$-dependence in the constant $C$.
\end{corollary}
\begin{proof}
We apply standard elliptic estimates in the region $\{r_H\leq r\leq r_I\}$ to estimates all higher-order derivatives in terms of $T$-derivatives of first-order quantities. Then we use that\\ $[(g_{M,Q,\Lambda}^{-1})^{\mu\nu}(^{\widehat{A}}D)_{\mu}(^{\widehat{A}}D)_{\nu},T]=0$  to be able to directly apply \eqref{eq:iedphysicalspace} with $\psi$ replaced by $T^k\psi$.
 \end{proof}
 
 We can also apply red shift estimates to improve the estimate in Proposition \ref{prop:iedphysloword} \textbf{at the expense of making the constant in the estimate $\kappa_+,\kappa_c$-dependent}.
 \begin{corollary}
 Let $N\in \N_0$ and $\kappa_+,\kappa_c>0$. Then there exists a constant\\ $C_{\kappa_+,\kappa_c}=C_{\kappa_+,\kappa_c}(\epsilon,p, r_H,r_I,\kappa_+,\kappa_c,N)>0$, such that:
\begin{multline}
\label{eq:iedphysicalspaceredshift}
\sum_{k_1+k_2+k_3\leq N}\sup_{\tau\geq 0}\int_{\Sigma_{\tau}\cap\{r\leq r_H\}} \mathcal{E}_{2}[T^{k_3}\psi]\,d\sigma dr\\
+\int_{0}^{\infty}\int_{\Sigma_{\tau}\cap\{r_H\leq r\leq r_I\}}  |\snabla_{\s^2}^{k_1}Y_*^{k_2+1}T^{k_3}\psi|^2+\upzeta(r)(|\snabla_{\s^2}^{k_1}Y_*^{k_2}T^{k_3+1}\psi|^2+|\snabla_{\s^2}^{k_1+1}Y_*^{k_2}T^{k_3}\psi|^2)+|\snabla_{\s^2}^{k_1}Y_*^{k_2}T^{k_3}\psi|^2\,d\sigma dr d\tau\\
+\int_{0}^{\infty}\int_{\Sigma_{\tau}\setminus\{r_H\leq r\leq r_I\}} |XT^{k_3}\psi|^2+|T^{k_3+1}\psi|^2+|\snabla_{\s^2}T^{k_3}\psi|^2+|T^{k_3}\psi|^2\,d\sigma drd\tau\\
\leq C_{\kappa_+,\kappa_c}\left[ \sum_{k\leq N}\int_{\Sigma_0}\mathcal{E}_{2}[T^k\psi]\,d\sigma dr+\int_{0}^{\infty}\left[\int_{\Sigma_{\tau}} |rT^kG_{\widehat{A}}|^2+\Omega^2(1-\upzeta)|rT^{k+1}rG_{\widehat{A}}|^2\,d\sigma dr \right]\,d\tau\right].
\end{multline}
If $\psi=\psi_{\leq L_0}$, with $L_0\in \N_0$, then we can omit the factor $\upzeta$ in \eqref{eq:iedphysicalspaceredshift} at the expense of adding $L_0$-dependence in the constant $C$.
\end{corollary}
 \begin{proof}
 	The proof of red-shift estimates is standard. We can combine \eqref{eq:iedphysicalspaceho} with energy estimates in the regions $\{r\leq r_H\}$ and $\{r\geq r_I\}$, and consider the vector field multipliers $\Omega^{-2}\underline{L}$ and $\Omega^{-2}\underline{L}$, respectively. See \cite{lecturesMD}[\S7] for derivation of red-shift estimates in a general setting. The additional zeroth- and first-order terms appearing in \eqref{eq:CSF} with $\q\neq 0$ compared to $\q=0$ do not cause issues in the argument.
 \end{proof}
 
\subsection{Verifying future integrability}
\label{sec:verifyfutureint}
Let $\psi$ be a solution to \eqref{eq:CSF} with $A=\widehat{A}$ arising from initial data in $C_c^{\infty}(\Sigma)\times C_c^{\infty}(\Sigma)$, with $G_{\widehat{A}}$ compactly supported in spacetime. Assume moreover that $\psi=\psi_{\ell}$ with $\ell\in \N_0$.

Consider the set
\begin{equation*}
\mathcal{A}=\{\q\in \R\,|\, \textnormal{$\psi$ is future integrable}\}.
\end{equation*}
The aim of this section is to show that $\mathcal{A}=\R$ if $\kappa_+>0$ and $\kappa_c>0$.

\begin{proposition}
Let $\kappa_+>0$ and $\kappa_c>0$. Then $(-\q_0,\q_0)\subset \mathcal{A}$ for suitably small $\q_0$.
\end{proposition}
\begin{proof}
We can directly apply \cite{bes20}[Theorem 4.2] to conclude future-integrability of fixed-$\ell$ solutions to \eqref{eq:CSF} in the case $\kappa_+>0$ and $\kappa_c>0$ with $|\q|\leq \q_0$, for $\q_0$ sufficiently small (given initial data in $C_c^{\infty}(\Sigma)\times C_c^{\infty}(\Sigma)$).

Alternatively, we can deduce future integrability of fixed-$\ell$ solutions for $|\q|\leq \q_0$ from integrated energy estimates in the $\q=0$ case (for fixed $\ell$) and applying the analogue of Corollary \ref{cor:iedsmallq}. The required integrated energy estimates on sub-extremal Reissner--Nordstr\"om-de Sitter can be found in \cite{gon24}[Proposition 11].
\end{proof}

\begin{proposition}
\label{prop:contargopen}
	Let $\kappa_+>0$ and $\kappa_c>0$. Then $\mathcal{A}$ is open.
\end{proposition}
\begin{proof}
We assume that $|\q|\geq \q_0$ and consider the operator $L_{\q}$ defined as follows:
\begin{multline*}
L_{\q}f=X(\Omega^2Xf)+r^{-2}\slashed{\Delta}_{\s^2}f-\left[r^{-1}\frac{d\Omega^2}{dr} 
+\frac{2\Lambda}{3}\right]f-2(1-\h)(T+i\q r^{-1})Xf\\
-\h \widetilde{\h} T^2f+\left(\frac{d\h}{dr}-2i \q r^{-1} \h\widetilde{\h}\right)Tf+\left[i\q r^{-1}\frac{d\h}{dr}+\q^2  \mathbbm{h}\widetilde{\mathbbm{h}}r^{-2}+i \q r^{-2}(1-\h)\right]f.
\end{multline*}
A solution to \eqref{eq:CSF} with $A=\widehat{A}$ satisfies $L_{\q}\psi=rG_{\widehat{A}}$.

Now we consider the operator $L_{\q_{\tau}}$ where $\q_{\tau}=\xi_{\tau}\q_1+(1-\xi_{\tau}) \q_{\star}$, with $\xi_{\tau}: [0,\infty)$ a smooth cut-off function such that $\xi_{\tau}(\tau')=1$ for $\tau'\leq \tau-\delta_0$ and $\xi_{\tau}(\tau')=0$ for $\tau' \geq \tau$, with $\delta_0>0$.

Then
\begin{equation*}
	(L_{\q_{\tau}}-L_{\q_{\star}})\psi^{\tau}=(\q_{\star}-\q_{\tau})\left[-2i r^{-1}X\psi^{\tau}+2i r^{-1} \h \widetilde{\h} T\psi^{\tau}-\left(i r^{-1} \frac{d\h}{dr}+(\q_{\tau}+\q_0) r^{-2} \h \widetilde{\h}+i r^{-2}(1-\h)\right)\psi^{\tau} \right].
\end{equation*}
Hence, $L_{\q_{\star}}\psi^{\tau}=r G_{\widehat{A}}+rG_{\q_{\star},\q_{\tau}}$, with
\begin{equation*}
	rG_{\q_{\star},\q_{\tau}}:=-(L_{\q_{\tau}}-L_{\q_{\star}})\psi^{\tau}.
\end{equation*}

Suppose now that $\q_{\star}\in \mathcal{A}$. Since $G_{\widehat{A}}$ is sufficiently integrable and $rG_{\q_{\star},\q_{\tau}}$ is compactly supported in time, and hence also sufficiently integrable, we must have that $\psi^{\tau}$ is future integrable.

There exists a constant $C>0$, such that $r^{-2}|r^3G_{\q_{\star},\q_{\tau})}|^2\leq C (\q_{\star}-\q_1)^2(|X\psi|^2+r^{-4}|T\psi|^2+r^{-4}|\psi|^2)$, so we can apply \eqref{eq:iedphysicalspaceredshift}, using that $\ell$ is bounded to omit the factor $\upzeta$, and estimate:
\begin{multline*}
\sum_{k_1+k_2+k_3\leq K}\sup_{\tau'\geq 0}\int_{\Sigma_{\tau'}} \mathcal{E}_{2}[T^{k_3}\psi^{\tau}]\,d\sigma dr\\
+\sum_{k_1+k_2+k_3\leq K+1}\int_{0}^{\infty}\int_{\Sigma_{\tau'}\cap\{r_H\leq r\leq r_I\}} |\snabla_{\s^2}^{k_1}Z_*^{k_2}T^{k_3}\psi^{\tau}|^2\,d\sigma dr d\tau'\\
+\sum_{k\leq K}\int_{0}^{\infty}\int_{\Sigma_{\tau'}\setminus\{r_H\leq r\leq r_I\}} |XT^{k}\psi^{\tau}|^2+|T^{k+1}\psi^{\tau}|^2+|\snabla_{\s^2}T^{3}\psi^{\tau}|^2+|T^{k}\psi^{\tau}|^2\,d\sigma drd\tau'\\
\leq C_{\kappa_+,\kappa_c}\Bigg[\int_{\Sigma_0} \sum_{k\leq K}\mathcal{E}_{2}[T^k\psi^{\tau}]\,d\sigma dr\\
+\int_{0}^{\infty}\left[\int_{\Sigma_{\tau'}} |rT^kG_{\widehat{A}}|^2+(\q_{\star}-\q_1)^2(|XT^k\psi^{\tau}|^2+|T^{k+1}\psi^{\tau}|^2+|T^k
+\psi^{\tau}|^2)\,d\sigma dr \right]\,d\tau'\Bigg].
\end{multline*}
For $|\q_{\star}-\q_1|$ suitably small depending on $\kappa_+,\kappa_c$, we can absorb the terms with a factor $(\q_{\star}-\q_1)^2$ into the left-hand side. Here we use that $\ell$ is bounded.

Using that $\psi^{\tau}=\psi$ for $\tau'\leq \tau$ and that the constant $C_{\kappa_+,\kappa_c}$ does not depend on $\tau$, we can take the limit $\tau\to \infty$ to conclude that $\psi$ is future integrable and hence $\q_1\in \mathcal{A}$.
\end{proof}

\begin{proposition}
Let $\kappa_c,\kappa_+>0$. The set $\mathcal{A}$ is closed in $\R$ and \eqref{eq:iedphysicalspace}.
\end{proposition}
\begin{proof}
Consider a sequence $\{q_n\}_{n\in \N}$ in $\mathcal{A}$ such that $\q_n\to \q\in \R$, i.e. solutions $\psi^{\q_n}$ to $L_{\q_n}\psi^{\q_n}=rG_{\widehat{A}}$ arising from smooth compactly supported initial data are future integrable. We fix the initial data for $\psi^{\q_n}$ to be independent of $n$. We then need to show that the solution $\psi^{\q}$ to $L_{\q}\psi=rG_{\widehat{A}}$ with the same initial data as $\psi^{\q_n}$ is future integrable to conclude that ${\q}\in \mathcal{A}$ and hence, that $\mathcal{A}$ is closed.

We consider the differences $\psi^{\q_n}-\psi^{\q_m}$ and $L_{\q_m}(\psi^{\q_n}-\psi^{\q_m})=rG_{\q_m,\q_n}$. We then apply \eqref{eq:iedphysicalspaceredshift} to estimate:
\begin{multline*}
\sum_{k_1+k_2+k_3\leq K}\sup_{\tau'\geq 0}\int_{\Sigma_{\tau'}} \mathcal{E}_{2}[T^{k_3}(\psi^{\q_n}-\psi^{\q_m})]\,d\sigma dr\\
+\sum_{k_1+k_2+k_3\leq K+1}\int_{0}^{\infty}\int_{\Sigma_{\tau'}\cap\{r_H\leq r\leq r_I\}} |\snabla_{\s^2}^{k_1}Z_*^{k_2}T^{k_3}(\psi^{\q_n}-\psi^{\q_m})|^2\,d\sigma dr d\tau'\\
+\sum_{k\leq K}\int_{0}^{\infty}\int_{\Sigma_{\tau'}\setminus\{r_H\leq r\leq r_I\}} |XT^{k}(\psi^{\q_n}-\psi^{\q_m})|^2+|T^{k+1}(\psi^{\q_n}-\psi^{\q_m})|^2+|\snabla_{\s^2}T^{3}(\psi^{\q_n}-\psi^{\q_m})|^2\\
+|T^{k}(\psi^{\q_n}-\psi^{\q_m})|^2\,d\sigma drd\tau'\\
\leq C_{\kappa_+,\kappa_c}(\q_n-\q_m)^2\int_{0}^{\infty}\left[\int_{\Sigma_{\tau'}} (|XT^k(\psi^{\q_n}-\psi^{\q_m})|^2+|T^{k+1}(\psi^{\q_n}-\psi^{\q_m})|^2+|T^k(\psi^{\q_n}-\psi^{\q_m})|^2)\,d\sigma dr \right]\,d\tau'.
\end{multline*}
It therefore follows that $\{\psi^{\q_n}\}$ is a Cauchy sequence in the Hilbert space corresponding to the norm on the right-hand side and must therefore converge with limit $\psi^{\q}$. We conclude that $\q\in \mathcal{A}$.
\end{proof}

\begin{corollary}
\label{cor:removesuffintkappapos}
	Let $\kappa_+>0$ and $\kappa_c>0$. Then $\mathcal{A}=\R$ and \eqref{eq:iedphysicalspace} and \eqref{eq:iedphysicalspaceho} hold for all solutions to \eqref{eq:CSF} with $A=\widehat{A}$.
\end{corollary}
\begin{proof}
	Since $\mathcal{A}$ is open and closed and contained in $\R$, which is connected, we conclude conclude future integrability for fixed $\ell$, smooth, compactly supported initial data and a smooth and compactly supported \eqref{eq:iedphysicalspaceho}. Therefore, \eqref{eq:iedphysicalspace} and \eqref{eq:iedphysicalspaceho} hold. By a standard density argument, we can then consider more general initial data and inhomogeneities $G_{\widehat{A}}$ for which the norms on the right-hand side of \eqref{eq:iedphysicalspace} or \eqref{eq:iedphysicalspaceho} remain finite.
\end{proof}

By applying the uniformity in $\kappa_+$ and $\kappa_c$ of \eqref{eq:iedphysicalspaceho}, together with the $\kappa_+>0$ and $\kappa_c>0$ result in Corollary \ref{cor:removesuffintkappapos}, we can take $\kappa_+\downarrow 0$ or $\kappa_c\downarrow 0$.

\subsection{$(\Omega^{-1}r)^{p}$-weighted energy estimates in physical space}
\label{sec:physrpest}
In this section, we derive additional, physical-space based $(\Omega^{-1}r)^{p}$-weighted energy estimates in the regions $\{r\leq r_H\}$ and $\{r\geq r_I\}$. The integrated energy estimate in Proposition \ref{prop:iedphysloword} involves $(\Omega^{-1} r)^{p}$ weights on both sides of the equation with an additional $\epsilon$-loss (only for bounded angular frequencies!) when comparing the powers appearing on the LHS. 

In this section, we will show that for $|\q|<\left(\frac{1}{2}\ell+\frac{1}{4}\right)$ or, equivalently, $\re \beta_{\ell}>\sqrt{3}\left(\ell+\frac{1}{2}\right)$, we can remove this $\epsilon$-loss.

\begin{proposition}
	\label{prop:rpestphysspace}
	Let $\psi$ be as in Corollary \ref{cor:iedphysicalspaceho}. Assume that $|\q|<\left(\frac{1}{4}+\frac{1}{2}\ell_0\right)$, or equivalently, $\re \beta_{\ell_0}>\sqrt{3}\left(\ell_0+\frac{1}{2}\right)$. Let $1<p<1+\sqrt{1-\frac{16\q^2}{(2\ell_0+1)^2}}$. Then, there exists a constant $C=C(\delta,p, r_H,r_I)>0$, such that:
	\begin{multline}
\label{eq:iedphysicalspace0noloss}
\sup_{\tau\geq 0}\sum_{k\leq K}\int_{\Sigma_{\tau}} \mathcal{E}_{p}[T^k\psi]\,d\sigma dr\\
+\sum_{k_1+k_2+k_3\leq K}\int_{0}^{\infty}\int_{\Sigma_{\tau}\cap\{r_H\leq r\leq r_I\}} |\snabla_{\s^2}^{k_1}Y_*^{k_2+1}T^{k_3}\widetilde{\psi}|^2+|\snabla_{\s^2}^{k_1}Y_*^{k_2}T^{k_3}\widetilde{\psi}|^2\\
+\upzeta\cdot (|\snabla_{\s^2}^{k_1}Y_*^{k_2}T^{k_3+1}\widetilde{\psi}|^2+|\snabla_{\s^2}^{k_1+1}Y_*^{k_2}T^{k_3}\widetilde{\psi}|^2)\,d\sigma drd\tau\\
+\int_{0}^{\infty}\int_{\Sigma_{\tau}\setminus\{r_H\leq r\leq r_I\}} (\rho_++\kappa_+)(\rho_c+\kappa_c)(r^{-1}\Omega)^{-p}\Omega^2|X\psi|^2\\
+(\rho_++\kappa_+)(\rho_c+\kappa_c)(r^{-1}\Omega)^{\delta}\Omega^{-2}(\mathbf{1}_{r\leq r_H}|K\psi|^2+\mathbf{1}_{r\geq r_I}|K_c\psi|^2)\\
+(\rho_++\kappa_+)(\rho_c+\kappa_c)(r^{-1}\Omega)^{2-p}\Omega^{-2}(|\psi|^2+|\snabla_{\s^2}\psi|^2+|T\psi|^2)\,d\sigma drd\tau\\
\leq C  \sum_{k\leq K}\int_{\Sigma_0} \mathcal{E}_{p}[T^{k}\psi]\,d\sigma dr\\
+C\int_{0}^{\infty}\left[\int_{\Sigma_{\tau}}\max\{(r^{-1}\Omega)^{-p}\rho_+\rho_c,\mathbf{1}_{0\leq \tau\leq 1}\}r^{-2}|r^3T^kG_{\widehat{A}}|^2+\Omega^2(1-\upzeta)|T^{k+1}(rG_{\widehat{A}})|^2\,d\sigma dr \right]\,d\tau.
\end{multline}
\end{proposition}
\begin{proof}
	For the sake of convenience, we will assume that $\widetilde{\h}\equiv 0$ in $\{r\leq r_H\}$ and $\h\equiv 0$ in $\{r\geq r_I\}$. The general case follows by applying additionally the applying local energy estimates from Theorem \ref{thm:localenest}. Without of generality, we consider the region $\{r\geq r_I\}$. The region $\{r\leq r_H\}$ can be treated entirely analogously.
	
By \eqref{eq:maineqradfield}, with $\h\equiv 0$, we obtain:
	\begin{equation*}
rG_{\widehat{A}}= X(\Omega^2X\psi)+r^{-2}\slashed{\Delta}_{\s^2}\psi-2 (T+i\q r^{-1})X\psi-\left[\frac{d\Omega^2}{dr} r^{-1}+\frac{2\Lambda}{3}-i \q r^{-2}\right]\psi.
\end{equation*}
Therefore,
\begin{multline}
\label{eq:mainrpid}
	-\chi_{R_1}\re\left((r^{-1}\Omega)^{-p}\Omega^2\overline{X\psi}\cdot  rG_{\widehat{A}}\right)=-\frac{1}{2}\chi_{R_1}(r^{-1}\Omega)^{-p}X\left(|\Omega^2X\psi|^2\right)+\frac{1}{2}(r^{-1}\Omega)^{2-p}\chi_{R_1}X\left(|\snabla_{\s^2}\psi|^2\right)+\textnormal{div}_{\s^2}(\ldots)\\
	+T\left(\chi_{R_1}(r^{-1}\Omega)^{-p}\Omega^2|X\psi|^2\right)+\chi_{R_1}(r^{-1}\Omega)^{-p}\Omega^2\re\left(\left[\frac{d\Omega^2}{dr} r^{-1}+\frac{2\Lambda}{3}-i \q r^{-2}\right]\psi\overline{X\psi}\right)\\
	=-\frac{1}{2}X\left[\chi_{R_1}(r^{-1}\Omega)^{-p}|\Omega^2X\psi|^2-\chi_{R_1}(r^{-1}\Omega)^{2-p}|\snabla_{\s^2}\psi|^2\right]+T\left(\chi_{R_1}(r^{-1}\Omega)^{-p}\Omega^2|X\psi|^2\right)+\textnormal{div}_{\s^2}(\ldots)\\
	+\frac{p}{2}(r^{-1}\Omega)^{-p}\left[\rho_c(1+O_{\infty}(\rho_c))+\kappa_c(1+O_{\infty}(\rho_c))\right]\chi_{R_1}\Omega^2|X\psi|^2\\
	+\frac{1}{2}(2-p)(r^{-1}\Omega)^{2-p}\left[\rho_c(1+O_{\infty}(\rho_c))+\kappa_c(1+O_{\infty}(\rho_c))\right]\chi_{R_1}|\snabla_{\s^2}\psi|^2\\
	-\q(r^{-1}\Omega)^{2-p}\chi_{R_1}\re\left(i \psi\overline{X\psi}\right)+(r^{-1}\Omega)^{2-p}(O_{\infty}(\rho_c)+\kappa_c O_{\infty}(\rho_c^0))\chi_{R_1}\re\left(\psi\overline{X\psi}\right)\\
	+\frac{1}{2}\frac{d\chi_{R_1}}{dr}(r^{-1}\Omega)^{-p}|\Omega^2X\psi|^2-\frac{1}{2}\frac{d\chi_{R_1}}{dr}(r^{-1}\Omega)^{2-p}|\snabla_{\s^2}\psi|^2,
\end{multline}
where we applied \eqref{lm:metricest} and 
\begin{equation*}
2\Omega^2\left[r^{-1}\frac{d\Omega^2 }{dr}+\frac{2\Lambda}{3}\right]
=r^{-2}\Omega^2 (O_{\infty}(\rho_c)+\kappa_c O_{\infty}(\rho_c^0))
\end{equation*}
to arrive at the final equality.

We apply Young's inequality together with $r^{-2}\Omega^2\leq C \rho_c(\rho_c+\kappa_c)$ to estimate the LHS of \eqref{eq:mainrpid} as follows:
\begin{equation*}
	2\left|(r^{-1}\Omega)^{-p}\Omega^2\overline{X\psi}\cdot  rG_{\widehat{A}})\right|\leq \eta (r^{-1}\Omega)^{-p} (\rho_c+\kappa_c)\Omega^2|X\psi|^2+C\eta^{-1}\rho_c (r^{-1}\Omega)^{-p}r^{-2}|r^3G_{\widehat{A}}|^2.
\end{equation*}

We now integrate both sides of \eqref{eq:mainrpid} in $\{\rho_c\leq 2 \rho_c(R_1)\}\cap\{\tau_1\leq \tau \leq \tau_2\}$ to obtain:
\begin{multline}
\label{eq:intrpestphys}
	\int_{\Sigma_{\tau_2}\cap\{\rho_c\leq 2 \rho_c(R_1)\}}(r^{-1}\Omega)^{-p}|X\psi|^2\,d\sigma dr\\
	+\int_{\tau_1}^{\tau_2}\int_{\Sigma_{\tau}\cap\{\rho_c\leq 2 \rho_c(R_1)\}}p(r^{-1}\Omega)^{-p}\left[(1-\eta)(\rho_c+\kappa_c)(1+O_{\infty}(\rho_c))+\kappa_c O_{\infty}(\rho_c)\right]\chi_{R_1}\Omega^2|X\psi|^2\\
	+(2-p)(r^{-1}\Omega)^{2-p}\left[\rho_c(1+O_{\infty}(\rho_c))+\kappa_c(1+O_{\infty}(\rho_c))\right]\chi_{R_1}|\snabla_{\s^2}\psi|^2\,d\sigma drd\tau\\
	\leq \int_{\Sigma_{\tau_1}\cap\{\rho_c\leq 2 \rho_c(R_1)\}}(r^{-1}\Omega)^{-p}|X\psi|^2\,d\sigma dr+\int_{\tau_1}^{\tau_2}\int_{\Sigma_{\tau}\cap\{\rho_c\leq 2 \rho_c(R_1)\}}2(r^{-1}\Omega)^{2-p}|q||X\psi||\psi|\,d\sigma dr d\tau\\
	+C\int_{\tau_1}^{\tau_2}\int_{\Sigma_{\tau}\cap\{\rho_c\leq 2 \rho_c(R_1)\}}(r^{-1}\Omega)^{2-p}(\rho_c+\kappa_c)|X\psi||\psi|\,d\sigma dr d\tau\\
	+C\int_{\tau_1}^{\tau_2}\int_{\Sigma_{\tau}\cap\{\rho_c(R_1)\leq \rho_c\leq 2 \rho_c(R_1)\}}(r^{-1}\Omega)^{2-p}|\snabla_{\s^2}\psi|^2\,d\sigma dr d\tau\\
	+C\eta^{-1}\int_{\tau_1}^{\tau_2}\int_{\Sigma_{\tau}\cap\{\rho_c\leq 2 \rho_c(R_1)\}}\rho_c (r^{-1}\Omega)^{-p}r^{-2}|r^3G_{\widehat{A}}|^2\,d\sigma dr d\tau .
\end{multline}
We can estimate the RHS of \eqref{eq:intrpestphys} by applying Corollary \ref{cor:iedphysicalspaceho}, which imposes the restriction $p>1$. We will estimate the remaining terms on the RHS of \eqref{eq:intrpestphys} via a suitably Hardy inequality, which will introduce a restriction on $\beta_{\ell}$.

Without loss of generality, assume that $\psi$ is supported on a fixed angular $\ell$-mode. Then we split using Young's inequality:
\begin{multline*}
	2(r^{-1}\Omega)^{2-p}|\q||X\psi||\psi|\leq \mu p(r^{-1}\Omega)^{-p}\rho_c \Omega^2|X\psi|^2\\
	+\left[\frac{\q^2}{p \mu \rho_c}-(1-\eta)(2-p)r^{2}(\rho_c+\kappa_c)\ell(\ell+1)\right]r^{-2}(r^{-1}\Omega)^{2-p}|\psi|^2\\
	+(1-\eta)(2-p)\ell(\ell+1)(\rho_c+\kappa_c)(r^{-1}\Omega)^{2-p}|\psi|^2\\
	\leq  \mu p(r^{-1}\Omega)^{-p}\rho_c \Omega^2|X\psi|^2\\
	+\left[\frac{\q^2-(1-\eta)p(2-p)\ell(\ell+1)}{p \mu \rho_c}\right](r^{-1}\Omega)^{2-p}|\psi|^2\\
	+(1-\eta+O(\kappa_c))(2-p)\ell(\ell+1)(\rho_c+\kappa_c)(r^{-1}\Omega)^{2-p}|\psi|^2.
\end{multline*}
We have that:
\begin{multline*}
	\frac{d}{dr}\left((r^{-2}\Omega^2)^{1-\frac{p}{2}}\right)=-\left(1-\frac{p}{2}\right)r^{-2}\frac{d}{d\rho_c}(r^{-2}\Omega^2)(r^{-1}\Omega)^{-p}=-(2-p)r^{-2}(\rho_c(1+O(\rho_c))\\
	+\kappa_c(1+O(\rho_c))(r^{-1}\Omega)^{-p}\\
	=-(2-p)r^{-2}\rho_c^{-1}(r^{-1}\Omega)^{2-p}-\kappa_c(2-p) r^{-2}(r^{-1}\Omega)^{-p}+r^{-2}\left[\kappa_cO(\rho_c)+O(\rho_c^2)\right](r^{-1}\Omega)^{-p}.
\end{multline*}

We can therefore estimate:
\begin{multline*}
	\int_{r_+}^{r_c}\rho_c^{-1}(r^{-1}\Omega)^{2-p} r^{-2}\chi_{R_1}|\psi|^2\,dr\leq (2-p)^{-1}\int_{r_+}^{r_c}-\frac{d}{dr}\left((r^{-2}\Omega^2)^{1-\frac{p}{2}}\right)\chi_{R_1}|\psi|^2\,dr+\ldots\\
	\leq 2(2-p)^{-1}\int_{r_+}^{r_c}(r^{-1}\Omega)^{2-p}|\psi||X\psi|\,dr +C \int_{\rho_c(R_1)}^{\rho_c(2R_1)}|\psi|^2\,d\rho_c,
\end{multline*}
where the terms in $\dots$ contain an additional factor of $\rho_c+\kappa_c$.

Hence,
\begin{equation*}
	\int_{r_+}^{r_c}\rho_c^{-1}(r^{-1}\Omega)^{2-p} r^{-2}\chi_{R_1}|\psi|^2\,dr\leq 4(2-p)^{-1}\int_{r_+}^{r_c}\rho_c (r^{-1}\Omega)^{-p}\Omega^2|X\psi|^2\,dr +C \int_{\rho_c(R_1)}^{\rho_c(2R_1)}|\psi|^2\,d\rho_c.
\end{equation*}
We therefore need:
\begin{equation*}
	\mu+\mu^{-1}\left[\frac{4\q^2}{p(2-p)}-\ell(\ell+1)\right]<1
\end{equation*}
The LHS is minimized for $\mu=\sqrt{\frac{4q^2}{p(2-p)}-\ell(\ell+1)}$, in which case we are left with:
\begin{equation*}
	\frac{q^2}{p(2-p)}-\frac{1}{4}\ell(\ell+1)<1
\end{equation*}
or equivalently, $p(2-p)>4q^2(\ell+\frac{1}{2})^2$. Since $p(2-p)\in (0,1)$, this requires $|\q|<\frac{1}{2}(\ell+\frac{1}{2})$, which corresponds to:
\begin{equation*}
	\beta_{\ell}=\sqrt{4\left(\ell+\frac{1}{2}\right)^2-4q^2}>\sqrt{3}\left(\ell+\frac{1}{2}\right).
\end{equation*}
Correspondingly, the allowed range for $p$ is:
\begin{equation*}
	1<p<1+\sqrt{1-\frac{16\q^2}{(2\ell+1)^2}}.\qedhere
\end{equation*}
\end{proof}

\subsection{Proof of Theorem \ref{thm:main}}
\label{sec:finishpfmainthm}
Without loss of generality, we consider the electromagnetic gauge $A=\widehat{A}$. Then $D_X=X$, $D_{\s^2}=\snabla_{\s^2}$ and $D_T=T+i\q r^{-1}$ (when acting on functions). Furthermore, $D_{Y_*}\widetilde{\psi}=Y_*\widetilde{\psi}$. We will first remove the conditions $\kappa_+>0$ and $\kappa_c>0$ in Corollary \ref{cor:removesuffintkappapos} by appealing to the uniformity in $\kappa_+$ and $\kappa_c$ of the constants $C$ in \eqref{eq:iedphysicalspace} and \eqref{eq:iedphysicalspaceho}. We assume that $\psi=\psi_{\geq \ell}$.
	
	Let $\kappa_+\geq 0$ and $\kappa_c\geq 0$. Let $\{\kappa_{+,n}\}$ and $\{\kappa_{c,n}\}$ be sequences in $(0,\infty)$ such that $\kappa_{+,n}\to 0$ and $\kappa_{c,n}\to 0$ as $n\to \infty$.
	
	Let $L_{\q}$ be the operator in the proof of Proposition \ref{prop:contargopen} corresponding to the parameters $\kappa_+$ and $\kappa_c$. Denote with $L_{\q,n}$ the operator corresponding instead to the parameters $\kappa_{+,n}$ and $\kappa_{c,n}$.
	
	Let $L_{\q}\psi=rG_{\widehat{A}}$. Then $L_{\q,n}\psi=rG_{\widehat{A}}+(L_{\q,n}-L_{\q})\psi$.
	
	By multiplying $\psi$ with a smooth cut-off function that is compactly supported in $0\leq \tau'\leq \tau+1$ and equal to 1 on $0\leq \tau'\leq \tau$ and applying Corollary \ref{cor:removesuffintkappapos}, we obtain:
\begin{multline}
\label{eq:auxeestiedlimitext}
\sum_{k_1+k_2+k_3\leq K}\int_{0}^{\tau }\int_{\Sigma_{\tau}\cap\{r_H\leq r\leq r_I\}} |\snabla_{\s^2}^{k_1}Y_*^{k_2+1}T^{k_3}\widetilde{\psi}|^2+\upzeta(r)(|\snabla_{\s^2}^{k_1}Y_*^{k_2}T^{k_3+1}\widetilde{\psi}|^2+|\snabla_{\s^2}^{k_1+1}Y_*^{k_2}T^{k_3}\widetilde{\psi}|^2)\\
+|\snabla_{\s^2}^{k_1}Y_*^{k_2}T^{k_3}\widetilde{\psi}|^2\,d\sigma dr d\tau\\
+\sum_{k\leq K}\int_{0}^{\tau}\int_{\Sigma_{\tau}\setminus\{r_H\leq r\leq r_I\}} \rho_+^{2\epsilon}\rho_c^{2\epsilon}(\rho_++\kappa_+)(\rho_c+\kappa_c)(r^{-1}\Omega)^{-p}\Omega^2|XT^k\psi|^2\\
+(\rho_++\kappa_+)(\rho_c+\kappa_c)(r^{-1}\Omega)^{\delta}\Omega^{-2}(\mathbf{1}_{r\leq r_H}|KT^k\psi|^2+\mathbf{1}_{r\geq r_I}|K_cT^k\psi|^2)\\
+\rho_+^{2\epsilon}\rho_c^{2\epsilon}(\rho_++\kappa_+)(\rho_c+\kappa_c)(r^{-1}\Omega)^{2-p}\Omega^{-2}(|T^k\psi|^2+|\snabla_{\s^2}T^k\psi|^2+|T^{k+1}\psi|^2)\,d\sigma drd\tau\\
\leq C\sum_{k\leq K}\Bigg[\int_{\Sigma_0} \mathcal{E}_{p}[T^k\psi]\,d\sigma dr\\
+\int_{0}^{\tau+1}\left[\int_{\Sigma_{\tau'}} \max\{(r^{-1}\Omega)^{-p}\rho_+^{1-2\epsilon}\rho_c^{1-2\epsilon},\mathbf{1}_{0\leq \tau\leq 1}\}r^{-2}|rT^kG_{\widehat{A}}|^2+\Omega^2(1-\upzeta)|T(rT^kG_{\widehat{A}})|^2\,d\sigma dr \right]\,d\tau'\Bigg]\\
+C\sum_{k\leq K}\int_0^{\tau+1}\int_{\Sigma_{\tau'}}\Bigg[\int_{\Sigma_{\tau'}} \max\{(r^{-1}\Omega)^{-p}\rho_+^{1-2\epsilon}\rho_c^{1-2\epsilon},\mathbf{1}_{0\leq \tau\leq 1}\}r^{-2}|T^k(L_{\q,n}-L_{\q})\psi|^2\\
+\Omega^2(1-\upzeta)|T^{k+1}(L_{\q,n}-L_{\q})\psi|^2\,d\sigma dr \Bigg]d\tau'.
\end{multline}

	We can estimate:
	\begin{equation*}
		|(L_{\q,n}-L_{\q})\psi|^2\leq (\kappa_{+,n}-\kappa_+)^2 (\rho_+^2 |X^2\psi|^2+|X\psi|^2+|\psi|^2)+(\kappa_{c,n}-\kappa_c)^2(\rho_c^2 |X^2\psi|^2+|X\psi|^2+|\psi|^2).
	\end{equation*}
	By applying a Gr\"onwall inequality and standard elliptic estimates, there exists a uniform constant $C>0$, such that we can estimate for $1<  p<\min\{2, 1+\re\beta_{\ell}\}+\epsilon$:
	\begin{multline*}
		\sum_{k\leq K}\int_0^{\tau+1}\int_{\Sigma_{\tau'}}\Bigg[\int_{\Sigma_{\tau'}} \max\{(r^{-1}\Omega)^{-p}\rho_+^{1-2\epsilon}\rho_c^{1-2\epsilon},\mathbf{1}_{0\leq \tau\leq 1}\}r^{-2}|T^k(L_{\q,n}-L_{\q})\psi|^2\\
+\Omega^2(1-\upzeta)|T^{k+1}(L_{\q,n}-L_{\q})\psi|^2\,d\sigma dr \Bigg]d\tau'\\
		\leq C \left[(\kappa_{+,n}-\kappa_+)^{2}+(\kappa_{c,n}-\kappa_c)^{2}\right]e^{C\tau}\sum_{k_1+k_2\leq N+2}\int_{\Sigma_0}\mathcal{E}_2[X^{k_1}T^{k_2}\psi]\,d\sigma dr.
	\end{multline*}
	For any fixed $\tau$ and for initially smooth and compactly supported $\psi$, the right-hand side above is finite and vanishes in the limit $n\to \infty$.
		
	Furthermore, as the constant $C$ appearing on the RHS of \eqref{eq:auxeestiedlimitext} does not depend on $n$, we can conclude that \eqref{eq:auxeestiedlimitext} also holds without the third integral on the RHS and we obtain in particular future integrability of $\psi$.
	
	By a standard density argument, we conclude that \eqref{eq:iedphysicalspace} and \eqref{eq:iedphysicalspaceho} also hold \underline{without} the assumption of future integrability in the case of more general initial data and inhomogeneities for which the integrals appearing on the right-hand side of \eqref{eq:iedphysicalspace} and \eqref{eq:iedphysicalspaceho}  are finite.
	
	To include the energy flux term on the right-hand side of \eqref{eq:mainied}, we simply couple \eqref{eq:iedphysicalspaceho} with the physical-space based $(\Omega^{-1}r)^p$-weighted energy estimates of Proposition \ref{prop:rpestphysspace}.

\appendix
\section{A physical-space integrated energy estimate modulo zeroth-order terms}
\label{sec:purelyphysied}
This section provides a purely physical-space-based energy method of deriving integrated energy estimates for \eqref{eq:maineqradfieldtilde} \textbf{with spacetime integrals of zeroth-order terms on the right-hand side} and may be read independently of the remainder of the paper. Furthermore, when restricting to sufficiently large angular frequencies $\ell$, the zeroth-order terms on the right-hand side can be removed to obtain a genuine integrated energy estimate.

The motivation for considering a purely-physical space based method is that, when reformulated in terms of \emph{twisted energy currents}, it generalizes to spacetime backgrounds that are close to Reissner--Nordstr\"om spacetimes, and hence, it may be applied to derive higher-order energy estimates for quasilinear problems, where simply treating nonlinearities as inhomogeneous terms of a linear problem would result in a loss of derivatives. We refer the reader to the discussion and set-up in \cite{dhrt22, dhrt24}.

\textbf{In this section, we will assume that $\kappa_c=0$ for the sake of convenience.} Let $\widetilde{\psi}$ be a solution to \eqref{eq:maineqradfieldtilde}, i.e.\
\begin{multline}
\label{eq:maineqradfieldtilde2}
r\Omega^2 G_{\widetilde{A}}=Y_*^2\widetilde{\psi}- T^2\widetilde{\psi}+r^{-2}\Omega^2\slashed{\Delta}_{\s^2}\widetilde{\psi}-2i \q r^{-1} T\widetilde{\psi}+\left[\q^2r^{-2}- r^{-1}\Omega^2\frac{d\Omega^2}{dr} \right]\widetilde{\psi}\\
=Y_*^2\widetilde{\psi}- K_+^2\widetilde{\psi}+r^{-2}\Omega^2\slashed{\Delta}_{\s^2}\widetilde{\psi}+2i \q r^{-1} K_+\widetilde{\psi}+\left[\q^2\rho_+^2- r^{-1}\Omega^2\frac{d\Omega^2}{dr} \right]\widetilde{\psi}.
\end{multline}
Define:
\begin{equation*}
	V_{\q}(r):=\q^2r^{-2}- r^{-1}\Omega^2\frac{d\Omega^2}{dr}.
\end{equation*}

In the proposition below, we obtain an integrated energy estimate up to a spacetime integral of $|\widetilde{\psi}|^2$, supported away from the photon sphere at $r=r_{\sharp}$.
\begin{proposition}
Let $r^{-1}\psi$ be a solution to \eqref{eq:CSF} with $A=\widehat{A}$ and $\Lambda=0$. Denote with $r^{-1}\widetilde{\psi}$ the associated solution corresponding to $A=\widetilde{A}$, which solves \eqref{eq:maineqradfieldtilde2}.

Then, for any $\delta>0$, there exists a constant $C=C(\h, \q,\delta)>0$ such that for all $0\leq \tau_1<\tau_2<\infty$:
	\begin{multline}
\label{eq:physspacecurrentfull}
\int_{\Sigma_{\tau_2}}\mathcal{E}_0[\psi]\,d\sigma dr +\int_{\mathcal{I}^+\cap \{\tau_1\leq \tau\leq \tau_2\}}|T\psi|^2\,d\sigma d\tau+\int_{\mathcal{H}^+\cap \{\tau_1\leq \tau\leq \tau_2\}}|K_+\psi|^2\,d\sigma d\tau\\
+\int_{\tau_1}^{\tau_2}\int_{\Sigma_{\tau}}(1-r^{-1}r_{\sharp})^2(r^{-3}(\rho_++\kappa_+)|\snabla_{\s^2}\psi|^2+r^{-2}|T\widetilde{\psi}|^2)+(r^{-1}\Omega)^{\delta}r^{-1}\rho_+^{-1}|Y_*\widetilde{\psi}|^2+(1-\upzeta)|\psi|^2\,d\sigma dr d\tau\\
	\leq C\int_{\tau_1}^{\tau_2}\int_{\Sigma_{\tau}}\upzeta r^{-2}|\psi|^2\,d\sigma dr d\tau+C\q^2\int_{\Sigma_{\tau_2}}r^{-2}|\psi|^2\,d\sigma dr+C\int_{\tau_1}^{\tau_2}\int_{\Sigma_{\tau}}(\Omega^{-1}r)^{\delta}r^{-1}\rho_+r^{-2}|r^3G_{\widehat{A}}|^2\,d\sigma dr\\
	+C\left|\int_{\tau_1}^{\tau_2}\int_{\Sigma_{\tau} }\re\left(\overline{T\psi} rG_{\widehat{A}}\right)\,d\sigma dr d\tau\right|.
\end{multline}
Furthermore, there exists an $L\in \N_0$, so that if we assume in addition that $\psi=\psi_{\geq L}$, then the spacetime integral of $\upzeta r^{-2}|\psi|^2$ and the $\Sigma_{\tau_2}$-integral of $\q^2 r^{-2}|\psi|^2$ on the right-hand side can be removed.
\end{proposition}
\begin{proof}
	Let $f:\R_{r_*}\to \R$ and assume that $r_*(r_{\sharp})=0$. We consider:
	\begin{equation*}
		\re\left((f' \widetilde{\psi}+2fY_*\widetilde{\psi})r\Omega^2\overline{G_{\tilde{A}}}\right).
	\end{equation*}
	We can apply Lemma \ref{lm:bulkterms} with $\kappa_c=0$ and with $\widetilde{\psi}$ taking the role of $u$ and $\ell(\ell+1)|u|^2$ replaced by $|\snabla_{\s^2}\widetilde{\psi}|^2$, with the main difference being the appearance of additional total $T$-derivative terms. For this reason, we will focus on the following term:
	\begin{equation*}
		-\re\left((f' \widetilde{\psi}+2fY_*\widetilde{\psi})\overline{(T^2\widetilde{\psi}+2i\q r^{-1}T\widetilde{\psi}})\right).
	\end{equation*}
	Note that
	\begin{equation*}
		-\re\left(f' \widetilde{\psi}\overline{T^2\widetilde{\psi}}\right)=-T\left[\re\left(f' \widetilde{\psi}\overline{T\widetilde{\psi}}\right)\right]+f'|T\widetilde{\psi}|^2
	\end{equation*}
	and that
	\begin{multline*}
		-\re\left(2fY_*\widetilde{\psi}\overline{T^2\widetilde{\psi}}\right)=-T\left[\re\left(2fY_*\widetilde{\psi}\overline{T\widetilde{\psi}}\right)\right]+\re\left(2fY_*T\widetilde{\psi}\overline{T\widetilde{\psi}}\right)\\
		=T\left[\re\left(2fY_*\widetilde{\psi}\overline{T\widetilde{\psi}}\right)\right]+Y_*\left[f|T\widetilde{\psi}|^2\right]-f'|T\widetilde{\psi}|^2.
			\end{multline*}
	Furthermore,
	\begin{multline*}
	4\re(i\q r^{-1}fY_*\widetilde{\psi}\overline{T\widetilde{\psi}})=	4T\left[\re(i\q r^{-1}fY_*\widetilde{\psi}\overline{\widetilde{\psi}})\right]-4\re(i\q r^{-1}fY_*T\widetilde{\psi}\widetilde{\psi})\\ 
	=4T\left[\re(i\q r^{-1}fY_*\widetilde{\psi}\overline{\widetilde{\psi}})\right]-4Y_*\left[\re(i\q r^{-1}fT\widetilde{\psi}\overline{\widetilde{\psi}})\right]+4(fr^{-1})'\re(i\q T\widetilde{\psi}\overline{\widetilde{\psi}})+4\re(i\q r^{-1}fT\widetilde{\psi}Y_*\overline{\widetilde{\psi}}).
	\end{multline*}
Note that the left-hand side is equal to minus the term on the very right-hand side, so we obtain:
\begin{equation*}
	4\re(iqr^{-1}fY_*\widetilde{\psi}\overline{T\widetilde{\psi}})=2T\left[\re(i\q r^{-1}fY_*\widetilde{\psi}\overline{\widetilde{\psi}})\right]-2Y_*\left[\re(i\q r^{-1}fT\widetilde{\psi}\overline{\widetilde{\psi}})\right]+2(fr^{-1})'\re(\widetilde{\psi}\overline{i\q T\widetilde{\psi}}).
\end{equation*}
	Hence, we conclude that:
	\begin{multline*}
		-\re\left((f' \widetilde{\psi}+2fY_*\widetilde{\psi})\overline{(T^2\widetilde{\psi}+2i\q r^{-1}T\widetilde{\psi}})\right)=T\left[\re\left(2fY_*\widetilde{\psi}\overline{T\widetilde{\psi}}\right)-\re\left(f' \widetilde{\psi}\overline{T\widetilde{\psi}}\right)+2\re(i\q r^{-1}fY_*\widetilde{\psi}\overline{\widetilde{\psi}})\right]\\+Y_*\left[f|T\widetilde{\psi}|^2-2\re(i\q r^{-1}fT\widetilde{\psi}\overline{\widetilde{\psi}})\right]-2fr^{-2}\Omega^2\re(\widetilde{\psi}\overline{i\q T\widetilde{\psi}}).
	\end{multline*}
	We are left with the following identity:
	\begin{multline*}
		2f'|Y_*\widetilde{\psi}|^2-f\Omega^2\frac{d}{dr}(\Omega^2r^{-2})|\snabla_{\s^2}\widetilde{\psi}|^2-fV'_{\q}|\widetilde{\psi}|^2-\frac{1}{2}f'''|\widetilde{\psi}|^2-2fr^{-2}\Omega^2\re(\widetilde{\psi}\overline{i\q T\widetilde{\psi}})\\
		=T\left[\re\left(2fY_*\widetilde{\psi}\overline{T\widetilde{\psi}}\right)-\re\left(f' \widetilde{\psi}\overline{T\widetilde{\psi}}\right)+2\re(i\q r^{-1}fY_*\widetilde{\psi}\overline{\widetilde{\psi}})\right]+{\rm div}_{\s^2}(\ldots)\\
		+Y_*\left[f|T\widetilde{\psi}|^2-2\re(i\q r^{-1}fT\widetilde{\psi}\overline{\widetilde{\psi}})+f|Y_*\widetilde{\psi}|^2-f\Omega^2r^{-2}|\snabla_{\s^2}\widetilde{\psi}|^2-fV_q|\widetilde{\psi}|^2+f'\re(Y_*\widetilde{\psi}\overline{\widetilde{\psi}})-\frac{1}{2}f''|\widetilde{\psi}|^2\right]\\
		-\re\left((f' \widetilde{\psi}+2fY_*\widetilde{\psi})r\Omega^2\overline{G_{\tilde{A}}}\right).
	\end{multline*}
Let $f(r_*)=-1+(r^{-2}\Omega^2)^{\frac{\delta}{2}}$ for $r_*\leq r_*(r_3)$, $f(r_*)=1-(r^{-2}\Omega^2)^{\frac{\delta}{2}}$ for $r_*\geq r_*(R_3)$ and $f(r_*)=\frac{2}{\pi} \arctan(\alpha r_*)$ for $r_*(r_2)\leq r_*\leq r_*(R_2)$. Then, 
\begin{equation*}
	\Omega^{-2}f'(r_*)=\begin{cases}
		\frac{\delta}{2} r^{-2}(r^{-2}\Omega^2)^{\frac{\delta}{2}-1}\left[\rho_+(1+O_{\infty}(\rho_+))+\kappa_+(1+O_{\infty}(\rho_+))\right]\quad &(r\leq r_3),\\
		\frac{\delta}{2} r^{-1-\delta}(1+O_{\infty}(r^{-1}))\quad &(r\geq R_3)
	\end{cases}
\end{equation*}
and for $r_*(r_2)\leq r_*\leq r_*(R_2)$:
\begin{align*}
f'(r_*)=&\: \frac{2\alpha}{\pi(1+\alpha^2r_*^2)},\\
f''(r_*)=&\: -\frac{4\alpha^3r_*	}{\pi(1+\alpha^2r_*^2)^2},\\
f'''(r_*)=&\: -\frac{4\alpha^3(1-3\alpha^2r_*^2)}{\pi(1+\alpha^2r_*^2)^3}.
\end{align*}
We can choose $f$ in $(r_3,r_2)\cup (R_2,R_3)$ such that $f$ is a $C^3$ function, $f'>0$ globally, and there exists a constant $C>0$ such that globally:
\begin{align*}
|f'''(r_*)|\leq &\: C |r_*|^{-3}.
\end{align*}
See also Proposition \ref{prop:srhighfreqest}, where the same $f$ appears.

Note that for $|r-r_{\sharp}|\leq 2\eta$:
\begin{equation*}
	|\arctan(\alpha r_*)|\leq \alpha |r_*|\leq 2\alpha \Omega^{-2}(r_{\sharp})|r-r_{\sharp}|,
\end{equation*}
if we take $\eta$ suitably small.

We can therefore apply Young's inequality to estimate for $|r-r_{\sharp}|\leq \eta$:
\begin{multline*}
	|2fr^{-2}\Omega^2\re(\widetilde{\psi}\overline{i\q T\widetilde{\psi}})|\leq \epsilon \arctan^2(\alpha r_*)r^{-2}\Omega^2|T\widetilde{\psi}|^2+\epsilon^{-1}\q^2\Omega^2r^{-2}|\widetilde{\psi}|^2\\
	\leq \epsilon \alpha^2 r^{-2}(1-r^{-1}r_{\sharp})^2 \Omega^2 |T\widetilde{\psi}|^2+C\q^2\epsilon^{-1}\Omega^2r^{-2}|\widetilde{\psi}|^2.
\end{multline*}
We obtain for $|r_*|\leq \frac{1}{2}\alpha$ and $r_*(r_2)\leq r_*\leq r_*(R_2)$:
\begin{equation*}
	-\frac{1}{2}f'''|\psi|^2\geq \frac{3}{2\pi (\frac{5}{4})^3}\alpha^3|\psi|^2
\end{equation*}
Taking $\eta>0$ suitably small and $\alpha^3\gg \epsilon^{-1}$, we therefore obtain for $|r-r_{\sharp}|\leq \eta$:
\begin{equation*}
	|2fr^{-2}\Omega^2\re(\widetilde{\psi}\overline{i\q T\widetilde{\psi}})|\leq \epsilon \alpha^2 r^{-2}(1-r^{-1}r_{\sharp})^2 \Omega^2 |T\widetilde{\psi}|^2-\frac{3}{4\pi (\frac{5}{4})^3}\alpha^3|\psi|^2.
\end{equation*}

Furthermore, there exists a constant $c>0$, such that:
\begin{equation*}
	-f\frac{d}{dr}(\Omega^2r^{-2})|\snabla_{\s^2}\widetilde{\psi}|^2\geq c r^{-3}(1-r^{-1}r_{\sharp})^2(\rho_++\kappa_+)|\snabla_{\s^2}\widetilde{\psi}|^2.
\end{equation*}
Hence, taking $\epsilon=\alpha^{-\frac{5}{2}}$, we have that $\alpha^3\gg \epsilon^{-1}\gg \alpha^{2}$ and we can integrate in $\tau_1\leq \tau \leq \tau_2$ to obtain:
\begin{multline}
\label{eq:physspacefcurrent}
	c\int_{\tau_1}^{\tau_2}\int_{\Sigma_{\tau}}(1-r^{-1}r_{\sharp})^2(\rho_++\kappa_+)r^{-3}|\snabla_{\s^2}\psi|^2+(1-\upzeta)|\psi|^2+(r^{-1}\Omega)^{\delta}r^{-1}\rho_+^{-1}|Y_*\widetilde{\psi}|^2\,d\sigma dr d\tau\\
	\leq C\int_{\tau_1}^{\tau_2}\int_{\Sigma_{\tau}}\upzeta r^{-2}|\psi|^2+ \alpha^{-\frac{1}{2}} r^{-2}(1-r^{-1}r_{\sharp})^2 |T\widetilde{\psi}|^2\,d\sigma dr d\tau\\
	+2\int_{\mathcal{H}^+\cap \{\tau_1\leq \tau\leq \tau_2\}}|K_+\widetilde{\psi}|^2\,d\sigma d\tau+2\int_{\mathcal{I}^+\cap \{\tau_1\leq \tau\leq \tau_2\}}|T\widetilde{\psi}|^2\,d\sigma d\tau\\
	 +C\sum_{i=1}^2\int_{\Sigma_{\tau_i}}\mathcal{E}_0[\psi]\,d\sigma dr d\sigma+C\int_{\tau_1}^{\tau_2}\int_{\Sigma_{\tau}}|rG_{\widehat{A}}|(|{Y_*}\widetilde{\psi}|+|\widetilde{\psi}|)\,d\sigma dr.
\end{multline}

We can add control over $T$-derivatives by introducing a function $y$ and considering the vector field multiplier $-2y\overline{Y_*\widetilde{\psi}}$. 
By the computations above involving $T\widetilde{\psi}$ together with Lemma \ref{lm:bulkterms}, we obtain:
\begin{multline*}
	y'(|Y_*\widetilde{\psi}|^2+|T\widetilde{\psi}|^2)-\Omega^2\frac{d}{dr}(y \Omega^2r^{-2})|\snabla_{\s^2}\widetilde{\psi}|^2-(yV_0)'|\widetilde{\psi}|^2-2(fr^{-1})'\re(\widetilde{\psi}\overline{iqT\widetilde{\psi}})\\
	=T\left[\re\left(2fY_*\widetilde{\psi}\overline{T\widetilde{\psi}}\right)+2\re(i\q r^{-1}fY_*\widetilde{\psi}\overline{\widetilde{\psi}})\right]\\
	+Y_*\left[f|T\widetilde{\psi}|^2-2\re(i\q r^{-1}fT\widetilde{\psi}\overline{\widetilde{\psi}})+y|Y_*\widetilde{\psi}|^2-\Omega^2r^{-2}y |\snabla_{\s^2}\widetilde{\psi}|^2-V_0|\widetilde{\psi}|^2\right]-2\re\left(y Y_*\widetilde{\psi}\cdot r\Omega^2\overline{G_{\tilde{A}}}\right).
\end{multline*}
Now take $y(r)=(1-r^{-1}r_{\sharp})^3$ and combine the above estimate with \eqref{eq:physspacefcurrent}, with $\alpha$ suitably large to obtain:
\begin{multline}
\label{eq:physspacecurrent}
	c\int_{\tau_1}^{\tau_2}\int_{\Sigma_{\tau}}(1-r^{-1}r_{\sharp})^2((\rho_++\kappa_+)r^{-3}|\snabla_{\s^2}\psi|^2+r^{-2}|T\widetilde{\psi}|^2)+(r^{-1}\Omega)^{\delta}r^{-1}\rho_+^{-1}|Y_*\widetilde{\psi}|^2+(1-\upzeta)|\psi|^2\,d\sigma dr d\tau\\
	\leq C\int_{\tau_1}^{\tau_2}\int_{\Sigma_{\tau}}\upzeta r^{-2}|\psi|^2\,d\sigma dr d\tau+2\int_{\mathcal{H}^+\cap \{\tau_1\leq \tau\leq \tau_2\}}|K_+\widetilde{\psi}|^2\,d\sigma d\tau+2\int_{\mathcal{I}^+\cap \{\tau_1\leq \tau\leq \tau_2\}}|T\widetilde{\psi}|^2\,d\sigma d\tau\\
	 +C\sum_{i=1}^2\int_{\Sigma_{\tau_i}}\mathcal{E}_0[\psi]\,d\sigma dr d\sigma+C\int_{\tau_1}^{\tau_2}\int_{\Sigma_{\tau}}|rG_{\widehat{A}}|(|{Y_*}\widetilde{\psi}|+(r^{-1}\Omega)^{\delta}r^{-1}\rho_+^{-1}|\widetilde{\psi}|)\,d\sigma dr.
\end{multline}
We are left with estimating the energy fluxes on the RHS. 

We evaluate:
\begin{equation*}
	-\re\left(( \chi_{r_1}\overline{K_+\widetilde{\psi}}+ (1-\chi_{r_1})\overline{T\widetilde{\psi}}) rG_{\widetilde{A}}\right)
\end{equation*}
to obtain:
\begin{multline*}
	\frac{1}{2}T\left[(1-\chi_{r_1})|T\widetilde{\psi}|^2+\chi_{r_1}|K_+\widetilde{\psi}|^2+|Y_*\widetilde{\psi}|^2+r^{-2}\Omega^2|\snabla_{\s^2}\widetilde{\psi}|^2-\q^2((1-\chi_{r_1})r^{-2}+\chi_{r_1}\rho_+^2)|\widetilde{\psi}|^2+r^{-1}\Omega^2\frac{d\Omega^2}{dr}|\widetilde{\psi}|^2\right]\\
-Y_*\left[\chi_{r_1}\re(\overline{K_+\widetilde{\psi}}\cdot Y_*\widetilde{\psi})\right]-Y_*\left[(1-\chi_{r_1})\re(\overline{T\widetilde{\psi}}\cdot Y_*\widetilde{\psi})\right]\\
	=-\re\left(( \chi_{r_1}\overline{K_+\widetilde{\psi}}+ (1-\chi_{r_1})\overline{T\widetilde{\psi}}) rG_{\widetilde{A}}\right)+\frac{d\chi_{r_1}}{dr}\re(\overline{i \q \widetilde{\psi}}\cdot Y_*\widetilde{\psi}).
\end{multline*}
After integration, we obtain:
\begin{multline*}
	\int_{\mathcal{I}^+\cap \{\tau_1\leq \tau\leq \tau_2\}}|T\psi|^2\,d\sigma d\tau+\int_{\mathcal{H}^+\cap \{\tau_1\leq \tau\leq \tau_2\}}|K_+\psi|^2\,d\sigma d\tau+\int_{\Sigma_{\tau_2}}\mathcal{E}_0[\psi]\,d\sigma dr \leq C \int_{\Sigma_{\tau_1}}\mathcal{E}_0[\psi]\,d\sigma dr\\
+\int_{\tau_1}^{\tau_2}\int_{\Sigma_{\tau}\cap \supp \frac{d\chi_{r_1}}{dr} } \epsilon |Y_*\widetilde{\psi}|^2+C\q^2\epsilon^{-1}|\widetilde{\psi}|^2\,d\sigma dr d\tau+C\q^2\int_{\Sigma_{\tau_2}}r^{-2}|\psi|^2\,d\sigma dr\\
+C\left|\int_{\tau_1}^{\tau_2}\int_{\Sigma_{\tau} }\re\left(( \chi_{r_1}\overline{K_+\widetilde{\psi}}+ (1-\chi_{r_1})\overline{T\widetilde{\psi}}) rG_{\widetilde{A}}\right)\,d\sigma dr d\tau\right|.
\end{multline*}
 By combining the above estimate with \eqref{eq:physspacecurrent}, we obtain:
	\begin{multline*}
\sup_{\tau\in [\tau_1,\tau_2]}\int_{\Sigma_{\tau}}\mathcal{E}_0[\psi]\,d\sigma dr +\int_{\mathcal{I}^+\cap \{\tau_1\leq \tau\leq \tau_2\}}|T\psi|^2\,d\sigma d\tau+\int_{\mathcal{H}^+\cap \{\tau_1\leq \tau\leq \tau_2\}}|K_+\psi|^2\,d\sigma d\tau\\
+\int_{\tau_1}^{\tau_2}\int_{\Sigma_{\tau}}(1-r^{-1}r_{\sharp})^2(r^{-3}(\rho_++\kappa_+)|\snabla_{\s^2}\psi|^2+r^{-2}|T\widetilde{\psi}|^2)+(r^{-1}\Omega)^{\delta}r^{-1}\rho_+^{-1}|Y_*\widetilde{\psi}|^2+(1-\upzeta=|\psi|^2\,d\sigma dr d\tau\\
	\leq C\int_{\tau_1}^{\tau_2}\int_{\Sigma_{\tau}}\upzeta r^{-2}|\psi|^2\,d\sigma dr d\tau+C\q^2\int_{\Sigma_{\tau_2}}r^{-2}|\psi|^2\,d\sigma dr+C\int_{\tau_1}^{\tau_2}\int_{\Sigma_{\tau}}|rG_{\widehat{A}}|(|{Y_*}\widetilde{\psi}|+(r^{-1}\Omega)^{\delta}r^{-1}\rho_+^{-1}|\widetilde{\psi}|)\,d\sigma dr\\
	+C\left|\int_{\tau_1}^{\tau_2}\int_{\Sigma_{\tau} }\re\left(( \chi_{r_1}\overline{K_+\widetilde{\psi}}+ (1-\chi_{r_1})\overline{T\widetilde{\psi}}) rG_{\widetilde{A}}\right)\,d\sigma dr d\tau\right|.
\end{multline*}

We conclude the proof of \eqref{eq:physspacecurrent} by applying Young's inequality to estimate $|rG_{\widehat{A}}|(|{Y_*}\widetilde{\psi}|+(r^{-1}\Omega)^{\delta}r^{-1}\rho_+^{-1}|\widetilde{\psi}|)$. Note that we can absorb the spacetime integral of $\upzeta |\psi|^2$ into the spacetime integral of $|\snabla_{\s^2}\psi|^2$ on the left-hand if we restrict to $\psi=\psi_{\geq L}$ with $L$ suitably large and we apply a Poincar\'e inequality on $\s^2$. Similarly, we can absorb the $\q^2 r^{-2}|\psi|^2$ flux term into the $r^{-2}|\snabla_{\s^2}\psi|^2$ term of $\mathcal{E}_0[\psi]$ when $\psi=\psi_{\geq L}$.
\end{proof}

We can further improve the above integrated estimates by coupling them to the $(\Omega^{-1}r)^p$-estimates from Proposition \ref{prop:rpestphysspace} modulo zeroth-order terms:

\begin{corollary}
\label{cor:iedmoduleo0th}
Let $\psi$ be a solution to \eqref{eq:maineqradfield}. Let $0< p<2$ and $N\in \N_0$. Let $r_+<r_H<r_I<\infty$.

Then there exists a constant $C=C(\h, p, q,r_H,r_I)>0$ such that for all $0\leq \tau_1<\tau_2<\infty$:
\begin{multline}
\label{eq:iledmod0thorder}
\sum_{k\leq N}\int_{\Sigma_{\tau_2}}\mathcal{E}_p[T^k\psi]\,d\sigma dr d\sigma+\int_{\mathcal{I}^+\cap \{\tau_1\leq \tau\leq \tau_2\}}|T^{k+1}\psi|^2\,d\sigma d\tau+\int_{\mathcal{H}^+\cap \{\tau_1\leq \tau\leq \tau_2\}}|K_+T^k\psi|^2\,d\sigma d\tau\\
+\int_{\tau_1}^{\tau_2}\int_{\Sigma_{\tau}\cap \{r\notin (r_H,r_I)\}}\mathcal{E}_{p-1}[T^k\psi]\,d\sigma dr d\tau\\
+\sum_{k_1+k_2+k_3\leq N}\int_{\tau_1}^{\tau_2}\int_{\Sigma_{\tau}\cap \{r\in (r_H,r_I)\}}(1-r^{-1}r_{\sharp})^2(|\snabla_{\s^2}^{k_1+1}T^{k_2}Y_*^{k_3}\psi|^2+|\snabla_{\s^2}^{k_1}T^{k_2+1}Y_*^{k_3}\widetilde{\psi}|^2)\\
+|\snabla_{\s^2}^{k_1}T^{k_2}Y_*^{k_3+1}\widetilde{\psi}|^2+|\snabla_{\s^2}^{k_1}T^{k_2}Y_*^{k_3}\psi|^2\,d\sigma dr d\tau\\
\leq C\int_{\tau_1}^{\tau_2}\int_{\Sigma_{\tau}}\max\{r^{-1}\rho_+(r^{-1}\Omega)^{2-p}\Omega^{-2},r^{-2}\}\upzeta |\psi|^2\,dr+C\q^2\int_{\Sigma_{\tau_2}}r^{-2}|\psi|^2\,d\sigma dr\\
+C\sum_{k\leq N}\int_{\Sigma_{\tau_1}}\mathcal{E}_p[T^k\psi]\,d\sigma dr d\sigma+C\int_{\tau_1}^{\tau_2}\int_{\Sigma_{\tau}}\rho_+ r^{-1} (\Omega^{-1}r)^pr^{-2}|T^k(r^3G_{\widehat{A}})|^2\,d\sigma dr d\tau\\
+C\left|\int_{\tau_1}^{\tau_2}\int_{\Sigma_{\tau} }\re\left(\overline{T^{k+1}\psi} rT^kG_{\widehat{A}}\right)\,d\sigma dr d\tau\right|.
\end{multline}
Furthermore, there exists an $L\in \N_0$, so that if we assume in addition that $\psi=\psi_{\geq L}$, then the spacetime integral of $\upzeta r^{-2}|\psi|^2$ and the $\Sigma_{\tau_2}$-integral of $\q^2 r^{-2}|\psi|^2$ on the right-hand side can be removed.
\end{corollary}

\section{Restricted mode stability for sub-extremal Reissner--Nordstr\"om away from extremality}
\label{sec:restrmodestab}
  In this section, we establish an additional mode stability result on sub-extremal Reissner--Nordstr\"om for a restricted set of frequencies. This result is independent from the rest of the paper and is not involved in the remaining propositions. The aim of this section is to elucidate the extent to which standard methods for proving mode stability apply.
 \begin{theorem}
\label{thm:modestabsubext}
Let $\kappa_c=0$, $\kappa_+>0$ and $(\omega,\ell)\in\mathcal{F}_{\flat,\sim }\cup \mathcal{F}_{\flat, +}$ with $\q\omega\notin (0,\frac{\q^2}{r_++r_-})$. Then there exists a constant $K_{\kappa_+}=K_{\kappa_+}(\kappa_+,\gamma,L_0)>0$ such that:
\begin{equation*}
|\mathfrak{W}|\gtrsim K_{\kappa_+}.
\end{equation*}
\end{theorem}
\begin{proof}
	It is straightforward to show that the quantity $\widetilde{R}(r)=e^{-\gamma r}(r-r_+)^{-\xi}(r-r_-)^{-\eta}R$, with $R=r^{-1}u$, satisfies the equation:
\begin{equation}
\label{eq:conflheun1}
\mathcal{T}_{r} (\widetilde{R})=e^{-\gamma r}(r-r_+)^{-\xi}(r-r_-)^{-\eta} r H,
\end{equation}
with $\mathcal{T}_r$ a confluent Heun operator, defined as follows:
\begin{multline*}
\mathcal{T}_{r} :=(r-r_+)(r-r_-)\frac{d^2}{dr^2}+\left[(2\eta+1)(r-r_+)+(2\xi+1)(r-r_-)+2\gamma (r-r_+)(r-r_-)\right]\frac{d}{dr}\\
+(2\gamma r-L)\mathbf{1},
\end{multline*}
where
\begin{align*}
\eta=&\:i\frac{\omega r_-^2-\q r_-}{r_+-r_-},\\
\xi=&-i\frac{\omega r_+^2-\q r_+}{r_+-r_-},\\
\gamma=&-i\omega,\\
L=&\:\ell(\ell+1)-i\q.
\end{align*}

In Kerr spacetimes with $a^2<M^2$, $\widetilde{R}=e^{-\gamma r}(r-r_+)^{-\xi}(r-r_-)^{-\eta}R$ also satisfies \eqref{eq:conflheun1}, but with $(\eta,\xi,\gamma,L)$ replaced by $(\eta_{\rm Kerr},\xi_{\rm Kerr},\gamma_{\rm Kerr},L_{\rm Kerr})$, where
\begin{align*}
\eta_{\rm Kerr}=&\: i \frac{\omega (r_-^2+a^2)-am}{r_+-r_-},\\
\xi_{\rm Kerr}=&- i \frac{\omega (r_+^2+a^2)-am}{r_+-r_-},\\
\gamma_{\rm Kerr}=&-i\omega,\\
L_{\rm Kerr}=&\:\lambda+a^2\omega^2-2a m\omega,
\end{align*}
see \cite{costa20}[Eq. (3.22)].

Note that in Ker $r_{\pm}=M\pm \sqrt{M^2-a^2}$, whereas in Reissner--Nordstr\" om $r_{\pm}=M\pm \sqrt{M^2-Q^2}$. With the translations $a\leftrightarrow Q$ and $(\eta_{\rm Kerr},\xi_{\rm Kerr}, L_{\rm Kerr}) \leftrightarrow (\eta, \xi, L)$, the confluent Heun operators are therefore are the same, so we can directly apply methods from the sub-extremal Kerr setting.

In particular, it follows from \cite{costa20}[Proposition 3.8] that
\begin{equation*}
\mathbb{U}(x):=\lim_{y\to 0}(x^2+Q^2)^{\frac{1}{2}}(x-r_+)^{\xi+\eta}e^{\gamma x}\int_{r_+}^{\infty}e^{-2\gamma (r_+-r_-)^{-1} (x+iy-r_-)(r-r_-)}(r-r_-)^{\eta}(r-r_+)^{\xi}e^{-\gamma r}R(r)\,dr
\end{equation*}
is well-defined as a limit in $L^2_x([r_+,\infty))$ and is a smooth solution to the ODE
\begin{equation*}
\frac{d^2\mathbb{U}}{dx_*^2}+\mathbb{V}\mathbb{U}=\frac{(x-r_-)(x-r_+)}{x^2+Q^2}\mathbb{H},
\end{equation*}
with $\frac{dx_*}{dx}=\frac{x^2+Q^2}{(x-r_-)(x-r_+)}$,
\begin{equation*}
\mathbb{H}(x)=\lim_{y\to 0}(x^2+Q^2)^{\frac{1}{2}}(x-r_+)^{\xi+\eta}e^{\gamma x}\int_{r_+}^{\infty}e^{-2\gamma (r_+-r_-)^{-1} (x+iy-r_-)(r-r_-)}(r-r_-)^{\eta}(r-r_+)^{\xi}e^{-\gamma r}(r\Omega^{-2}H)\,dr
\end{equation*}
and for $M=1$:
\begin{multline*}
\mathbb{V}(x)=\frac{x-r_-}{(x^2+Q^2)^2}\Big[(r_+-r_-)^{-1}\omega^2((r_+-r_-)x+r_-^2+3Q^2)(x-r_-)^2-\ell(\ell+1)(x-r_+)+\q^2(x-r_-)\\
-4\q \omega(r_+-r_-)^{-1} (x-r_-)^2-(x^2+Q^2)^{-2}(x-r_+)(2x^3+Q^2x^2-4Q^2x+Q^4)\Big],
\end{multline*}
which can be derived analogously to the Kerr case in the proof of \cite{costa20}[Proposition 3.8]. Note that
\begin{align*}
\mathbb{V}(\infty)=&\:\omega^2,\\
\mathbb{V}(r_+)=&\:\frac{(r_+-r_-)^2}{4r_+^2}\left((r_++r_-)\omega-\q\right)^2.
\end{align*}
In comparison, in the Kerr case we have:
\begin{align*}
\mathbb{V}(\infty)=&\:\omega^2,\\
\mathbb{V}(r_+)=&\:\frac{(r_+-r_-)^2}{4r_+^2}\left(2\omega\right)^2.
\end{align*}

Furthermore,
\begin{align*}
\left[\frac{d\mathbb{U}}{dx_*}-i\omega \mathbb{U}\right](\infty)=&\: 0,\\
\left[\frac{d\mathbb{U}}{dx_*}+i\left(\omega-\frac{\q}{r_++r_-}\right)\frac{r_+-r_-}{r_+} \mathbb{U}\right](-\infty)=&\:0,\\
|\mathbb{U}|^2(\infty)=&\: |\Gamma(2\xi+1)|^2\left|\frac{r_+-r_-}{2\omega}\right|^2r_+^{-2}|u(r_+)|^2.
\end{align*}
Defining $ \mathbb{J}^T[\tilde{u}]:=-\omega\re(i \frac{d\mathbb{U}}{dx_*}\overline{\mathbb{U}})$, and using that $\mathbb{V}$ is real, we obtain 
\begin{align*}
\frac{d}{dx_*}\left( \mathbb{J}^T[\mathbb{U}]\right)= -\omega\frac{(x-r_-)(x-r_+)}{x^2+Q^2}\re(\overline{i \mathbb{H}} \mathbb{U}).
\end{align*}
By applying the above boundary conditions on $\tilde{u}$, we moreover obtain for $\omega \tomega <0$:
\begin{align*}
\mathbb{J}^T[\mathbb{U}](\infty)=&\:\frac{1}{2}\left(\left|\frac{d\mathbb{U}}{dx_*}\right|^2(\infty)+\omega^2|\mathbb{U}|^2(\infty)\right),\\
\mathbb{J}^T[\mathbb{U}](-\infty)=&\:\frac{1}{2}\left(\left|\frac{d\mathbb{U}}{dx_*}\right|^2(-\infty)+\frac{r_+-r_-}{r_+}\omega \left(\omega-\frac{\q}{r_++r_-}\right)|\mathbb{U}|^2(-\infty)\right),
\end{align*}
so, for $\omega \left(\omega-\frac{\q}{r_++r_-}\right)>0$, we obtain:
\begin{equation*}
|\mathbb{J}^T[\mathbb{U}](-\infty)|+|\mathbb{J}^T[\mathbb{U}](\infty)|\leq\left| \omega \int_{\R_{x_*}}\frac{(x-r_-)(x-r_+)}{x^2+Q^2}\re(\overline{i \mathbb{H}} \mathbb{U})\,dx_*\right|.
\end{equation*}

The proof of the Wronskian estimate in the $|a|<M$ case in \cite{costa20}[Theorem 5.1] then still applies when $\omega  \left(\omega-\frac{\q}{r_++r_-}\right)>0$.

Note that in the sub-extremal Kerr setting, the factor  $\omega-\frac{\q}{r_++r_-}$ above is instead simply $\omega$, which ensures mode stability everywhere away from $\omega=0$.
\end{proof}

\section{ODE error estimates}
\label{sec:ODEest}
In Proposition \ref{prop:genoderrorest}, we will derive general error estimates for solutions to Schr\"odinger ODEs with perturbed potentials.

The following lemma concerning Volterra integrals forms the main ingredient for these error estimates and is proved in \cite{olv97}[Theorem 6.10.2].

\begin{lemma}[{\cite{olv97}[Theorem 6.10.2]}]
\label{lm:volterraint}
Let $x\in (x_1,x_2)$, with $-\infty\leq x_1<x_2\leq \infty$ and consider the integral equation:
\begin{equation}
\label{eq:volterraint}
\varepsilon(x)=\int_{x_1}^{x}K(x,y)\vartheta(y)\left[W(y)+\varepsilon(y)\right]\,dy,
\end{equation}
with $K,W$ continuous complex-valued functions, such that:
\begin{align*}
K(x,x)=&\:0,\\
|K(x,y)|\leq &\: P_0(x)Q(y)\quad y\in (x_1,x),\\
\left|\frac{\partial K}{\partial x}(x,y)\right|\leq &\: P_1(x)Q(y)\quad y\in (x_1,x),\\
\left|\frac{\partial^2 K}{\partial x^2}(x,y)\right|\leq &\: P_2(x)Q(y)\quad y\in (x_1,x)
\end{align*}
for some continuous functions $Q$ and $P_{\mu}$, $\mu\in\{0,1,2\}$.

Assume moreover that:
\begin{align*}
\int_{x_1}^{x_2}|\vartheta(y)|\,dy < \infty,\\
\sup_{y\in (x_1,x_2)} Q(y)|W(y)|<\infty,\\
\sup_{y\in (x_1,x_2)} P_0(y)|Q(y)|<\infty.
\end{align*}
Then $\varepsilon$ satisfies the following estimates: for all $x\in (x_1,x_2)$
\begin{align*}
\frac{\varepsilon(x)}{P_0(x)}\leq \frac{\sup_{y\in (x_1,x)} Q(y)|W(y)|}{\sup_{y\in (x_1,x)} P_0(y)Q(y)}\left(e^{\sup_{y\in (x_1,x)} P_0(y)Q(y)\int_{x_1}^{x_2}|\vartheta(y)|\,dy}-1\right),\\
\frac{\frac{d \varepsilon}{dx}(x)}{P_1(x)}\leq \frac{\sup_{y\in (x_1,x)} Q(y)|W(y)|}{\sup_{y\in (x_1,x)} P_0(y)Q(y)}\left(e^{\sup_{y\in (x_1,x)} P_0(y)Q(y)\int_{x_1}^{x_2}|\vartheta(y)|\,dy}-1\right).
\end{align*}
\end{lemma}

\begin{proposition}
\label{prop:genoderrorest}
Consider the following differential equation in the interval $(x_1,x_2)$ with $-\infty\leq x_1<x_2\leq \infty$:
\begin{equation}
\label{eq:genhomode}
U''(x)=\left[\mathcal{V}_0(x)+\vartheta(x)\right]U(x),
\end{equation}
with $\mathcal{V}_0$ and $\vartheta$ smooth functions. Assume that there exist solutions $W_1$ and $W_2$ to \eqref{eq:genhomode} without the $\vartheta$ term, such that their Wronskian $\mathcal{W}(W_1,W_2)$ is non-vanishing.

Let $f: (x_1,x_2)\to (0,\infty)$ and assume moreover that there exists a constant $C>0$ and functions\\ $P_0,P_1,Q: (x_1,x_2)\to [0,\infty)$ such that:
\begin{align*}
|f(x)W_1(x)f(y)W_2(y)|+|f(y)W_1(y)f(x)W_2(x)|\leq &\: C P_0(x)Q(y)\quad y\in (x_1,x_2),\\
\left|(fW_1)'(x)(fW_2)(y)\right|+\left|(fW_1)(y)(fW_2)'(x)\right|\leq &\: C P_1(x)Q(y)\quad y\in (x_1,x_2),\\
\left|f'(x)W_1'(x)f(y)W_2(y)\right|+\left|f(y)W_1(y)f'(x)W_2'(x)\right|\leq &\: C P_2(x)Q(y)\quad y\in (x_1,x_2),\\
\sup_{x\in (x_1,x_2)}(|f(x)W_1|(x)+|f(x)W_2|(x))Q(x)<&\: \infty,\\
\sup_{x\in (x_1,x_2)} P_0(x)Q(x)<&\: \infty,\\
\int_{x_1}^{x_2}\frac{|\vartheta|(y)}{f^2(y)}\,dy<&\:\infty.
\end{align*}
Then, there exist constants $C_1,C_2\in \C$ such that
\begin{equation*}
f U=C_1(fW_1+f\varepsilon_1)+C_2(fW_2+f\varepsilon_2),
\end{equation*}
where $\varepsilon_i$, $i=1,2$, satisfy for all $x\in (x_1,x_2)$: there exists a constant $C>0$ such that
\begin{align*}
\left|f(x)\varepsilon_i\right|(x)\leq&\: CP_0(x)\frac{\sup_{y\in (x_1,x)}|fW_i|(y)Q(y)}{\sup_{y\in (x_1,x)} P_0(y)Q(y)}\left(e^{\frac{C}{|\mathcal{W}(W_1,W_2)|}\sup_{y\in (x_1,x)} P_0(y)Q(y) \int_{x_1}^{x}\frac{|\vartheta(y)|}{f^2(y)}\,dy}-1\right),\\
\frac{\left|\frac{d (f\varepsilon_i)}{dx}\right|(x)}{P_1(x)}\leq&\: C\frac{\sup_{y\in (x_1,x)}|fW_i|(y)Q(y)}{\sup_{y\in (x_1,x)} P_0(y)Q(y)}\left(e^{\frac{C}{|\mathcal{W}(W_1,W_2)|}\sup_{y\in (x_1,x)} P_0(y)Q(y) \int_{x_1}^{x}\frac{|\vartheta(y)|}{f^2(y)}\,dy}-1\right).
\end{align*}
The above estimates hold also with the interval $(x_1,x)$ in the suprema and the integrals replaced with $(x,x_2)$.
\end{proposition}
\begin{proof}
We write $U(x)=C_1(W_1(x)+\varepsilon_1(x))+C_2(W_2(x)+\varepsilon_2(x))$, with $C_{\pm}\in \C$. Then $U$ is a solution to \eqref{eq:genhomode} if:
\begin{align}
\label{eq:whitterror1}
\varepsilon_1''-\mathcal{V}_0\varepsilon_1=&\:\vartheta (\varepsilon_1+W_1),\\
\label{eq:whitterror2}
\varepsilon_2''-\mathcal{V}_0\varepsilon_2=&\:\vartheta (\varepsilon_2+W_2).
\end{align}
Consider the following Wronskian: $\mathcal{W}(W_1,W_2)=W_1 W_2'-W_1'W_2$, which is moreover constant in $x$. Then, by Green's formula (variation of parameters), \eqref{eq:whitterror1} and  \eqref{eq:whitterror2} are satisfied and $\varepsilon_i(x_1)=0$ and:
\begin{align*}
\varepsilon_1(x)=&\: \frac{1}{\mathcal{W}(W_1,W_2)}\int_{x_1}^{x}(W_2(x)W_1(y)-W_1(x)W_2(y))\vartheta(y)(\varepsilon_1(y)+W_1(y))\,dy,\\
\varepsilon_2(x)=&\: \frac{1}{\mathcal{W}(W_1,W_2)}\int_{x_1}^{x}(W_1(x)W_2(y)-W_2(x)W_1(y))\vartheta(y)(\varepsilon_2(y)+W_2(y))\,dy.
\end{align*}
We can insert the function $f$ in the above integral equalities in the following way:
\begin{align*}
f(x)\varepsilon_1(x)=&\:\int_{x_1}^{x}K(x,y)\frac{\vartheta(y)}{f^2(y)}(f(y)\varepsilon_1(y)+f(y)W_1(y))\,dy,\\
f(x)\varepsilon_2(x)=&\:\int_{x_1}^{x}(-K(x,y)))\frac{\vartheta(y)}{f^2(y)}(f(y)\varepsilon_2(y)+f(y)W_2(y))\,dy,
\end{align*}
with $K(x,y)=\frac{1}{\mathcal{W}(W_1,W_2)}[f(x)W_2(x)f(y)W_1(y)-f(x)W_{1}(x)f(y)W_{2}(y)]$.

Note that $K(x,x)=0$. By the assumptions on $W_1$ and $W_2$, we can estimate
\begin{align*}
|\mathcal{W}(W_1,W_2)||K(x,y)|\leq &\: |f(x)W_1(x)f(y)W_2(y)|+|f(y)W_1(y)f(x)W_2(x)|\leq C P_0(x)Q(y)\quad y\in (x_1,x),\\
|\mathcal{W}(W_1,W_2)|\left|\frac{\partial K}{\partial x}(x,y)\right|\leq &\: \left|(fW_1)'(x)f(y)W_2(y)\right|+\left|f(y)W_1(y)(fW_2)'(x)\right|\leq C P_1(x)Q(y)\quad y\in (x_1,x)
\end{align*}
Furthermore,
\begin{multline*}
|\mathcal{W}(W_1,W_2)|\left|\frac{\partial^2 K}{\partial x^2}(x,y)\right|\leq  |\mathcal{V}_0(x)| (|f(x)W_1(x)f(y)W_2(y)|+|f(y)W_1(y)f(x)W_2(x)|)\\
+ \left|(f'(x)W_1'(x)f(y)W_2(y)\right|+\left|f(y)W_1(y)f'(x)W_2'(x)\right|\\
\leq  (|\mathcal{V}_0(x)|P_0(x)+P_2(x))Q(y)\quad y\in (x_1,x).
\end{multline*}
We can now apply Lemma \ref{lm:volterraint} to conclude the estimates in the proposition. One may repeat the above argument with the interval $(x_1,x)$ replaced by $(x,x_2)$ and with $\varepsilon_i(x_2)=0$.
\end{proof}

\begin{lemma}
	\label{lm:propwhitt}
	Consider the Whittaker equation
	\begin{equation}
\label{eq:whitt}
\frac{d^2U}{d\xi^2}=\left[-\frac{1}{4}-\sigma \xi^{-1}+\left(\mu^2-\frac{1}{4}\right)\xi^{-2}\right] U,
\end{equation}
with either $\xi\in (0,\infty)$ or $\xi\in (-\infty,0)$ and $\sigma\in \R$, $\mu\in [0,\infty)\times i(0,\infty)$.

Then there exist solutions $W_{-i\sigma,\mu}(i \xi)$ and $M_{-i\sigma,\mu}(i \xi)$ to \eqref{eq:whitt} (Whittaker functions), with principal branches corresponding to the range $\arg z\in (-\pi,\pi]$ that satisfy the following properties:
\begin{align}
\label{eq:whittasymp1}
W_{ -i\sigma,\mu}(i \xi)=&\: e^{\sign(\xi)\frac{\pi}{2}\sigma }e^{-  i\frac{\xi}{2}} |\xi|^{ -i\sigma}(1+O(|\xi|^{-1}))\quad |\xi|>1,\\
\label{eq:whittasymp2}
M_{-i \sigma,\mu}( i\xi)=&\: \frac{\Gamma(1+2\mu)}{\Gamma(\frac{1}{2}+\mu-i\sigma)}e^{-\sign(\xi)\pi(\sigma-i\mu-\frac{i}{2})}W_{ -i\sigma,\mu}(i \xi)+ \frac{\Gamma(1+2\mu)}{\Gamma(\frac{1}{2}+\mu+i\sigma)}e^{-\sign(\xi)\pi\sigma}W_{ i\sigma,\mu}(-i \xi),\\
\label{eq:whittasymp3}
M_{-i \sigma,\mu}(i \xi)=&\: e^{\sign(\xi)i\frac{\pi}{2}(\frac{1}{2}+\mu)}|\xi|^{\frac{1}{2}+\mu}(1+O(|\xi|^{-1})) \quad |\xi|\leq1,\\
\label{eq:whittasymp4}
W_{ -i\sigma,\mu}( i\xi)=&\: \frac{\Gamma(-2\mu)}{\Gamma(\frac{1}{2}-\mu +i\sigma)}M_{ -i\sigma,\mu}(i \xi)+\frac{\Gamma(2\mu)}{\Gamma(\frac{1}{2}+\mu +i \sigma)}M_{-i \sigma,-\mu}(i \xi)\quad (2\mu\notin \N_0),\\
\label{eq:whittasymp5}
W_{ -i\sigma,\mu}( i\xi)=&\: \frac{\Gamma(2\mu)}{\Gamma(\frac{1}{2}+\mu +i\sigma)}M_{ -i\sigma,-\mu}(i \xi)(1+O(|\xi|^{-1})) \quad |\xi|\leq1, \quad (\re 2\mu>1),\\
\label{eq:whittasymp6}
W_{ -i\sigma,\frac{1}{2}}( i\xi)=&\: \frac{1}{\Gamma(1 +i\sigma)}M_{ \sigma,-\frac{1}{2}}(i \xi)(1+\log |\xi|O(|\xi|^{-1}))\quad |\xi|\leq1,\\
\label{eq:whittasymp7}
W_{ -i\sigma,0}(i \xi)=&\: -\frac{1}{\Gamma(\frac{1}{2}+i\sigma)}M_{ -i\sigma,0}( i\xi)\left(\log \xi+\frac{\Gamma'(\frac{1}{2}+i\sigma)}{\Gamma(\frac{1}{2}+i\sigma)}+2\gamma_{\rm Euler}+\log |\xi| O(|\xi|)\right) \quad |\xi|\leq1,\\
\label{eq:whittwronskian}
\mathcal{W}&\left(W_{  -i\sigma,\mu}( i \xi), M_{  -i\sigma,\mu}( i \xi)\right)=i \frac{\Gamma(1+2\mu)}{\Gamma(\frac{1}{2}+\mu+i\sigma)}.
\end{align}
with $\Gamma$ the Gamma function and $\gamma_{\rm Euler}$ the Euler--Mascheroni constant. Here the constants in the Big-O notation can depend on $\sigma$ and $\mu$.
\end{lemma}
\begin{proof}
	We apply standard properties of Whittaker functions, see for example \cite{NIST:DLMF}[\S 13.14] with the constant $\kappa$ replaced by $-i\sigma$.
	\end{proof}
\begin{corollary}
\label{cor:errorwhitt}
Let $U$ be a solution to
\begin{equation}
\label{eq:whitterror}
\frac{d^2U}{d\xi^2}=\left[-\frac{1}{4}-\sigma \xi^{-1}+\left(\mu^2-\frac{1}{4}\right)\xi^{-2}+\vartheta(\xi)\right] U,
\end{equation}
with either $\xi\in (0,\infty)$ or $\xi\in (-\infty,0)$ and $\sigma\in \R$, $\mu\in [0,\infty)\cup i(0,\infty)$. Let $0\leq \xi_0\leq \infty$. Assume that:
 \begin{equation*}
 \int_{|\xi|}^{ \xi_0} (\eta^{-2}+1)^{-\frac{1}{2}}|\vartheta(\eta)|\,d|\eta|<\infty. \end{equation*}
 Let
 \begin{align*}
 P_0(\xi):=&\: (\xi^{-2}+1)^{\frac{1}{2}\re \mu}.
\end{align*}

Then there exist complex constants $B_1,B_2\in \C$ and a constant $C>0$ that depends on $\sigma$ and $\mu$, such that we can write:
\begin{equation*}
U(\xi)=B_1(W_{-i\sigma,\mu}(i\xi)+\varepsilon_1(\xi))+B_2(M_{-i\sigma,\mu}(i \xi)+\varepsilon_2(\xi)),
\end{equation*}
with $\varepsilon_i(\xi)$, with $i\in \{1,2\}$, satisfying the following bounds:\begin{align*}
\left|\varepsilon_i\right|(\xi)\leq&\: C(\xi^{-2}+1)^{-\frac{1}{4}} (1+\delta_{\mu0}\log(2+\xi^{-2}))P_0(\xi)\left(e^{C \int^{\xi_0}_{|\xi|} (1+\delta_{\mu0} \log^2(2+\eta^{-2}))(\eta^{-2}+1)^{-\frac{1}{2}}|\vartheta(\eta)|\,d|\eta|}-1\right),\\
\left|\frac{d \varepsilon_i}{d\xi}\right|(\xi)\leq&\: C(\xi^{-2}+1)^{\frac{1}{4}}(1+\delta_{\mu0}\log(2+\xi^{-2}))P_0(\xi)   \left(e^{C \int^{\xi_0}_{|\xi|}(1+\delta_{\mu0} \log^2(2+\eta^{-2})) (\eta^{-2}+1)^{-\frac{1}{2}}|\vartheta(\eta)|\,d|\eta|}-1\right).
\end{align*}
\end{corollary}
\begin{proof}
Let $W_1(\xi):=W_{-i\sigma,\mu}(i \xi)$ and $W_2(\xi):=M_{-i\sigma,\mu}(i \xi)$. Let $f(\xi):=(\xi^{-2}+1)^{\frac{1}{4}}$ if $\mu\neq 0$ and $f(\xi):=(\xi^{-2}+1)^{\frac{1}{4}}(\log(2+\xi^{-2}))^{-1}$ if $\mu=0$.

We define $P_i$, with $i\in \{0,1\}$, and $Q$ as follows:
\begin{align*}
Q(\xi):=&\:(\xi^{-2}+1)^{-\frac{1}{2}\re \mu} \quad (\mu\neq 0)\\
Q(\xi):=&\:(\xi^{-2}+1)^{-\frac{1}{2}\re \mu}(\log(2+\xi^{-2}))^{-1}\quad (\mu= 0),\\
P_0(\xi):=&\: (\xi^{-2}+1)^{\frac{1}{2}\re \mu},\\
P_1(\xi):=&\: (\xi^{-2}+1)^{\frac{1}{2}}P_0(\xi).
\end{align*}

It follows from the above asymptotic properties of $W_1$ and $W_2$ from Lemma \ref{lm:propwhitt} that:
\begin{align*}
|(fW_1)(\xi)(fW_2)(\eta)|+|(fW_1)(\eta)(fW_2)(\xi)|\leq &\: C P_0(\xi) Q(\eta)\quad |\eta|\in (|\xi|,\infty),\\
\left|(fW_1)'(\xi)(fW_2)(\eta)\right|+\left|(fW_1)(\eta)(fW_2)'(\xi)\right|\leq &\: C P_1(\xi)Q(\eta)\quad |\eta|\in (|\xi|,\infty),\\
\left|(f'W_1)(\xi)(f'W_2)(\eta)\right|+\left|(fW_1)(\eta)(f'W_2)(\xi)\right|\leq &\: C P_1(\xi)Q(\eta)\quad |\eta|\in (|\xi|,\infty),\\
\sup_{|\xi|\in (|\xi_0|,\infty)}[|fW_1|(\xi)+|fW_2|(\xi)]Q(\xi)<&\: \infty \quad \textnormal{for any $\xi_0>0$},\\
\sup_{|\xi|\in (|\xi_0|,\infty)} P_0(\xi)Q(\xi)<&\:\infty\quad \textnormal{for any $\xi_0>0$},
\end{align*}
Hence, the conditions in Proposition \ref{prop:genoderrorest} with $x=|\xi|$ and $x_1=\infty$ are satisfied.

Note moreover that there exists a constant $C>0$, such that:
\begin{align*}
\sup_{|\eta|\in (|\xi|,\infty)} P_0(\eta)Q(\eta)= &\:\begin{cases}	
1\quad  &\mu\neq 0,\\
(\log(2+\xi^{-2}))^{-1} &\mu= 0,
\end{cases}\\
\sup_{|\eta|\in (|\xi|,\infty)} |fW_2(\xi)|Q(\eta)\leq &\: C,\\
\sup_{|\eta|\in (|\xi|,\infty)} |fW_1(\xi)|Q(\eta)\leq &\: C
\end{align*}
We can therefore apply Proposition \ref{prop:genoderrorest} to conclude the proof.
\end{proof}

\begin{lemma}
\label{lm:whittnozero}
The Whittaker functions $W_{\mp i\sigma,\mu}(\pm i \xi)$ have no zeroes for $\xi\in \R$.
\end{lemma}
\begin{proof}
Let $W(\xi)=W_{   -i\sigma,\mu}( i\xi)$. Then, using that $W$ solve a Schr\"odinger equation with a real potential, we obtain:
\begin{equation*}
\frac{d}{d\xi}\re\left(i W \frac{d\overline{W}}{d\xi}\right)=\re\left(i W\overline{W}''\right)+\re\left(i \left|\frac{d\overline{W}}{d\xi}\right|^2\right)=0
\end{equation*}
Hence, we have that for all $\xi_1,\xi_2\in \R$:
\begin{equation*}
\re\left(i W \frac{d\overline{W}}{d\xi}\right)(\xi_1)=\re\left(i W \frac{d\overline{W}}{d\xi}\right)(\xi_2)
\end{equation*}
In particular, 
\begin{multline*}
\lim_{\xi\to \infty}\re\left(i W_{ -i\sigma,\mu}(- i\xi) \frac{d\overline{W}_{ -i\sigma,\mu}(- i\xi) }{d\xi}\right)\\
=\lim_{\xi\to \infty}\re\left( i e^{\sign(\xi)\frac{\pi}{2}\sigma }e^{-  i\frac{\xi}{2}} |\xi|^{-i\sigma}(1+O(|\xi|^{-1}))\cdot  \frac{i}{2}e^{-i\sign(\xi)\frac{\pi}{2}\overline{\sigma} }e^{  i\frac{\xi}{2}} |\xi|^{ i\overline{\sigma}}(1+O(|\xi|^{-1}))\right)=-\frac{1}{2}e^{\sign(\xi) \pi \re  \sigma}.
\end{multline*}
Hence, $W_{ -i\sigma,\mu}(- i\xi) \neq 0$ for any $\xi\in \R$. The same argument also applies to $W_{ +i\sigma,\mu}( -i\xi) $.
\end{proof}

 \begin{lemma}
 	\label{prop:hypgeomfundsoln}
 	Consider the hypergeometric equation
 	\begin{equation}
 	\label{eq:hypgeom}
\zeta(1-\zeta)\frac{d^2F}{d\zeta^2}+(c-(a+b+1)\zeta)\frac{dF}{d\zeta}-a b F=0,
\end{equation}
with $a,b,c\in \C$ and $\zeta\in (0,1)$.

Define and assume that there exist constants $\gamma,\mu_0>0$, such that:
\begin{align*}
	\mu:=&\:\frac{1}{2}(c-1)\in [0,\mu_0)\cup i (0,\mu_0),\\
	-iw:=&\:a+b-c\in i(-\gamma,\gamma),\\
	v:=&\:ab-\frac{1}{2}c(1+a+b-c),\\
	i\sigma_a:=&\:a-\frac{c}{2}\in i\R,\\
	i\sigma_b:=&\:b-\frac{c}{2}\in i\R.
\end{align*}
We can express:
\begin{align*}
	\sigma_a=&-\frac{w}{2}\pm \sqrt{\frac{1}{4}-\mu^2+v+\frac{1}{4}w^2},\\
	\sigma_b=&-\frac{w}{2}\mp \sqrt{\frac{1}{4}-\mu^2+v+\frac{1}{4}w^2}.
\end{align*}

Let $\tilde{s}:=\frac{\zeta}{1-\zeta}$ and define $U(\tilde{s})=\zeta^{\frac{c}{2}}(\tilde{s})(1-\zeta(\tilde{s}))^{\frac{1}{2}(a+b-c-1)}F(\zeta(s))$.
Then $U$ satisfies the equation
\begin{equation}
\label{eq:hypgeomalt}
	\frac{d^2U}{d\tilde{s}^2}=(1+\tilde{s})^{-2}\left[-\frac{1}{4}(1+w^2)+v \tilde{s}^{-1}+\left(\mu^2-\frac{1}{4}\right)\tilde{s}^{-2}\right]U.
\end{equation}
Furthermore, the following functions are solutions to \eqref{eq:hypgeomalt}:
\begin{align}
\label{eq:hypgeomdefG}
G_{\underline{\sigma},\mu}(\tilde{s}):=&\:\zeta^{\frac{c}{2}}(\tilde{s})(1-\zeta(\tilde{s}))^{\frac{1}{2}(a+b-c-1)} {}_2F_1(a,b,c;\zeta(\tilde{s})),\\
\label{eq:hypgeomdefF1}
	F_{\underline{\sigma},\mu}(\tilde{s}):=&\:\zeta^{\frac{c}{2}}(\tilde{s})(1-\zeta(\tilde{s}))^{\frac{1}{2}(a+b-c-1)} {}_2F_1(a,b,a+b+1-c;1-\zeta(\tilde{s})).
\end{align}
where ${}_2F_1(a,b,a+b+1-c;1-\zeta)$ and ${}_2F_1(a,b,c;\zeta)$ are Gauss hypergeometric functions, which solve \eqref{eq:hypgeom} and satisfy:
\begin{align}
\label{eq:asymphypgeomsmallarg1}
	{}_2F_1(a,b,a+b+1-c;1-\zeta)=&\:\sum_{k=0}^{\infty}\frac{\Gamma(a+k)\Gamma(b+k)\Gamma(a+b+1-c+k)}{\Gamma(a)\Gamma(b)\Gamma(a+b+1-c)}(1-\zeta)^k=1+O_{\infty}(1-\zeta),\\
	\label{eq:asymphypgeomsmallarg2}
	{}_2F_1(a,b,c;\zeta)=&\:\sum_{k=0}^{\infty}\frac{\Gamma(a+k)\Gamma(b+k)\Gamma(c+k)}{\Gamma(a)\Gamma(b)\Gamma(c)}\zeta^k=1+O_{\infty}(\zeta).
\end{align}

There exists a constant $C=C(\mu_0,\gamma)>0$, such that the following Wronskian bounds hold:
\begin{equation*}
	\left|\mathcal{W}(F_{\underline{\sigma},\mu},G_{\underline{\sigma},\mu})\right|\geq C_{\gamma,\mu_0}^{-1}
\end{equation*}
and we can write:
\begin{align}
\label{eq:relhypgeomfunct1}
	F_{\underline{\sigma},\mu}(\tilde{s}(\zeta))=&\:\Gamma\left(1+i(\sigma_a+\sigma_b)\right)\Bigg[\frac{\Gamma(-2\mu)}{\Gamma(\frac{1}{2}-\mu+i\sigma_a)\Gamma(\frac{1}{2}-\mu+i\sigma_b)}G_{\underline{\sigma},\mu}(\tilde{s}(\zeta))\\ \nonumber
	+&\:\frac{\Gamma(2\mu)}{\Gamma(\frac{1}{2}+\mu+i\sigma_a)\Gamma(\frac{1}{2}+\mu+i\sigma_b)}G_{\underline{\sigma},-\mu}(\tilde{s}(\zeta))\Bigg] \quad (2\mu\notin \N_0),\\
	\label{eq:relhypgeomfunct2}
	F_{\underline{\sigma},\mu}(\tilde{s}(\zeta))=&\:\Gamma\left(1+i(\sigma_a+\sigma_b)\right)\Bigg[-\frac{(-1)^{-2\mu}}{(2\mu)!}\log \zeta\\ \nonumber
	+&\:\frac{\Gamma(2\mu)}{\Gamma(\frac{1}{2}+\mu+i\sigma_a)\Gamma(\frac{1}{2}+\mu+i\sigma_b)}G_{\underline{\sigma},-\mu}(\tilde{s}(\zeta))(1+O_{\infty}(\zeta))\Bigg] \quad (2\mu\in \N_1),\\
	\label{eq:relhypgeomfunct3}
	F_{\underline{\sigma},0}(\tilde{s}(\zeta))=&\:-\frac{\Gamma\left(1+i(\sigma_a+\sigma_b)\right)}{\Gamma(\frac{1}{2}+i\sigma_a)\Gamma(\frac{1}{2}+i\sigma_b)}G_{0,\sigma}(\tilde{s}(\zeta))\\ \nonumber
\times&\: \Bigg[\log \zeta+\left(2\gamma_{\rm Euler}+\frac{\Gamma'(a)}{\Gamma(a)}+\frac{\Gamma'(b)}{\Gamma(b)}\right)+\log \zeta O_{\infty}(\zeta)\Bigg],\\
\label{eq:relhypgeomfunct0}
G_{\underline{\sigma},\mu}(\tilde{s}(\zeta))=&\:\frac{\Gamma(1+2\mu) \Gamma(iw)}{\Gamma(\frac{1}{2}+\mu-i\sigma_a)\Gamma(\frac{1}{2}+\mu-i\sigma_b)}F_{\underline{\sigma},\mu}(\tilde{s}(\zeta))\\ \nonumber
+&\:\frac{\Gamma(1+2\mu) \Gamma(-iw)}{\Gamma(\frac{1}{2}+\mu+i\sigma_a)\Gamma(\frac{1}{2}+\mu+i\sigma_b)}F_{-\underline{\sigma},\mu}(\tilde{s}(\zeta))\quad (w\neq 0),\\
\label{eq:relhypgeomfunct0b}
G_{\underline{\sigma},\mu}(\tilde{s}(\zeta))=&\:-\frac{\Gamma(1+2\mu)}{\Gamma(\frac{1}{2}+\mu+i\sigma_a)\Gamma(\frac{1}{2}+\mu+i\sigma_b)}\zeta^{\mu+\frac{1}{2}}(1-\zeta)^{-\frac{1}{2}}\\ \nonumber
\times &\: \left(\log(1-\zeta)
 	+2\gamma_{\rm Euler}+\frac{\Gamma'(a)}{\Gamma(a)}+\frac{\Gamma'(b)}{\Gamma(b)}+O_{\infty}((1-\zeta))\right)\quad (w= 0).
\end{align}
\end{lemma}
\begin{proof}
We will derive the equivalence of \eqref{eq:hypgeom} and \eqref{eq:hypgeomalt} by starting with \eqref{eq:hypgeomalt} and expressing $a,b,c$ in terms of $\mu,w,v$.

We can express: $\tilde{s}=\frac{\zeta}{1-\zeta}$ and $\frac{d}{d\tilde{s}}=(1+\tilde{s})^{-2}\frac{d}{d\zeta}=(1-\zeta)^2\frac{d}{d\zeta}$. We therefore obtain:
\begin{equation*}
(1+\tilde{s})^2\frac{d^2 U}{d\tilde{s}^2}=\frac{d}{d\zeta}\left((1-\zeta)^2\frac{dU}{d\zeta}\right)=(1-\zeta)^2\frac{d^2 U}{d\zeta^2}-2(1-\zeta)\frac{d U}{d\zeta}.
\end{equation*}
Therefore, \eqref{eq:hypgeomalt} is equivalent to:
\begin{equation}
\label{eq:mainsubextzeta}
0=\zeta(1-\zeta)\frac{d^2 U}{d\zeta^2}-2\zeta\frac{d U}{d\zeta}+\left[\frac{1}{4}(1+w^2)\frac{\zeta}{1-\zeta}-v -\left(\mu^2-\frac{1}{4}\right)\left(\frac{1-\zeta}{\zeta}\right)\right] U.
\end{equation}
Let $-iw=a+b-c$. Then $U(\tilde{s})=\zeta^{\frac{1}{2}+\mu}(\tilde{s})(1-\zeta(\tilde{s}))^{\frac{1}{2}(-1-iw)}F(\zeta(\tilde{s}))$. Then we can write in terms of $F(\zeta)$:
\begin{equation*}
\frac{dU}{d\zeta}=\zeta^{\frac{1}{2}+\mu}(1-\zeta)^{-\frac{1}{2}-\frac{iw}{2}}\frac{dF}{d\zeta}+\left[\left(\frac{1}{2}+\mu\right)\zeta^{-1}+\left(\frac{1}{2}+\frac{iw}{2}\right)(1-\zeta)^{-1}\right]\zeta^{\frac{1}{2}+\mu}(1-\zeta)^{-\frac{1}{2}-\frac{iw}{2}}F
\end{equation*}
so
\begin{multline*}
\zeta^{-\frac{1}{2}-\mu}(1-\zeta)^{\frac{1}{2}+\frac{iw}{2}}\frac{d^2U}{d\zeta^2}=\frac{d^2F}{d\zeta^2}+\left[\left(1+2\mu\right)\zeta^{-1}+\left(1+iw\right)(1-\zeta)^{-1}\right]\frac{dF}{d\zeta}\\
+\left[\left(\mu^2-\frac{1}{4}\right)\zeta^{-2}+\left(\frac{1}{2}+\mu\right)\left(1+iw\right)\zeta^{-1}(1-\zeta)^{-1}+\left(\frac{1}{2}+\frac{iw}{2}\right)\left(\frac{3}{2}+\frac{iw}{2}\right)(1-\zeta)^{-2}\right]F.
\end{multline*}
Therefore, \eqref{eq:mainsubextzeta} is equivalent to:
\begin{multline*}
0=\zeta(1-\zeta)\frac{d^2 F}{d\zeta^2}+\left[\left(1+2\mu\right)(1-\zeta)+\left(-1+iw\right)\zeta\right]\frac{dF}{d\zeta}-2\left[\left(\frac{1}{2}+\mu\right)+\left(\frac{1}{2}+\frac{iw}{2}\right)\zeta(1-\zeta)^{-1}\right]F\\
+\left[\left(\mu^2-\frac{1}{4}\right)\zeta^{-1}(1-\zeta)+\left(\frac{1}{2}+\mu\right)\left(1+iw\right)+\left(\frac{1}{2}+\frac{iw}{2}\right)\left(\frac{3}{2}+\frac{iw}{2}\right)\zeta (1-\zeta)^{-1}\right]F\\
+\left[\frac{1}{4}(1+w^2)\zeta (1-\zeta)^{-1}-v -\left(\mu^2-\frac{1}{4}\right)\left(\frac{1-\zeta}{\zeta}\right)\right] F\\
=\zeta(1-\zeta)\frac{d^2 F}{d\zeta^2}+\left[-\left(2+2\mu-iw\right)\zeta+\left(1+2\mu\right)\right]\frac{dF}{d\zeta}-\left[\left(\frac{1}{2}+\mu\right)\left(1-iw\right)+v\right]F.
\end{multline*}
To obtain equivalence of \eqref{eq:hypgeomalt} and \eqref{eq:hypgeomalt}, we must have that
\begin{align*}
	c=&\:1+2\mu,\\
	a+b+1=&\:2+2\mu-iw=c+1-iw,\\
	ab=&\:\left(\frac{1}{2}+\mu\right)(1-iw)+v.
\end{align*}
Then
we also obtain
\begin{equation*}
	v=ab-\frac{c}{2}(a+b+1-c).
\end{equation*}

Write
\begin{align*}
a=&\:\frac{1}{2}+\mu -\frac{iw}{2}+i\lambda,\\
b=&\:\frac{1}{2}+\mu -\frac{iw}{2}-i\lambda.
\end{align*}
Then we also obtain:
\begin{equation*}
ab=\frac{1}{4}\left(1+2\mu -iw\right)^2+\lambda^2=\mu^2-\frac{1}{4}+\frac{1}{4}\left(iw\right)^2+\left(\frac{1}{2}+\mu\right)\left(1-iw\right)+\lambda^2,
\end{equation*}
from which it follows that
\begin{equation*}
\lambda^2=\frac{1}{4}-\mu^2+v+\frac{1}{4}w^2
\end{equation*}
and $\lambda=\pm \sqrt{\frac{1}{4}-\mu^2+v+\frac{1}{4}w^2}$. 

If we write $a=\frac{c}{2}+i\sigma_a$ and $b=\frac{c}{2}+i\sigma_b$, then
\begin{align*}
	i\sigma_a=&-\frac{iw}{2}\pm i\sqrt{\frac{1}{4}-\mu^2+v+\frac{1}{4}w^2},\\
	i\sigma_b=&-\frac{iw}{2}\mp i\sqrt{\frac{1}{4}-\mu^2+v+\frac{1}{4}w^2}.
\end{align*}

	Let $2\mu\notin \N_0$, which is equivalent to $c\notin \N_1$. Then:
 \begin{multline*}
 	\mathcal{W}(F_{\underline{\sigma},\mu},G_{\underline{\sigma},\mu})=\zeta^{c}(\tilde{s})(1-\zeta(\tilde{s}))^{a+b-c-1} \frac{d\zeta}{d\tilde{s}} \mathcal{W}\left({}_2F_1(a,b,a+b+1-c;1-\zeta), {}_2F_1(a,b,c;\zeta)\right)\\
 	=\zeta^{c}(\tilde{s})(1-\zeta(\tilde{s}))^{a+b-c+1}\mathcal{W}\left({}_2F_1(a,b,a+b+1-c;1-\zeta), {}_2F_1(a,b,c;\zeta)\right).
 \end{multline*}
 Furthermore, by \cite{NIST:DLMF}[\S15.10(ii)] we obtain
 \begin{multline}
 \label{eq:hypgeomconnectionform}
\mbox{}_2F_1(a,b,a+b+1-c;1-\zeta)=\frac{\Gamma(1-c)\Gamma(a+b-c+1)}{\Gamma(a-c+1)\Gamma(b-c+1)}\mbox{}_2F_1(a,b,c;\zeta)\\
+\frac{\Gamma(c-1)\Gamma(a+b-c+1)}{\Gamma(a)\Gamma(b)}\zeta^{1-c}\mbox{}_2F_1(a-c+1,b-c+1,2-c;\zeta)
\end{multline}
and by \cite{NIST:DLMF}[\S15.10(i)]:
\begin{equation*}
	\mathcal{W}\left( {}_2F_1(a,b,c,\zeta),\zeta^{1-c}\mbox{}_2F_1(a-c+1,b-c+1,2-c;\zeta)\right)=(1-c)\zeta^{-c}(1-\zeta)^{c-a-b-1}.
\end{equation*}
Combining the above identities, we obtain:
\begin{equation*}
	\mathcal{W}(F_{\underline{\sigma},\mu},G_{\underline{\sigma},\mu})=(c-1)\frac{\Gamma(c-1)\Gamma(a+b-c+1)}{\Gamma(a)\Gamma(b)}= \frac{\Gamma(1+2\mu) \Gamma(1+i(\sigma_a+\sigma_b))}{\Gamma(\frac{1}{2}+\mu+i\sigma_a)\Gamma(\frac{1}{2}+\mu+i\sigma_b)}.
\end{equation*}
Gamma functions have no zeroes and the product $\Gamma(\frac{1}{2}+\mu+i\sigma_a)\Gamma(\frac{1}{2}+\mu+i\sigma_b)$ has no poles. We conclude that there exists a uniform constant $C_{\gamma}>0$, such that
\begin{equation*}
	\left|\mathcal{W}(F_{\underline{\sigma},\mu},G_{\underline{\sigma},\mu})\right|\geq C_{\gamma,\mu_0}^{-1}.
\end{equation*}

Note that $G_{\underline{\sigma},-\mu}(\tilde{s})=\zeta^{1-\frac{c}{2}}(\tilde{s})(1-\zeta(\tilde{s}))^{\frac{1}{2}(a+b-c-1)} {}_2F_1(a-c+1,b-c+1,2-c;\zeta(s))$, so \eqref{eq:hypgeomconnectionform} implies that for $2\mu \notin \N_0$:
\begin{equation}
\label{eq:relFG}
F_{\underline{\sigma},\mu}(\tilde{s})=\Gamma\left(1+i(\sigma_a+\sigma_b)\right)\left[\frac{\Gamma(-2\mu)}{\Gamma(\frac{1}{2}-\mu+i\sigma_a)\Gamma(\frac{1}{2}-\mu+i\sigma_b)}G_{\underline{\sigma},\mu}(\tilde{s})+\frac{\Gamma(2\mu)}{\Gamma(\frac{1}{2}+\mu+i\sigma_a)\Gamma(\frac{1}{2}+\mu+i\sigma_b)}G_{\underline{\sigma},-\mu}(\tilde{s})\right].
\end{equation}

Now let $2\mu\in \N_0$, or equivalently, $c\in \N_1$. Then we apply \cite{NIST:DLMF}[\S15.8(ii)] to expand:
\begin{multline*}
	{}_2F_1(a,b,a+b+1-c;1-\zeta)=(1-\delta_{\mu0})\zeta^{-2\mu}\frac{\Gamma(a+b+1-c)}{\Gamma(a)\Gamma(b)}((2\mu-1)!+O_{\infty}(\zeta))\\
	-(-1)^{2\mu}\frac{\Gamma(a+b+1-c)}{(2\mu)!\Gamma(a-2\mu)\Gamma(b-2\mu)}\left(\gamma_{\rm Euler}-\frac{\Gamma'(2\mu+1)}{\Gamma(2\mu+1)}+\frac{\Gamma'(a)}{\Gamma(a)}+\frac{\Gamma'(b)}{\Gamma(b)}\right)\\
	-(-1)^{2\mu}\frac{\Gamma(a+b+1-c)}{\Gamma(a-c+1)\Gamma(b-c+1)}\log \zeta \sum_{k=0}^{\infty}\frac{\Gamma(a+k)\Gamma(b+k)}{\Gamma(a)\Gamma(b)k!(2\mu+k)!}\zeta^k\\
	=-\frac{(-1)^{-2\mu}}{(2\mu)!}\log \zeta F_1(a,b,c;\zeta)+(1-\delta_{\mu0})\zeta^{-2\mu}\frac{\Gamma(c-1)\Gamma(a+b-c+1)}{\Gamma(a)\Gamma(b)}(1+O_{\infty}(\zeta))\\
	-(-1)^{2\mu}\frac{\Gamma(a+b+1-c)}{(2\mu)!\Gamma(a-2\mu)\Gamma(b-2\mu)}\left(\gamma_{\rm Euler}-\frac{\Gamma'(2\mu+1)}{\Gamma(2\mu+1)}+\frac{\Gamma'(a)}{\Gamma(a)}+\frac{\Gamma'(b)}{\Gamma(b)}\right)\\
	=\begin{cases}
		-\frac{(-1)^{-2\mu}\Gamma(a+b)}{(2\mu)!\Gamma(a)\Gamma(b)}\log \zeta F_1(a,b,c;\zeta)+\zeta^{-2\mu}\frac{\Gamma(c-1)\Gamma(a+b-c+1)}{\Gamma(a)\Gamma(b)}(1+O_{\infty}(\zeta)) \quad (2\mu\in \N_1),\\
		-\frac{\Gamma(a+b)}{\Gamma(a)\Gamma(b)}\log \zeta F_1(a,b,c;\zeta)-\frac{\Gamma(a+b)}{\Gamma(a)\Gamma(b)}\left(2\gamma_{\rm Euler}+\frac{\Gamma'(a)}{\Gamma(a)}+\frac{\Gamma'(b)}{\Gamma(b)}\right)+O_{\infty}(\zeta)\quad (2\mu=0).
	\end{cases}
\end{multline*}
We conclude that for $2\mu\in \N_1$:
\begin{multline*}
	F_{\underline{\sigma},\mu}(\tilde{s})=\Gamma\left(1+i(\sigma_a+\sigma_b)\right)\Bigg[-\frac{(-1)^{-2\mu}\Gamma(a+b)}{(2\mu)!\Gamma(a)\Gamma(b)}\log \zeta(\tilde{s})\\
	+\frac{\Gamma(-2\mu)}{\Gamma(\frac{1}{2}-\mu+i\sigma_a)\Gamma(\frac{1}{2}-\mu+i\sigma_b)}(G_{\underline{\sigma},\mu}(\tilde{s})+\log \zeta(\tilde{s}) O_{\infty}(\zeta(\tilde{s}))\\
	+\frac{\Gamma(2\mu)}{\Gamma(\frac{1}{2}+\mu+i\sigma_a)\Gamma(\frac{1}{2}+\mu+i\sigma_b)}G_{\underline{\sigma},-\mu}(\tilde{s})\Bigg].
\end{multline*}
We conclude that, as in the $2\mu\notin \N_0$ case:
\begin{equation*}
	\mathcal{W}(F_{\underline{\sigma},\mu},G_{\underline{\sigma},\mu})= \frac{\Gamma(1+2\mu) \Gamma(1+i(\sigma_a+\sigma_b))}{\Gamma(\frac{1}{2}+\mu+i\sigma_a)\Gamma(\frac{1}{2}+\mu+i\sigma_b)}.
\end{equation*}
For $2\mu=0$, we obtain:
\begin{equation*}
	F_{\underline{\sigma},0}(\tilde{s})=-G_{\underline{\sigma},0}(\tilde{s})\left[\frac{\Gamma(a+b)}{\Gamma(a)\Gamma(b)}\log \zeta+\frac{\Gamma(a+b)}{\Gamma(a)\Gamma(b)}\left(2\gamma_{\rm Euler}+\frac{\Gamma'(a)}{\Gamma(a)}+\frac{\Gamma'(b)}{\Gamma(b)}\right)+O_{\infty}(\zeta)\right].
\end{equation*}

In this case, we therefore have that:
\begin{multline*}
\frac{d\zeta}{d\tilde{s}}(\zeta)\mathcal{W}(F_{\underline{\sigma},0}(\tilde{s}(\zeta)),G_{\underline{\sigma},0}(\tilde{s}(\zeta))=\left[(1-\zeta)^2{F}_{\underline{\sigma},0}\frac{dG_{\underline{\sigma},0}}{d\zeta}-(1-\zeta)^2\frac{d{F}_{\underline{\sigma},0}}{d\zeta}G_{\underline{\sigma},0}\right]|_{\zeta=0}\\
=\frac{\Gamma(a+b)}{\Gamma(a)\Gamma(b)}(1-\zeta)^2(\log \zeta G_{\underline{\sigma},\mu})G'_{\underline{\sigma},\mu}-\frac{\Gamma(a+b)}{\Gamma(a)\Gamma(b)}(1-\zeta)^2(\log \zeta G_{\underline{\sigma},\mu})'G_{\underline{\sigma},\mu}+(1-\zeta)^2O_{\infty}(\zeta)\Bigg|_{\zeta=0}\\
=-\frac{\Gamma(a+b)}{\Gamma(a)\Gamma(b)}\lim_{\zeta\downarrow 0}\zeta^{-1}(1-\zeta)^2G_{\underline{\sigma},\mu}^2(\zeta)=-\frac{\Gamma(a+b)}{\Gamma(a)\Gamma(b)}.
\end{multline*}
Hence, $\mathcal{W}(F_{\underline{\sigma},\mu},G_{\underline{\sigma},\mu})=-\frac{\Gamma(a+b)}{\Gamma(a)\Gamma(b)}$.

By \cite{NIST:DLMF}[\S15.10(ii)], we have that for $c\neq a+b$:
 \begin{multline}
 \label{eq:hypgeomconnectionformb}
\mbox{}_2F_1(a,b,c;\zeta)=\frac{\Gamma(c)\Gamma(c-a-b)}{\Gamma(c-a)\Gamma(c-b)}\mbox{}_2F_1(a,b,a+b+1-c;1-\zeta)\\
+\frac{\Gamma(c)\Gamma(a+b-c)}{\Gamma(a)\Gamma(b)}(1-\zeta)^{c-a-b}\mbox{}_2F_1(c-a,c-b,c-a-b+1;1-\zeta).
\end{multline}

We can view $(a,b,c)$ as functions of $(\mu,\underline{\sigma})$. Under the transformation $(\mu,\underline{\sigma})\mapsto (\mu, -\underline{\sigma})$, we have that $(a,b,c)\mapsto (c-a,c-b,c)$, so we can express:
 \begin{align*}
 	F_{\underline{\sigma},\mu}(\tilde{s}):=&\:\zeta^{\frac{c}{2}}(\tilde{s})(1-\zeta(\tilde{s}))^{\frac{1}{2}(a+b-c-1)} {}_2F_1(a,b,a+b+1-c;1-\zeta(\tilde{s})),\\
 	F_{-\underline{\sigma},\mu}(\tilde{s}):=&\:\zeta^{\frac{c}{2}}(\tilde{s})(1-\zeta(\tilde{s}))^{-\frac{1}{2}(a+b-c+1)} {}_2F_1(c-a,c-b,c-a-b+1;1-\zeta(\tilde{s})).
 \end{align*}
We conclude that \eqref{eq:relhypgeomfunct0} holds.

By \cite{NIST:DLMF}[\S15.8(ii)], we moreover have that in the case $c=a+b$, or $w=0$:
 \begin{equation*}
 	\mbox{}_2F_1(a,b,a+b;\zeta)=-\frac{\Gamma(a+b)}{\Gamma(a)\Gamma(b)}\left(\log(1-\zeta)
 	+2\gamma_{\rm Euler}+\frac{\Gamma'(a)}{\Gamma(a)}+\frac{\Gamma'(b)}{\Gamma(b)}+O_{\infty}((1-\zeta))\right),
 \end{equation*}
 from which \eqref{eq:relhypgeomfunct0b} follows.
\end{proof}

 \begin{proposition}
 	\label{prop:hypgeomerror}
Assume that
\begin{equation*}
\int_{\tilde{s}}^{\infty}\eta |\vartheta(\eta)|\,d\eta<\infty	
\end{equation*}
and let $U$ be a solution to:
\begin{equation}
\label{eq:hypgeomalterr}
	\frac{d^2U}{d\tilde{s}^2}=(1+\tilde{s})^{-2}\left[-\frac{1}{4}(1+w^2)+v \tilde{s}^{-1}+\left(\mu^2-\frac{1}{4}\right)\tilde{s}^{-2}\right]U+\vartheta U.
\end{equation}
Let $\zeta(\tilde{s})=\frac{\tilde{s}}{1+\tilde{s}}$ and $f(\tilde{s}):=(1+|\log(1-\zeta(\tilde{s}))^{-1}|)\zeta^{-\frac{1}{2}}(\tilde{s})(1-\zeta(\tilde{s}))^{\frac{1}{2}}=\tilde{s}^{-\frac{1}{2}}(1+\log(\tilde{s}+1))$. Then there exists constants $A_1,A_2\in \C$ and a positive constant $C>0$, such that we can write:
\begin{equation*}
	U=A_1(F_{\underline{\sigma},\mu}+\varepsilon_1)+A_2(G_{\underline{\sigma},\mu}+\varepsilon_2),
\end{equation*}
with
\begin{align*}
\left|\varepsilon_1\right|(\tilde{s})\leq&\: C|F_{\underline{\sigma},\mu}|(\tilde{s})(1+\log(\tilde{s}+1))\left(e^{C \int_{\tilde{s}}^{\infty}\eta \log^2(2+\eta)|\vartheta(\eta)|\,d\eta}-1\right)\quad (\mu\neq 0),\\
\left|\varepsilon_1\right|(\tilde{s})\leq&\: C|F_{\underline{\sigma},\mu}|(\tilde{s})(1+\log(\tilde{s}+1))\left(e^{C\log(1+\zeta^{-1}(\tilde{s})) \int_{\tilde{s}}^{\infty}\eta \log^2(2+\eta)|\vartheta(\eta)|\,d\eta}-1\right)\quad (\mu= 0),\\
\left|\varepsilon_2\right|(\tilde{s})\leq&\: C|F_{\underline{\sigma},\mu}|(\tilde{s})(1+\log(\tilde{s}+1))\left(e^{C \int_{\tilde{s}}^{\infty}\eta \log^2(2+\eta)|\vartheta(\eta)|\,d\eta}-1\right)\quad (\mu\neq 0),\\
\left|\varepsilon_2\right|(\tilde{s})\leq&\: C\frac{|F_{\underline{\sigma},\mu}|(\tilde{s})(1+\log(\tilde{s}+1))}{\log(1+\zeta^{-1}(\tilde{s}))}\left(e^{C\log(1+\zeta^{-1}(\tilde{s})) \int_{\tilde{s}}^{\infty}\eta \log^2(2+\eta)|\vartheta(\eta)|\,d\eta}-1\right)\quad (\mu= 0),\\
\left|\frac{d(f \varepsilon_1)}{d\tilde{s}}\right|(\tilde{s})\leq&\: C\tilde{s}^{-1}f(\tilde{s}) (1-\zeta(\tilde{s}))|F_{\underline{\sigma},\mu}|(\tilde{s})\left(e^{C \int_{\tilde{s}}^{\infty}\eta \log^2(2+\eta)|\vartheta(\eta)|\,d\eta}-1\right)\quad (\mu\neq 0),\\
\left|\frac{d(f \varepsilon_1)}{d\tilde{s}}\right|(\tilde{s})\leq&\: C\tilde{s}^{-1}f(\tilde{s}) (1-\zeta(\tilde{s}))|F_{\underline{\sigma},\mu}|(\tilde{\tilde{s}})\left(e^{C\log(1+\zeta^{-1}(\tilde{s})) \int_{\tilde{s}}^{\infty}\eta \log^2(2+\eta)|\vartheta(\eta)|\,d\eta}-1\right)\quad (\mu= 0),\\
\left|\frac{d(f \varepsilon_2)}{d\tilde{s}}\right|(\tilde{s})\leq&\: C\tilde{s}^{-1}f(\tilde{s}) (1-\zeta(\tilde{s}))|F_{\underline{\sigma},\mu}|(\tilde{s})\left(e^{C \int_{\tilde{s}}^{\infty}\eta \log^2(2+\eta)|\vartheta(\eta)|\,d\eta}-1\right)\quad (\mu\neq 0),\\
\left|\frac{d(f \varepsilon_2)}{d\tilde{s}}\right|(\tilde{s})\leq&\: C\tilde{s}^{-1}f(\tilde{s}) (1-\zeta(\tilde{s}))\frac{|F_{\underline{\sigma},\mu}|(\tilde{s})}{\log(1+\zeta^{-1}(\tilde{s}))}\left(e^{C\log(1+\zeta^{-1}(\tilde{s})) \int_{\tilde{s}}^{\infty}\eta \log^2(2+\eta)|\vartheta(\eta)|\,d\eta}-1\right)\quad (\mu= 0).
\end{align*}
 \end{proposition}
 \begin{proof}
We need to check that the assumptions in Proposition \ref{prop:genoderrorest} are satisfied. Let $U_1(\tilde{s}):=F_{\underline{\sigma},\mu}(\tilde{s})$ and $U_2(\tilde{s}):=F_{\underline{\sigma},\mu}(\tilde{s})$.

By definition of hypergeometric functions, we have that for any triple of complex constants $(\tilde{a},\tilde{b},\tilde{c})$, ${}_2F_1(\tilde{a},\tilde{b},\tilde{c};\zeta)=1+O_{\infty}(\zeta)$.

Hence,
\begin{align*}
|U_1|(\tilde{s})\leq&\:(1-\zeta(\tilde{s}))^{-\frac{1}{2}}(1+O_1(1-\zeta(\tilde{s}))),\\
|U_2|(\tilde{s})\leq &\:(1-\zeta(\tilde{s}))^{-\frac{1}{2}}(1+|\log(1-\zeta(\tilde{s}))|((1+O_1(1-\zeta(\tilde{s}))),\\
|U_1|(\tilde{s})\leq&\:\zeta^{\frac{1}{2}-{\re \mu}}(\tilde{s})(1+O_1(\zeta(\tilde{s}))+O_1(\zeta^{2{\re \mu}}(\tilde{s})))\quad (\mu\neq 0),\\
|U_1|(\tilde{s})\leq&\:\zeta^{\frac{1}{2}}(\tilde{s})\log\zeta(\tilde{s})(1+O_1(\zeta(\tilde{s})))\quad (\mu= 0),\\
|U_2|(\tilde{s})\leq&\:\zeta^{\frac{1}{2}+{\re \mu}}(\tilde{s})(1+O_1(\zeta(\tilde{s}))).
\end{align*}

We define the functions $P_i,Q: \R\to (0,\infty)$, with $i\in \{0,1,2\}$, as follows:
\begin{align*}
Q(\tilde{s}):=&\:\zeta(\tilde{s})^{\frac{1}{2}{\re \mu}},\\
P_0(\tilde{s}):=&\: \zeta(\tilde{s})^{-\frac{1}{2}{\re \mu}}\quad (\mu\neq 0),\\
P_0(\tilde{s}):=&\:\log(1+\zeta^{-1}(\tilde{s}))\quad (\mu=0),\\
P_1(\tilde{s}):=&\: \zeta^{-1}(\tilde{s})\frac{d\zeta}{d\tilde{s}}(\tilde{s})P_0(\tilde{s})=2 \zeta^{-1}(\tilde{s})(1-\zeta(\tilde{s}))^2P_0(\tilde{s})=\tilde{s}^{-1} (1-\zeta(\tilde{s}))P_0(\tilde{s}),\\
P_2(\tilde{s}):=&\: \zeta^{-1}(\tilde{s})(1-\zeta(\tilde{s}))^{-1}\frac{d\zeta}{d\tilde{s}}(\tilde{s})P_0(\tilde{s})= \zeta^{-1}(\tilde{s})(1-\zeta(\tilde{s}))P_0(\tilde{s})=\tilde{s}^{-1} P_0(\tilde{s}).
\end{align*}

Then
\begin{align*}
|(fU_1)(s)(fU_2)(\eta)|+|(fU_1)(\eta)(fU_2)(s)|\leq &\: C P_0(s) Q(\eta)\quad \eta\in (s,\infty),\\
\left|(fU_1)'(s)(fU_2)(\eta)\right|+\left|(fU_1)(\eta)(fU_2)'(s)\right|\leq &\: C P_1(s)Q(\eta)\quad \eta\in (s,\infty),\\
\left|(f'U_1)(s)(f'U_2)(\eta)\right|+\left|(fU_1)(\eta)(f'U_2)(s)\right|\leq &\: C P_2(s)Q(\eta)\quad \eta\in (s,\infty),\\
\sup_{ \eta \in (\eta_0,\infty)}[|fU_1|(\eta)+|fU_2|(\eta)]Q(\eta)<&\: \infty \quad \textnormal{for any $\eta_0>0$},\\
\sup_{\eta \in (\tilde{s},\infty)} P_0(\eta)Q(\eta)<&\:\infty\quad \textnormal{for any $\tilde{s}>0$},
\end{align*}

Note moreover that there exists a constant $C>0$, such that:
\begin{align*}
\sup_{\eta\in (\tilde{s},\infty)} P_0(\eta)Q(\eta)= &\:1\quad (\mu\neq 0),\\
\sup_{\eta\in (\tilde{s},\infty)} P_0(\eta)Q(\eta)= &\:\log(1+\zeta^{-1}(\tilde{s}))\quad (\mu= 0),\\
\sup_{\eta\in (\tilde{s},\infty)} |fU_2(\eta)|Q(\eta)\leq &\: C,\\
\sup_{\eta\in (\tilde{s},\infty)} |fU_1(\eta)|Q(\eta)\leq &\: C\quad (\mu\neq 0),\\
\sup_{\eta\in (\tilde{s},\infty)}  |fU_1(\eta)|Q(\eta)\leq &\: C\log(1+\zeta^{-1}(\tilde{s}))\quad (\mu=0).
\end{align*}

We can now apply Proposition \ref{prop:genoderrorest} to conclude the estimates for $\varepsilon_i$.
  \end{proof}
  
  \begin{corollary}
  \label{cor:esthypgeomG}
  Let $\delta>0$ and $0\leq \tilde{s}\leq e^{ \delta|w|^{-1}}$. Then there exists a uniform constant $C>0$ (independent of $w$), such that:
  	\begin{align*}
  		|G_{\underline{\sigma},\mu}|(\tilde{s})\leq &\:C (1-\zeta(\tilde{s}))^{-\frac{1}{2}}\log(1-\zeta(\tilde{s}))\leq C \delta |w|^{-1}(1-\zeta(\tilde{s}))^{-\frac{1}{2}},\\
  		\left|\frac{dG_{\underline{\sigma},\mu}}{d\tilde{s}}\right|(\tilde{s})\leq &\:C \tilde{s}^{-1}(1-\zeta(\tilde{s}))^{-\frac{1}{2}}\log(1-\zeta(\tilde{s}))\leq C \delta \tilde{s}^{-1}|w|^{-1}(1-\zeta(\tilde{s}))^{-\frac{1}{2}}.
  	\end{align*}
  \end{corollary}
  \begin{proof}
  We apply Proposition \ref{prop:hypgeomerror} with $w=0$ to \eqref{eq:hypgeomalt} by taking $\vartheta(\tilde{s})=-\frac{w^2}{4}(1+\tilde{s})^{-2}$.
 	
  \end{proof}

 \begin{lemma}
 \label{lm:coshratios}
 Let $c>0$ and consider $f_c:\R\to (0,\infty)$, with $f_{c}(x):=\frac{\cosh(x+c)}{\cosh(x-c)}=\frac{\cosh(-x-c)}{\cosh(c-x)}$. Then
 \begin{align*}
 	f_c'>&\:0,\\
 	f_c(0)=&\:1,\\
 	\lim_{x\to \infty}f_c(x)=&\:e^{2c},\\
 	\lim_{x\to -\infty}f_c(x)=&\:e^{-2c}.
 \end{align*}
 \end{lemma}
 \begin{proof}
 	The limiting properties of $f_c$ follow immediately. To conclude that $f_c'>0$, we consider $\log f_c(x)$ and compute:
 	\begin{equation*}
 		(\log f_c)'(x)=(\log(\cosh((\cdot)+a))-\log(\cosh((\cdot)-a))'(x)=\tanh(x+a)-\tanh(x-a).
 	\end{equation*}
 	We conclude that the right-hand side above is strictly positive by using that $(\tanh)'=\frac{1}{\cosh^2}>0$.
 \end{proof}

\bibliographystyle{alpha}


	\end{document}